\documentclass[twoside,11pt,openright]{report}

\usepackage[american]{babel}
\usepackage{latexsym}
\usepackage{amssymb}
\usepackage{amsmath}
\usepackage{amsbsy}
\usepackage{graphics}
\usepackage{color,graphicx}
\usepackage{cite} 
\usepackage{bbm}    
\usepackage{framed,xspace}
\usepackage{url}
\usepackage{qip}
\usepackage[active]{srcltx}
\usepackage{color,graphicx}
\usepackage{tikz}
\usepackage{makeidx}
\usepackage{vmargin}
\usepackage{afterpage}
\usepackage{fancyheadings}
\usepackage{rotating}
\usepackage{makeidx}

\newenvironment{proof}[1][] {\noindent{\bf Proof#1.}\hspace{0.75em}}
	       {\hspace{\fill}$\blacksquare$\vspace{0.3cm}}
\newtheorem{definition}{Definition}[chapter]
\newtheorem{proposition}{Proposition}[chapter]
\newtheorem{theorem}{Theorem}[chapter]
\newtheorem{lemma}{Lemma}[chapter]
\newtheorem{corollary}{Corollary}[chapter]
\newtheorem{claim}{Claim}[chapter]

\newcommand{\clearemptydoublepage}{
  \newpage{\pagestyle{empty}\cleardoublepage}}

\newcommand{\Index}[1]{\index{#1}#1}

\newcommand{\Pl}{{\sf P}}
\newcommand{\dP}{{\sf P}'}
\newcommand{\hP}{\hat{\sf P}}
\newcommand{\dhP}{\hat{\sf P}'}

\newcommand{\A}{{\sf A}}
\newcommand{\dA}{{\sf A}'}
\newcommand{\hA}{\hat{\sf A}}
\newcommand{\dhA}{\hat{\sf A}'}

\newcommand{\B}{{\sf B}}
\newcommand{\dB}{{\sf B}'}
 
\newcommand{\hB}{\hat{\sf B}}
\newcommand{\dhB}{\hat{\sf B}'}

\newcommand{\E}{{\sf E}}
\newcommand{\dhE}{\hat{\E}}

\newcommand{\dAlice}{\mathfrak{A}}
\newcommand{\dBob}{\mathfrak{B}}
\newcommand{\dPlayer}{\mathfrak{P}}

\newcommand{\dBobBQSM}{\dBob_{\text{\sc bqsm}}}
\newcommand{\dBobPoly}{\dBob_{\mathrm{poly}}}
\newcommand{\dAlicePoly}{\dAlice_{\mathrm{poly}}}
\newcommand{\dPlayerPoly}{\dPlayer_{\mathrm{poly}}}
\newcommand{\poly}{\mathit{poly}}

\newcommand{\negl}[1]{{\sc{negl}}\,(#1)} 
\newcommand{\eps}{\varepsilon}

\newcommand{\x}{\times} 
\newcommand{\+}{+}	

\newcommand{\hmin}[1]{\ensuremath{H_{\infty}(#1)}}
\newcommand{\hzero}[1]{\ensuremath{H_0(#1)}}
\newcommand{\Hmin}{H_{\infty}}
\newcommand{\Hmax}{H_0}
  
\newcommand{\approxp}{\stackrel{\text{\it\tiny p}}{\approx}}    
\newcommand{\approxs}{\stackrel{\text{\it\tiny s}}{\approx}}    
\newcommand{\approxc}{\stackrel{\text{\it\tiny c}}{\approx}}    
\newcommand{\approxq}{\stackrel{\text{\it\tiny q}}{\approx}}    

\newcommand*{\assign}
	    {\ensuremath{\kern.5ex\raisebox{.1ex}{\mbox{\rm:}}\kern
	   -.3em =}}
\newcommand{\eqq}{\stackrel{\raisebox{-1ex}{\tiny?}}{=}}

\newcommand{\set}[1]{\{#1\}}
\newcommand{\Set}[2]{\{ #1 : #2\}}
\newcommand{\zo}{\{0,1\}}

\newcommand{\Hil}{\mathcal{H}}
\newcommand{\op}[1]{\mathtt{#1}}
\newcommand{\M}{\mathcal{M}}
\newcommand{\X}{\mathcal{X}}

\newcommand{\test}{t\hspace{-0.15ex}e\hspace{-0.1ex}st}
\newcommand{\Test}{T\hspace{-0.3ex}e\hspace{-0.15ex}st}
\newcommand{\err}{e\hspace{-0.1ex}r\hspace{-0.15ex}r}

\newcommand{\F}{\mathcal{F}}
\newcommand{\G}{\mathcal{G}} 

\newcommand{\ot}[2]{\mathtt{#1\text{--}#2 \ OT}\,}

\newcommand{\MC}[3]{#1 \leftrightarrow #2 \leftrightarrow #3}

\newcommand{\pro}[1]{\operatorname{Pr}_{\mathtt{#1}}}
\newcommand{\prob}[1]{\operatorname{Pr \,}[#1]}

\newcommand{\FF}{\mathbb{F}}
\newcommand{\NN}{\mathbb{N}}


\newcommand{\emphbf}[1]{\emph{\textbf{#1}}}

\newcommand{\CO}{\mathtt{Commit\&Open}}

\newcommand{\ignore}[1]{}

\newcommand{\mycitation}[2]{
  {\begin{flushright}
      \emph{#1} \\ 
      \textrm{--- #2}
    \end{flushright}}}

\newcounter{itm}
\newenvironment{myprotocol}[1]
  {\begin{minipage}{\columnwidth} 
    \begin{framed}\hspace{0ex} 
     \begin{minipage}{0.99\columnwidth}
       {\bf #1:}
       \setcounter{itm}{1}
       \begin{list}{\arabic{itm}.}{\usecounter{itm}}}
   {    \end{list}
       \vspace{-1.5ex} 
       \end{minipage} 
     \end{framed} 
    \end{minipage}\vspace{-0.6ex}}

\newcommand{\idF}{\mathcal{F_{\sf{ID}}}}
\newcommand{\QID}{\mathtt{QID}}
\newcommand{\QIDplus}{\compile(\QID^+)}
\newcommand{\IDplus}{{\sf \textsl{$ID^+$}}\xspace} 

\newcommand{\otF}{\mathcal{F_{\sf{OT}}}}
\newcommand{\QOT}{{\tt 1\text{-}2 \, QOT^\ell}}

\newcommand{\compile}{{\cal C}^{\alpha}}

\newcommand{\commitx}[2]{\mathtt{commit}\,(#1,#2)\,}
\newcommand{\commitk}[3]{\mathtt{commit}\,_{#3}\,(#1,#2)\,}
\newcommand{\commit}{\mathtt{commit}}
\newcommand{\Commit}{\mathtt{COMMIT}}
\newcommand{\Commitk}[3]{\mathtt{COMMIT}\,_{#3}\,\big( #1,#2 \big)\,} 
\newcommand{\open}[2]{\mathtt{open}\,(#1,#2)\,}
\newcommand{\sss}{\mathtt{sss}}

\newcommand{\sk}{{\tt sk}}
\newcommand{\pkH}{{\tt pkH}}
\newcommand{\pkB}{{\tt pkB}}
\newcommand{\GH}{{\cal G}_{\tt H}}
\newcommand{\GB}{{\cal G}_{\tt B}}

\newcommand{\xtr}[2]{{\tt{xtr}}_{#2}(#1)} 
\newcommand{\xtrx}{\tt{xtr}}

\newcommand{\coin}{\mathtt{coin}}
\newcommand{\start}{\mathtt{start}}

\newcommand{\la}{\leftarrow}
\newcommand{\ra}{\rightarrow}

\newcommand{\cPi}[2]{\Pi_{#1,#2}^{\, \COIN}}	  
\newcommand{\cxPi}[3]{\Pi_{#1,#2}^{\, #3-\COIN}}

\newcommand{\cF}{\F_{\COIN}}		
\newcommand{\cxF}[1]{\F_{{#1}-\COIN}}	
\newcommand{\clF}{\F_{\ell-\COIN}}	

\newcommand{\uncont}{\texttt{uncont}}
\newcommand{\force}{\texttt{force}}
\newcommand{\random}{\texttt{random}}
\newcommand{\defterm}[1]{\textsf{#1}}

\newcommand{\D}[1]{\mathcal{D}^{\, #1}}
\newcommand{\U}{\mathcal{U}}

\newcommand{\IQZK}{\mathtt{IQZK}}
\newcommand{\IQZKF}{\mathtt{IQZK^{\cxF{\kappa}}}}
\newcommand{\COIN}{\mathtt{COIN}}
\newcommand{\NIZK}{\mathtt{NIZK}} 
\newcommand{\SNIZK}{\sf{\hat{S}}_{\NIZK}}
\newcommand{\SIQZK}{\sf{\hat{S}}_{\IQZK}}
\newcommand{\SIQZKF}{\sf{\hat{S}}_{\IQZK^{\cxF{\kappa}}}}

\newcommand{\crs}{\omega}

\newcommand{\Forg}{\tilde{\mathcal{F}}} 
\newcommand{\Dist}{\tilde{\mathcal{D}}} 

\newcommand{\abort}{\texttt{abort}}
\newcommand{\success}{\texttt{success}}

\newcommand{\ZKPK}{\mathtt{ZKPK(\Rel)}}

\newcommand{\zkpkF}{\F_{\ZKPK}}

\newcommand{\Rel}{\mathcal{R}}
\newcommand{\NP}{\ensuremath{\mathcal{NP}}}
\newcommand{\Poly}{\ensuremath{\mathcal{P}}}

\setpapersize{A4} 
\setmarginsrb{3.2cm}{2.3cm}{2.40cm}{2.8cm}{1.3cm}{0.7cm}{0.1cm}{1.3cm}
\renewcommand{\chaptermark}[1]%
{\markboth{#1}{}}
\renewcommand{\sectionmark}[1]%
{\markright{\thesection\ #1}}
\newcommand{\inmargin}[1]{
  \marginpar{\flushleft
    \vspace{4.10cm}
    \begin{minipage}{0.70cm}
      \begin{rotate}{90}
        \mbox{\parbox{6cm}{\vspace{0.7cm}\par{}\flushright{}#1}}
      \end{rotate}
    \end{minipage}}}
\makeatletter
\def\@makechapterhead#1{%
  {\parindent \z@ \raggedright \normalfont \ifnum \c@secnumdepth
    >\m@ne \flushright\inmargin{\LARGE\bfseries{\sc{}chapter\space
        \Huge\thechapter}}
    \vspace*{41\p@}%
    \par\nobreak
    \vskip 35\p@ \fi \interlinepenalty\@M
    \flushleft\begin{minipage}{13.5cm}\flushleft\Huge \normalfont
    \bfseries #1\end{minipage}\par\nobreak
    \vspace*{35\p@}}}
\makeatother
\makeindex

\begin{document}


\pagestyle{empty} 
\pagenumbering{roman} 
\setcounter{secnumdepth}{-1}
\vspace*{\fill}
\begin{flushright}  
  {\huge\sf Cryptographic Protocols under Quantum Attacks}\\[3ex]
  {\LARGE\sf Carolin Lunemann} 
\end{flushright}
\noindent\rule{\linewidth}{1mm}\\[-.5ex]
\noindent\rule{\linewidth}{2.5mm}
\vfill
\begin{center}
  {\huge\sf PhD Dissertation}\\[\fill]
  {\sf Department of Computer Science\\Aarhus University\\Denmark}
\end{center}
\vspace*{\fill}

\cleardoublepage

\begin{center}
  \vspace*{\stretch{1}}
  {\huge Cryptographic Protocols under Quantum Attacks}\\[\fill]
  A Dissertation\\
  Presented to the Faculty of Science\\
  of Aarhus University\\
  in Partial Fulfillment of the Requirements for the\\
  PhD Degree\\[\stretch{2}]
  by\\
  Carolin Lunemann\\
  August 31, 2010
\end{center}
\vspace*{\stretch{1}}


\clearemptydoublepage
\pagestyle{empty}
\chapter*{{\Huge Abstract}}
\addcontentsline{toc}{chapter}{Abstract} 

The realm of this thesis is cryptographic protocol theory in the
quantum world. We study the security of quantum and classical
protocols against adversaries that are assumed to exploit quantum
effects to their advantage. Security in the quantum world means that
quantum computation does not jeopardize the assumption, underlying the
protocol construction. But moreover, we encounter additional setbacks
in the security proofs, which are mostly due to the fact that some
well-known classical proof techniques are forbidden by certain
properties of a quantum environment. Interestingly, we can exploit
some of the very same properties to the benefit of quantum
cryptography. Thus, this work lies right at the heart of the conflict
between highly potential effects but likewise rather demanding
conditions in the quantum world.


\clearemptydoublepage
\chapter*{{\Huge Acknowledgments}}
\addcontentsline{toc}{chapter}{Acknowledgments}

Thanks to all the people, I met on the way of my PhD education, for
introducing me to new ways of thinking, inspiring me by fascinating
ideas, and helping me in so many other ways. Thanks to all the people
that stayed with me on this way without questioning these very same
aspects. To mention just a few $\ldots$\\[-2ex]

$\ldots$ special thanks to my supervisors Ivan Damg{\aa}rd---who
taught me so much about cryptography---and Louis Salvail---who taught
me so much about quantum---and often vice versa. Thanks to my
co-authors Christian Schaffner, Jesper Buus Nielsen, and Serge Fehr
for their support throughout my PhD studies and for teaching me many
fascinating details.

$\ldots$ thanks to Prof.\ Claude Cr\'{e}peau from McGill University in
Montr\'eal and Prof.\ Stefan Wolf from ETH Z\"urich as well as
Prof. Susanne B{\o}dker from Aarhus University for agreeing to
constitute the evaluation committee for my PhD thesis.

$\dots$ thanks for proof-reading some parts of the thesis to Chris,
Dorthe, and Jesper as well as for valuable comments to Claude and
Louis.

$\ldots$ thanks to the Crypto-Group with all its members that left,
stayed or newly arrived. This environment is a great place to do
research in (while drinking espresso and eating cake), and many
non-research activities (like launching julefrokost, fridaybaring,
dancing danske folkedanse, and often absurd discussions about
everything under the sun) are unforgettable. Thanks to the always
helpful staff at the department, especially to Dorthe, Ellen, and
Hanne.

$\ldots$ thanks to Aarhus University for enabling me to travel around
the world, and in that respect also to IACR, Universit\'e de
Montr\'eal, INTRIQ, and Princeton University. Furthermore, my studies
were supported by the MOBISEQ research project, which is funded by Det
Strategiske Forskningsr{\aa}ds Programkomite for Nanovidenskab og
-teknologi, Bioteknologi og IT (NABIIT). Many thanks also to Institut
Mittag-Leffler of the Royal Swedish Academy of Sciences for providing
time, grant, and space to think---as well as to prepare the defense
and to revise the thesis. Thanks to Prof.\ Mary Beth Ruskai for
pointing to the quantum information program at the institute.

$\dots$ thanks to Prof.\ Christoph Kreitz from the University of
Potsdam for sparking my interest in crypto in the first place.

$\ldots$ and thanks to my family and many friends for their support, often
from the distance. Grazie a Claudio, with whom I went through all
phases of ever-changing three years of PhD studies. Vielen Dank an
MaPa, Steffi, und Dande, whose support and (blind) acceptance was most
helpful during my education. And last but not least, mange tak til
Jesper.

\vspace{0.5ex}
\begin{flushright}
  \emph{Carolin Lunemann}\\
  \emph{{\AA}rhus, August 31, 2010}
\end{flushright}


\clearemptydoublepage
\tableofcontents


\clearemptydoublepage
\pagenumbering{arabic}
\setcounter{secnumdepth}{3}
\setcounter{tocdepth}{2}
\pagestyle{fancy}
\pagenumbering{arabic}


\clearemptydoublepage
\chapter{Introduction}
\label{chapter:intro}


\section{On Cryptography}
\label{section:crypto}

\mycitation
{\\The multiple human needs and desires that demand privacy among two or
more people \\ in the midst of social life must inevitably lead to
cryptology wherever men thrive \\ and wherever they write.}
{David Kahn}

Cryptography is the art of \emph{secret writing} (from Greek $\kappa
\rho \upsilon \pi \tau o \varsigma$ and $\gamma \rho \alpha \varphi
\omega$) and may be considered almost as old as writing
itself. Cryptography played a crucial role throughout the history of
any society that depended on information, from the Greek Scytale and
the Roman Caesar cipher, over the Vigen\`ere cipher, electromechanical
rotor machines and encryption standards, to forming the backbone of
electronic infrastructures in modern life (see e.g.~\cite{Singh00}
for a historic survey of cryptography).

The first cryptographic methods are known as \emph{secret-key
cryptography}, based on one secret key shared between the
communicating parties and used both for encryption and
decryption. Already apparent from this description derives its main
problem, which lies in the logistics of distributing the key securely:
Prior to any secret communication, the involved parties must be in
possession of the same secret key. Nevertheless, secret-key
cryptography was in use for thousands of years, adjusting its
complexity to ever-increasing developments in technique and
technology.

\emph{Public-key cryptography} was the technological revolution,
solving the key distribution problem. The idea was independently
discovered by Diffie and Hellman in~\cite{DH76} with Rivest, Shamir,
and Adleman, providing the first implementation~\cite{RSA78}, and
slightly earlier but in secrecy, by Ellis, followed by Cocks' and
Williamson's practical application (e.g.~\cite{Cocks08}). Public-key
cryptography is based on a pair of keys for each communicating party,
namely a public key for encryption and a corresponding secret key for
decryption, where it must hold that it is computationally infeasible
(in polynomial time) of deriving the secret key from the public
one. Then, we require a family of trapdoor one-way functions defining
the encryption and decryption procedure. Informally, that means that
encryption is a one-way operation, which is efficiently computable,
given the public key, whereas the decryption function is hard to
compute, unless the trapdoor is known, i.e.\ the secret key. Thus, the
public key can be published without compromising security, and hence,
public-key cryptography does not suffer from key distribution
problems. Due to that and to the fact that the technique additionally
allows for digital signatures that are verifiable with the public key
and yet unforgeable without the secret key, the concept of public-key
cryptography is highly used and required in the age of the Internet
and the proliferation of electronic communication systems.

New potential in cryptography emerged with \emph{quantum
cryptography}, starting with Wiesner's groundbreaking
paper~\cite{Wiesner83}\footnote{The paper was written in the early
1970ies but rejected and only published retroactively in 1983.},
suggesting that ``quantum mechanics allows us novel forms of coding
without analogue [in classical physics]'' (p.~78). His approach of
conjugate coding did not only lay the foundations of the new
cryptographic technique but also suggested a system for sending ``two
mutually exclusive messages'' (p.~83), which is today known as the
powerful primitive of oblivious transfer. It took several years (and
the Caribbean sea) to establish quantum cryptography as a scientific
discipline, accomplished by Bennett and Brassard, mainly by the
BB84-protocol for quantum key distribution (QKD) in~\cite{BB84} after
preceding work such as~\cite{BBBW82, BB83}, culminating in the first
practical realizations~\cite{BB89, BBBSS92}. An alternative QKD scheme
was independently proposed by Ekert in~\cite{Ekert91}, based on a
different approach using quantum entanglement. Since then, QKD was
further researched, both on a theoretical and an experimental level.
Today, conjugate BB84-coding also forms the basis for various more
general quantum cryptographic tasks other than key distribution.

Modern cryptography concerns, besides the secrecy and authenticity of
communication, also the security of various other task. For instance,
theoretical research in the sub-field of \emph{cryptographic protocol
theory} covers cryptographic primitives with fundamental properties for
secure multi-party computation. Each primitive can be seen as a
building block that implements some basic functionality. Composition
of such primitives within outer protocols yield applications that
implement a specific task securely over a distance.


\section{On the Quantum World}
\label{sec:quantum}

\mycitation 
{\\Anyone who is not shocked by quantum theory has not
understood it.}{Niels Bohr}

In the quantum world, we consider the behavior of systems at the
smallest scale, which cannot be explained nor described by classical
physics. A \emph{quantum}\footnote{quantus (Latin) - how much} is the
smallest unit of a physical entity, and the fundamental concept in
\Index{quantum information theory} is a quantum bit, or short,
a\index{qubit} qubit. Quantum information theory was established at
the beginning of the last century, but has been subject to different
interpretations ever since---both scientific and philosophical. This
thesis is divided into two subareas of quantum information theory,
constituting the following two main parts,
Part~\ref{part:quantum.cryptography} and
Part~\ref{part:cryptography.in.quantum.world}
(Part~\ref{part:preliminaries} is dedicated to preliminaries).

Part~\ref{part:quantum.cryptography} is in the realm of
quantum cryptography, where---informally speaking---the transmission
of qubits followed by some classical post-processing is employed to
accomplish a cryptographic task. The security is mainly derived by
the special properties of the qubits during and after transmission, and
therewith, directly from physical laws.

Part~\ref{part:cryptography.in.quantum.world} on cryptography in a
quantum world refers to the study of cryptography with completely
classical messages exchange, but where the environment around is
quantum. In other words, the security of the classical schemes must
withstand powerful quantum computing capabilities.

We now present---in brief and on a (counter-)intuitive level---the
aspects unique to the quantum world, which are relevant in the context
of this work. Interestingly, these quantum features can be exploited
to the benefit of quantum cryptography. However, the very same
properties impose intriguingly new challenges in classical
cryptography. In other words, ``what quantum mechanics takes away with
one hand, it gives back with the other''~\cite[p.~582]{NC00}. And so,
this work lies right at the heart of the conflict between highly
potential effects, but likewise rather demanding conditions.

  \paragraph{\sc Information gain vs.\ disturbance.}

  This aspect might be argued to constitute the most outstanding
  advantage of quantum cryptography over the classical world, and
  forms ``the engine that powers quantum
  cryptography''~\cite[p.~1]{Fuchs96}. In the classical case, a bit can
  simply be read in transmission, and the information gain solely
  depends on the security of the respective encryption used. In
  quantum cryptography, information is typically encoded in two
  complementary types of photon polarization or, in other words, a
  qubit is prepared in one out of two conjugate bases with orthogonal
  basis states. To gain information about such an unknown qubit, it
  must be observed, but observing in the quantum world means
  measuring. Measuring, or more precisely distinguishing between two
  non-orthogonal quantum states, is destructive and therewith any
  measurement disturbes the system. This is explained in the
  \emph{Heisenberg uncertainty principle}, which states that certain
  pairs of quantum properties are complementary in that measuring one
  of them necessarily disturbs the other.
  
  Consequently, eavesdropping on a qubit transmission disturbs the
  system, and can therefore be noticed in a statistically detectable
  way. Moreover, the quantitative trade-off between information gain
  and disturbance is useful not only against an external adversary,
  but it is also a main ingredient when proving security against a
  dishonest player. This fact is inherent in the basic security
  aspects of all our quantum two-party protocols, discussed later in
  Part~\ref{part:quantum.cryptography}.

  \paragraph{\sc An unknown quantum state cannot be copied.}

  This fact---unheard of in the case of classical data---is formalized
  in the \emph{no-cloning theorem}\index{no-cloning
  theorem}~\cite{WZ82}. The peculiar property constitutes another
  major security feature in quantum communications and underlies all
  our quantum protocols in
  Part~\ref{part:quantum.cryptography}. However, it also sets severe
  restriction in the theory of quantum computing. This becomes
  apparent in Part~\ref{part:cryptography.in.quantum.world}, where the
  commonly used classical proof technique \emph{rewinding}, which is
  also shortly discussed below, requires to copy certain data, and so
  has to be carefully reviewed in the quantum world.

  \paragraph{\sc Quantum memory is limited.}

  \index{quantum storage} A more practical issue concerns the
  limitation of the amount of qubits that can be stored and then
  retrieved undisturbed. This may be seen as a snapshot of current
  state of the art. However, ongoing research strongly suggest that it
  is---and will be---much easier to transmit and measure qubits than
  it is to store them for a non-negligible time.

  We will make use of this given fact in our quantum protocols in
  Chapter~\ref{chap:hybrid.security.applications}, which are designed
  such that dishonest parties would need large quantum memory to
  attack successfully---a security property that classical protocols
  cannot achieve. Yet, we do not exclusively rely on this condition
  only, but investigate a wider diversification of security that is
  not threatened by potential breakthroughs in developing quantum
  storage.

  \paragraph{\sc Quantum rewinding is tricky.}
  
  \index{rewinding!quantum} As already indicated, this statement is a
  key aspect in Part~\ref{part:cryptography.in.quantum.world}, and
  originates from most of the above mentioned properties ``all wrapped
  up together''. Rewinding is a very powerful technique in
  simulation-based proofs against a classical dishonest party: We can
  prove security against a cheating player by showing that a run of a
  protocol between him and the honest player can be efficiently
  simulated without interacting with the honest player, but with a
  simulator instead. A simulator is a machine which does not know the
  secrets of the honest party but yet it sends messages like the
  honest player would do but with more freedom, e.g.\ in how and when
  to generate these. Then to conclude the proof, we have to show that
  the running time of the simulation as well as the distribution of
  the conversation are according to expectations. A simulator
  basically prepares a valid conversation and tries it on the
  dishonest party. Now, in case this party does not send the expected
  reply, we need the possibility to rewind him.\footnote{More
  precisely, we model the player---similar to the simulator---as a
  machine, and thus, we can just set back this machine to an earlier
  status, i.e., erase parts of the memory and start a new
  conversation. In that sense, rewinding can be thought of as, for
  instance, rebooting a computer after it crashed.}

  Unfortunately, rewinding as a proof technique can generally not be
  directly applied in the quantum world, i.e., if the dishonest machine
  is a quantum computer. First, we cannot trivially copy and store an
  intermediate state of a quantum system, and second, quantum
  measurements are in general irreversible. In order to produce a
  classical transcript, the simulator would have to partially measure
  the quantum system without copying it beforehand, but then it would
  become impossible to reconstruct all information necessary for
  correct rewinding.

  Due to these difficulties, no simple and straightforward security
  proof for the quantum case was known. However, Watrous recently
  showed that in a limited setting an efficient quantum simulation,
  relying on the newly introduced \emph{quantum rewinding theorem}
  (see~\cite{Watrous09} and Section~\ref{sec:quantum.rewinding}), is
  possible. We will discuss this aspect in more detail in
  Chapters~\ref{chap:coin.flip} and \ref{chap:framework}: We will show
  that the quantum rewinding argument can also be applied to classical
  non-constant round coin-flipping in the quantum world, and propose a
  framework to weaken certain assumptions on the coin, in quest for a
  quantum-secure constant round protocol.

  \paragraph{\sc Spooky actions at a distance.} 

  This famous naming by Einstein\footnote{``Spooky actions at a
  distance'' was put down originally as ``spukhafte Fernwirkung''
  in~\cite{Einstein71}.} describes the phenomenon of
  \Index{entanglement}. Informally, two qubits are called entangled,
  if their state can only be described with reference to each
  other. This has the effect that a measurement on one particle has an
  instantaneous impact on the other one---despite any distance
  separating the qubits spatially.

  Entanglement is definitely a unique resource to the quantum world
  only. In the words of Schr\"{o}dinger, entanglement is not
  ``\emph{one} but rather \emph{the} characteristic trait of quantum
  mechanics, the one that enforces its entire departure from classical
  lines of thought''~\cite[p.~555]{Schroedinger35}. Besides
  constituting a disturbing aspect---intuitively and philosophically,
  entanglement opens up for interesting applications such as
  \emph{quantum teleportation}~\cite{BBCJPW93} and \emph{superdense
  coding}~\cite{BW92}, as well as for various aspects in quantum
  cryptography and computing. We will use entanglement as a thought
  experiment in our quantum protocols when analyzing an equivalent
  \emph{purified EPR-version}\footnote{An EPR-pair denotes a pair of
  entangled qubits. The name (ironically) originates from the
  paper~\cite{EPR35} by Einstein, Podolsky, and Rosen, in which they
  criticized quantum mechanics as an incomplete theory---due to
  entanglement.} (Chapter~\ref{chap:hybrid.security}).


\section{Contributions}
\label{sec:contributions}

This dissertation is based on research done during the three years of
my PhD studies at the Department of Computer Science, Aarhus
University, Denmark. Part of the research was conducted while visiting
Universit\'e de Montr\'eal, Qu\'ebec, Canada. The realm of this work
is quantum cryptography and classical cryptography in the quantum
world. More specifically, the thesis covers aspects of (quantum)
cryptographic protocol theory, based on cryptographic primitives. The
main results are outlined in the following sections and pictorially
represented in Figure~\ref{fig:pic.thesis}.


\subsection{The Importance of Mixed Commitments}
\label{sec:cont.commit}
\index{rewinding}

  Classical mixed (or dual-mode) commitments are of great significance
  for most constructions discussed in this work. Here, we explain the
  challenges that the quantum world imposes on commitments in general
  and summarize the results of~\cite{DFLSS09,DL09,LN10} in that
  aspect.\\
 
  Security for classical constructions in the quantum world means that
  quantum computation does not jeopardize the underlying mathematical
  assumption that guarantees the security, for instance, in the
  context of commitments, the hiding and binding property. However, we
  encounter even more setbacks in the context of actually proving such
  constructions secure in an outer protocol, which, in regard of this
  work with its main focus on simulation-based security, are mostly
  due to the strong restrictions on rewinding in the quantum world.
  
  The first difficulty in any attempt to rewind the adversary regards
  the fact that the reduction from the computational security of an
  outer protocol to the \emph{computationally binding} property of a
  commitment does not simply translate from the classical to the
  quantum world. Computational binding means that if a dishonest party
  can open a commitment to two different values, then the
  computational assumption does not hold. In the classical case, a
  simulator simulates a run of the outer protocol with the committer,
  such that the latter outputs a valid commitment at some point during
  the execution. Later in the protocol he must then provide a correct
  opening. The simulator has the possibility to rewind the player to
  any step in the protocol execution, e.g.\ to a point after the
  commitment was sent. Then it can repeat the simulation of the outer
  protocol, which can now be adapted to the simulator's knowledge of
  the committed value. If the dishonest committer opened the same
  commitment to a different value than previously, he could break the
  underlying assumption guaranteeing computational binding. In other
  words, two valid openings of the same commitment imply the inversion
  of the underlying one-way function, which concludes the proof. Such
  a technique, however, is impossible to justify in the quantum world,
  since we cannot trivially copy and store an intermediate state, and
  measurements are in general irreversible. In order to succeed, the
  simulator would have to partially measure the quantum system without
  copying it beforehand to obtain the first transcript, but then it
  would become impossible to reconstruct all information necessary for
  correct rewinding.

  The second challenge we encounter is to prove an outer protocol with
  an embedded \emph{computationally hiding} commitment
  secure. Generally speaking, in a classical simulation of the outer
  protocol, the simulator aims e.g.\ at hitting an ideal outcome to a
  function of which it then commits. Then, if the reply from the
  possibly dishonest counterpart matches this prepared function such
  that both sides conclude on the ideal value as their result and the
  transcript is indistinguishable from a real run of the protocol, the
  simulation was successful. Otherwise, the simulator rewinds the
  dishonest player completely and repeats the simulation. We show a
  natural and direct translation of this scenario to the quantum world
  in Chapter~\ref{chap:coin.flip}, where we use a technique that
  allows quantum rewinding in this very setting when using bit
  commitments (see Section~\ref{sec:cont.coin.flip}). In case of
  string commitments however, we cannot rewind the other player in
  poly-time to hit the guess, since that guess consists of a
  bit-string. A possible solutions for simulating against a classical
  adversary is to let him commit to his message before the simulator
  commits. Then the player's message can be extracted and the
  simulation can be matched accordingly. This technique, however, is
  again doomed to fail in the quantum realm, since it reduces to the
  previous case where the simulator cannot preserve the other party's
  intermediate status as required during such a simulation.

  We will circumvent both of the above aspects by introducing mixed
  commitment schemes in our protocols. Generally speaking, the notion
  of mixed commitments requires some trapdoor information, given to
  the simulator in the ideal world. Depending on the instantiation,
  the trapdoor provides the possibility for \emph{extraction} of
  information out of the commitments or for \emph{equivocability} when
  opening the commitments. This allows us to circumvent the necessity
  of rewinding in the proof, while achieving high security in the real
  protocol. The idea of mixed commitment schemes is described in more
  detail in Section~\ref{sec:mixed.commit.idea} and a quantum-secure
  instantiation is proposed in
  Section~\ref{sec:mixed.commit.instantiation}. Various extensions are
  then discussed to match the construction to respective requirements
  in different outer protocols
  (Sections~\ref{sec:extended.commit.compiler},~\ref{sec:extended.commit.coin},
  and~\ref{sec:mixed.commit.trapdoor.opening}).


\subsection{Improving the Security of Quantum Protocols}
\label{sec:cont.hybrid.security}

  The following results are joint work with Damg{\aa}rd, Fehr,
  Salvail, and Schaffner~\cite{DFLSS09} and will be addressed in
  detail in Chapter~\ref{chap:hybrid.security}.\\

  We propose a general compiler for improving the security of a large
  class of two-party quantum protocols, implementing different
  cryptographic tasks and running between mutually distrusting players
  Alice and Bob. Our main result states that if the original protocol
  is secure against a so-called \emph{benign} Bob, who is only
  required to treat the qubits ``almost honestly'' but can deviate
  arbitrarily afterwards, then the compiled protocol is secure against
  a \emph{computationally bounded} quantum Bob. The unconditional
  security against Alice is preserved during compilation and it
  requires only a constant increase of transmitted qubits and
  classical messages.

  The consequences of such a compiler are twofold. First, the basic
  assumption in designing new protocols for any two-party
  functionality is reduced to the relatively weak assumption on
  benignity. On the other hand, the proofs for already existing
  protocols within the specific class typically go through under the
  assumption (at least after some minor adaptions). And second,
  security in the bounded-quantum-storage model implies benign
  security. Therefore, by compilation of such protocols, we can
  achieve \emph{hybrid security}, which means that the adversary now needs
  \emph{both} large quantum memory \emph{and} large quantum computing
  power to break these new protocols.
 
  In more detail, the protocols we consider here start with a qubit
  transmission from Alice to Bob, where each qubit is encoded in one
  of two conjugate bases. This implies that, whenever Bob measures in
  the complementary basis, he obtains a random outcome. The second
  part of the protocol consist of arbitrary classical messages and
  local computations, depending on the task at hand but typically
  relying on the fact that a dishonest Bob has high uncertainty about
  a crucial piece of information.

  The basic technique to construct the compiler was already suggested
  in the first quantum oblivious transfer protocol~\cite{CK88}. We
  want to force Bob to measure by asking him to commit (using a
  classical scheme) to all his basis choices and measurement results,
  and then require him to open some of them later. While classical
  intuition suggests that the commitments should force Bob to measure
  (almost) all the qubits, it was previously very unclear what exactly
  it would achieve in the quantum world. To our best knowledge, it was
  never formally proven that the classical intuition also holds for a
  quantum Bob. We now give a full characterization of the
  commit\&open approach in general quantum settings, namely that it
  forces Bob to be benign.

  We propose a formal definition for \emph{benignity}, which might be
  of independent interest. A benign Bob is characterized by the
  following two conditions, which must be satisfied after the qubit
  transmission. First, his quantum storage is very small, and second,
  there exists a basis-string such that the uncertainty about Alice's
  encoded bit is essentially one bit whenever the encoding basis does
  not match the basis indicated in that string. These two conditions
  imply that a successfully passed opening of his commitments for a
  random test subset puts Bob in a situation, which is close to a
  scenario in which he measured as supposed to: His quantum memory is
  essentially of size zero, and furthermore, measuring the untested
  qubits in a basis complementary to the one Bob (claims to have)
  used, leads to a result with large uncertainty. The bounds on Bob's
  uncertainty and his quantum memory are proven for an ideal state
  that is negligible close to the real state. For the ideal state, we
  can then show that the remaining subsystem after the test is a
  superposition of states with \emph{relative Hamming distance upper
  bounded by the test estimate}.

  To conclude the proof, we assume that the original protocol
  implements some ideal functionality with statistical security
  against benign Bob. Then we show that the compiled protocol with the
  commitments also implements that functionality but now with security
  against any computationally bounded (quantum) Bob. To preserve the
  unconditional security of the original protocol, we require an
  unconditionally hiding commitment scheme. Since the common reduction
  from the computational security of the protocol to the computational
  binding property of a commitment would require rewinding, we use a
  mixed dual-mode commitment, which allows us to avoid rewinding Bob
  in this step (see also Section~\ref{sec:cont.commit}).

  We generalize our result to noisy quantum communication and show
  that the compilation does not render sequential composability
  insecure. We then extend the underlying commitment scheme for a more
  general composability guarantee and obtain that any compiled
  protocol \emph{computationally quantum-UC-emulates} its
  corresponding ideal functionality.


\subsection{Classical Coin-Flipping in the Quantum World}
\label{sec:cont.coin.flip}

  The result on quantum-secure single coin-flipping is based
  on~\cite{DL09}, co-authored with Damg{\aa}rd, and will be fully
  discussed in Chapter~\ref{chap:coin.flip}. The proposed
  amplification framework for obtaining strong coin-strings from weak
  initial assumption on the coins is joint work with
  Nielsen~\cite{LN10} and will be addressed in more detail in
  Chapter~\ref{chap:framework}.\\

  We first investigate the standard coin-flipping protocol with
  classical messages exchange but where the adversary is assumed to be
  capable of quantum computing. The output of the protocol is a
  uniformly random unbiased bit, and the construction does not require
  any set-up assumptions. Therewith, the communicating parties can
  interactively generate true randomness from scratch in the quantum
  world. Our result constitutes the most \emph{direct quantum
  analogue} of the classical security proof by using a recent result
  of Watrous~\cite{Watrous09} that allows for quantum rewinding in
  this restricted setting and when flipping a single coin.

  The full potential of coin-flipping lies in the possibility of
  flipping a string of coins instead of a bit, such that the parties
  can interactively generate a common random string from
  scratch. Therewith, it is possible, for instance, to implement the
  theoretical assumption of the common-reference-string-model, which
  then implies that various interesting applications can be realized
  in a simple manner without any set-up assumptions.

  We show that with our definitions, the single coin-flipping protocol
  composes sequentially. Additionally, we sketch an extended
  construction of the underlying commitment scheme, allowing for
  efficient simulation on both sides, with which we achieve more
  general composition guarantees. Both compositions, however, are not
  fully satisfactory. Sequential coin-flipping allows for
  implementations without set-up assumptions but leads to a
  non-constant round application. In contrast, parallel composition
  achieves much better efficiency with constant round complexity but
  requires some set-up assumptions in our proposed construction
  here. Unfortunately, we do not know how to extend Watrous quantum
  rewinding to the case of bit-strings, while keeping the running time
  of the simulator polynomial. The proof technique in the purely
  classical setting is impossible to apply in the quantum world (see
  also Section~\ref{sec:cont.commit}). Other techniques to achieve
  constant round coin-flipping are not known to date.

  Our framework in Chapter~\ref{chap:framework} can be understood as a
  step towards \emph{constant round coin-flipping}. We first
  investigate different security degrees of a string of coins. We then
  propose protocol constructions that allow us to amplify the
  respective degrees of security such that weaker coins are converted
  into very strong ones. The final result constitutes an amplification
  towards a coin-flipping protocol with long outcomes, which is fully
  poly-time simulatable on both sides against quantum adversaries. The
  protocol can be implemented with quantum-computational security in
  the plain model without any set-up assumptions. It only assumes
  mixed commitment schemes, which we know how to construct with
  quantum security, and no other assumptions are put forward. With
  this solution, we still have to compose the single coin-flip as
  sketched above sequentially to obtain long outcomes, but we achieve
  coins with stronger security.

  Our method of amplifying the security strength of coins also applies
  to potential constant round coin-flipping. If the underlying weak
  protocol already produces string outcomes and is constant round,
  then the resulting strong protocol is also constant round, and we
  consider it a contribution in itself to define the weakest security
  notion for any potential candidate that allows to amplify to the
  final strong protocol using a constant round reduction.


\subsection{Applications}
\label{sec:cont.applications}

  We consider our applications in both parts of the thesis
  (Chapters~\ref{chap:hybrid.security.applications}
  and~\ref{chap:coin.flip.applications}) well suited as examples for
  the respective precedent main results, since they all have some
  special properties. Depending on the context they are proposed in,
  they appeared in~\cite{DFLSS09,DL09,LN10}.\\

  The first quantum protocol in Section~\ref{sec:hybrid.security.ot}
  implements \emph{oblivious transfer} (OT), which constitutes a
  highly relevant cryptographic primitive that is complete for general
  two-party computation. Interestingly, the idea behind this primitive
  was introduced in the context of quantum cryptography, namely, in
  the pioneering paper of Wiesner~\cite{Wiesner83} that also paved the
  way for quantum cryptography by introducing the concept of conjugate
  coding. The very nature of conjugate coding implies oblivious
  transfer, and with that, it can be understood as a natural quantum
  primitive.

  Classical and quantum OT cannot be implemented without any
  additional restrictions. However, in contrast to classical OT,
  quantum OT reduces to classical commitment. The idea of using a
  classical commitment within quantum protocols was already suggested
  in the first quantum oblivious transfer protocol~\cite{CK88} and its
  follow-up work in~\cite{BBCS91}. Various partial results followed,
  such as assuming a perfect ideal
  commitment~\cite{Yao95,Mayers96,Unruh10} or a (theoretical) quantum
  string commitment~\cite{CDMS04}. Based on the analysis of our
  compilation (sketched in Section~\ref{sec:cont.hybrid.security}), we
  can now rather simply apply our compiler to (a variant of) the
  original quantum OT-protocol, and therewith, give a complete proof
  for a concrete commitment scheme.

  In a rather straightforward way, oblivious transfer as a building
  block easily extends to \emph{password-based identification}, which
  is needed for any authenticated set-up. The quantum identification
  scheme in Section~\ref{sec:hybrid.security.id} allows for
  identification by solely proving the knowledge of a secret password
  without actually announcing it in clear. Furthermore, it has some
  special properties, which indicates its utility value in
  practice. First, the only option without being in possession of the
  password is to guess it, which implies that the same password may be
  safely reused for a long time. Second, the scheme tolerates a
  possibly non-uniform password, which translates to a realistic
  assumption of user-memorizable passwords. And last, a typical
  setting for identification is not necessarily required to run over
  large distances to be considered useful, and as such, it can
  actually be implemented with existing technology. Naturally, an
  identification scheme, secure under diversified assumptions and
  against any external adversary, is an important step towards an
  actual implementation.

  The classical \emph{generation of commitment keys} in
  Section~\ref{sec:key.generation.coin} nicely combines the above
  applications with the results on quantum-secure coin-flipping,
  fulfilling the requirement on our mixed commitment construction. By
  running a coin-flipping protocol as an initial step in the quantum
  protocols above, the communicating players can interactively
  generate their commitment keys for compilation. This allows us to
  avoid the common-reference-string-model and yields implementations
  of \emph{entire} protocols in the quantum world without any set-up
  assumptions.

  The two application in the context of zero-knowledge are interesting
  in that the interactive generation of coins at the beginning or
  during outer protocols allows for quantum-secure realizations of
  classical schemes from scratch. First in
  Section~\ref{sec:coin.iqzk}, we show a simple transformation from
  non-interactive zero-knowledge to \emph{interactive quantum
  zero-knowledge}. Then in Section~\ref{sec:coin.zkpk}, we propose a
  \emph{quantum-secure zero-knowledge proof of knowledge}, which relies
  not only on initial randomness but also on enforceable randomness
  and is based on a witness encoding scheme providing a certain degree
  of extractability, defined for the quantum context to resemble
  special soundness of classical schemes. Both zero-knowledge
  constructions nicely highlight that the realization of coin-flipping
  as a stand-alone tool allows for using it rather freely in various
  contexts.

  \begin{figure}[h]
    \includegraphics[width=\textwidth]{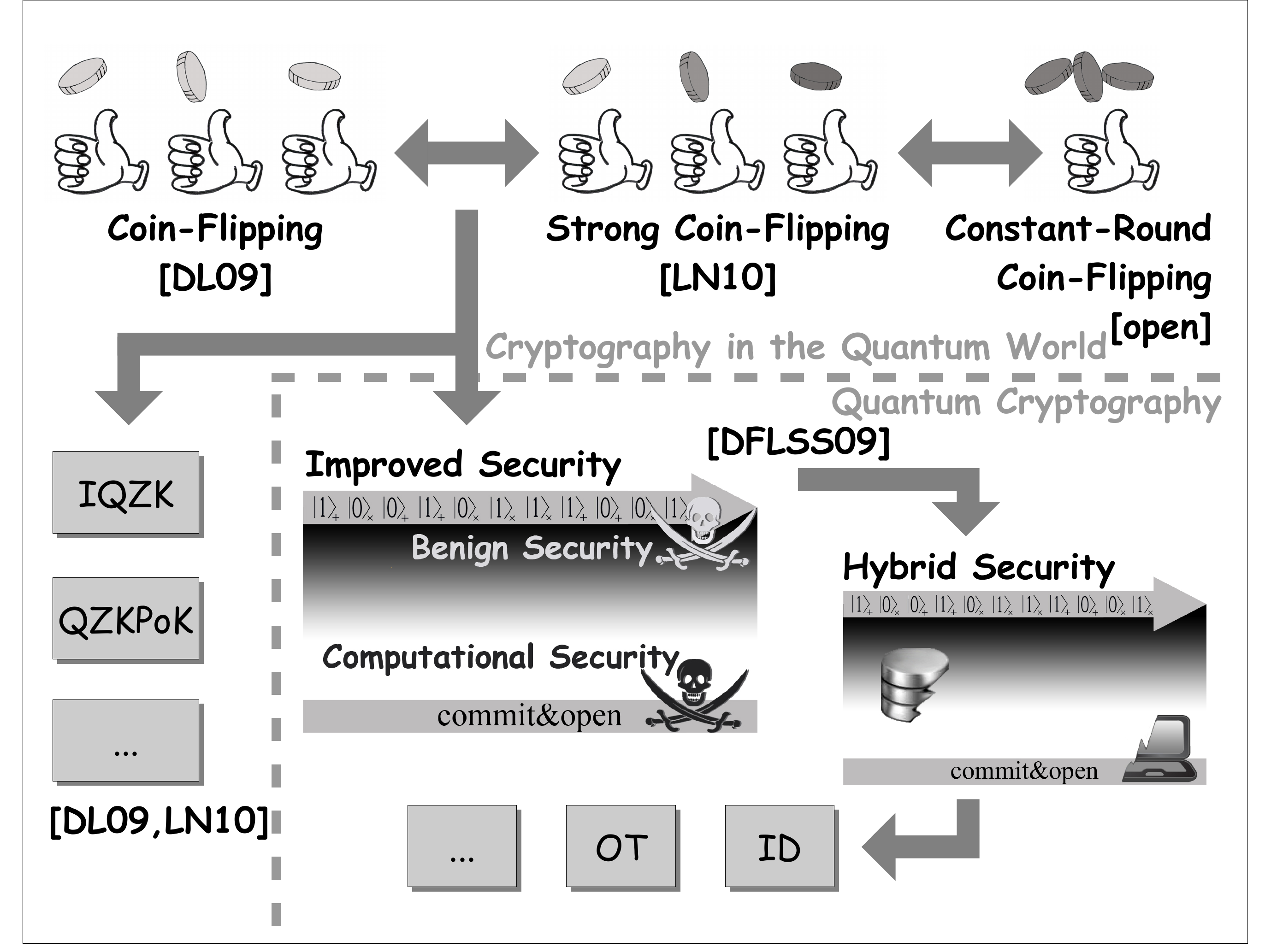}
    \vspace{-3ex}
    \caption{Picture of the Thesis.}
    \label{fig:pic.thesis}
  \end{figure}


\clearemptydoublepage
\part{Setting The Stage}
\label{part:preliminaries}


\chapter{Cryptographic Toolbox}
\label{chap:crypto.tools}

In this work, we are interested in classical and quantum cryptographic
two-party protocols, i.e., our focus lies on enabling two players to
accomplish a specific task securely by communicating over a
distance. In a perfect world of gentlemen, we could, of course, just
communicate over a distance without using cryptographic security
precautions. In an ideal world, we can simply assume a ``black-box''
that solves what we want while not leaking anything of
importance. However, we operate in the real world. This means that we
do not only have to take various dishonest players into account when
implementing our protocols, but also that we have to work within a
restricted framework of given conditions and existing
resources.\footnote{Note that, throughout this work, we will use the
terms \emph{ideal world} and \emph{real world} also in the more formal
context of the so-called \emph{two-world paradigm} (see
Section~\ref{sec:worlds}) for simulation-based proofs.}

In the following sections we formalize this intuitive description in
cryptographic terms. The chapter is not intended to provide a full
introduction to cryptographic protocol theory, but rather to give a
brief but complete overview of notation, tools, conditions, and
settings we will use, and to fix terminology that may vary in standard
literature. In short, we are setting the stage for the results in this
thesis.


\section{Players}
\label{sec:players}

Our main characters are Alice and Bob, who are subject to different
roles and cheating capabilities. The \emph{correctness} of our
two-party protocols is ensured, if they implement the task at hand in
the desired way. This scenario only concerns
\emph{honest}\index{player!honest} parties Alice and Bob, who may have
different roles, such as sender, receiver, committer, verifier, user
and server, depending on the respective functionality to be carried
out. An honest player is denoted by~$\Pl$.

\emph{Security} is shown by investigating the case where one of the
parties is \emph{dishonest}\index{player!dishonest}. More precisely, a
dishonest party $\dP$ can try, for instance, to bias the outcome of
the protocol or to succeed illegitimately.

Between these two extremes, there are various nuances of cheating. For
instance, the common notion of
\emph{semi-honest}\index{player!semi-honest} describes an
``honest-but-curious'' player who is curious enough in trying to gain
additional information while following the protocol honestly. We will
in Chapters~\ref{chap:hybrid.security}
and~\ref{chap:hybrid.security.applications} use another intermediate
notion that captures \emph{benignly dishonest}\index{player!benign}
behavior in quantum protocols. The protocols consist of a quantum
transmission phase and some classical post-processing. A benign
receiver of qubits is assumed to treat these ``almost honestly'',
which means he immediately measures most of the qubits upon reception
in the specified bases. Afterwards during the classical
post-processing, he can deviate arbitrarily. Thus, in some sense, he
wants to cheat but is incapable of mastering the quantum information
in any other way than simply measuring it. We will define this newly
introduced notion in greater detail later on, as it forms the
foundation of our improved quantum protocols.

A very different external adversary is the so-called
\emph{man-in-the-middle}\index{player!man-in-the-middle} Eve (denoted
by $\E$), who tries to eavesdrop on the classical and quantum
communication between Alice and Bob, with the intention to break the
protocol---or at least gain some information---without being
detected. Quantum cryptography provides its protocols with
\emph{automatic intrusion detection}, due to the fact that here any
kind of intrusion will inevitably disturb the system. However, we have
to thoroughly implement the testing of qubits for interference as well
as investigate the potential information leakage of the classical
communication.


\section{Security Flavors, Assumptions, and Models}
\label{sec:flavors}

  The purpose and objective of theoretical cryptography is to design
  protocols with the highest security possible under \emph{any}
  condition, this means without any restriction on adversarial
  resources such as computing power and memory size. However, this
  \emph{unconditional security}\index{security!unconditional} is
  extremely hard to obtain for both players simultaneously in the
  classical and in the quantum world. In fact, some tasks are proven
  to be impossible to achieve with unconditional security for both
  players. The most well-known example thereof might be the
  impossibility results on unconditionally secure classical and
  quantum bit commitment (proven in the quantum case
  by~\cite{Mayers97,LC97}). Furthermore, for two distrusting parties,
  the only applications actually proven to be unconditionally secure
  regarding confidentiality are Vernam's symmetric one-time pad
  encryption~\cite{Vernam26,Shannon49} as well as quantum key
  distribution~\cite{BB84,SP00}.

  Thus, the level of security has to be lowered for implementing other
  functionalities, and we have to achieve a reasonable balance between
  realistic assumptions under consideration of current and future
  technology---as weak as possible---and yet meaningful security---as
  strong as possible. For that purpose, we specify cryptographic
  models to capture various notions of security and to impose
  realistic restrictions on the adversary. To mention just a few, such
  models consider limited computing power, limited memory
  size~\cite{Maurer92,DFSS05}, a common resource with special
  properties (e.g.\ initially shared randomness), noisy
  storage~\cite{WST08} or restricted quantum measurement (e.g.\ a
  limited set of measurements~\cite{KMP04} or a limited set of qubits
  to be measured at the same time~\cite{Salvail98}).

  \paragraph{\sc Computational Security.}
  \index{security!computational} Restricting the adversarial classical
  computing power and time is currently the most applied model in
  practical public-key cryptography. Thus, it is known as the
  \emph{plain model}\index{plain model}, achieving \emph{computational
  security} based on classical hardness assumptions that some problems
  are computationally infeasible to solve in polynomial\index{polynomial}
  time\footnote{An algorithm is poly-time, if its running time is
  upper bounded by a polynomial in the size of its input,
  i.e.~$O(n^c)$. In more detail, there exist constants $c > 1$ and $n_0$
  such that $poly(n) \leq n^{c}$ for all $n > n_0$. As synonyms, we
  often use \emph{feasible} or \emph{efficient}.}. Usually, security
  is shown by reducing the security of the actual scheme to that of a
  well-known mathematical problem. However, the hardness of such
  complexity assumptions is unproven.

  It should also not go unnoted that with the emergence of quantum
  computers which, due to their speed-up in running time, have great
  potential to solve several of the basic assumptions in polynomial
  time, security of various crypto-systems would fold. To give
  examples, Shor showed algorithms for efficiently factoring large
  integers~\cite{Shor97}, which would jeopardize the RSA assumption,
  and for the related problem of computing discrete logarithms
  underlying e.g.\ the ElGamal encryption system. Grover's algorithm
  for conducting a search through some unstructured search space shows
  a quadratic speed-up over classical computation. This, for instance,
  also affects the time of performing exhaustive search over the set
  of possible keys, used in symmetric crypto-systems (e.g.\ DES). Of
  course, these algorithms only yield profitable results, if
  large-scale quantum computers can be built. Interestingly, the very
  quantum effects that makes them so powerful, also makes them so
  difficult to control---so far.

  \paragraph{\sc Quantum-Computational Security.}
  \index{security!quantum-computational} Recently, the new sub-field
  of so-called post-quant\-um crypto\-gra\-phy has emerged within
  public-key cryptography.\footnote{The common classification might be
  slightly confusing, in that the notion ``post-quantum'' relates to
  the time after the successful development of large-scale
  \emph{quantum computers} as opposed to \emph{quantum cryptography.}}
  There, the focus lies on researching assumptions which are believed
  to be hard even on a quantum computer, and thus, on achieving
  \emph{quantum-computational security}. Post-quantum crypto-schemes
  include, for instance, the McEliece crypto-system based on a
  coding-theoretic problem~\cite{McEliece78} and lattice-based
  crypto-systems (e.g.~\cite{Ajtai96, Regev05}). The latter provide,
  besides good efficiency when en- and decoding, the merit that
  breaking the security of such protocols implies to solve a hard
  lattice problem in the \emph{worst case}. However, we should stress
  also in this context that this hardness is again assumed; formal
  proofs are still to come. In this work, we will use lattice-based
  crypto-systems for implementing mixed commitment schemes, secure in
  the quantum world (Chapters~\ref{chap:hybrid.security} and
  \ref{chap:framework}).

  \paragraph{\sc Quantum Security.}
  In contrast to security through mathematical hardness assumptions in
  classical cryptography, the security in quantum cryptography is
  based on quantum mechanical laws. Proofs for physical limitations
  are not by reduction as for computational limitations but in
  information-theoretic terms. That means that in such models, an
  adversary does not learn \emph{any} information, except with at most
  negligible\index{negligible} probability.\footnote{Negligible in $n$
  means that any function of $n$ is smaller than the inverse of any
  polynomial, provided $n$ is sufficiently large, i.e., for all
  constants $c$ there exists a constant $n_c$ such that $\negl{n} \leq
  n^{-c}$ for all $n > n_c$.}

  \paragraph{\sc Bounded-Quantum-Storage Model.}
  \index{bounded-quantum-storage model} In the quantum cryptographic
  setting, one such physical limitation is formalized in the
  \emph{bounded-quantum-storage model} (BQSM), proposed
  in~\cite{DFSS05}. The intuitive idea behind the model is that the
  most sensitive information is encoded in qubits that are transmitted
  in the first phase of the protocol. Then, at some later point,
  typically an announcement of the encoding bases follows to complete
  the task at hand. Now, under the assumption that an adversary's
  quantum memory size is limited, he cannot store all of the qubits
  but has to measure some fraction. Thus, by converting quantum
  information into classical information without complete knowledge of
  the right bases, information gets irreversibly destroyed.

  The protocols in this model achieve unconditional protection against
  cheating by one of the players, while if the other is corrupted, the
  protocols are secure under the sole assumption that his quantum
  storage is of limited size, namely of size at most a constant
  fraction of the qubits sent. Such a bound can also be applied to an
  external eavesdropper's quantum memory by slightly extending the
  respective original protocol. The underlying motivation for the BQSM
  is the fact that transmission and measurement of qubits is well
  within reach of current technology. Storing quantum information
  however requires keeping the qubit state stable under controlled
  conditions for a non-negligible time, which still constitutes a
  major technological challenge, and an attack would require large
  quantum storage with a long lifetime. In contrast, honest parties,
  following the protocol, do not need quantum memory at
  all. Furthermore, neither honest nor dishonest parties are bounded
  with respect to their classical storage or computing power. We want
  to stress that the impossibility results against the
  bounded-classical-storage model (see
  e.g.~\cite{Maurer90,Maurer92,CCM98,DM04}) do not hold in the quantum
  setting.\footnote{The bounded-classical-storage model ensures
  security as long as the adversary's memory size is at most quadratic
  in the memory size of the honest players. A favorably larger gap
  between the storage assumptions on honest and dishonest parties was
  shown to be impossible~\cite{DM04}.} Hence, the BQSM is realistic
  for fundamental physical reasons and potentially useful in practice.

  Many two-party applications investigated in the BQSM (like
  identification) are not necessarily required to run over large
  distances to be considered useful. Thus, such protocols can actually
  be implemented with existing devices, and many applications have
  been proven BQSM-secure~\cite{DFSS05,DFSS07,Schaffner07}. We will
  work in this model in
  Chapter~\ref{chap:hybrid.security.applications}, where it
  constitutes one of the security layers in our quantum protocols.

  \paragraph{\sc Common-Reference-String-Model.}
  \index{common-reference-string-model} Another useful model, which we
  will consider, is the \emph{common-reference-string-model}
  (CRS-model). In this model, as the name suggests, the parties are
  provided with a classical common public string before communication,
  taken from some fixed distribution that only depends on the security
  parameter. For efficiency and composability, we will often assume
  the model to allow for techniques, which require an initially shared
  random string. However, we consider a random string ``in the sky'' a
  set-up, which is only theoretically useful. To meet more practical
  demands, we suggest in Chapter~\ref{chap:coin.flip} a quantum-secure
  implementation of the CRS-model ``from scratch''.


\section{Worlds}
\label{sec:worlds}

  \paragraph{\sc Classical vs.\ Quantum World.}
  We are interested in cryptography in the quantum world, covering
  both quantum and classical cryptographic protocol theory, which is
  evident in the separation of the thesis in the two main parts,
  Part~\ref{part:quantum.cryptography} on quantum cryptography and
  Part~\ref{part:cryptography.in.quantum.world} on classical
  cryptography in the quantum world. Thus, throughout this work, we
  consider quantum potential---achieving very high security in the
  first case but also imposing new demands in the latter. In contrast,
  the (pure) classical world of cryptography does traditionally not
  assume adversarial quantum effects. However, we emphasize our very
  strong requirement also for all classical protocols and proofs to be
  quantum-computationally secure, which implies both the exclusive use
  of post-quantum crypto-schemes, and the avoidance or carefully
  adaption of classical proof techniques.

  \paragraph{\sc Ideal vs.\ Real World.}
  \index{two-world paradigm} \index{ideal world} \index{ideal-world
  adversary} \index{ideal functionality} \index{real world}
  \index{real-world adversary} 
  For the definition of security, we work
  in two different worlds, which are captured in the \emph{two-world
  paradigm} of simulation-based proofs. The basic idea of the paradigm
  is to first specify the ideal functionality $\F$ that models the
  intended behavior of the protocol, or in other words, the properties
  we would have in an \emph{ideal world}. The ideal functionality can
  be thought of as a trusted third party or simply a black-box that
  gets private inputs from the players, accomplishes a specific task
  without leaking any information, and then outputs the result to the
  respective player. Honest and dishonest players in the ideal world
  are modeled by probabilistic poly-time machines, denoted by $\hP$
  and $\dhP$, respectively. The \emph{real world} captures the actual
  protocol $\Pi$, consisting of message exchange between the parties
  and local computations. Recall that real-world players are indicated
  by honest $\Pl$ and dishonest $\dP$.

  Now, the input-output behavior of $\F$ defines the required
  input-output behavior of $\Pi$. Intuitively, if the executions are
  indistinguishable, security of the protocol in real life follows. In
  other words, a dishonest real-world player $\dP$ that attacks
  protocol $\Pi$ cannot achieve (significantly) more than an
  ideal-world adversary $\dhP$, attacking the corresponding ideal
  functionality $\F$. We will make this aspect more formal in
  Section~\ref{sec:security.definition}.


\section{Primitives}
\label{sec:primitives}

In the following, we will describe those two-party cryptographic
\Index{primitives}, along with some known facts about them, that are
relevant in the context of this work. Primitives are fundamental
problems that are later used as basic building blocks in larger outer
protocols. Discussed on their own, primitives might seem to be
somewhat limited but still constitute intriguing thought
experiments. For clarification, an identification scheme, as discussed
in Section~\ref{sec:primitives.id}, may commonly not count as a
primitive per se, although it may well constitute a building block in
a larger outer protocol. Our prime purpose for introducing it in the
context of primitives, however, is the close relation to oblivious
transfer in its construction.


\subsection{Commitments}
\label{sec:primitives.commitment}
\index{commitment}
\index{commitment!binding}
\index{commitment!hiding}

  Commitment schemes constitute a very important building block within
  cryptographic protocols. In fact, all our protocols proposed here
  implementing a wide range of cryptographic tasks, make use of
  various types of commitment schemes, which may indicate the
  significance of the construction. Commitments can be realized with
  classical schemes or through quantum communication. Here, we will
  only discuss and construct commitments from classical crypto
  schemes, but with a strong requirement of quantum-computational
  security.

  Intuitively, a commitment scheme allows a player to commit to a
  value while keeping it hidden ({\em hiding property}), yet
  preserving the possibility to reveal the value fixed at commitment
  time later during the so-called opening phase ({\em binding
  property}). More formally, a basic commitment scheme
  $\commitx{m}{r}$ takes a message $m$ and some random variable $r$ as
  input. Depending on the respective scheme, the message $m$ can be a
  single bit (\emph{bit commitment}) or a bit sequence (\emph{string
  commitment}). The length of the randomness $r$ is polynomial in the
  security parameter. It is also possible to construct a so-called
  \emph{keyed commitment schemes} of the form $\commitk{m}{r}{K}$,
  which takes key $K$ as additional input. The most common way of
  opening commitment $\commitx{m}{r}$ to reveal the committed message
  $m$ when time is ripe, is to send values $m$ and $r$ in plain, so
  that the receiver of the commitment can check its validity. In
  Chapter~\ref{chap:framework}, we will change this way of opening a
  commitment, due to the special requirements of the particular
  construction there.

  The hiding property is formalized by the non-existence of a
  distinguisher able to distinguish with non-negligible advantage
  between two commitments, i.e., we have indistinguishability between
  two commitments with $\commitx{m_1}{r_1} \approx
  \commitx{m_2}{r_2}$. The binding property is fulfilled, if it is
  infeasible for a forger to open a commitment to more than one valid
  value, i.e., we have $\commitx{m_1}{r_1} \neq \commitx{m_2}{r_2}$
  for $m_1 \neq m_2$. Each property, hiding and binding, can be
  satisfied unconditionally or subject to a complexity assumption.
  The ideal case of unconditionally secure commitments, i.e.\
  unconditionally hiding and unconditionally binding at the same time,
  is impossible. Consequently, we have to decide on one of the two
  flavors of commitment schemes, namely unconditionally hiding and
  computationally binding or unconditionally binding and
  computationally hiding.\footnote{Note that certain
  applications---beyond the scope of this work---have computational
  security simultaneously for both properties hiding and binding.} For
  completeness, it is worth noting that the same applies in quantum
  cryptography~\cite{Mayers97,LC97}, where perfect commitments can
  only be achieved when assuming some restrictions on the adversary,
  for instance, the BQSM-model~\cite{DFSS05,DFRSS07}.

  In the context of oblivious transfer (OT; see
  Section~\ref{sec:primitives.ot}), we know that a classical
  commitment does not imply classical OT without any additional
  requirement (such as key agreement). In contrast, a classical
  commitment implies quantum OT, which is all the more interesting as
  OT is complete for secure two-party computation. This implication in
  the quantum case was realized in~\cite{CK88} and proven partially
  in~\cite{Yao95,Mayers96,CDMS04}. We will give the first full proof
  in Section~\ref{sec:hybrid.security.ot}.

  Commitments are equivalent to \Index{one-way function}s, i.e., a
  function $f: \zo^* \ra \zo^*$ for which it is easy to compute
  $f(x)$, given $x$. But given only $y = f(x)$ where $x$ is random, it
  is computationally infeasible in poly-time to compute any element in
  $f^{-1}(y)$. Thus, from an appropriate one-way function, secure
  against quantum adversaries, we can construct quantum-secure
  commitment schemes (e.g.~\cite{Naor91}). Bit commitments, in turn,
  imply a quantum-secure coin-flip, which we will show in
  Chapter~\ref{chap:coin.flip}. Naturally, the hiding, respectively
  binding, property holds with unconditional security in the classical
  and the quantum setting, if the distinguisher, respectively the
  forger, is unrestricted with respect to his (quantum) computational
  power. Recall that in case of a poly-time bounded classical
  distinguisher, respectively forger, the commitment is
  computationally hiding, respectively binding. The computationally
  hiding property translates to the quantum world by simply allowing
  the distinguisher to be quantum. However, the case of a quantum
  forger cannot be handled in such a straightforward manner, since the
  commonly used classical proof technique relies on rewinding the
  possibly dishonest committer, which is in general prohibited by the
  laws of quantum mechanics.

  Another restriction on rewinding occurs when committing to a string
  instead of a single bit. Solutions for proving string commitments
  secure are known for the classical case, but they cannot be adapted
  to the quantum world. Thus, solutions for quantum-secure constant
  round coin-flipping are yet to come (see
  Chapter~\ref{chap:framework} and also
  Section~\ref{sec:primitives.coin.flip}).


\subsection{Oblivious Transfer}
\label{sec:primitives.ot}
\index{oblivious transfer}

  As already indicated, another highly relevant primitive in
  cryptography is oblivious transfer, commonly abbreviated by
  OT. Interestingly, the basic idea for OT was first proposed by
  Wiesner in the context of quantum cryptography, where he suggests
  conjugate coding as ''a means for transmitting two messages either
  but not both of which may be received''~\cite[p.~79]{Wiesner83}. OT
  as a cryptographic concept was then introduced by Rabin ($\tt
  Rabin\text{--}OT$ in~\cite{Rabin81}) and Even, Goldreich, and Lempel
  ($\ot{1}{2}$ in~\cite{EGL85}). OT is a \emph{complete} cryptographic
  primitive, i.e., it is sufficient for secure two-party
  computation~\cite{Kilian88}, meaning that secure 1-2 OT allows for
  implementing any cryptographic two-party functionality.
 
  In this work, we are mainly interested in $\ot{1}{2}^\ell$, i.e.\
  one[message]-out-of-two[messages] oblivious transfer, with message
  length $\ell$. In an $\ot{1}{2}^\ell$ protocol, the sender sends two
  $\ell$-bit strings $s_0$ and $s_1$ to the receiver. The receiver can
  choose which string to receive, i.e.~$s_c$ according to his choice
  bit $c$, but does not learn anything about the other message
  $s_{1-c}$. At the same time, the sender does not learn $c$, i.e., he
  does not learn which string the other party has chosen.

  As in the classical case, quantum OT cannot be implemented without
  any additional restrictions, such as bounded quantum memory in the
  BQSM~\cite{DFSS05, DFRSS07}. However, in contrast to classical OT,
  quantum OT reduces to classical commitment, as already discussed
  before (more in Section~\ref{sec:hybrid.security.ot}).

  $\tt Rand\text{--}OT$ is a randomized variation of general
  $\ot{1}{2}$ and essentially coincides, except that the sender does
  not input the two messages himself, rather they are generated
  uniformly at random during the protocol (and then output to the
  sender). For completeness, we note that $\tt{Rabin\text{--}OT}$ is
  another slightly varied but equivalent version of $\ot{1}{2}$, where
  the sender transmits a message $s$ with probability $1/2$. However,
  he remains oblivious about whether or not the receiver actually got
  $s$. Thus, $\tt{Rabin\text{--}OT}$ can be seen as a \emph{secure
  erasure channel}.

  We conclude this introduction by mentioning two natural
  generalizations of $\ot{1}{2}$. First, $\ot{1}{n}$ allows the
  receiver to obtain exactly one element out of a set of $n$
  elements. This application is similar to private information
  retrieval in database settings but constitutes a stronger notion
  than the latter, as it additionally requires that the user is
  oblivious to all other items (as in database privacy). An even
  further generalization is $\ot{m}{n}$, in which the receiver can
  choose a subset of $m$ elements out of the entire set of size $n$.
  Interestingly, $\ot{1}{n}$ underlies the construction of a quantum
  identification scheme in~\cite{DFSS05}, which exemplifies the
  significance of the primitive. More details on this transformation
  are given in Section~\ref{sec:primitives.id}.


\subsection{Password-Based Identification}
\label{sec:primitives.id}
\index{identification}

  A password-based identification scheme (ID, in short) allows a user
  to identify himself to a server by proving his knowledge of a
  previously agreed secret password. In addition, we will put forward
  the following security requirement: Any party that is not in
  possession of the valid password can (essentially) not succeed by
  any other means but trying to guess. This means that a user without
  password---or in other words, a user who pretends to be someone
  else---cannot delude the server with a probability that exceeds the
  probability of guessing the respective password. Similarly, the
  server can only guess a user's password and then learn whether the
  guess is correct or not---but no information beyond that. This in
  particular implies that the same password may be safely reused in
  further runs of the protocol. Furthermore, our aim is to develop a
  scheme that tolerates a possibly non-uniform password, or in short,
  a realistic user-memorizable password (such as a PIN code) without
  jeopardizing security.

  For reasons of their significance in any authenticated set-up, a
  wide range of classical and quantum ID-schemes can be found in the
  literature (see Section~\ref{sec:hybrid.security.id}). Here, we will
  however focus on the quantum identification scheme, proposed
  in~\cite{DFSS05} and proven secure against any dishonest server with
  bounded quantum storage. Interestingly, in the context of
  primitives, it is constructed out of an extension of a
  \emph{randomized} $\ot{1}{2}^\ell$ to a \emph{randomized}
  $\ot{1}{n}^\ell$. We will briefly sketch the intuitive idea here:
  Recall that such a $\ot{1}{n}^\ell$ supplies the user with $n$
  random $\ell$-bit strings but yields only one of the strings on the
  server's side. Such a scheme can then be used for the purpose of
  identification, when the server ``chooses'' the one specific string
  indexed by the password, and the user proves which of the $n$
  strings obtained is the one with indices matching the password. Note
  that this last step of comparison must be secured by another
  cryptographic technique such as a hash-function and the strings must
  have large Hamming distance, which is not covered by the OT
  application itself. However, by the nature of secure OT, a dishonest
  user does not gain any information on the server's choice and thus,
  does not know which string is the one getting accepted. A dishonest
  server can likewise not do better than guessing a choice, and so the
  string he later receives from the user is most probably random to
  him and hence, contains no information on the password. We want to
  stress again that for simplicity, we skip many subtle but important
  details of the final ID-scheme as well as the means regarding better
  efficiency. More details are given in
  Section~\ref{sec:hybrid.security.id}, where we propose an extension
  of the scheme towards higher and more diverse security.


\subsection{Coin-Flipping}
\label{sec:primitives.coin.flip}
\index{coin-flipping}

  True randomness is a crucial ingredient in cryptographic
  applications. Therefore, coin-flipping (or coin-tossing) is yet
  another essential primitive in this work. Secure coin-flipping
  allows two parties to agree on a uniformly random bit in a fair way,
  which means that neither party can influence the value of the coin
  to his advantage. Intuition suggest that this should be easily
  obtainable for an actual coin-toss if the parties met, flipped a
  coin together and simply looked at the outcome. Now, we want to
  achieve a similar fairness even when the parties are communicating
  over a distance. This problem was first formalized in cryptographic
  terms by Blum as \emph{coin-flipping by telephone}~\cite{Blum81}.

  An ideal coin-flip can be modeled as follows: Each player inputs a
  bit of his choice, independently of each other, and the box then
  outputs the exclusive disjunction of the two bits as the coin. When
  implementing the primitive however, we must consider that one party
  must make a first move during communication, and therefore the other
  one may choose his bit accordingly. The most straightforward way to
  achieve fairness also over a distance is by bit commitments as
  follows. The first player chooses a random bit $x_1$ and commits to
  it, the other one then sends his bit $x_2$ in plain, then the
  commitment is opened, and the resulting coin is $x_1 \oplus
  x_2$. Thus, bit commitment implies secure coin-flipping, since the
  first player is bound to his bit, but can still keep it hidden until
  the second player makes his move.

  Secure implementations for coin-flipping have been proposed also by
  means of quantum communication. For instance, solutions for a
  strong\index{coin-flipping!strong} coin-flip with a potential,
  optimal coin bias of approx.\ $0.2$ and for the
  weaker\index{coin-flipping!weak} notation with arbitrary small
  bias. Note that in the quantum literature, ``strong'' or ``weak''
  indicates weather the dishonest party cannot bias the coin more than
  specified or the dishonest party can influence the coin entirely
  towards one outcome but only by the specified bias towards the other
  value, respectively (see e.g.~\cite{Wehner08} for an overview). We
  want to stress that throughout this work, we use the (intuitive)
  literal interpretation of a ``weak'' and ``strong'' coin, indicating
  its degrees of security.

  We are interested in the standard coin-flipping protocol with
  classical messages exchange, but where the adversary is assumed to
  be capable of quantum computing. Even when basing the embedded
  commitment on a computational assumption that withstands quantum
  attacks, the security proof of the entire coin-flipping and its
  integration into other applications could previously not be
  naturally translated from the classical to the quantum world. We
  will propose a solution based on Watrous' quantum rewinding in
  Chapter~\ref{chap:coin.flip}. Certainly, the desirable protocol
  would be constant round\index{constant round complexity}, meaning
  that a string of coins can be flipped in a constant number of
  rounds, instead of having the number of rounds depending on the
  number of coins. Towards this aim, we present a framework that
  transforms weaker demands on the coins into very strong properties,
  with the final result of a fully simulatable coin-flipping protocol,
  secure against poly-sized quantum adversaries, which can be
  implemented in the plain model from scratch (see
  Chapter~\ref{chap:framework}). On a side note, implementing constant
  round coin-flipping is an open problem in the quantum
  setting. Interestingly, the first quantum application, namely
  quantum key distribution (QKD), enables two parties to produce a
  secret random bit-string (which is then used as a key in symmetric
  crypto-systems). However, by assumption on its purpose, the
  QKD-setting does not have to hold against an internal dishonest
  party. The requirements for secure coin-flipping are much stronger
  in this sense, and it turns out that in a typically QKD-protocol,
  the key could theoretically always be biased by one of the parties.

  We conclude here by stressing the importance of truly random, fair
  coins for cryptographic purposes. Namely, by producing a string of
  coins, the communicating parties can interactively generate a common
  random string from scratch. The generation can then be integrated
  into other (classical or quantum) cryptographic protocols that work
  In the common-reference-string-model. This way, various interesting
  applications can be implemented entirely in a simple manner without
  any set-up assumptions. We will discuss some examples thereof in
  Chapter~\ref{chap:coin.flip.applications}.


\subsection{Zero-Knowledge}
\label{sec:primitives.zk}
\index{zero-knowledge}

  Informally, a zero-knowledge (ZK) proof system is ``both convincing
  and yet yield nothing beyond the validity of the
  assertion''~\cite{Goldreich10}[p.~1]. Thus, only this one bit of
  knowledge is communicated from prover to verifier. Such building
  blocks are typically used in outer cryptographic protocols for
  enforcing that potentially dishonest players behave according to the
  protocol specification, namely, they are required to prove in
  zero-knowledge the correctness of a secret-based action without
  leaking the secret. As examples, we want to mention zero-knowledge
  proofs for Graph Isomorphism and Graph 3-Coloring, proven secure in
  the classical and quantum setting
  by~\cite{GMW91} and~\cite{Watrous09}, respectively. For a survey
  about zero-knowledge, we refer e.g.\
  to~\cite{Goldreich01,Goldreich02,Goldreich10}.

  On a very intuitive level, such proof systems typically proceed in
  several rounds of a protocol. In each round, the prover must answer
  a \emph{challenge} from the verifier which he does not know
  beforehand. In order to be able to answer all challenges in all
  rounds, the prover must know whatever he claims. We differentiate
  between \emph{proofs}\index{zero-knowledge!proofs} and \emph{proofs
  of knowledge}. The respective definitions are given by two
  properties, which vary and are informally stated below. Loosely
  speaking, the distinction between proofs and proofs of knowledge is
  drawn in the content of the assertion: In a proof the prover claims
  the existence of an object. In contrast, in a proof of knowledge, he
  claims knowledge of an object. We stress that a proof of existence
  cannot be modeled via an ideal functionality in the natural way,
  whereas a proof of knowledge can. The third property of
  zero-knowledge does not differ in both systems.
  
  \paragraph{\sc Zero-Knowledge Proofs.}
  \index{zero-knowledge!proofs!completeness}
  \index{zero-knowledge!proofs!soundness}

  Informally, a zero-knowledge proof for set $\cal L$ on common input
  $x$ yields no other knowledge than the validity of membership $x \in
  \cal{L}$, which holds if the following three requirements are
  satisfied. First, if the statement is true, i.e.~$x \in \cal{L}$, an
  honest verifier will be convinced of this fact by an honest prover,
  and thus accept the proof (\emph{completeness}). This holds with
  overwhelming probability. Second, if the statement is false, i.e.~$x
  \notin \cal{L}$, a dishonest prover cannot convince an honest
  verifier of the contrary, except with low probability
  (\emph{soundness}). And last, if the statement is true, a dishonest
  verifier learns nothing beyond this fact
  (\emph{zero-knowledge}). The latter is shown by formally arguing
  that, given only the statement, a simulator can (by itself) produce
  a transcript that is indistinguishable from a real interaction
  between honest prover and dishonest verifier. The degree of
  indistinguishability then specifies the flavor of
  zero-knowledge. Note also that the first two properties are general
  aspects of interactive proof systems. However, in this context, they
  are defined in probabilistic terms, and we require the completeness
  and the soundness error to be negligible, at least after sufficient
  (sequential) repetitions.

  The notion of (interactive) zero-knowledge first appeared
  in~\cite{GMR85} by Goldwasser \emph{et al.} Then in~\cite{GMW86}, it
  was shown that ZK proofs exist for any \NP-language under the
  assumption that commitments exist, which in turn is implied in the
  existence of one-way functions~\cite{Naor91,HILL99}.\footnote{As in
  standard literature, $\NP$ (\emph{non-deterministic polynomial
  time})\index{complexity class!NP ($\NP$)} refers to the set of all
  decision problems, where the "yes"-instances can be recognized in
  polynomial time by a non-deterministic Turing machine. The class
  $\Poly$ (\emph{deterministic polynomial time})\index{complexity
  class!P ($\Poly$)} contains all decision problems which can be solved
  by a deterministic Turing machine in polynomial time. Note that
  every set in $\Poly$ has a trivial zero-knowledge proof in which the
  verifier proves membership by himself.} Blum \emph{et al.}  showed
  that the interaction between prover and verifier in any ZK proof can
  be replaced by sharing a short common reference string available to
  all parties from the start of the protocol~\cite{BFM88}. Note that a
  reference string is a weaker requirement than interaction. The
  requirement for non-interactive zero-knowledge is simpler than for
  general zero-knowledge, since all information is communicated
  mono-directional from prover to verifier. The verifier does not
  influence the distribution in the real world. Thus, in the ideal
  world, we require a simulator that only produces output that is
  indistinguishable from the real distribution of the output. We will
  use such a generic construction in Section~\ref{sec:coin.iqzk},
  where we show a simple transformation from non-interactive
  zero-knowledge to interactive zero-knowledge in the quantum world.

  \paragraph{\sc Zero-Knowledge Proofs of Knowledge.}
  \index{zero-knowledge!proofs of knowledge}  
  \index{zero-knowledge!proofs of knowledge!completeness}
  \index{zero-knowledge!proofs of knowledge!special soundness}

  Intuitively, a zero-knowledge proof of knowledge for relation $\Rel$
  with common instance $x$ and prover's private witness $w$ yields no
  other knowledge to the verifier than the validity of $(x,w) \in
  \Rel$. Especially, it holds that witness $w$ is not leaked. This is
  formulated by the following three requirements. First, if the prover
  follows the protocol and knows $w$, such that $(x,w) \in \Rel$, he
  will always convince the verifier. Note that this holds with
  probability 1, or in other words, \emph{completeness} is defined
  deterministically rather than probabilistically. Second, if the
  (possibly dishonest) prover can with whatever strategy convince the
  verifier to accept, then he knows $w$. This holds, except with
  probability determined by the knowledge error, which again must be
  negligible in the length of the challenge (\emph{special
  soundness}). Note here that in the context of machines, we interpret
  knowledge via behavior. In more detail, to define knowledge, we
  specify a knowledge extractor for which it holds that if the
  extractor can extract $w$ from the prover, for instance, by
  simulating two accepting conversations via rewinding, we say that
  the prover knows $w$. This idea prevents the prover to output the
  knowledge itself, and therewith, the last requirement, i.e.\ the
  property of \emph{zero-knowledge}, capturing that a dishonest
  verifier learns (essentially) nothing, remains unchanged from the
  description above.

  The concept of proofs of knowledge was first introduced also
  in~\cite{GMR85} and formulated in greater detail in~\cite{BG93}. We
  will propose a quantum-secure zero-knowledge proof of knowledge
  based on simulatable witness encoding in
  Section~\ref{sec:coin.zkpk}.

  \paragraph{\sc $\Sigma$-protocols.} 
  A \Index{$\Sigma$-protocol} is a special case of the above, in that
  it is an honest-verifier zero-knowledge proof of knowledge. Such a
  protocol is of three-move-form, starting with the prover's message
  $\tt a^\Sigma$, followed by the verifier's challenge $\tt c^\Sigma$,
  and concluded with the prover's response $\tt z^\Sigma$. Its name
  originates from this form, as the ``$\Sigma$'' visualizes first the
  common input $x$, and then the flow of communication (from top to
  bottom). The flavor of honest-verifier
  zero-knowledge\index{zero-knowledge!honest-verifier} (HVZK),
  although weaker than general zero-knowledge, still allows for useful
  building blocks, which would be impossible to implement with a
  stronger notion in certain settings. As the name suggests, it
  captures a scenario in which, instead of covering any feasible
  verifier strategy, the verifier behaves honest (or rather
  honest-but-curious), and maintains and outputs a transcript of the
  entire interaction.

  By its nature of being a proof of knowledge, special soundness holds
  for a $\Sigma$-protocol, and therewith, that from two accepting
  conversations with different challenges a $w$ can be extracted such
  that $(x,w) \in \Rel$. We will use an honest-verifier simulator as a
  black-box in Sections~\ref{sec:extended.commit.compiler}
  and~\ref{sec:extended.commit.coin} to receive, on input $x$, a valid
  conversation $\big( {\tt a^{\Sigma}, c^\Sigma, z^{\Sigma}}
  \big)$. Intuitively, the purpose of using $\Sigma$-protocols then
  lies in the fact that only one valid conversation could have been
  produced unequivocally without knowing the witness.

  
\subsection{Secure Secret Sharing}
\label{sec:primitives.sss}
\index{secret sharing}

  Secure secret sharing refers---as the name suggests---to a method
  for distributing one secret in several shares amongst the
  players. The secret can only be reconstructed by combining a
  sufficient number of shares (threshold), but any individual share or
  any number of shares below the threshold does not contain any useful
  information on its own.

  Classical secret sharing schemes were introduced independently
  in~\cite{Shamir79} and~\cite{Blakely79}, and quantum secret sharing
  was first proposed in~\cite{HBB99,CGL99}. Classical secret sharing
  is an extremely powerful primitive and is widely used in multi-party
  computation. We will use secret sharing as a building block for
  equipping our mixed commitments with trapdoor openings
  (Section~\ref{sec:mixed.commit.trapdoor.opening}). This extended
  construction will then constitute one essential step in
  bootstrapping fully simulatable coin-flipping from weak
  coin-flipping (Chapter~\ref{chap:framework}).


\chapter{Quantum Tools}
\label{chap:quantum.tools}

\emph{Quantum} refers to a discrete unit of a physical quantity at the
smallest scale, for which quantum mechanics constitutes the underlying
mathematical framework. For the main part of this thesis, we will work
with abstract mathematical objects, as our focus lies on theory, as
opposed to realizing, for instance, a qubit as an actual physical
system such as a ``light quantum'', encoded by polarization of a
photon.

In this chapter, we give an overview of the aspects of quantum
mechanics, essential for this work. The connection between the
mathematical description and physical reality is best reflected in the
postulates of quantum mechanics, which are covered in
Section~\ref{sec:postulates}. This section is also intended to fix the
terminology we will use later on. Next, we will describe distance
measures (Section~\ref{sec:distinguishability}) and uncertainty
measures (Section~\ref{sec:entropies}). Then we will discuss the
concept of information reconciliation and privacy amplification
(Section~\ref{sec:reconciliation.amplification}) as well as the
problems of rewinding in general quantum systems and the technique of
quantum rewinding (Section~\ref{sec:rewinding}). Finally in
Section~\ref{sec:security.definition}, we will introduce the
definitions of security, which underlie all our following main
results.


\section{Postulates and Terminology}
\label{sec:postulates}

We now briefly introduce the field of quantum mechanics on the basis
of its postulates, capturing quantum-physical events and processes
in mathematical formalisms. We will closely follow the descriptions
given in~\cite{NC00} and refer thereto for more details.

  \paragraph{\sc Description of an isolated system.}

  \index{quantum state} A general $d$-dimensional quantum state, where
  $d \in \N$, is described mathematically by a positive semi-definite
  \emph{density matrix} $\rho$ defined in the complex Hilbert space of
  dimension $d$, i.e., a complete inner product space denoted by
  $\Hil_d$. The standard notation to write a \emph{pure quantum
  state}\index{quantum state!pure} is represented in Dirac's
  \emph{bra-ket notation} by a vector as $\ket{\Psi} \in \Hil_d$, and
  is given, for complex coefficients $\alpha_i \in \C$, as
  \begin{eqnarray}
    \ket{\Psi} = \sum_{i = 0}^{d-1} \alpha_i \ket{i} \, . \label{eq.state}
  \end{eqnarray}
  The orthonormal \emph{basis}\index{basis} is denoted by the set
  $\{\ket{0},\ldots,\ket{d-1}\}$, i.e.\ the linearly independent
  spanning set of mutually orthogonal unit vectors. The form of a pure
  state as given in Eq.~\eqref{eq.state} as linear combinations nicely
  reflects an interference phenomenon unique to the quantum world,
  namely the \emph{superposition}\index{superposition} of basis
  states. Informally speaking, it highlights the fact that a quantum
  particle is in all possible basis states at once. And thus, a
  complete description of such a particle must include the description
  of every possible state as well as the probability of the particle
  being in that state, given by $|\alpha_i|^2$ for each respective
  $\ket{i}$. By the normalization condition, the total sum of
  probabilities, i.e.\ $\sum_i |\alpha_i|^2$, equals 1.
  
  A \emph{mixed quantum state}\index{quantum state!mixed} is a
  statistical ensemble of pure states $\{ \lambda_i, \ket{i} \}$,
  where again $\{ \ket{i} \}_i$ forms a basis, and can be represented as
  density matrix by
  \begin{eqnarray}
    \rho = \sum_{i} \lambda_i \proj{i} \, ,
    \label{eq.mixed.state}
  \end{eqnarray}
  with eigenvalues $\lambda_i$ and eigenstates $\ket{i}$. Again, it
  holds that the system is in state $\ket{i}$ with probability
  $\lambda_i$, where $\lambda_i \geq 0$ and, by the normalization
  condition, we have $\sum_i \lambda_i = 1$.

  More specifically, a \emph{qubit}\index{qubit} is a two-dimensional
  pure quantum state living in $\Hil_2$. The \emph{computational
  basis}\index{basis!computational} (also called $\+\,$-basis,
  standard basis, canonical basis, or rectilinear basis) is defined by
  the pair $\{ \ket{0}, \ket{1} \}$, where
  \begin{equation}
    \ket{0} = 
    \left(\begin{array}{c}
      1\\
      0
    \end{array}\right)
    \text{ and } 
    \ket{1} = 
    \left(\begin{array}{c}
      0\\
      1
    \end{array}\right) \, .
  \end{equation}
  The pair $\{ \ket{+}, \ket{-} \}$ denotes the \emph{diagonal
  basis}\index{basis!diagonal} (also named the $\x$-basis or Hadamard
  basis), where
  \begin{align}
    & \ket{+} = (\ket{0}+\ket{1})/\sqrt{2} \ \text{ and}\\
    & \ket{-} =(\ket{0}-\ket{1})/\sqrt{2} \, .
  \end{align}
  Another common denotation is $\{ \ket{0}_+, \ket{1}_+ \}$ for the
  computational basis and $\{ \ket{0}_\x, \ket{1}_\x \}$ for the
  diagonal basis. We use $\{ +,\x \}$ as shorthand to refer to the set
  of these two most commonly used \emph{conjugate bases}.

  \paragraph{\sc Evolution in a closed system.}
  
  The dynamics that apply in a closed systems as described above are
  captured in the description of a \emph{unitary
  transform}\index{transformation!unitary} $\op{U}$. $\op{U}$ is
  unitary, if it holds that $\op{U^\dag U} = \id$. Unitary operations
  preserve inner products between vectors, which yields their more
  intuitive expression in outer product representation as
  follows. Define $\ket{out_i} = \op{U} \ket{in_i}$ to be the
  transformation from ``input'' basis $\{ \ket{in_i} \}_i$ into
  ``output'' basis $\{ \ket{out_i} \}_i$. Then,
  \begin{eqnarray}
    \op{U} = \sum_{i} \ket{out_i}\bra{in_i} \, .
  \end{eqnarray}
  From the requirement of unitarity, it is evident that such a
  transformation must be \emph{reversible}\index{reversible}. That
  means that undoing operation $\op{U}$ on $\ket{in}$ corresponds to
  applying its inverse $\op{U^\dag}$ on $\ket{out}$ and recreates
  $\ket{in}$.
  
  For completeness we note that, although part of this postulate, we
  will not consider the refined version of time evolution, defined by
  the Schr\"odinger equation.

  In the more specific case of single qubits, the transformation from
  the computational basis to the diagonal basis, and vice versa, is
  obtained by applying the \emph{Hadamard
  operation}\index{transformation!Hadamard} $\op{H}$, where
  \begin{equation}
    \op{H} = \frac{1}{\sqrt{2}}
    \left(\begin{array}{lr}
      1 & 1\\
      1 &-1
    \end{array}\right) \, ,
  \end{equation}
  and note that $\op{H} = \op{H}^\dag$.  The two-dimensional
  \emph{Identity operator}\index{transformation!Identity} $\id$ is represented
  by matrix
  \begin{equation}
    \id = 
    \left(\begin{array}{lr}
      1 & 0\\
      0 & 1
    \end{array}\right)\, ,
  \end{equation} 
  other important operations are described by the Pauli matrices
  \begin{equation}
    \op{\sigma_X} = 
    \left(\begin{array}{lr}
      0 & 1\\
      1 & 0    
    \end{array}\right) \ \ \text{and} \ \ 
    \op{\sigma_Z} = 
    \left(\begin{array}{lr}
      1 & 0\\
      0 &-1     
    \end{array}\right) \, .
  \end{equation}
  Operator $\op{\sigma_X}$ describes a
  \emph{bit-flip}\index{transformation!bit-flip}. Matrix
  $\op{\sigma_Z}$ defines a \emph{phase-flip}
  operation\index{transformation!phase-flip}, adding a phase factor of
  \ -1 for non-zero entries, and otherwise leaving the bit
  invariant. For completeness, we also explicitly state
  \begin{equation}
    \op{\sigma_Y} = \left(\begin{array}{lr} 0 &-i\\ i & 0
    \end{array}\right) \, ,
  \end{equation}
  but note that $\op{\sigma_Y} = i\op{\sigma_X\sigma_Z}$.
  
  The \emph{controlled-NOT
  operation}\index{transformation!controlled-NOT} $\op{CNOT}$ is a
  combination of $\id$ and $\op{\sigma_X}$ and is defined for two
  input qubits as
  \begin{equation}
    \op{CNOT} = \left(\begin{array}{lccr} 
      1 & 0 & 0 & 0\\ 
      0 & 1 & 0 & 0\\ 
      0 & 0 & 0 & 1\\
      0 & 0 & 1 & 0
    \end{array}\right) \, .
  \end{equation} 
Thus, if the control qubit is 1, $\op{CNOT}$ flips the target
qubit. Otherwise, $\id$ is applied to the target qubit. Or in other
words, the value of the second output qubit corresponds to the
classical exclusive disjunction (XOR).
  
  \paragraph{\sc Quantum measurements.} 
  \index{quantum measurement} To extract information of a quantum
  system, it must be measured. The following descriptions of
  measurements illustrate the \emph{irreversible}\index{irreversible}
  nature of quantum measurements in general, and therewith, the
  disturbance caused by observation. In other words, some information
  about a state before measurement is lost after measurement. This
  fact stands in sharp contrast to the reversible transformations
  within a closed system as described previously.

  Quantum measurements are described by a collection of measurement
  operators\index{quantum measurement!operator} $\M = \{ \op{M}_m
  \}_m$, where $m$ denotes the measurement outcome. The
  \emph{probability} $\prob{m}$ to obtain outcome $m$ when measuring
  state $\ket{\psi}$ with $\M$ is given by
  \begin{eqnarray}
    \prob{m} = \bra{\psi} \op{M^\dag}_m \op{M}_m \ket{\psi} \, ,
    \label{eq.general.measurement.prob}
  \end{eqnarray}
  with completeness equation $\sum_m \op{M^\dag}_m \op{M}_m = \id$, or
  equivalent, $\sum_m \bra{\psi} \op{M^\dag}_m \op{M}_m \ket{\psi} =
  1$. Conditioned on having obtained $m$, the \emph{post-measurement
  state}\index{quantum measurement!post-measurement state} must be
  renormalized to
  \begin{eqnarray}
    \rho_m = \frac{\op{M}_m \ket{\psi}}{\sqrt{\bra{\psi} \op{M}^\dag_m
    \op{M}_m \ket{\psi}}} \, \ .
    \label{eq.general.measurement.state}
  \end{eqnarray}
  We also want to stress that quantum measurements do not necessarily
  commute, that means that different measurement orders may yield
  different measurement outcomes.

  If all operators $\op{M}_m$ are orthogonal projectors, denoted by
  $\op{P}_m = \op{M}_m^\dag \op{M}_m$, we call the measurement
  \emph{projective}\index{quantum measurement!projective} and $\op{M}
  = \sum_m m \op{P}_m$ its\index{quantum measurement!observable}
  \emph{observable}. The respective probability and post-measurement
  state are then given by
    \begin{eqnarray}
      \prob{m} = \bra{\psi}\op{P}_m\ket{\psi}
    \end{eqnarray}
  and
    \begin{eqnarray}
      \frac{\op{P}_m\ket{\psi}}{\sqrt{\prob{m}}} \, ,
    \end{eqnarray}
  Measuring in basis $\{ \ket{m} \}_m$ means to apply a projective
  measurement defined by projectors $\op{P}_m = \proj{m}$. 

  When only specifying mappings $\op{E}_m = \op{M}_m^\dag \op{M}_m$,
  we obtain an expression in the \emph{positive operator-valued
  measure formalism}\index{quantum measurement!positive
  operator-valued measure} (POVM), similar to
  Eq.~\eqref{eq.general.measurement.prob}, namely,
    \begin{eqnarray}
      \prob{m} = \tr(E_m \rho) \, ,
    \end{eqnarray}
  where $\mathcal{E} = \{ \op{E}_m \}_m$ is the POVM, denoting the set
  of Hermitian operators such that $\sum_m \op{E}_m = \id$ and
  $\op{E}_m \geq 0$. This formalism is simpler than the general
  expressions in Eqs.~\eqref{eq.general.measurement.prob}
  and~\eqref{eq.general.measurement.state}, but sufficient for many
  purposes, as it yields simple measurement statistics. It also
  becomes evident here that for a complete description of measuring
  the observable of a quantum system, the formulation of a quantum
  system must include uncertainty in that the probability for all
  possible outcomes must be encoded in it.

  Again more specifically, measuring a single qubit in the
  computational or diagonal basis\index{basis!computational}
  \index{basis!diagonal} means applying the measurement described by
  projectors $\proj{0}$ and $\proj{1}$ or projectors $\proj{+}$ and
  $\proj{-}$, respectively. We want to point out a very important
  consequence of using such conjugate bases (also called mutually
  unbiased bases). Measuring a qubit, prepared in one of two conjugate
  bases, is equivalent to distinguishing between two non-orthogonal
  quantum states. Non-orthogonal states however cannot be
  distinguished (with arbitrary precision), which can be derived from
  the above formalisms. Thus, any measurement must destroy information
  and therewith disturb the system---except, of course, a measurement
  of a basis state in its own basis. In other words, a state with
  fixed measurement outcome in one basis implies maximal uncertainty
  about the measurement outcome in the other basis.

  \paragraph{\sc Composite systems.}

  The joint state of a \emph{multipartite system} in
  $\Hil_{2^n}^{\otimes n}$ is given by the tensor product
  $\ket{\Psi}_1 \otimes \cdots \otimes \ket{\Psi}_n$. For simplicity,
  we consider a bipartite joint state $\rho_{AB} \in \Hil^A \otimes
  \Hil^B$ shared between Alice and Bob, i.e.,
  \begin{eqnarray}
    \rho_{AB} = \ket{\Psi}_A \ket{\Psi}_B = \sum_i \alpha_i \ket{i}_A
    \sum_j \beta_j \ket{j}_B \, , \label{eq.product.state}
  \end{eqnarray}
  with orthonormal bases $\{ \ket{i}_A \}_i$ for $\Hil^A$ and $\{
  \ket{j}_B \}_j$ for $\Hil^B$. The form of the state in
  Eq.~\eqref{eq.product.state} indicates a \emph{product
  state}\index{product state}, which is \emph{separable}, since it can
  be decomposed into two definite pure states.

  For string $x = (x_1,\ldots,x_n) \in \{0,1\}^n$, encoded in bases
  $\theta = (\theta_1,\ldots,\theta_n) \in \{\+,\x\}^n$, we write
  $\ket{x}_\theta = \ket{x_1}_{\theta_1} \otimes \cdots \otimes
  \ket{x_n}_{\theta_n}$. For $S \subseteq \{ 1, \ldots, n \}$ of size
  $s$, we define $x|_S \in \{ 0,1 \}^s$ and $\theta|_S \in \{ +,\x
  \}^s$ to be the restrictions $(x_i)_{i \in S}$ and $(\theta_i)_{i
  \in S}$, respectively. If all qubits are encoded in the same basis
  $\theta \in [ +,\x ]$, then $\ket{x}_\theta = \ket{x_1 \ldots
  x_n}_\theta$.
  
  In contrast to the product states of Eq.~\eqref{eq.product.state},
  we can also have pure composite systems in some \emph{entangled
  states}\index{entanglement} of the form
  \begin{eqnarray}
    \rho_{AB} &=& \sum_{i,j} \gamma_{ij} \, \ket{i}_A \ket{j}_B
    \label{eq.entangled}
  \end{eqnarray}
  with $\gamma_{ij} \neq \alpha_i \beta_j$. Entangled components mean
  that they can only be described with reference to each
  other. Special cases thereof are the maximally entangled
  \emph{EPR-pairs} (or Bell states):
  \begin{equation}
    \begin{array}{ll} 
    &\ket \Phi_{00} = (\ket{00} + \ket{11}) / \sqrt{2} \, , \\ 
    &\ket \Phi_{11} = (\ket{00} - \ket{11}) / \sqrt{2} \, , \\ 
    &\ket \Phi_{01} = (\ket{01} + \ket{10}) / \sqrt{2} \, , \text{ and} \\ 
    &\ket \Phi_{10} = (\ket{01} - \ket{10}) / \sqrt{2} \, .
    \end{array}
  \end{equation} 
  Important for cryptographic purposes are the following
  observations. First, as Eq.~\eqref{eq.entangled} indicates, upon
  observing one of the two particles, entangled in one single state,
  the system will collapse, and thus, the other particle will at least
  partially be determined---even though the particles may be spatially
  separated. On a side note, the outcome of the first measurement is
  random, and therewith the state, to which the composite system
  collapses into, is so as well. Hence, information (i.e.\ a
  non-random message) cannot be transmitted faster than the speed of
  light by shared entanglement. Second, entanglement is
  basis-independent, e.g.\ $\ket \Phi_{00} = (\ket{00} + \ket{11}) /
  \sqrt{2} = (\ket{++} + \ket{--}) / \sqrt{2} \ $. And last, if an
  entangled state $\rho_{AB}$ is pure, then it cannot be entangled
  with any other state, for instance, one in Eve's hands, so it holds
  that $\rho_{ABE} = \rho_{AB} \otimes \rho_E$. Thus, under the
  assumption of it being pure, entanglement is monogamic.

  Subsystems of a composite system can be described by the
  \emph{reduced density operator} computed by the \emph{partial
  trace}. Let $\rho_{AB} = \big(\ket{a_1}\bra{a_2} \, \otimes \,
  \ket{b_1}\bra{b_2} \big)$ and assume that only subsystem $A$ is
  accessible. Then, we have
  \begin{eqnarray}
    \label{eq.partial.trace}
    \tr_B(\rho_{AB}) 
    = \langle b_2 | b_1 \rangle \, |a_1\rangle\langle{a_2}| \, .
  \end{eqnarray}
  Trivially, when tracing system $B$ out of a product state, we
  have $\rho_{AB} = tr_B(\rho_A \otimes \rho_B) = \rho_A$. However,
  the reduced density operator in an entangled EPR-pair is a complete
  mixture with trace distance $1/2$ (see next
  Section~\ref{sec:distinguishability}). Thus interestingly, the joint
  state of two entangled qubits is pure and can be completely
  determined, yet its subsystems alone are completely mixed.


\section{Distance, Distinguishability, and Dependence}
\label{sec:distinguishability}

We will need various measures to determine the distance between
classical and quantum states. Distance measures possess an important
operational meaning in the context of distinguishability between two
systems.

  \paragraph{\sc Distance.}
  For classical information, the distance between two binary strings
  of equal length can be measured by means of the \emph{Hamming
  distance}\index{Hamming distance} $d_H$, which is the number of
  positions at which the strings differ, or more formally, for strings
  $x,y \in \set{0,1}^n$, we have
    \begin{eqnarray}\label{eq.hamming}
      d_H(x,y) \assign \left|\Set{i}{x_i \neq y_i}\right| \, .
    \end{eqnarray}
  We will also need the\index{Hamming distance!relative}
  \emph{relative Hamming distance}
    \begin{eqnarray}
      \label{eq.relative.hamming}
      r_H(x,y) \assign \frac{d_H(x,y)}{n} \, \ .
    \end{eqnarray}
  For completeness, we note that the \emph{Hamming
  weight}\index{Hamming weight} $w_H$ is the Hamming distance to $x$
  from the all-zero string (of same length), i.e.~$w_H(x) \assign
  \left|\Set{i}{x_i = 1}\right|$.

  In the classical world, the \emph{statistical or variational
  distance}\index{statistical distance} between two classical
  probability distributions $P$ and $Q$ over the same finite set $\X$
  with events $E \subseteq \X$ is determined by
    \begin{eqnarray}
      \delta\big( P,Q \big) \assign \frac{1}{2} \sum_{x \in \X} |P(x)
      - Q(x)| = \max_{E} | P(E) - Q(E)| \, .
      \label{eq.classical.distance}
    \end{eqnarray}
  A measure of proximity is given by the
  \emph{fidelity}\index{fidelity}
    \begin{eqnarray}
      F\big( P,Q \big) \assign \sum_{x \in \X} \sqrt{P(x) Q(x)} \, .
      \label{eq.classical.fidelity}
    \end{eqnarray}

  The classical notions of distance and fidelity can be generalized to
  the distance and proximity of two quantum states $\rho$ and
  $\sigma$. The quantum analogue to the classical distance in
  Eq.~\eqref{eq.classical.distance} is the \emph{trace
  distance}\index{trace distance}, given as
    \begin{eqnarray}
      \delta \big( \rho,\sigma \big) \assign \frac{1}{2} \tr \big(
      |\rho - \sigma| \big)\label{eq.trace.distance} \, ,
    \end{eqnarray}
  where $|A| = \sqrt{A^\dag A}$ is the trace norm of any operator
  $A$. The notion of fidelity translates to \emph{quantum
  fidelity}\index{fidelity!quantum} by
    \begin{eqnarray}
      F \big( \rho,\sigma \big) \assign \tr\sqrt{\sqrt{\rho} \ \sigma
      \sqrt{\rho}} \, .
      \label{eq.quantum.fidelity}
    \end{eqnarray}
  The relation between classical variational distance and quantum
  trace distance can be made more explicit by
    \begin{eqnarray}
      \delta \big( \rho,\sigma \big) = \max_\mathcal{E} \ \delta \big(
      \mathcal{E}(\rho),\mathcal{E}(\sigma)\big) \, ,
      \label{eq.trace.distance.povm}
    \end{eqnarray}
  where the maximum is taken over all POVMs $\mathcal{E}$, and
  $\rho,\sigma$ indicate the probability distributions obtained when
  measuring $\rho$ or $\sigma$ using $\mathcal{E}$. Moreover, it is
  worth pinpointing that, for mixtures of pure quantum states $\rho =
  \sum_i \lambda_i \proj{i}$ and $\sigma = \sum_i \gamma_i \proj{i}$
  with same orthonormal basis $\set{\ket{i}}_i$ but potentially
  different eigenstates $\lambda_i$ and $\gamma_i$, the quantum
  measure naturally reduces to the classical one between the
  eigenvalue distributions $\lambda = \{ \lambda_i \}_i$ and $\gamma
  = \{ \gamma_i \}_i$ by
    \begin{eqnarray}
      \delta(\rho,\sigma) = \frac{1}{2}\tr|\rho - \sigma| 
      = \frac{1}{2} \Big| \sum_i (\lambda_i - \gamma_i) \proj{i} \Big|
      = \frac{1}{2}  \sum_i |\lambda_i - \gamma_i | 
      = \delta(\lambda,\gamma) \, .
    \end{eqnarray}
  A similar reduction can be obtained for the fidelity.
  
  Trace distance and quantum fidelity are, in general, equivalent
  concepts---but with partly different characteristics and
  properties, so we will use one or the other, depending on the
  respective context (see~\cite{FG99} or~\cite{NC00} for a more detailed
  discussion). However, they are closely related in that we have
    \begin{eqnarray}
      1 - F \big( \rho,\sigma \big) \leq \delta \big( \rho,\sigma
      \big) \leq \sqrt{1-F \big( \rho,\sigma \big)^2} \, .
    \end{eqnarray}
  
  For pure states $\rho = \proj{\psi}$ and $\sigma = \proj{\phi}$,
  expressions \eqref{eq.trace.distance} and
  \eqref{eq.quantum.fidelity} simplify to
    \begin{eqnarray}
      \label{eq.trace.distance.pure}
      \delta \big( \rho,\sigma \big) = \sqrt{1- | \langle \psi | \phi
      \rangle |^2} \; \text{ and } \ F \big( \delta, \sigma \big) = |
      \langle \psi | \phi \rangle | \, ,
    \end{eqnarray}
  where the latter can be seen as transition probability. Furthermore,
  the fidelity measure for a pure state $\rho = \proj{\psi}$ and an
  arbitrary quantum state $\sigma$ is given by
    \begin{eqnarray}
      F \big( \rho,\sigma \big) = \sqrt{\bra{\psi}\sigma\ket{\psi}}
      \label{eq.quantum.fidelity.simple} \, ,
    \end{eqnarray}
  and shows that the square root of the overlap between the states
  determines the fidelity. 

  \paragraph{\sc Distinguishability.}
  \index{indistinguishability} The importance of both quantum measures
  is due to their operational meaning of distinguishability. The
  fidelity can be seen as an ``upside down'' trace distance in that
  the limits 0 and 1 in $0 \leq F \big( \rho,\sigma \big) \leq 1$
  meaning perfectly distinguishable and perfectly indistinguishable,
  respectively. In contrast, the trace distance $0 \leq \delta \big(
  \rho,\sigma \big) \leq 1$ increases for decreasing
  indistinguishability, such that we get $\delta \big( \rho,\sigma
  \big) = 0$ for $\rho = \sigma$ and $\delta \big( \rho,\sigma \big) =
  1$ for $\rho$ orthogonal to $\sigma$.

  Coming back to Eq.~\eqref{eq.trace.distance.povm} in this context,
  it is worth noting that the POVM $\mathcal{E}$ that achieves the
  maximum is the optimal POVM for distinguishing $\rho$ and
  $\sigma$. Furthermore, we want to single out two important
  properties by means of the trace distance. First, we have unitary
  invariance with $\delta \big( \rho, \sigma \big) = \delta \big(
  \op{U} \rho \op{U^\dag}, \op{U} \sigma \op{U^\dag} \big)$, meaning
  that the distance between the states does not change when a unitary
  operation $\op{U}$ is applied to both of them. And second, any
  trace-preserving quantum operation $\op{T}$ is contractive
  (monotonicity under quantum operations) with $\delta \big(
  \op{T}(\rho), \op{T}(\sigma) \big) \leq \delta \big( \rho, \sigma
  \big)$. Informally, no physical process can achieve an increased
  distance, or in other words, no modification on the states can help
  to better distinguish two states. An important special case relating
  the partial trace shows that $\delta \big(
  \tr_B(\rho_{AB}),\tr_B(\sigma_{AB}) \big) \leq \delta \big(
  \rho_{AB},\sigma_{AB} \big)$, which again informally states that two
  systems are at least as hard to distinguish when only a part of them
  is accessible.
  
  Two families of probability distributions $\set{P_n}_{n \in
  \naturals}$ and $\set{Q_n}_{n \in \naturals}$ are called
  \emph{perfectly}
  indistinguishable\index{indistinguishability!perfect}, denoted by $P
  \approxp Q$, if their output distributions on each input are
  identical, namely $P_n = Q_n$ for all $n \in \naturals$. In other
  words, an unbounded adversary cannot distinguish the outcomes, which
  holds with probability 1. Relaxing this condition defines
  \emph{statistical}\index{indistinguishability!statistical}
  indistinguishability ($P \approxs Q$), which holds if the
  statistical distance $\delta\big( P_n,Q_n \big)$ is negligible (in
  the length of the input). This covers the setting, in which an
  unbounded adversary cannot distinguish the outcomes, except with
  negligible probability. For $\delta\big( P_n,Q_n \big) \leq \eps$,
  and therewith, indistinguishability except with probability $\eps$,
  we also call the distributions \emph{$\eps$-close}. Thus, perfect
  and statistical indistinguishability are defined in the
  information-theoretic sense and we call the resulting security
  flavor \emph{unconditional}.

  In the computational setting, we require that the two distributions
  cannot be distinguished by any computationally efficient
  procedure. More formally, let \\\noindent$\prob{{\cal A}(P(x)) = 1
  \, | \, x \la P}$ denote the probability that an algorithm $\cal A$
  is successful in that it outputs ``P'', if the input $x$ comes from
  $P$, and analogue for $Q$. To claim
  \emph{computational}\index{indistinguishability!computational}
  indistinguishability between $P$ and $Q$, denoted by $P \approxc Q$,
  for any probabilistic poly-time algorithm $\cal A$, it must hold
  that the (distinguishing) advantage $adv$, i.e.,
  $$ 
  adv ({\cal A}) = \vert \, \prob{{\cal A}(P_n) = 1} - \prob{{\cal
  A}(Q_n) = 1} \, \vert \, ,
  $$ is negligible in the length of the
  input. \emph{Quantum-computational}
  \index{indistinguishability!quantum-computational}
  indistinguishability ($P \approxq Q$) is defined similarly for the
  case of a \emph{quantum} algorithm $\cal A$. In other words,
  (quantum) computational security holds with overwhelming probability
  against a poly-time (quantum) adversary.

  Consider a quantum algorithm consisting of a uniform family $\{
  C_n\}_{n \in \naturals}$ of quantum circuits, which is said to run
  in polynomial time, if the number of gates of $C_n$ is polynomial in
  $n$. Then, two families of quantum states $\set{\rho_n}_{n \in
  \naturals}$ and $\set{\sigma_n}_{n \in \naturals}$ are called
  \emph{perfectly}
  indistinguishable\index{indistinguishability!perfect} with $\rho
  \approxp \sigma$, if $\delta \big( \rho_n,\sigma_n \big) = 0$ in the
  case of unrestricted running time. We have \emph{statistical}
  indistinguishability\index{indistinguishability!statistical} with
  $\rho \approxs \sigma$, if $\delta \big( \rho_n,\sigma_n \big) \leq
  \varepsilon$, for $\varepsilon$ negligible in~$n$, and without any
  restriction on the running time. Again, for $\delta(\rho,\sigma)
  \leq \varepsilon$, we call the quantum states
  $\varepsilon$-close---or indistinguishable, except with probability
  $\varepsilon$. Then, to prove sufficient closeness between an ideal
  system and the real state, we require $\varepsilon$ to be negligible
  (in the security parameter). Last, we have
  \emph{quantum-computationally}
  indistinguishable\index{indistinguishability!quantum-computational},
  denoted by $\rho \approxq \sigma$, if any polynomial-time quantum
  algorithm has negligible advantage $\varepsilon$ of distinguishing
  $\rho_n$ from $\sigma_n$.

  \paragraph{\sc Dependence.}
  \index{independence} We will often use upper case letters for random
  variables (for proofs) that describe respective values (in the
  actual protocol). Let $P_X$ denote the probability distribution of a
  classical random variable $X \in \X$ over finite set $\X$.

  Let 
    \begin{eqnarray}
      \rho_X = \sum_{x \in \X} P_X(x) \proj{x}
      \label{eq:quantum.representation.X}
    \end{eqnarray}
  denote the quantum representation of the classical random variable
  $X$. Let $\rho_E^x$ denote a state in register $E$, depending on
  value $x \in \X$ of random variable $X$ over $\X$ with distribution
  $P_X$. Then, from the view of an observer, who holds register $E$
  but does not know $X$, the system is in state
    \begin{eqnarray}
      \rho_E = \sum_{x \in \X} P_X(x) \rho_E^x  \, ,
      \label{eq.dependence.classical.quantum}
    \end{eqnarray}
  where $\rho_E$ depends on $X$ in the sense that $E$ is in state
  $\rho_E^x$ exactly iff $X = x$.

  \emph{Independence} in a bipartite joint state with classical and
  quantum parts can be expressed as
    \begin{eqnarray}
      \rho_{XE} = \sum_{x \in \X} P_X(x) \proj{x} \otimes \rho_E^x \, .
    \end{eqnarray}
  Such a state is formally called a \emph{cq-state}. Note that
  naturally, $\rho_E = tr_X(\rho_{XE}) = \sum_x P_X(x) \rho_E^x$, and
  that the notation can be extended to states depending on more
  classical variables, i.e.\ \emph{ccq-states}, \emph{cccq-states}
  etc. Full independence of classical and quantum parts within one
  state is given iff $\rho_{E}^x = \rho_E$ for any $x$ and therewith
  $\rho_{XE} = \rho_X \otimes \rho_E$. This means in particular that
  no information on $X$ is gained by observing only $\rho_E$. However,
  full independence is often too strong a requirement. For our
  purposes, it suffices that the real state is close to the ideal
  situation.

  Last in this context, we want to express that a random variable $X$
  is independent of a quantum state $\rho_E$ when given a random
  variable $Y$. Independence in this case means that, when given $Y$,
  the state $E$ gives no additional information on $X$. Yet another
  way to understand \emph{conditional
  independence}\index{independence!conditional} is that $E$ is
  obtained from $X$ and $Y$ by solely processing $Y$. Formally,
  adopting the notion introduced in~\cite{DFSS07}, we require that
  $\rho_{X Y E}$ equals $\rho_{X\leftrightarrow Y \leftrightarrow E}$,
  where the latter is defined as
  \begin{eqnarray}
    \label{eq:conditional.independence}
    \rho_{X\leftrightarrow Y \leftrightarrow E} :=
    \sum_{x,y}P_{XY}(x,y)\proj{x} \otimes \proj{y} \otimes \rho_{E}^y \, ,
  \end{eqnarray}
  where $\rho_{E}^y = \sum_x P_{X|Y}(x|y) \rho_E^{x,y}$. In other
  words, $\rho_{X Y E} = \rho_{X\leftrightarrow Y \leftrightarrow E}$
  precisely if $\rho_E^{x,y} = \rho_E^{y}$ for all $x$ and $y$.


\section{Entropies}
\label{sec:entropies}
\index{entropy}

Entropies are useful measures of ``information, choice and
uncertainty''. We will give a brief recap here, only covering the
concepts most important in the context of this work. For a general
introduction we refer to e.g.~\cite{NC00,Renner05,Schaffner07} for
more details and proofs.

The \emph{Shannon entropy}\index{entropy!Shannon}~\cite{Shannon48}
  \begin{eqnarray}
    H(X) \assign - \log \left( \sum_{x}
    P_X(x) \right) = - \sum_{x} p_x \log p_x
  \end{eqnarray}
applies to a classical probability distribution $P_X$ over $\X$ with
probabilities $p_x$, and as such quantifies the information gain on
average after learning $X$, or complementary, the average uncertainty
before learning $X$.\footnote{Note that the logarithmic base is 2 for
a result in bits.} The binary version thereof, namely the \emph{binary
entropy function}\index{entropy!binary}, is defined for the case of
two possibilities as
  \begin{eqnarray}
    h(\mu) \assign -\big(\mu\log({\mu}) + (1-\mu)\log{(1-\mu)}\big)
  \end{eqnarray}
with $0 \leq \mu \leq \frac{1}{2}$. We will use that, given the ball
of all $n$-bit strings at Hamming distance at most $\mu n$ from $x$,
denoted as $\mathrm{B}^{\mu n}(x)$, we have that $|\mathrm{B}^{\mu
n}(x)| \leq 2^{h(\mu) n}$.

For a cryptographic scenario with not necessarily independent
repetitions, its generalization is given by the \emph{R\'enyi
entropy}\index{entropy!Renyi@R\'enyi}~\cite{Renyi61} of order $\alpha$
as
  \begin{eqnarray}
    H_\alpha(X) = \frac{1}{1-\alpha} \log \left( \sum_{x \in \X}
    P_X(x)^\alpha \right)
  \end{eqnarray}
for $\alpha \geq 0$. Note that the Shannon entropy is the special case
for limit $\alpha \rightarrow 1$. 

The \emph{joint entropy}\index{entropy!joint} of a pair of random
variables $(X_0, X_1)$ measures the total uncertainty about the pair
and is naturally defined as
  \begin{eqnarray}
    H(X_0 X_1) = - \log \left( \sum_{{x_0},{x_1}}
    P_{{X_0}{X_1}}(x_0,x_1) \right) \, .
  \end{eqnarray}
Assume now that $X_1$ is learned, and therewith, $H(X_1)$ bits of
information about $(X_0,X_1)$. Then, the remaining uncertainty of
$X_0$, conditioned on knowing $X_1$, is given by the \emph{conditional
entropy}\index{entropy!conditional}
  \begin{eqnarray}
    H(X_0|X_1) \assign H(X_0 X_1) -  H(X_1) \, .
  \end{eqnarray}

R\'enyi entropies can also be defined for the quantum world, i.e.,
where a density matrix $\rho$ replaces the probability distribution,
and we have
  \begin{eqnarray}
    H_\alpha(\rho) \assign \frac{1}{1-\alpha} \log \big( \tr(\rho^\alpha)
    \big) \, ,
  \end{eqnarray}
for $\alpha \in [0,\infty]$. The \emph{von Neumann}
entropy\index{entropy!von Neumann} is then given by
  \begin{eqnarray}
    H(\rho) \assign - \tr(\rho \log \rho) \, ,
  \end{eqnarray}
which corresponds to the Shannon entropy when measuring quantum state
$\rho_X$ in basis $\{ \ket{x}\bra{x} \}$, or in other words $H(\rho) =
- \sum_x \lambda_x \log \lambda_x$, where $\lambda_x$ are the
eigenvalues of $\rho_X$. Thus, it naturally holds that
$H_\alpha(\rho_X) = H_\alpha(X)$, whenever classical variable $X$ is
encoded in quantum state $\rho_X$.

A special entropy measure is obtained when taking the limit $\alpha
\rightarrow \infty$, namely the
\emph{min-entropy}\index{entropy!min-entropy}. The notion of
min-entropy is used in the context of randomness extraction and
privacy amplification in the presence of a dishonest receiver or an
eavesdropper on the transmission (see
Section~\ref{sec:reconciliation.amplification}). Intuitively, the
(classical) min-entropy is determined by the highest peak in a
distribution, and therewith, describes the maximum amount of
potentially leaked information, which in turn formalizes security for
cryptographic applications in the worst case. In other words, the
min-entropy measures the probability of an adversary's best guess
about an unknown value.

  \begin{definition}[Min-Entropy]
    \label{def:min.entropy}
    Let $X$ be a random variable over alphabet $\X$ with probability
    distribution $P_X$. The min-entropy of $X$ is defined as
    $$\hmin{X} = -\log\bigl( \max_x P_X(x) \bigr) \, .$$
  \end{definition}

Another important special case is the
\emph{max-entropy}\index{entropy!max-entropy} with values for $\alpha$
approaching zero. Its definition captures a R\'enyi entropy, in which
all possible events are considered equally, regardless of their
probabilities. Its operational meaning lies in information
reconciliation (see also
Section~\ref{sec:reconciliation.amplification}).
  \begin{definition}[Max-Entropy]
    \label{def:max.entropy}
    The max-entropy of a density matrix $\rho$ is defined as
    $$\hzero{\rho} = \log\bigl( \rank{\rho}\bigr)\, .$$
  \end{definition}

For completeness, we note that another notion of R\'enyi entropies
with a (non-negative) smoothing parameter $\epsilon$ was introduced in
\cite{Renner05,RW05}\index{entropy!smooth}. Intuitively, it holds that
for two random variables $X_0$ and $X_1$ with almost the same
probability distribution (e.g.\ $X_0 = X_1$ with high probability),
the difference between $H^\epsilon_\alpha(X_0)$ and
$H^\epsilon_\alpha(X_1)$ is small. However, in this work we will only
use the ``un-smoothed'' R\'enyi entropies as discussed above.

Last, we conclude with the following lemma, which we will need in the
context of oblivious transfer. Informally, it states that if the joint
entropy of two random variables $X_0$ and $X_1$ is large, then at
least one of them has half of the original entropy---in a randomized
sense.
  \begin{lemma}[Min-Entropy-Splitting Lemma~\cite{Wullschleger07,DFRSS07}]
    \label{lemma:splitting}
    Let $X_0, X_1$ be random variables with $\hmin{X_0 X_1} \geq
    \alpha$. Then, there exists a binary random variable $K \in \{ 0,1
    \}$ such that $\hmin{X_{1-K}K} \geq \alpha/2$.
  \end{lemma}


\section{Information Reconciliation and Privacy Amplification}
\label{sec:reconciliation.amplification}
\index{information reconciliation}
\index{privacy amplification}

Errors and eavesdropping affect the communication in our quantum
protocols such that the honest parties might end up with bit-strings
of measurement outcomes that differ or have leaked in some
positions. Countermeasures were proposed already in the first
practical implementation of QKD~\cite{BBBSS92}. The honest parties
first \emph{reconcile} their shared data by public discussion to
obtain consistent strings. Note that this process has to be
accomplished without revealing more information than absolutely
necessary to an adversary eavesdropping on the public (classical)
channel. The simplest procedure involves a test on a subset of all
shared (qu)bits to compute the error rate, i.e., the relative number
of all positions with different outcomes. In that case, these publicly
announced bits must later be discarded, which in turn means that more
qubits have to be sent at the beginning of the protocol. According to
the error rate in the testset, error correction must be applied to the
untested remaining set. Since the transmission of qubits is very
efficient in practice and good error correction techniques are known,
we will use this simple technique in our quantum protocols.

After successful reconciliation, the honest parties are in possession
of identical bit-strings. To turn these strings into completely secure
ones, \emph{privacy amplification}~\cite{BBR88} can be applied, which
intuitively distills a shorter but (essentially) private shared
string. More concretely, privacy amplification employs
\emph{two-universal hashing}\index{two-universal hashing} (see
Definition~\ref{def:hashing}) to transform a partially secret string
into a highly secure ``hashed down'' string, about which any adversary
only has negligible information and which looks essentially random to
him. Note that two-universal hashing also works against quantum
adversaries, i.e., in the case when the attacker holds quantum
information about the initial string~\cite{KMR05,RK05,Renner05}. In
fact, it is essentially the only efficient way to perform privacy
amplification against quantum adversaries.
 
  \begin{definition}[Two-Universal Hashing]
  \label{def:hashing}
     A class $\mathcal{F}: \{ 0,1 \}^n \rightarrow \{ 0,1 \}^\ell$ of
     hashing functions is called {\em two-universal}, if for any pair
     $x, y \in \{ 0,1 \}^n$ with $x \neq y$, and $F$ uniformly chosen
     from $\mathcal{F}$, it holds that 
     $$ 
     \prob{F(x)=F(y)} \leq \frac{1}{2^{\ell}} \, .
     $$
  \end{definition}
In the slightly stronger notion of {\em strongly two-universal}
hash-functions\index{two-universal hashing!strong}, we require the
random variables $F(x)$ and $F(y)$ to be independent and uniformly
distributed over $\{ 0,1 \}^\ell$.

Let classical $X$ be correlated with classical part $U$ and quantum
part $E$, i.e., $\rho_{XUE} = \sum_{x \in \{ 0,1 \}^n} P_X(x)
\rho^x_{UE}$. Let $F$ be a hash-function chosen uniformly from
$\mathcal{F}$. After applying $F$ to $X$, we obtain the
\emph{cccq-state} $\rho_{F(X)FUE}$ of form
  \begin{eqnarray}
    \rho_{F(X)FUE} = \sum_{f \in \mathcal{F}} \sum_{z \in \{ 0,1
    \}^\ell} \proj{z} \otimes \proj{f} \otimes \sum_{x \in f^{-1}(z)}
    P_X(x) \rho^x_{UE} \, .
  \end{eqnarray}

The basic theorem for privacy amplification in the quantum world was
introduced in~\cite{RK05} and~\cite{Renner05}, and confined
in~\cite{Schaffner07}. Here, we give the version from~\cite[Corollary
2.25]{Schaffner07} but in its un-smoothed form and tailored to our
context.
  \begin{theorem}[Privacy Amplification]
  \label{theo:privacy-amplification}
    Let $\rho_{XUE}$ be a \emph{ccq-state} with classical $X$
    distributed over $\{ 0,1 \}^n$, classical $U$ in the finite domain
    $\mathcal{U}$, and quantum state $\rho_E$. $U$ and $\rho_E$ may
    depend on $X$. Let $F$ be the random and independent choice of a
    member of a universal-2 class of hash-functions $\mathcal{F}: \{
    0,1 \}^n \rightarrow \{ 0,1 \}^\ell$. Then,
    $$\delta\bigl(\rho_{F(X)FUE},{\textstyle\frac{1}{2^\ell}} \mathbbm{1}
    \otimes \rho_{FUE}\bigr) \leq \frac{1}{2}
    2^{-\frac{1}{2}\big(\hmin{X|U}-\hzero{\rho_E}-\ell\big)} \, .
    $$
  \end{theorem}
Note that if the rightmost term of the theorem is negligible, then we
are in a situation where $F(X)$ is essentially uniform and independent
of $F$ and $E$.


\section{Rewinding}
\label{sec:rewinding}
\index{rewinding}

We require for classical schemes in the quantum world that quantum
computation does not jeopardize the underlying mathematical assumption
that guarantees the security. But we encounter even more setbacks in
the context of actually proving a cryptographic protocol secure in a
quantum environment, which in the realm of this work are mostly due to
the strong restrictions on general rewinding---a common proof
technique for showing the security of different protocols in the
computational setting.


  \subsection{Problems with General Rewinding}
  \label{sec:general.rewinding}
  \index{rewinding!problems}

  Recall that in the context of simulation-based security, we prove
  security against a cheating player by showing that a run of a
  protocol between him and the honest player can be efficiently
  simulated without interacting with the honest player but with a
  simulator instead. Basically, such a simulator prepares a valid
  conversation and tries it on the dishonest party. In case this party
  does not send the expected replies, a classical simulator rewinds
  the machine of the corrupted player to an earlier status and repeats
  the simulation. Note that if the dishonest party sends an invalid
  reply, the simulation is aborted. To conclude the proof, we then
  show that the running time of the simulation as well as the
  distribution of the conversation are according to expectations.

  Such a technique, however, is impossible to justify in the quantum
  world. Generally speaking, the simulator had to partially measure
  the quantum system without copying it beforehand to obtain the
  protocol transcript. But then it would become impossible to
  reconstruct all information necessary for correct rewinding. The
  problem of rewinding in general quantum systems was originally
  observed in~\cite{vdGraaf97}, detailed discussions can also be found
  e.g.\ in~\cite{DFS04,Watrous09}. In the context of this work, there
  are two relevant rewinding settings.

  The first setting applies to simulations intended to collect several
  transcripts of conversations. An example thereof is the classical
  simulation for protocols with embedded \emph{computationally
  binding} commitments. Recall that computational binding means that
  if a dishonest party can open a commitment to two different values,
  then the computational assumption does not hold. In a classical
  simulation, the simulator simulates a run of the outer protocol with
  the committer, such that the latter outputs a valid commitment and
  later provides a correct opening. Now, the simulator has the
  possibility to rewind the player to a point after the commitment was
  sent and repeat the simulation, which can be adapted to the
  simulator's knowledge of the committed value. The event of obtaining
  a different opening for the same commitment in this second run
  implies the inversion of the underlying one-way function, which is
  assumed to be infeasible. In such a simulation, the simulator must
  store the previous transcript before rewinding. Another example of
  this setting occurs when proving \emph{special soundness} in a proof
  of knowledge. There, a classical simulator simulates a run of a
  protocol against a dishonest prover. It then keeps a transcript of
  the simulation and rewinds him. From two accepting conversations,
  the simulator can extract the prover's witness. Again, the simulator
  must store transcripts of the communication before rewinding.

  The second setting requires the simulator to rewind the dishonest
  player to the beginning of a protocol, if the reply from the
  dishonest party does not match the prepared outcome of the protocol
  such that both sides conclude on the ideal values as their
  result. This setting applies, for instance, when proving an outer
  protocol with an embedded \emph{computationally hiding} commitment
  secure. Fortunately, if such a simulation complies with a restricted
  setting, the newly introduced \emph{quantum rewinding lemmas}
  of~\cite{Watrous09} can be applied. Therewith, rewinding is possible
  in a restricted quantum setting. We will discuss this technique in
  more detail in the following section, but in short, it requires a
  one bit reply from the dishonest party (e.g.\ a bit reply to a
  previous bit commitment), the simulation circuit must be unitary,
  and in case of rewinding, we do not intend to keep intermediate
  transcripts nor collect all possible results (see
  Section~\ref{sec:quantum.rewinding}). Unfortunately, we do not know
  how to translate this technique to a multi-bit reply, while keeping
  the running time of the simulator polynomially bounded. In that
  case, the classical simulation would again reduce to the first
  setting above, in which the simulator must store previous
  transcripts, namely a previous message from the dishonest party that
  commits him to his multi-bit reply beforehand.


  \subsection{Quantum Rewinding}
  \label{sec:quantum.rewinding}
  \index{rewinding!quantum}

  Recall that we consider the second setting of the previous
  section. In a classical simulation against dishonest Bob, a
  poly-time simulator guesses, for instance, a valid reply $b'$ of
  dishonest Bob and prepares the protocol transcript according to
  it. When the simulator finally receives Bob's actual reply $b$, it
  checks if the values coincide ($b = b'$), i.e., if its guess was
  correct and therewith, if the simulation was successful. If that is
  not the case, the simulator rewinds Bob and repeats the simulation
  until $b = b'$. No previous information has to be stored nor
  collected.

  Recently, Watrous proposed a quantum analogue of such a simulator
  with the potential of rewinding, and proved therefore, that quantum
  zero-knowledge is possible in an unrestricted model. We will sketch
  the most important aspects of his construction here but refer
  to~\cite{Watrous09} for further details and proofs. More
  specifically, Watrous proved how to construct a quantum
  zero-knowledge proof system for Graph Isomorphism and introduced two
  so-called \emph{quantum rewinding lemmas}; one for an exact setting
  and one that holds for slightly weaker assumptions and therewith
  covers a scenario with perturbations. The investigated protocol
  proceeds as a $\Sigma$-protocol, i.e., a protocol in three-move
  form, where the verifier flips a single coin in the second step and
  sends this challenge to the prover. Thus, the setting applies to the
  case where the reply $b$ from above is a single bit. This will also
  be the case for our simulation in Chapter~\ref{chap:coin.flip}, and
  therefore, we can use Watrous' result in a black-box
  manner. Unfortunately, we do not know how to translate his technique
  to a multi-bit reply, while keeping the running time of the
  simulator polynomially bounded.

  The quantum rewinding procedure is implemented by a general quantum
  circuit $R$, which receives Bob's input registers $\big( W,X \big)$,
  where $W$ contains any $n$-qubit auxiliary input $\ket{\psi}$ and
  $X$ is a working register, initialized to the all-zero state of size
  $k$. As a first step, $R$ applies a unitary quantum circuit $Q$ to
  all registers to simulate the conversations, obtaining as output a
  multi-qubit register $Y$ and a single-qubit register $G$. Register
  $G$ contains the outcome of the $\op{CNOT}$-operation on the
  dishonest party's bit $b$ (as control) and the simulator's guess
  $b'$. Thus, by measuring this register in the computational basis,
  the simulator can determine whether the simulation was successful.

  In more detail, the transformation from $\big( W,X \big)$ to $\big(
  G,Y \big)$ by applying $Q$ can be written as
  \begin{eqnarray*}
    Q \ket{\psi}_{W} \ket{0^{k}}_{X} = \sqrt{p} \ket{0}_{G}
    \ket{\phi_{good}(\psi)}_{Y} + \sqrt{1-p} \ket{1}_{G}
    \ket{\phi_{bad}(\psi)}_{Y}\, ,
  \end{eqnarray*}
  where $0 < p < 1$ and $\ket{\phi_{good}(\psi)}$ denotes the state,
  we want the system to be in for a successful simulation. The qubit
  in register $G$ is then measured with respect to the standard basis,
  which indicates success or failure of the simulation. A successful
  execution (where $b = b'$) results in outcome 0 with probability
  $p$. In that case, $R$ outputs $Y$. A measurement outcome 1
  indicates $b \neq b'$, and hence, implies an unsuccessful
  simulation. In that case, $R$ quantumly rewinds the system by
  applying the reverse circuit $Q^\dag$, and then a phase-flip
  transformation (on register $X$) before another iteration of $Q$ is
  applied, i.e.,
  \begin{eqnarray*}
    && Q \bigg( 2 \Big( \mathbb{I} \otimes
    \ket{0^{k}}\bra{0^{k}} \Big) - \mathbb{I} \bigg)
    Q^\dag \ket{1}_G\ket{\phi_{bad}(\psi)}_Y \\ 
    &=& 2\sqrt{p(1-p)}\ket{0}_G\ket{\phi_{good}(\psi)}_Y + (1-2p)
    \ket{1}_G\ket{\phi_{bad}(\psi)}_Y \, .
  \end{eqnarray*}
  Thus, after this rewinding, the amplitudes of the ``good'' and the
  ``bad'' states are increased and decreased, respectively. Thus, a
  measurement of register $G$ in the computational basis will result
  in outcome $0$ with higher probability $4p(1-p)$. Note that for the
  special case where $p$ equals $1/2$ and is independent of
  $\ket{\psi}$, the simulation terminates after at most one rewinding.

  Watrous' ideal quantum rewinding lemma (without perturbations) then
  states the following: Under the condition that the probability $p$
  of a successful simulation is non-negligible and independent of any
  auxiliary input, $R$ is poly-time and its output $\rho(\psi)$ has
  square-fidelity close to 1 with state $\ket{\phi_{good}(\psi)}$ of a
  successful simulation, i.e.,
  $$
  \bra{\phi_{good}(\psi)}\rho(\psi)\ket{\phi_{good}(\psi)} \geq 1 -
  \eps \, ,
  $$ 
  with error bound $\eps > 0$.

  However, we cannot apply the exact version of Watrous' rewinding
  lemma in our simulation in Chapter~\ref{chap:coin.flip}, since we
  simulate against a dishonest party with an underlying commitment
  that only provides quantum-computational hiding against this
  party. Therefore, we can only claim that the party's input is
  \emph{close to independent} from the probability $p$. In other
  words, we must allow for small perturbations in the quantum
  rewinding procedure and the slightly weaker notion of Watrous'
  quantum rewinding lemma, as stated below, applies.
  \begin{lemma}
    [Quantum Rewinding Lemma with small
    perturbations~\cite{Watrous09}]
    \label{lemma:qrewind}
    Let $Q$ be the unitary $(n,k)$-quantum circuit and let $R$ be the
    general quantum circuit describing the quantum rewinding
    procedure. Let $p_0, q \in (0,1)$ and $\varepsilon \in
    (0,\frac{1}{2})$ be real numbers such that
    \begin{enumerate}
    \item $|p - q| < \eps$
    \item $p_0(1-p_0) \leq q(1-q)$, and
    \item $p_0 \leq p$ 
    \end{enumerate}
    for all n-qubit states $\ket{\psi}$. Then there exists $R$ of size
    $$O \left( \frac{log(1/\varepsilon) size(Q)}{p_0(1-p_0)} \right)$$
    such that, for every $n$-qubit state $\ket{\psi}$, the output
    $\rho(\psi)$ of $R$ satisfies
    \begin{eqnarray*}
      \bra{\phi_{good}(\psi)}\rho(\psi)\ket{\phi_{good}(\psi)} 
      \geq 1 - \eps'
    \end{eqnarray*}
    where $\eps' = 16 \eps \frac{log^2( 1 / \eps)}{p_0^2(1-p_0)^2}$.
  \end{lemma}
 
  Intuitively, Requirement~(1.)~allows for small perturbation between
  the actual probability $p$ and the ideal probability $q$. Thus,
  $\eps$ can be understood as the advantage of the dishonest party. It
  follows that if $\eps$ is negligible, we can argue that $p$ is
  \emph{close} to $q$ and therefore, close to independent of the
  auxiliary input. Probability $p_0$ in Requirement~(3.)~denotes the
  lower bound on the actual probability, for which the procedure
  guarantees correctness and terminates in poly-time. Instead of using
  $p$ in circuit $R$, we use $p_0$. Furthermore, $Q$ is replaced by
  $U$ with $U = VQ$. Lemma~\ref{lemma:qrewind} reflects these
  replacements. On a very intuitive level, the general input state
  $\ket{\psi}$ is analyzed in more detail, i.e.~$\ket{\psi} =
  \sum_{i=1}^{2^n} \alpha_i \ket{\psi_i}$ leading to
  $$ 
  \ket{\phi_{good}(\psi)} = \sum_{i = 1}^{2^n} \alpha_i
  \sqrt{\frac{p(\psi_i)}{p(\psi)}} \ket{\phi_{good}(\psi_i)} \, ,
  $$ and similar for $\ket{\phi_{bad}(\psi)}$. This more detailed
  description allows that in any position, the probability is only
  near-independent of the input. The slight variations must then be
  addressed by an operator $V$, such that $U = VQ$ is close to $Q$ but
  satisfies the exact case of rewinding. In other words, applying $U$
  on the perturbed input state gives the ideal outcome
  \begin{eqnarray*}
    U \ket{\psi}_{W} \ket{0^{k}}_{X} = \sqrt{q} \ket{0}_{G}
    \ket{\delta_{good}(\psi)}_{Y} + \sqrt{1-q} \ket{1}_{G}
    \ket{\delta_{bad}(\psi)}_{Y}\, .
  \end{eqnarray*}
  Transformation $V$ can therewith be understood as a correction. The
  bound in Requirement~(2.)~follows from proof details which will not
  be addressed here. Finally, note that the bounds are not necessarily
  tight. Important for our proof is, however, that all operations can
  be performed by polynomial-size circuits, and thus, the simulator
  has polynomial size (in the worst case). Furthermore, for negligible
  $\eps$, it follows that the ``closeness'' of output $\rho(\psi)$
  with good state $\ket{\phi_{good}(\psi)}$ is slightly reduced, but
  quantum rewinding remains possible and the output $\rho(\psi)$ of
  $R$ has still square-fidelity close to 1 with state
  $\ket{\phi_{good}(\psi)}$ of a successful simulation.


\section{Definition of Security} 
\label{sec:security.definition}
\index{security!definitions} 
\index{composition}
\index{composition!sequential}

We will now define security for our two-party protocols, along the
lines informally described in Section~\ref{sec:flavors}. To this end,
we will work in the framework put forward by Fehr and Schaffner
in~\cite{FS09}. There, they propose general definitions for
correctness and security for any quantum protocol that implements a
\emph{classical non-reactive two-party functionality}, meaning that
in- and output must be classical. We stress that the framework also
allows functionalities which behave differently in case of a dishonest
player. They then show that such a quantum protocol, complying with
the framework, composes \emph{sequentially} in a classical
environment, or in other words, within an outer classical
protocol. Their security definitions are phrased in simple
information-theoretic conditions, depending on the functionality,
which implies strong simulation-based security. For the sake of
simplicity, the framework does not assume additional entities such as
e.g.\ an environment, without of course compromising correctness in
the given setting.

Throughout this work, we are interested in quantum and classical
protocols that implement classical functionalities. As already
mentioned, such primitives are often used as building blocks in more
complicated classical (multi-party) protocols which implement more
advanced tasks. Therefore, it can be justified in
Part~\ref{part:quantum.cryptography} to restrict the focus to such
quantum protocols that run in a classical environment and have
classical in- and outputs. Furthermore, although the framework was
originally proposed for quantum protocols that compose in a classical
environment, we adapt it here for classical protocols against quantum
attacks, composing equally well when imposing the suggesting
restriction regarding the in- and outputs. Thus, we will use it also
in Part~\ref{part:cryptography.in.quantum.world} for defining security
of our classical protocols.

Although various other security and composition frameworks have been
proposed (such as~\cite{BM04,Unruh04,Unruh10,WW08}), we consider the
security level achieved in this framework as a reasonable balance
between weak demands and yet meaningful security. Furthermore, its
structure is as simple and clear as possible and compliance with the
definitions gives us sequential composition. Towards a general
composition, we must, of course, extend the basic protocols as shown
in Sections~\ref{sec:composability.compiler}
and~\ref{sec:composability.coin}.

We will now introduce the framework more formally for a general
functionality. We will use information-theoretic definitions in our
notions of unconditional security as investigated in~\cite{FS09}. In
addition, we will also show that computational security can be defined
similarly, although with some modifications.


\subsection{Correctness}
\label{sec:security.definition.correctness}

  A protocol $\Pi$ consists of an infinite family of interactive
  quantum circuits for players Alice and Bob, indexed by the security
  parameter. For instance, in our quantum protocols this security
  parameter $m$ corresponds to the number of qubits transmitted in the
  first phase. However, to ease notation, we often leave the
  dependence on the security parameter implicit.

  Since we assume the common input state $\rho_{UV}$ to be classical,
  i.e.,
  $$\rho_{UV} = \sum_{u,v} P_{UV}(u,v) \proj{u} \otimes \proj{v} \, ,
  $$ for some probability distribution $P_{UV}$, we understand $U,V$
  as random input variables. The same holds for the classical output
  state $\rho_{XY}$ with output $X,Y$. The input-output behavior of
  the protocol is uniquely determined by $P_{XY|UV}$, and we write
  $\Pi(U,V) = (X,Y)$. Then, a classical non-reactive two-party ideal
  functionality $\F$ is given by a conditional probability
  distribution $P_{\F(U,V)|UV}$ with $\F(U,V)$ denoting the
  ideal-world execution, where the players forward their inputs $U,V$
  to $\F$ and output whatever they obtain from $\F$. The definition of
  correctness of a protocol is now straightforward.
    \begin{definition}[Correctness] 
      \label{def:correctness}
      A protocol $\Pi(U,V) = (X,Y)$ correctly implements an ideal
      classical functionality $\F$, if for every distribution of the
      input values $U$ and $V$, the resulting common output satisfies
      \[ (U,V,X,Y) \approxs (U,V, \F(U,V)) \, .
      \]
    \end{definition}


\subsection{Information-Theoretic Security}
\label{sec:security.definition.unconditional}

  We define information-theoretic security based on~\cite[Proposition
  4.3]{FS09}. Note that in the following, we simplify the joint output
  representation (compared to~\cite{FS09}) in that we denote the
  output in the real world by $out_{\A,\B}^\Pi$ (which is equivalent
  to $\Pi_{\A,\B}\rho_{UV}$), and the output in the ideal world by
  $out_{\hA,\hB}^\F$ (equivalent to $(\F_{\hA,\hB}) \rho_{UV}$).

  Recall that $U$ denotes honest Alice's classical input, and let $Z$
  and $V'$ denote dishonest Bob's classical and quantum
  information. Then, any input state $\rho_{U Z V'}$ is restricted to
  be of form $$
  \rho_{\MC{U}{Z}{V'}} = \sum_{u,z} P_{UZ}(u,z) \proj{u} \otimes
  \proj{z} \otimes \rho_{V'}^z \, ,
  $$ where it holds here that $\rho_{V'}^z = \rho_{V'}^{u,z}$. This
  implies that Bob's quantum part~$V'$ is correlated with Alice's part
  only via his classical~$Z$.

    \begin{definition}
      [Unconditional security against dishonest Bob]
      \label{def:qmemoryBob} 
      A protocol $\ \Pi$ implements an ideal classical functionality
      $\F$ unconditionally securely against dishonest Bob, if for any
      real-world adversary $\dB$, there exists an ideal-world
      adversary $\dhB$ such that, for any input state with
      $\rho_{UZV'} = \rho_{\MC{U}{Z}{V'}}$, it holds that the outputs
      in the real and the ideal world are statistically
      indistinguishable, i.e.,
      $$
      out_{\A,\dB}^\Pi \approxs out_{\hA,\dhB}^\F \, .
      $$
    \end{definition}
  For completeness, we state these output states explicitly, i.e.,
  $$
  out_{\A,\dB}^\Pi = \rho_{UXZY'} \ \text{ and } \ 
  out_{\hA,\dhB}^\F = \rho_{\MC{UX}{Z}{Y'}} \, ,
  $$ 
  which shows that Bob's possibilities in the ideal world are
  limited. He can produce some classical input $V$ for $\F$ from his
  quantum input state $V'$, and then he can obtain a quantum state
  $Y'$ by locally processing $V$ and possibly $\F$'s classical reply
  $Y$. This description is also depicted in
  Figure~\ref{fig:protocol.functionality}.
  
  \begin{figure}
    \begin{framed}
      \noindent\hspace{-1.5ex} {\sc Protocol $\Pi_{\A,\dB}$
	\hspace{29ex}Functionality $\F_{\hA,\dhB}$} \\    
      \begin{center}
	\vspace{-3ex}
	\input{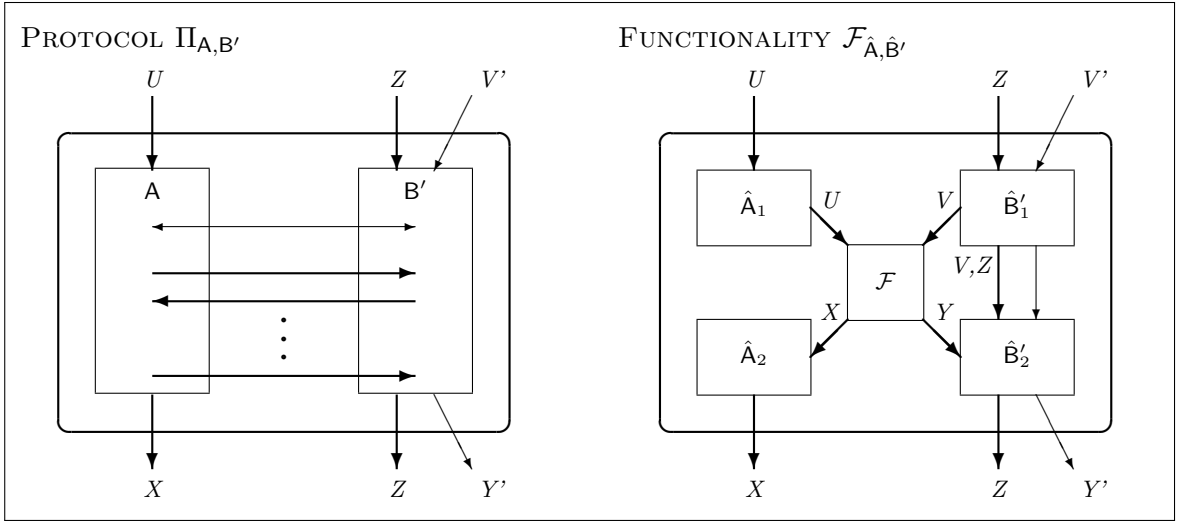}
	\vspace{-2ex}
      \end{center}
    \end{framed}
    \vspace{-1.5ex}
    \caption{Real World vs.\ Ideal World~\cite{Chris10}.}
    \label{fig:protocol.functionality}
  \end{figure}

  Analogously, we can define unconditional security for honest Bob
  against dishonest Alice. In this case, we assume a classical $Z$ and
  a quantum state $U'$ as dishonest Alice's input and a classical
  input $V$ of honest Bob.

    \begin{definition}
      [Unconditional security against dishonest Alice]
      \label{def:unboundedAliceNiceOrder}
      A protocol $\ \Pi$ implements an ideal classical functionality
      $\F$ unconditionally securely against dishonest Alice, if for
      any real-world adversary $\dA$, there exists an ideal-world
      adversary $\dhA$ such that, for any input state with
      $\rho_{U'ZV} = \rho_{\MC{U'}{Z}{V}}$, it holds that the outputs
      in the real and the ideal world are statistically
      indistinguishable, i.e.,
      $$
      out_{\dA,\B}^\Pi \approxs out_{\dhA,\hB}^\F \, .
      $$
    \end{definition}

  Note that in the definitions above, we do not require the running
  time of ideal-world adversaries to be polynomial whenever the
  real-life adversaries run in polynomial time. This way of defining
  unconditional security can lead to the (unwanted) effect that
  unconditional security does not necessarily imply computational
  security. However, as mentioned before, by extending our basic
  constructions we can achieve efficient ideal-life adversaries.

  Intuitively, the composition theorem below states that if quantum
  protocols $\pi_1\cdots\pi_\ell$ securely implement ideal
  functionalities $\F_1\cdots\F_\ell$, then a protocol
  $\Sigma^{\pi_1\cdots\pi_\ell}$ is \emph{essentially} as secure as a
  classical hybrid protocol $\Sigma^{\F_1\cdots\F_\ell}$ with
  sequential calls to $\F_1\cdots\F_\ell$. Note that for the hybrid
  protocol to be {\em classical}, we mean that it has classical in-
  and output (for the honest players), but also that all communication
  between the parties is classical.\footnote{We want to stress that a
  \emph{hybrid protocol} is a protocol that makes sequential calls to
  ideal functionalities. This term should not be confused with the
  notion of \emph{hybrid security} in
  Chapter~\ref{chap:hybrid.security}, which refers to quantum
  protocols providing twofold security in case of an adversary who is
  either bounded in quantum storage or bounded in
  quantum-computational power.} The above facts imply that such
  protocols compose sequentially. Below, we state (a simplified
  variant of) the theorem in~\cite{FS09}. We omit its proof here but
  note that it proceeds along similar lines as the proof of
  Theorem~\ref{thm:composition.computational}, translating sequential
  composition to the case of computational security.

    \begin{theorem}[Composition Theorem I~\cite{FS09}]
      \label{thm:composition.unconditional}
      Let $\Sigma^{\F_1\cdots\F_\ell}$ be a classical hybrid protocol
      which makes at most $k$ calls to $\F_1\cdots\F_\ell$, and for
      every $i \in \set{1,\ldots,\ell}$, let protocol $\pi_i$ be an
      $\eps$-secure implementation of $\F_i$ against $\dAlice$ and
      $\dBob$. Then the output of $\Sigma^{\pi_1\cdots\pi_\ell}$ is at
      distance at most $O(k\eps)$ to the output produced by
      $\Sigma^{\F_1\cdots\F_\ell}$.
    \end{theorem}

We want to explicitly state here that if the hybrid protocol is
secure, then so is the real-life protocol, and as such it could itself
be use as a sub-protocol in yet another classical outer protocol.

    \begin{corollary}
      If $\Sigma^{\F_1\cdots\F_\ell}$ is a $\delta$-secure
      implementation of $\cal G$ against $\dAlice$ and $\dBob$, and if
      $\pi_i$ is an $\eps$-secure implementation of $\F_i$ against
      $\dAlice$ and $\dBob$ for every $i \in \set{1,\ldots,\ell}$,
      then $\Sigma^{\pi_1\cdots\pi_\ell}$ is a $O(\delta +
      \eps)$-secure implementation of $\cal G$.
  \end{corollary}


\subsection{Computational Security}
\label{sec:security.definition.computational}

  One can define security against a computationally bounded dishonest
  Bob in the CRS-model analogously to information-theoretic security,
  with the two differences that the input given to the parties has to
  be sampled by an efficient quantum algorithm and that the output
  states of Definition~\ref{def:qmemoryBob} should be computationally
  indistinguishable. Recall that in the CRS-model, all participants in
  the real world have access to a classical public common reference
  string~$\crs$, which is chosen before any interaction starts,
  according to a distribution only depending on the security
  parameter. On the other hand, the participants in the ideal-world
  execution $\F_{\hA,\hB}$, interacting only with the ideal
  functionality, do not make use of string~$\crs$. Hence, an
  ideal-world adversary $\dhB$ that operates by simulating a
  real-world adversary $\dB$ is free to choose $\crs$ in any way he
  wishes.

  In order to define computational security against a dishonest Bob in
  the CRS-model, we consider a polynomial-size quantum circuit, called
  \emph{input sampler}, which takes as input the security parameter
  and the common reference string $\crs$ (chosen according to its
  distribution) and which produces the input state $\rho_{U Z
  V'}$. Again, $U$, $Z$, and $V'$ denote Alice's classical, Bob's
  classical, and Bob's quantum information, respectively, and we
  require from the input sampler that $\rho_{U ZV'} =
  \rho_{\MC{U}{Z}{V'}}$. In the following, we let $\dBobPoly$ be the
  family of all {\em polynomial-time strategies} for dishonest Bob.

    \begin{definition}
      [Computational security against dishonest Bob]
      \label{def:polyboundedBobCRS} 
      A protocol $\ \Pi$ implements an ideal classical functionality
      $\F$ computationally securely against dishonest Bob, if for any
      real-world adversary $\dB \in \dBobPoly$, who has access to the
      common reference string $\crs$, there exists an ideal-world
      adversary $\dhB \in \dBobPoly$, not using $\crs$, such that, for
      any efficient input sampler as described above, it holds that
      the outputs in the real and the ideal world are
      quantum-computationally indistinguishable, i.e.,
      $$
      out_{\A,\dB}^\Pi \approxq out_{\hA,\dhB}^\F \, .
      $$
    \end{definition}

  Protocols fulfilling the definition above provide sequential
  composition in a naturally weaker but otherwise similar notion as
  unconditionally secure protocols. We can therefore adapt the
  original composition theorem to the case of computational
  security. For completeness, we will include its proof as given
  in~\cite{DFLSS09}.

  Consider a dishonest $\dB$ and the common state \smash{$\rho_{U_j
  V'_j}$} at any point during the execution of the hybrid protocol
  when a call to functionality $\F_i$ is made. The requirement for the
  oracle protocol to be \emph{classical} is now expressed in that
  there exists a classical $Z_j$---to be understood as consisting of
  $\dhB$'s classical communication with $\hA$ and with the $\F_{i'}$'s
  up to this point---such that given $Z_j$, Bob's quantum state $V'_j$
  is not entangled with Alice' classical input and auxiliary
  information: \smash{$\rho_{U_j Z_j V'_j} =
  \rho_{\MC{U_j}{Z_j}{V'_j}}$}. Furthermore, we require that we may
  assume $Z_j$ to be part of $V'_j$ in the sense that for any $\dhB$
  there exists $\dhB'$ such that $Z_j$ is part of $V'_j$. This
  definition is motivated by the observation that if Bob can
  communicate only classically with Alice, then he can entangle his
  quantum state with information on Alice's side only by means of the
  classical communication.

  We also consider the protocol we obtain by replacing the ideal
  functionalities by quantum two-party sub-protocols
  $\pi_1\cdots\pi_\ell$ with classical in- and outputs for the honest
  parties, i.e., whenever $\Sigma^{\F_1\cdots\F_\ell}$ instructs $\hA$
  and $\hB$ to execute $\F_i$, they instead execute
  $\pi_i$ and take the resulting outputs. We then write
  $\Sigma^{\pi_1\cdots\pi_\ell}$ for the real quantum protocol we
  obtain this way.

  Recall that we require from the input sampler that $\rho_{U ZV'} =
  \rho_{\MC{U}{Z}{V'}}$, i.e., that $V'$ is correlated with Alice's
  part only via the classical~$Z$. When considering classical hybrid
  protocols $\Sigma^{\pi_1\cdots\pi_\ell}$ in the real world, where
  the calls are replaced with quantum protocols using a common
  reference string, it is important that every real protocol $\pi_i$
  uses a separate instance (or part) of the common reference string
  which we denote by $\crs_i$.

    \begin{theorem}[Composition Theorem II]
      \label{thm:composition.computational}
      Let $\Sigma^{\F_1\cdots\F_\ell}$ be a classical two-party hybrid
      protocol which makes at most $k=\poly(n)$ calls to the
      functionalities, and for every $i \in \set{1,\ldots,\ell}$, let
      protocol $\pi_i$ be a computationally secure implementation of
      $\F_i$ against $\dBobPoly$.

      Then, for every real-world adversary $\dB \in \dBobPoly$ who
      accesses the common reference string $\crs=\crs_1,\ldots,
      \crs_k$ there exists an ideal-world adversary $\dhB \in
      \dBobPoly$ who does not use $\crs$ such that for every efficient
      input sampler, it holds that the outputs in the real and the
      ideal world are quantum-computationally indistinguishable, i.e.,
      $$ 
      out_{\A,\dB}^{\Sigma^{\pi_1\cdots\pi_\ell}} \approxq
      out_{\hA,\dhB}^{\Sigma^{\F_1\cdots\F_\ell} } \, .
      $$
    \end{theorem}

  Note that we do not specify what it means for the hybrid protocol to
  be secure. In fact, Theorem~\ref{thm:composition.computational}
  guarantees that {\em whatever} the hybrid protocol achieves, an
  indistinguishable output is produced by the real-life protocol with
  the functionality calls replaced by protocols. Of course, if the
  hybrid protocol {\em is} secure in the sense of
  Definition~\ref{def:polyboundedBobCRS}, then so is the real-life
  protocol.

  \begin{corollary}
    If $\Sigma^{\F_1\cdots\F_\ell}$ is a computationally secure
    implementation of $\cal G$ against $\dBobPoly$, and if $\pi_i$ is
    a computationally secure implementation of $\F_i$ against
    $\dBobPoly$ for every $i \in \set{1,\ldots,\ell}$, then
    $\Sigma^{\pi_1\cdots\pi_\ell}$ with at most $k=\poly(n)$ oracle
    calls is a computationally secure implementation of~$\cal G$
    against $\dBobPoly$.
  \end{corollary}

  \begin{proof} 
\def\mark#1{\bar{#1}}

    We prove the claim in Theorem~\ref{thm:composition.computational}
    by induction on $k$. If no calls are made, we can set $\dhB
    \assign \dB$ and the claim holds trivially. Consider now a
    protocol $\Sigma^{\F_1\cdots\F_\ell}$ with at most $k > 0$ oracle
    calls. For simplicity, we assume that the number of oracle calls
    equals $k$, otherwise we instruct the players to make some ``dummy
    calls''. Let $\rho_{U_k Z_k V'_k}$ be the common state right
    before the $k$-th, and thus, last call to one of the sub-protocols
    $\pi_1,\ldots,\pi_\ell$ in the execution of the real protocol
    $\Sigma^{\pi_1,\ldots,\pi_\ell}$. To simplify notation in the rest
    of the proof, we omit the index $k$ and write
    \smash{$\rho_{\mark{U} \mark{Z} \mark{V}'}$} instead (see
    Figure~\ref{fig:composition.proof}). We know from the induction
    hypothesis for $k-1$ that there exists an ideal-world adversary
    $\dhB \in \dBobPoly$ not using the common reference string such
    that $\rho_{\mark{U} \mark{Z} \mark{V}'} \approxq \sigma_{\mark{U}
    \mark{Z} \mark{V}'}$ where $\sigma_{\mark{U} \mark{Z} \mark{V}'}$
    is the common state right before the $k$-th call to a
    functionality in the execution of the hybrid protocol
    $\Sigma^{\F_1\cdots\F_\ell}$ with input $\rho_{U Z V'}$. As
    described, $\mark{U}$ and $\mark{Z},\mark{V}'$ are to be
    understood as follows. $\mark{U}$ denotes $\A$'s (respectively
    $\hA$'s) input to the sub-protocol (respectively functionality)
    that is to be called next. $\mark{Z}$ collects the classical
    communication dictated by $\Sigma^{\F_1\ldots,\F_\ell}$ as well as
    $\dhB$'s classical inputs to and outputs from the previous calls
    and $\mark{V}'$ denotes the dishonest player's current quantum
    state. Note that the existence of $\mark{Z}$ is guaranteed by our
    formalization of {\em classical} hybrid protocols and
    $\sigma_{\mark{U} \mark{Z} \mark{V}'} =
    \sigma_{\MC{\mark{U}}{\mark{Z}}{\mark{V}'}}$.

    Let $\crs_i$ be the common reference string used in protocol
    $\pi_i$. For simplicity, we assume that the index $i$, which
    determines the sub-protocol $\pi_i$ (or functionality~$\F_i$) to
    be called next, is {\em fixed} and we just write $\pi$ and $\F$
    for $\pi_i$ and $\F_i$, respectively.

\begin{figure}
  \begin{framed}
    \centering 
    \input{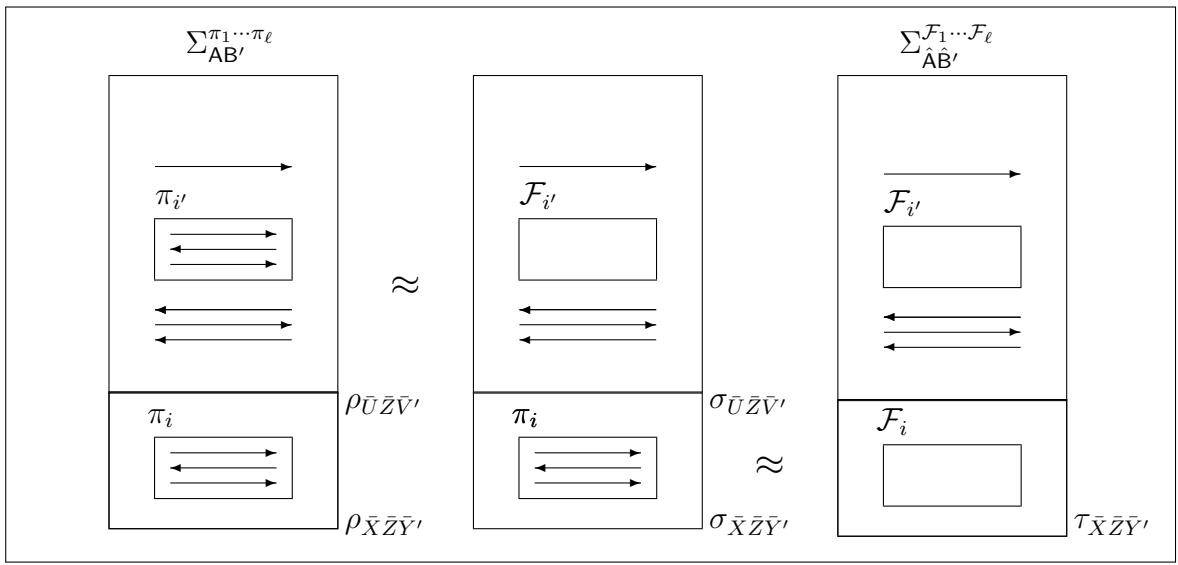}
  \end{framed}
  \caption{Steps of the Composability Proof} 
  \label{fig:composition.proof} 	
\end{figure}

    It follows from Definition~\ref{def:polyboundedBobCRS} of
    computational security that there exists $\dhB \in \dBobPoly$
    (independent of the input state) not using $\crs_i$ such that the
    corresponding output states $\sigma_{\mark{X}\mark{Z}\mark{Y}'}$
    and $\tau_{\mark{X}\mark{Z}\mark{Y}'}$ produced by $\F_{\hA,\dhB}$
    (as prescribed by the oracle protocol) and $\pi_{\A,\dB}$ run on
    the state $\sigma_{\mark{U} \mark{Z} \mark{V}'} =
    \sigma_{\MC{\mark{U}}{\mark{Z}}{\mark{V}'}}$ are
    quantum-computationally indistinguishable.

    The induction step is then completed with
    $$ out_{\A,\dB}^{\Sigma^{\pi}} = \rho_{\mark{X}\mark{Z} \mark{Y}'}
    = (\pi_{\A,\dB})\, \rho_{\mark{U}\mark{Z}\mark{V}'} \approxq
    (\pi_{\A,\dB})\,\sigma_{\mark{U}\mark{Z}\mark{V}'} =
    \sigma_{\mark{X}\mark{Z}\mark{Y}'} \approxq
    \tau_{\mark{X}\mark{Z}\mark{Y}'} = out_{\hA,\dhB}^{\Sigma^{\F}} \,
    , $$ 

    where $(\pi_{\A,\dB})\, \rho_X$ should be understood as running
    protocol $\pi_{\A,\dB}$ with input $\rho_X$.

    Note that the strategy of $\dhB$ does not depend on the state
    $\sigma_{\mark{U}\mark{Z}\mark{V}'}$, and hence, the overall
    ideal-world adversary $\dhB$ does not depend on the input state
    either. Furthermore, the concatenation of two polynomially bounded
    players is polynomially bounded, i.e.~$\dhB \in \dBobPoly$.
  \end{proof}


\clearemptydoublepage
\part{Quantum Cryptography}
\label{part:quantum.cryptography}


\chapter{Introduction}
\label{chap:intro.quantum.cryptography}

In this part of the thesis, we present our research in quantum
cryptography, which offers a secure alternative to some conventional
cryptographic schemes that are rendered insecure by the potential emerge
of large-scale quantum-computing. We also want to mention an actual
implementation of quantum protocols within the research project
MOBISEQ (``Mobile Quantum Security''), which is a joint project of the
cryptology group from the computer science department and the iNano
center at the physics department, both at Aarhus University. The main
goal of MOBISEQ is the development of technology for secure quantum
communication that can compete with conventional methods on
practicality, velocity and security and that can be integrated into
existing infrastructures. However, at the time of writing, the
implementation is still ``under construction''.

In the next sections, we will introduce the concept of mixed
(classical) commitment schemes, since they are an important underlying
construction in our quantum protocols.

In Chapter~\ref{chap:hybrid.security}, we discuss our main result on
improving the security of quantum protocols via a commit\&open
step, based on these mixed commitments. We first introduce the setting
and then propose a general compiler therein. We further show that the
construction remains secure in the case of noisy communication. We
then proceed with combining the compilation technique with the
bounded-quantum-storage model. Last, we show sequential composability
and further use the extended commitment construction, discussed in
Section~\ref{sec:extended.commit.compiler}, towards a more general
composition. 

In Chapter~\ref{chap:hybrid.security.applications}, we discuss
that the compiler can be applied to known protocols and show two
example applications, with the result of achieving hybrid-secure
protocols.


\section{Mixed Commitments}
\label{sec:mixed.commit}
\index{commitment!mixed}

Commitments were introduced on an intuitive level in
Section~\ref{sec:primitives.commitment} and capture the process of a
party being committed to his message by the binding characteristic
without immediately revealing it to the other party due to the hiding
aspect.


\subsection{Motivation}
\label{sec:mixed.commit.motivation}

  Our compiler construction in the following chapters requires a
  classical yet quantum-secure commitment from $\B$ to $\A$. Since we
  aim at preserving the unconditional security against $\A$ in the
  outer quantum protocols, the commitment can only be
  quantum-computationally binding. As described in
  Section~\ref{sec:rewinding}, the standard reduction from the
  computational security of the protocol to the computational binding
  property of the commitment would require rewinding $\dB$, which is
  not possible in the assumed protocol scenario.

  Therefore, we construct keyed commitment schemes, which are special
  in that they are \emph{mixed commitments} or \emph{dual-mode
  commitments}.\footnote{The notions are interchangeable. The term of
  mixed commitments was introduced in~\cite{DN02}. In~\cite{DFLSS09},
  the name dual-mode commitments was used to relate to the notion of a
  dual-mode crypto-system~\cite{PVW08}, which is similar in spirit,
  but slightly more involved. Last we want to mention that our schemes
  are similar to the commitment schemes used in~\cite{DFS04} but with
  extensions.} Generally speaking, the notion of mixed commitments
  requires some trapdoor information, related to the common reference
  string and given to the simulator in the ideal world. This trapdoor
  provides the possibility for \emph{extracting} information out of
  the commitments\index{commitment!extractability}, which finally
  allows us to circumvent the necessity of rewinding $\dB$. We will
  discuss this in detail in
  Section~\ref{sec:mixed.commit.idea}. Additionally, we require that
  the basic mathematical assumption, which guarantees the hiding and
  binding properties of the commitments, withstands quantum
  attacks. We will propose an actual instantiation in
  Section~\ref{sec:mixed.commit.instantiation}.


\subsection{Idea}
\label{sec:mixed.commit.idea}
  
  Recall that a keyed bit or string commitment $C =
  \commitk{m}{r}{pk}$ takes as input a message $m$ and some randomness
  $r$ of size polynomial in the security parameter, as well as a
  public key $pk$. The message $m$ can be a single bit $b$ for the
  implementation of bit commitments or, in order to achieve string
  commitments, a bit-string $m = b_0,\ldots,b_s$. In order to open the
  commitment, message $m$ and random variable $r$ are sent in plain
  and the receiver therewith checks the correctness of $C$. Hiding is
  typically formalized by the requirement $\big(
  pk,\commitk{m_1}{r_1}{pk} \big) \approx \big(
  pk,\commitk{m_2}{r_2}{pk} \big)$ with different flavors of
  indistinguishability, while binding prohibits that there exist
  $C,m_1,r_1,m_2,r_2$, such that $m_1 \neq m_2$, but
  $\commitk{m_1}{r_1}{pk} = C = \commitk{m_2}{r_2}{pk}$.

  We construct our commitments in the CRS-model such that they provide
  dual modes depending on the public key. In more detail, let $\tt
  commitK = (\commit, \GH, \GB, \xtrx)$ denote a (keyed) mixed
  commitment scheme. The commitment key $pk$ is generated by one of
  the two possible key-generation algorithms, $\GH$ or
  $\GB$. Generator $\GB$ takes as input the security parameter
  $\kappa$ and generates a key pair $(pk,sk) \la \GB$, where $pk \in
  \zo^\kappa$ is a public key and $sk$ is the corresponding secret
  key. $\xtrx$ is a poly-time extraction algorithm that takes $sk$ and
  $C$ as input and produces $m$ as output, i.e., $\xtr{C}{sk}=
  \xtr{\commitk{m}{r}{pk}}{sk} = m$, which must hold for all pairs
  $(pk,sk)$ generated by $\GB$ and for all values $m,r$. In other
  words, the secret key $sk$ allows to efficiently extract $m$ from
  $C$, and as such the commitment is unconditionally binding. We often
  denote this type of key therefore by $\pkB$. For a key $pk \la \GH$,
  the commitment scheme is unconditionally hiding (and we often refer
  to this type as $\pkH$). Furthermore, we need the unconditionally
  binding key $\pkB$ and the unconditionally hiding key $\pkH$ to be
  computationally indistinguishable even against quantum attacks,
  i.e., $\pkB \approxq \pkH$.

  We want to stress that we can even weaken the assumption on the
  hiding key in that we merely require that there exists a public-key
  encryption scheme where a random public key looks pseudo-random to
  poly-time quantum circuits. Thus, $\commit$ does not require actual
  unconditionally hiding keys, but we can use uniformly random strings
  from $\zo^\kappa$ as such. This is feasible in our proposed
  construction, sketched below, and still provides unconditional
  hiding, except with negligible probability. This fact also ensures
  that most keys of a specific domain are in that sense
  unconditionally hiding keys.

  Finally, to avoid rewinding we use the following proof method: In
  the real-world protocol, $\B$ uses the unconditionally hiding key
  $\pkH$ to maintain unconditional security against any unbounded
  $\A$. To argue security against a computationally bounded $\dB$, an
  information-theoretic argument involving the simulator $\dhB$ is
  given to prove that $\dB$ cannot cheat with the unconditionally
  binding key $\pkB$. Security in real life then follows from the
  quantum-computational indistinguishability of $\pkH$ and $\pkB$.


\subsection{Instantiations}
\label{sec:mixed.commit.instantiation}
  
  As a candidate for instantiating our commitment construction, we
  propose the lattice-based public-key encryption scheme of
  Regev~\cite{Regev05}. The crypto-system is based on the
  (conjectured) hardness of the learning with error (LWE) problem,
  which can be reduced from worst-case hardness of the approximation
  of the shortest vector problem (in its decision version). Thus,
  breaking Regev's crypto-system implies an efficient algorithm for
  approximating the lattice problem in the worst-case, which is
  assumed to be hard even with quantum computing power.

  In more detail, the crypto-system uses dimension $n$ as security
  parameter and is para\-metrized by two integers $m$ and $p$, where
  $p$ is a prime bounded by $n^2 \leq p \leq 2n^2$, and a probability
  distribution on $\mathbb{Z}_p$. A regular public key (in
  $\mathbb{Z}_p^{m \times n}$) for Regev's scheme is proven to be
  quantum-computationally indistinguishable from the case where a
  public key is chosen from the uniform distribution, and therewith,
  independently from a secret key. In this case, the ciphertext
  carries essentially no information about the message~\cite[Lemma
  5.4]{Regev05}. This proof of semantic security for Regev's
  crypto-system is in fact the property we require for our commitment,
  as the public key of a regular key pair can be used as the
  unconditionally binding commitment key $\pkB$ in the ideal-world
  simulation. Then, for the real protocol, an unconditionally hiding
  commitment key $\pkH$ can simply be constructed by uniformly
  choosing numbers in $\mathbb{Z}_p^{n \times m}$. Both public keys
  will be of size $\tilde{O}(n^2)$, and the encryption process
  involves only modular additions, which makes its use simple and
  efficient.\footnote{The notation $\tilde{O}(\cdot)$ is similar to
  the asymptotic Landau notation $O(\cdot)$ but ignores logarithmic
  factors.}
 
  For simplicity and efficiency, we use a common reference string,
  which allows us to use Regev's scheme in a simple way and, since it
  is relatively efficient, we get a protocol that is potentially
  practical. More specifically, in the CRS-model we assume the key
  $\pkB$ for the commitment scheme, generated by $\GB$, to be
  contained in the common reference string. We want to stress however
  that we show in Part~\ref{part:cryptography.in.quantum.world},
  Section~\ref{sec:key.generation.coin}, how to avoid the CRS-model at
  the cost of a non-constant round construction, where we let the
  parties generate a common reference string jointly by coin-flipping.
  
  For the compiler construction here, we will use Regev's original
  version, as we require bit commitments. However, a multi-bit variant
  of Regev's scheme is given in the full version of~\cite{PVW08}. All
  requirements as described above are maintained in this more
  efficient variant, which improves the performance of Regev's scheme
  by essentially a factor of $n$, e.g., the scheme can encrypt $n$
  bits using $\tilde{O}(n)$ bits. We use later in
  Part~\ref{part:cryptography.in.quantum.world},
  Chapter~\ref{chap:framework}, that this implies that we can flip a
  $\lambda$-bit string using $O(\lambda)$ bits of communication when
  $\lambda$ is large enough. We also rely on this multi-bit version
  for our extended commitment construction, which we will describe in
  the next Section~\ref{sec:extended.commit.compiler} and then use in
  Section~\ref{sec:general.composition.compiler}, where we show how to
  achieve efficient simulation also against a dishonest $\dA$.


  \subsection{Extended Construction}
  \label{sec:extended.commit.compiler}
  \index{commitment!extended} 

  To achieve efficient simulation against both players, i.e.\
  additional efficient simulation also against $\dA$ (in
  Section~\ref{sec:general.composition.compiler}), we need to extend
  our commitments by yet another trapdoor, which provides the
  commitment with
  \emph{equivocability}\index{commitment!equivocability}. Intuitively,
  this means that we now enable the simulator in the ideal world that
  it can construct commitments equivocally such that it can open them
  later to different bits. As we still need in addition the properties
  of the mixed commitment scheme of
  Section~\ref{sec:mixed.commit.idea} in its multi-bit variant, we
  will build the new scheme around it, such that its trapdoor can
  still be used for extraction.

  The new extension is based on the idea of UC-commitments~\cite{CF01}
  and requires a $\Sigma$-protocol for a (quantumly) hard relation
  $\Rel = \{(x,w)\}$, i.e.\ an honest-verifier perfect zero-knowledge
  proof of knowledge with instance $x$ and witness $w$ (see also
  Section~\ref{sec:primitives.zk}). Conversations are of form $\tt
  \big( a^{\Sigma}, c^{\Sigma}, z^{\Sigma} \big)$, where the prover
  sends $\tt a^{\Sigma}$, the verifier challenges him with bit $\tt
  c^{\Sigma}$, and the prover replies with $\tt z^{\Sigma}$. For
  practical candidates of $\Rel$, see e.g.~\cite{DFS04}. By special
  soundness, it holds that from two accepting conversations with
  different challenges, i.e.\ $\big( {\tt a^\Sigma},0,{\tt z^\Sigma_0}
  \big)$ and $\big( {\tt a^\Sigma},1,{\tt z^\Sigma_1} \big)$, the
  simulator can extract $w$ such that $(x,w) \in \Rel$.

  In real life, the common reference string consists of commitment key
  $\pkH$ and instance $x$. To commit to a bit $b$, the committer $\B$
  first runs the honest-verifier simulator to get, on input $x$, a
  conversation $\big( {\tt a^{\Sigma}}, b, {\tt z^{\Sigma}}
  \big)$. Then, he commits by sending $\big( {\tt a_\Sigma}, C_0, C_1
  \big)$, where $C_b= \commitk{{\tt z^\Sigma}_b}{r_b}{\pkH}$ and
  $C_{1-b} = \commitk{0^{z'}}{r_{1-b}}{\pkH}$ with randomness
  $r_b,r_{1-b}$ and $z' = |{\tt z^\Sigma}|$. To open a commitment,
  $\B$ reveals $b$ and opens $C_b$ by sending ${\tt z^\Sigma}_b$,
  $r$. The receiver checks that $\big({\tt a^{\Sigma}}, b, {\tt
  z_{\Sigma}} \big)$ is a valid conversation and that $C_b$ was
  correctly opened. Assuming that the $\Sigma$-protocol is
  honest-verifier perfect zero-knowledge and $\pkH$ provides
  unconditional hiding, the new commitment construction is again
  unconditionally hiding.

  In the ideal world, we assume that the simulator (simulating against
  $\dA$) knows $w$ such that $(x,w) \in \Rel$ (and public key
  $\pkH$). Therewith, it can compute two valid conversations $\big(
  {\tt a^\Sigma},0,{\tt z^\Sigma}_0 \big)$ and $\big( {\tt
  a^\Sigma},1,{\tt z^\Sigma}_1 \big)$ and set $C_0= \commitk{{\tt
  z^\Sigma}_0}{r_0}{\pkH}$ and $C_1= \commitk{{\tt
  z^\Sigma}_1}{r_1}{\pkH}$. This enables to open both ways, assuming
  the knowledge of the trapdoor $w$.

  We maintain extraction, since in the respective simulation against
  $\dB$, the public key is chosen in a different but indistinguishable
  way, namely as $(x,\pkB)$, where $\pkB$ is the binding commitment
  key, generated together with $sk$. Now, given a commitment
  $(a,C_0,C_1)$, the simulator can decrypt $C_0,C_1$ to determine
  which of them contains a valid reply ${\tt z^\Sigma}_b$ of the
  $\Sigma$-protocol. The only way this could fail is in the case where
  both $C_0$ and $C_1$ contain valid replies, which would imply that
  the committer $\dB$ could compute a valid $w$. For a polynomial-time
  bounded committer and a (quantumly) hard relation $\Rel$, however,
  this can occur only with negligible probability.


\clearemptydoublepage
\chapter{Improved Security for Quantum Protocols}
\label{chap:hybrid.security}

Here, we propose a general compiler\index{compiler} for improving the
security of two-party quantum protocols, implementing different
cryptographic tasks and running between mutually distrusting players
Alice and Bob. The compiler extends security against an ``almost
honest'' adversary by security against an arbitrary computationally
bounded (quantum) adversary. Furthermore, we can achieve hybrid
security such that certain protocols can only be broken by an
adversary who has large quantum memory {\em and} large computing
power. The results in this chapter are joint work with Damg{\aa}rd,
Fehr, Salvail and Schaffner, and appeared in~\cite{DFLSS09}.


\section{Motivation}
\label{sec:motivation.hybrid.security}

Our proposed compiler applies to a large class of quantum protocols,
namely to so-called \emph{BB84-type} protocols that follow a
particular but very typical construction design for quantum
communication. Our main result states that if the original protocol is
secure against a so-called\index{player!benign} \emph{benign} Bob who
is only required to treat the qubits ``almost honestly'' but can
deviate arbitrarily afterwards, then the compiled protocol is secure
against a \emph{computationally bounded} quantum Bob. The
unconditional security against Alice that BB84-type protocols usually
achieve is preserved during compilation and it requires only a
constant increase of transmitted qubits and classical messages.

In other words, with our compiler, one can build a protocol for any
two-party functionality by designing a protocol that only has to be
secure if Bob is benign, which is a relatively weak assumption. On the
other hand, many protocols following the BB84-type pattern (at least
after some minor changes) have been proposed, e.g.\ for Oblivious
Transfer, Commitment, and Password-Based Identification
\cite{CK88,DFSS08,DFRSS07,DFSS07}. Typically, their proofs go
through under our assumption. For instance, our compiler can easily be
applied to existing quantum protocols implementing ID and OT, which we
will show as example applications in
Chapter~\ref{chap:hybrid.security.applications}.

In more detail, the compiler incorporates the mixed commitment scheme,
discussed in Section~\ref{sec:mixed.commit}, into the basic protocols
with Bob as committer. Recall that we need such a mixed commitment to
preserve the unconditional security against Alice that BB84-type
protocols typically achieve but cannot apply the typical reduction
from the computational security of the protocol to the computational
binding property of the commitment, due to the restrictions on
rewinding in the quantum world (see Section~\ref{sec:rewinding}). The
idea of introducing a (plain) commitment in quantum protocols has
already been sketched in other works, for instance,
in~\cite{CK88,BBCS91}. Furthermore, there are partial results,
investigating this scenario, e.g.~\cite{Yao95,CDMS04,Mayers96}. We
will go into more details of preceding work in
Section~\ref{sec:hybrid.security.ot}.

Previously, it was very unclear what exactly such a $\CO$-step would
achieve in the quantum world. The intuition is clearly that if Bob
passes the test, he must have measured most of his qubits, also in the
remaining untested subset. But---to our best knowledge---it was never
formally proven that the classical intuition also holds for a quantum
Bob. We now give a full characterization of $\CO$ in our quantum
setting, namely that it forces Bob to be benign, for which we propose
a formal definition and which might be of independent interest. These
aspects are covered in Section~\ref{sec:compiler}. In this context, we
want to mention the follow-up work in~\cite{BF10}. They phrase the
$\CO$-approach more clearly as the quantum version of classical
sampling, and additionally, investigate sampling in quantum settings
more generally.

In Section~\ref{sec:compiler.noise}, we generalize our result to noisy
quantum communication. Furthermore, security in the
bounded-quantum-storage model that assumes the adversary's quantum
storage to be of limited size, implies benign security. Therefore by
compilation of such protocols, we can achieve hybrid security, which
means that the adversary now needs \emph{both} large quantum memory
\emph{and} large quantum computing power to break these new
protocols. The preservation of BQSM-security allows us to get security
properties that classical protocols cannot achieve, if the assumption
on the limited quantum memory holds---which definitely is the case
with current state-of-the-art
(Section~\ref{sec:compiler.hybrid.security}). However, if the
assumption should fail and the adversary could perfectly store all
qubits sent, the known protocols can be easily broken. Thus, by
applying our compiler, we obtain another security layer that equips
such protocols with additional quantum-computational security. Last,
we sketch that the compiled protocols in their basic form remain
sequentially composable. Moreover, by using the extended commitment
construction of Section~\ref{sec:extended.commit.compiler}, we achieve
efficient simulations on both sides, and therewith, a more general
composition. This result is discussed in
Section~\ref{sec:composability.compiler}.
 

\section{Introducing $\CO$}
\label{sec:compiler}

We now discuss our compiler construction in detail, starting from
describing the form of BB84-type protocols and formalizing our notion of
benignity. Then, we show the transformation from benign security
towards computational security and conclude with its proof.


\subsection{Initial situation}
\label{sec:initial.compiler}

  We consider quantum two-party protocols that follow a particular but
  very typical construction design. These protocols consist of two
  phases, called \emph{preparation} and \emph{post-processing}
  phase. We call such a protocol a \emph{BB84-type}
  protocol\index{BB84-type}, as they have the same structure and the
  same encoding scheme as the first (complete) quantum protocol by
  Bennett and Brassard in 1984 for quantum key
  distribution~\cite{BB84}. However, we want to stress again that we
  are interested in protocols for cryptographic tasks other than key
  distribution, and therewith, we also consider the case of dishonest
  players. A generic BB84-type protocol $\Pi$ is specified in
  Figure~\ref{fig:BB84-type}.
  
  \begin{figure}
    \begin{framed}
      \noindent\hspace{-1.5ex} {\sc Protocol $\Pi$}  \\[-4ex]
      \begin{description}\setlength{\parskip}{0.5ex}
      \item[{\it Preparation:}]$ $ \\ 
	$\A$ chooses $x \in_R \set{0,1}^n$ and $\theta \in_R
	\set{\+,\x}^n$ and sends $\ket{x}_{\theta}$ to~$\B$, and $\B$
	chooses $\hat{\theta} \in_R \set{0,1}^n$ and obtains $\hat{x}
	\in \set{0,1}^n$ by measuring $\ket{x}_{\theta}$ in bases
	$\hat{\theta}$.
      \item[{\it Post-processing:}]$ $\\ 
	Arbitrary classical communication and classical local
	computations.
      \end{description}
      \vspace{-1.5ex}
    \end{framed}
    \vspace{-1.5ex}
    \caption{The Generic BB84-type Quantum Protocol $\Pi$. } 
    \label{fig:BB84-type} 
  \end{figure}	

  In the preparation phase, Alice transmits $n$ random BB84-qubits to
  Bob. More specifically, Alice chooses a random bit string $x=
  x_1,...,x_n$ and a random basis-string $\theta
  =\theta_1,...,\theta_n$ from a set of two conjugate bases, encodes
  her qubits accordingly, i.e., $x_i$ is encoded in the state of the
  $i$th particle using basis $\theta_i$, and sends them to Bob. Bob
  chooses a basis-string $\hat{\theta} =
  \hat{\theta}_1,..,\hat{\theta} _n$ and measures the $i$th particle
  in basis $\hat{\theta}_i$. If Bob plays honestly, he learns $x_i$
  whenever the bases match, i.e.\ $\hat{\theta}_i=
  \theta_i$. Otherwise, he gets a random independent result. The
  second phase of the protocol, the post-processing, consist of
  arbitrary classical messages and local computations, depending on
  the task at hand.

  However, the fact that all BB84-type protocols have in common is
  that the classical post-processing typically relies on Bob's subsets
  of correct and random outcomes, or in other words, on the fact that
  a dishonest Bob has high uncertainty about a crucial piece of
  information. Thus, BB84-type protocols---in their basic form---may
  be broken by a dishonest Bob, who does not measure the qubits
  immediately. This is due to the fact that Alice typically reveals
  $\theta$ at a later stage so that Bob knows the correct
  subset. However, a dishonest Bob could measure all stored qubits in
  matching bases $\hat{\theta} = \theta$, and thus, learn more
  information than he was supposed to.

  This aspect is captured in our definition of security against a
  \emph{benign} Bob, or more precisely a ``benignly dishonest'' Bob,
  who treats the qubits ``almost honestly'' in the preparation phase
  but can deviate arbitrarily otherwise. Note that, in contrast to
  Bob's situation, BB84-type protocols typically achieve unconditional
  security against cheating by Alice in their default form. On a very
  intuitive level, it should now be evident that we want to enforce
  Bob's measurement upon qubit reception before any further
  announcement by Alice. In the next section, we will make this
  definition more formal.


\subsection{Security against Benign Bob} 
\label{sec:defbenign.compiler}
\index{player!benign}

  The following security definition captures information-theoretic
  security against a benign Bob. Recall that such a dishonest Bob is
  benign in that, in the preparation phase, he does not deviate (too
  much) from what he is supposed to do. In the post-processing phase,
  though, he may be arbitrarily dishonest.

  To make this description formal, we fix an arbitrary choice of
  $\theta$ and an arbitrary value for the classical information, $z$,
  which Bob may obtain as a result of the preparation phase (i.e.~$z =
  (\hat{\theta},\hat{x})$ in case Bob is actually honest). Let $X$
  denote the random variable describing the bit-string $x$, where we
  understand the distribution $P_X$ of $X$ to be conditioned on the
  fixed choices for $\theta$ and~$z$. Furthermore, let $\rho_E$ be the
  state of Bob's quantum register $E$ after the preparation
  phase. Note that, still with fixed $\theta$ and~$z$, $\rho_E$ is of
  the form $\rho_E = \sum_x P_X(x) \rho^x_E$, where $\rho^x_E$ is the
  state of Bob's quantum register in case $X$ takes on the value
  $x$. In general, the $\rho^x_E$ may be mixed, but we can think of
  them as being reduced pure states with $\rho^x_E = \tr_R \big(
  \proj{\psi^x_{ER}} \big)$ for a suitable register $R$ and pure
  states $\ket{\psi^x_{ER}}$. We then call the state $\rho_{ER} =
  \sum_x P_X(x) \proj{\psi^x_{ER}}$ a point-wise purification (with
  respect to $X$) of $\rho_E$. Obviously, in case Bob is honest, $X_i$
  is fully random whenever $\theta_i \neq \hat{\theta}_i$, and we have
  $$\Hmin\bigl(X|_I \,\big|\,X|_{\bar{I}} = x|_{\bar{I}}\bigr) =
  d_H\bigl(\theta|_I,\hat{\theta}|_I\bigr) \, ,
  $$ 
  for every $I \subseteq \set{1,\ldots,n}$ and every $x|_I$, where
  $\bar{I}$ denotes the complementary set. In that case, Bob does not
  store any non-trivial quantum state so that $R$ is ``empty'' and
  $$
  \Hmax(\rho_{ER}) = \Hmax(\rho_E) = 0 \, .
  $$

  A benign Bob $\dB$ is now specified to behave close-to-honestly in
  the preparation phase in that, after the preparation, he produces an
  auxiliary output $\hat{\theta}$. Given this output, we are in a
  certain sense close to the ideal situation where Bob really measured
  in basis $\hat{\theta}$ as far as the values of $\Hmin\bigl(X|_I
  \,\big|\,X|_{\bar{I}} = x|_{\bar{I}}\bigr)$ and $\Hmax(\rho_{ER})$
  are concerned.\footnote{The reason why we consider the point-wise
  purification of $\rho_E$ is to prevent Bob from artificially blowing
  up $\Hmax(\rho_{ER})$ by locally generating a large mixture or
  storing an unrelated mixed input state.} Informally speaking, the
  following definition states (under Point~(1.)) that there exists a
  string $\hat{\theta}$ of $\dB$'s measurement bases, such that the
  uncertainty about $\A$'s bit $x_i$ is essentially 1 whenever
  $\theta_i \neq \hat{\theta_i}$. Furthermore, $\dB$'s quantum storage
  is small.\index{entropy!min-entropy}\index{entropy!max-entropy}
    \begin{definition}
      [Unconditional security for Alice against {\em benign} Bob]
      \label{def:BenignBob}
      A BB84-type quantum protocol $\Pi$ securely implements $\F$
      against a {\em $\beta$-benign} $\dB$ for some parameter $\beta
      \geq 0$, if it securely implements $\F$ according to
      Definition~\ref{def:qmemoryBob}, with the following two
      modifications:
      \begin{enumerate}
        \item The quantification is over all $\dB$ with the following
	property: After the preparation phase $\dB$ either aborts, or
	else produces an auxiliary output $\hat{\theta} \in
	\set{\+,\x}^n$. Moreover, the joint state of $\A$ and $\dB$
	after $\hat{\theta}$ has been output is statistically
	indistinguishable from a state for which it holds that, for
	any fixed values for $\theta$, $\hat{\theta}$ and $z$, for any
	subset $I \subseteq \set{1,\ldots,n}$, and for any
	$x|_{\bar{I}}$,
	\begin{equation}\label{eq:staterequirements}
	  \Hmin\bigl(X|_I \,\big|\,X|_{\bar{I}} = x|_{\bar{I}}\bigr)
	  \geq d_H\bigl(\theta|_I,\hat{\theta}|_I\bigr) - \beta n
	  \qquad\text{and}\qquad \Hmax\bigl(\rho_{ER}\bigr) \leq \beta n
	  \, ,
	\end{equation}
	where $\rho_{ER}$ is a point-wise purification of $\rho_E$
	with respect to $X$.
        \item $\dhB$'s running time is polynomial in the running time
        of $\dB$.
      \end{enumerate} 
    \end{definition}


\subsection{From Benign to Computational Security}
\label{sec:begnign.computational.compiler}

  We now show a generic compiler which transforms any BB84-type
  protocol into a new quantum protocol for the same task. The compiler
  achieves that, if the original protocol is unconditionally secure
  against dishonest Alice and unconditionally secure against
  \emph{benign} Bob, then the compiled protocol remains to be
  unconditionally secure against dishonest Alice but is now
  \emph{computationally secure} against an \emph{arbitrary} dishonest
  Bob.

  The idea behind the construction of the compiler is to incorporate a
  commitment scheme and force Bob to behave benignly by means of the
  $\CO$-procedure. More precisely, we let Bob classically and
  position-wise commit to all his measurement bases and outcomes. Then
  Alice chooses a random test-subset of size $\alpha m$ and checks by
  Bob's openings that the bits coincide whenever the bases match. If
  the test is passed, the post-processing is conducted on the
  remaining unopened positions. Otherwise, Alice
  aborts. Figure~\ref{fig:compiled} shows the compilation of an
  arbitrary BB84-type protocol $\Pi$. The quantum communication is
  increased from $n$ to $m = n/(1-\alpha)$ qubits, where $0 < \alpha <
  1$ is an additional parameter that can be arbitrarily chosen, and
  the compiled protocol requires three more rounds of interaction.

  \begin{figure}
    \begin{framed}
      \noindent\hspace{-1.5ex} {\sc Protocol $\compile(\Pi)$ :}
      \\[-4ex]
      \begin{description}\setlength{\parskip}{0.5ex}
      \item[{\it Preparation:}] $ $\\
    	$\A$ chooses $x \in_R \set{0,1}^m$ and $\theta \in_R
	\set{\+,\x}^m$ and sends $\ket{x}_{\theta}$ to~$\B$. Then,
	$\B$ chooses $\hat{\theta} \in_R \set{0,1}^m$ and obtains
	$\hat{x} \in \set{0,1}^m$ by measuring $\ket{x}_{\theta}$ in
	bases $\hat{\theta}$. 
      \item[{\it Verification:}] 
	\begin{enumerate}
	\item[]
	\item
	  $\B$ commits to $\hat{\theta}$ and $\hat{x}$ position-wise
	  by $c_i := \commitx{(\hat{\theta}_i,\hat{x}_i)}{r_i}$
	  with randomness $r_i$ for $i = 1,\ldots,m$. He sends the
	  commitments to $\A$.
	\item\label{step:check} 
	  $\A$ sends a random test subset $T \subset \{1,\ldots,m \}$
	  of size $\alpha m$. $\B$ opens $c_i$ for all $i \in T$. $\A$
	  checks that the openings are correct and that $x_i =
	  \hat{x}_i$ whenever $\theta_i = \hat{\theta}_i$. If all
	  tests are passed, $\A$ accepts. Otherwise, she rejects and
	  aborts.
	\item 
	  The tested positions are discarded by both parties: $\A$ and
	  $\B$ restrict $x$ and $\theta$, respectively $\hat{\theta}$
	  and $\hat{x}$, to $i \in \bar{T}$.
	\end{enumerate}
      \item[{\it Post-processing:}] $ $\\
	As in $\Pi$ (with $x, \theta,\hat{x}$
	and $\hat{\theta}$ restricted to positions $i \in \bar{T}$).
      \end{description}
      \vspace{-1.5ex}
    \end{framed}
    \vspace{-1.5ex}
    \caption{The Compiled Protocol $\compile(\Pi)$.} 
    \label{fig:compiled} 
  \end{figure}

  Although apparently simple---intuition clearly suggests that if Bob
  passes the measurement test, he must have measured most of his
  qubits, also in the remaining untested subset---this $\CO$ approach
  is not trivial to rigorously prove for a quantum Bob. Moreover, in
  order to preserve unconditional security against dishonest Alice,
  the commitment scheme needs to be unconditionally hiding, and so can
  be at best quantum-computationally binding. For a plain commitment
  scheme however, the common reduction from computational security of
  the protocol $\compile(\Pi)$ to the computational binding property
  of a commitment scheme would require rewinding, but we do not know
  of any technique for our protocol structure (see also
  Section~\ref{sec:rewinding} for an elaborated discussion).

  Therefore, we use our mixed dual-mode commitment construction
  $\commit$ from Section~\ref{sec:mixed.commit} that allows use to
  circumvent the necessity of rewinding. Recall that $\commit$ is a
  keyed dual-mode commitment scheme with unconditionally hiding key
  $\pkH$, generated by $\GH$, and unconditionally binding key $\pkB$,
  generated by $\GB$ along with a secret key $\sk$ that allows to
  efficiently extract $m$ from $\commitk{m}{r}{\pkB}$. Furthermore, we
  have that $\pkH \approxq \pkB$. For simplicity and efficiency, we
  consider the CRS-model, and we assume the key $\pkB$ for the
  commitment scheme, generated according to $\GB$, to be contained in
  the common reference string. We discuss in
  Section~\ref{sec:generation.generation} how to avoid the CRS-model,
  at the cost of a non-constant round construction where the parties
  generate a common reference string jointly by coin-flipping. Such an
  approach allows us to implement the entire application without any
  set-up assumptions. With our dual-mode commitment scheme, we arrive
  at the following theorem, capturing the compilation of any protocol
  from benign security towards computational security.
  \begin{theorem}[Compiler]
    \label{thm:compiler}
    Let $\Pi$ be a BB84-type protocol, unconditionally secure against
    dishonest Alice and against $\beta$-benign Bob for some constant
    $\beta \geq 0$. Consider the compiled protocol $\compile(\Pi)$ for
    arbitrary $\alpha > 0$, where the commitment scheme is
    instantiated by a dual-mode commitment scheme. Then,
    $\compile(\Pi)$ is unconditionally secure against dishonest Alice
    and quantum-computationally secure against dishonest Bob in the
    CRS-model.
  \end{theorem}

  \begin{proof}
    We sometimes write $\compile_\pkH(\Pi)$ for the compiled protocol
    $\compile(\Pi)$ to stress that key $\pkH$, produced by $\GH$, is
    used for the dual-mode commitment scheme. Analogously, we write
    $\compile_\pkB(\Pi)$ when key $\pkB$, produced by $\GB$, is used
    instead.

    Correctness is trivially checked. In order to show unconditional
    security against $\dA$, first note that the unconditionally hiding
    property of the commitment ensures that $\dA$ does not learn any
    additional information. Furthermore, as the ideal-world adversary
    $\dhA$ is not required to be poly-time bounded, according to
    Definition~\ref{def:unboundedAliceNiceOrder}, $\dhA$ can break the
    binding property of the commitment scheme, and thereby, perfectly
    simulate the behavior of honest $\B$ towards $\dA$ attacking
    $\compile(\Pi)$. The issue of efficiency of the ideal-life
    adversaries will be addressed in
    Section~\ref{sec:composability.compiler}. 

    As for computational security against dishonest Bob, according to
    Definition~\ref{def:polyboundedBobCRS}, we need to prove that for
    every real-world adversary $\dB \in \dBobPoly$ attacking
    $\compile(\Pi)$, there exists a suitable ideal-world adversary
    $\dhB \in \dBobPoly$ attacking $\F$ such that
    \begin{equation}
      out_{\A,\dB}^{\compile(\Pi)} \approxq out_{\hA,\dhB}^\F \, .
    \end{equation}
    \newline
    First, note that by the computational indistinguishability of
    $\pkH$ and $\pkB$,
    \begin{equation}\label{eq:KeySwitch}
      out_{\A,\dB}^{\compile(\Pi)} = out_{\A,\dB}^{\compile_\pkH(\Pi)}
      \approxq out_{\A,\dB}^{\compile_{\pkB}(\Pi)} \, .
    \end{equation} 
    Then, we construct an adversary $\dB_\circ \in \dBobPoly$ who
    attacks the unconditional security against benign Bob of protocol
    $\Pi$, and which satisfies
    \begin{equation}\label{eq:BenignBob}
      out_{\A,\dB}^{\compile_\pkB(\Pi)} =
      out_{\A_\circ,\dB_\circ}^{\Pi} \, ,
    \end{equation} 
    where $\A_\circ$ honestly executes $\Pi$. We define $\dB_\circ$ in
    the following way. Consider the execution of $\compile(\Pi)$
    between $\A$ and $\dB$. We split entity $\A$ into two players
    $\A_\circ$ and $\tilde{\A}$, where we think of $\tilde{\A}$ as
    being placed in between $\A_\circ$ and $\dB$. The splitted
    entities of this proof are also depicted in
    Figure~\ref{fig:splitted.entities}. $\A_\circ$ plays honest $\A$'s
    part of $\Pi$. $\tilde{\A}$ can be understood as completing
    $\CO$. More specifically, $\tilde{\A}$ acts as follows. It
    receives $n$ qubits from $\A_\circ$ and produces $\alpha
    n/(1-\alpha)$ random BB84-qubits of its own. Then, it interleaves
    the produced qubits randomly with the received qubits and sends
    the resulting $m= n/(1-\alpha)$ qubits to $\dB$. $\tilde{\A}$ then
    completes the verification step of $\compile(\Pi)$ with $\dB$,
    asking him to have the commitments opened which correspond to
    $\tilde{\A}$'s produced qubits. If this results in accept,
    $\tilde{\A}$ lets $\A_\circ$ finish the protocol with $\dB$. Note
    that pair $(\A_\circ,\tilde{\A})$ does exactly the same as $\A$.

    However, we can also move the actions of $\tilde{\A}$ to $\dB$'s
    side, and define $\dB_\circ$ as follows. $\dB_\circ$ samples
    $(\pkB,\sk)$ according to $\GB$ and executes $\Pi$ with $\A$ by
    locally running $\tilde{\A}$ and $\dB$, using $\pkB$ as commitment
    key. If $\tilde{\A}$ accepts the verification, then $\dB_\circ$
    outputs $\hat{\theta} \in \zo^n$ (as required from a benign
    Bob), obtained by decrypting the unopened commitments with the
    help of $\sk$. Otherwise, $\dB_\circ$ aborts at this point. It is
    now clear that Eq.~\eqref{eq:BenignBob} holds. Exactly the same
    computation takes place in both ``experiments'', the only
    difference being that they are executed partly by different
    entities.
    \newline
    \newline
    The last step is to show that, for some $\dhB$,
    \begin{equation}\label{eq:ByAssumption}
      out_{\A_\circ,\dB_\circ}^{\Pi} \approxs out_{\hA,\dhB}^\F \, .
    \end{equation} 
    Eq.~\eqref{eq:ByAssumption} actually claims that $\dhA$ and $\dhB$
    successfully simulate $\A_\circ$ and $\dB_\circ$ executing
    $\Pi$. This follows by assumption of benign security of $\Pi$, if
    we can show that $\dB_\circ$ is $\beta$-benign, according to
    Definition~\ref{def:BenignBob}, for any $\beta \geq 0$. We show
    this in the following subsection, or more precisely, we prove that
    the joint state of $\A_\circ, \dB_\circ$ after the preparation
    phase is statistically indistinguishable from a state
    $\rho_{Ideal}$ which satisfies the bounds in
    Eq.~\eqref{eq:staterequirements} of
    Definition~\ref{def:BenignBob}. We conclude the current proof by
    claiming that Theorem~\ref{thm:compiler} follows from
    Eqs.~\eqref{eq:KeySwitch}~--~\eqref{eq:ByAssumption} together.
  \end{proof}

  \begin{figure}
    \begin{framed}
    \begin{center}
      \begin{tikzpicture}
	\tikzstyle{every node}=[minimum size=6mm] 
	\draw (0,0) node[draw] (zeroA) {$\A_\circ$}; 
	\draw (5,0) node[draw] (tildeA) {$\tilde{\A} $}; 
	\draw (10,0) node[draw] (dB) {$\dB$}; 
	\draw[rounded corners] (zeroA.east) -- node[below] {$\Pi$}
	(tildeA); 
	\draw[rounded corners] (tildeA.east) -- node[above]
	{$\compile(\Pi)$} (dB); 
	\draw[rounded corners] (zeroA.north) ++(-4mm,1mm) -- +(0mm,
	2mm) -| node[pos=0.25,above] {$\A$} ([shift={(4mm,1mm)}]
	tildeA.north); 
	\draw[rounded corners] (tildeA.south) ++(-4mm,-1mm) -- +(0mm,
	-2mm) -| node[pos=0.25,below] {$\dB_\circ$}
	([shift={(4mm,-1mm)}] dB.south);
      \end{tikzpicture}
    \end{center}
    \end{framed}
    \vspace{-3.5ex}
    \caption{Constructing an attacker $\dB_\circ$ against $\Pi$ from
      an attacker $\dB$ against $\compile(\Pi)$.}
    \label{fig:splitted.entities} 
  \end{figure}


\subsection{Proof of Bounds on Entropy and Memory Size}
\label{sec:bounds.compiler}

  Recall that $\A_\circ$ executing $\Pi$ with $\dB_\circ$ can
  equivalently be thought of as $\A$ executing $\compile_\pkB(\Pi)$
  with $\dB$ (Figure~\ref{fig:splitted.entities}). Furthermore, a joint
  state of $\A,\dB$ is clearly also a joint state of $\A_\circ,
  \dB_\circ$. To show the existence of $\rho_{Ideal}$ for $\A_\circ,
  \dB_\circ$ as promised above, it therefore suffices to show such a
  state for $\A,\dB$. In other words, we need to show that the
  execution of $\compile_\pkB(\Pi)$ with honest $\A$ and arbitrarily
  dishonest $\dB$---after verification---will be close to a state where
  Eq.~\eqref{eq:staterequirements} holds. 

  To show this closeness, we consider an equivalent EPR-version, where
  Alice creates $m$ EPR-pairs $(\ket{00}+\ket{11})/\sqrt{2}$, sends
  one qubit in each pair to Bob, and keeps the others in register
  $A$. Then, Alice can measures her qubits only when needed, namely,
  she measures the qubits within $T$ in Step~(\ref{step:check}.)~of
  the verification phase, and the remaining qubits at the end of the
  verification phase. With respect to the information Alice and Bob
  obtain, this EPR-version is {\em identical} to the original protocol
  $\compile_\pkB(\Pi)$ based on single qubits, since the only
  difference is the point in time when Alice obtains certain
  information.

  We can further modify the procedure without affecting
  Eq.~\eqref{eq:staterequirements} as follows. Instead of measuring
  her qubits in $T$ in {\em her} basis $\theta|_T$, Alice measures
  them in {\em Bob's} basis $\hat{\theta}|_T$. However, she still
  verifies only whether $x_i = \hat{x}_i$ for those $i \in T$ with
  $\theta_i = \hat{\theta}_i$. Because the positions $i \in T$ with
  $\theta_i \neq \hat{\theta}_i$ are not used in the protocol at all,
  this change has no effect. As the commitment scheme is
  unconditionally binding, if key $\pkB$ is used, Bob's basis
  $\hat{\theta}$ is well defined by his commitments (although hard to
  compute), even if Bob is dishonest. The resulting scheme is given in
  Figure~\ref{fig:epr.version}.

  \begin{figure}
    \begin{framed}
      \noindent\hspace{-1.5ex}{\sc Protocol }EPR-$\compile_\pkB(\Pi)$
      : \\[-4ex]
      \begin{description}\setlength{\parskip}{0.5ex}
      \item[{\it Preparation:}] $ $\\ 
	$\A$ prepares $m$ EPR-pairs
	$\frac{1}{\sqrt{2}}(\ket{00}+\ket{11})$ and sends the second
	qubit in each pair to $\B$, while keeping the other qubits in
	register $A = A_1\cdots A_m$. $\B$ chooses $\hat{\theta}
	\in_R \set{0,1}^m$ and obtains $\hat{x} \in \set{0,1}^m$ by
	measuring the received qubits in bases $\hat{\theta}$.
      \item[{\it Verification:}]
	\begin{enumerate}
	\item[]
	\item 
	  $\B$ commits to $\hat{\theta}$ and $\hat{x}$ position-wise
	  by $c_i := \commitx{(\hat{\theta}_i,\hat{x}_i)}{r_i}$ with
	  randomness $r_i$ for $i = 1,\ldots,m$. He sends the
	  commitments to $\A$.
	\item\label{step:EPRcheck} 
	  $\A$ sends a random test subset $T \subset \{1,\ldots,m \}$
	  of size $\alpha m$. $\B$ opens $c_i$ for all $i \in T$.
	  $\A$ chooses $\theta \in_R \set{\+,\x}^m$, measures
	  registers $A_i$ with $i \in T$ in basis $\hat{\theta}_i$ to
	  obtain $x_i$, and checks that the openings are correct and
	  that $x_i = \hat{x}_i$ whenever $\theta_i = \hat{\theta}_i$
	  for $i \in T$. If all tests are passed, $\A$
	  accepts. Otherwise,, she rejects and aborts.
	\item\label{step:EPRrest} 
	  $\A$ measures the remaining registers in basis
	  $\theta|_{\bar{T}}$ to obtain $x|_{\bar{T}}$. The tested
	  positions are discarded by both parties: $\A$ and $\B$
	  restrict $x$ and $\theta$, respectively $\hat{\theta}$ and
	  $\hat{x}$, to $i \in \bar{T}$.
	\end{enumerate}
      \item[{\it Post-processing:}] $ $\\
	As in $\Pi$ (with $x, \theta,\hat{x}$ and $\hat{\theta}$
	restricted to positions $i \in \bar{T}$).
      \end{description}
      \vspace{-1.5ex}
    \end{framed}
    \vspace{-1.5ex}
    \caption{The EPR-version of $\compile_\pkB(\Pi)$.}
    \label{fig:epr.version} 
  \end{figure}

  We consider an execution of EPR-$\compile_\pkB(\Pi)$ in
  Figure~\ref{fig:epr.version} with an honest $\A$ and a dishonest
  $\dB$, and we fix $\hat{\theta}$ and $\hat{x}$, determined by
  $\dB$'s commitments. Let $\ket{\varphi_{AE}} \in {\cal
  H}_A\otimes{\cal H}_E$ be the state of the joint system right before
  Step~(\ref{step:check}.)~of the verification phase. Since we are
  anyways interested in the point-wise purification of $\dB$'s state,
  we may indeed assume this state to be pure. If it is not pure, we
  purify it and carry the purifying register $R$ along with $E$.

  Clearly, if $\dB$ had honestly done his measurements then, for some
  $\ket{\varphi_E}\in {\cal H}_E$,
  \begin{equation*}
    \ket{\varphi_{AE}} = \ket{\hat{x}}_{\hat{\theta}} \otimes
    \ket{\varphi_E} \, .
  \end{equation*}
  In this case, the quantum memory $E$ would be empty, i.e.,
  \begin{equation*}
    \Hmax(\proj{\varphi_E}) = 0 \, ,
  \end{equation*}
  and the uncertainty about $X$, obtained when measuring
  $A|_{\bar{T}}$ in basis $\theta|_{\bar{T}}$ would be maximal in the
  sense that it would be exactly one bit in each position with
  non-matching bases, i.e.,
  \begin{equation*}
  \Hmin(X) = d_H(\theta|_{\bar{T}},\hat{\theta}|_{\bar{T}}) \, .
  \end{equation*}
  
  We now show that the verification phase enforces these properties
  for an arbitrary dishonest $\dB$, at least approximately in the
  sense of Eq.~\eqref{eq:staterequirements}. Recall that $T \subset
  \set{1,\ldots,m}$ is random subject to $|T| = \alpha
  m$. Furthermore, for fixed $\hat{\theta}$ but randomly chosen
  $\theta$, the subset $T' = \Set{i \in T}{ \theta_i =
  \hat{\theta}_i}$ is a random subset (of arbitrary size) of $T$. Let
  the random variable $\Test$ describe the choice of $\test = (T,T')$
  as specified above, and consider the state $\rho_{\Test AE}$,
  consisting of the classical $\Test$ and the quantum state
  $\ket{\varphi_{AE}}$ with
  \begin{equation*}
    \rho_{\Test AE} = \rho_{\Test} \otimes \proj{\varphi_{AE}} =
    \sum_{\test} P_{\Test}(\test) \proj{\test} \otimes
    \proj{\varphi_{AE}} \, .
  \end{equation*}

  Recall that $r_H(\cdot,\cdot)$ denotes the relative Hamming distance
  between two strings (see Eq.~\eqref{eq.relative.hamming}). The
  following lemma shows that, we are in state $\rho_{\Test AE}$ close
  to an ``ideal state'' $\tilde{\rho}_{\Test AE}$, capturing a
  situation , where for {\em any} choice of $T$ and $T'$ and for {\em
  any} outcome $x|_T$ when measuring $A|_T$ in basis
  $\hat{\theta}|_T$, the relative error $r_H(x|_{T'},\hat{x}|_{T'})$
  (the test estimate) gives an upper bound (which holds with
  probability 1) on the relative error
  $r_H(x|_{\bar{T}},\hat{x}|_{\bar{T}})$ one would obtain by measuring
  the remaining subsystems $A_i$ with $i \in \bar{T}$ in basis
  $\hat{\theta}_i$.
  \begin{lemma}
    \label{lemma:ideal.state}
    For any $\eps > 0$, $\hat{x}\in\{0,1\}^m$ and $\hat{\theta}\in
    \{\+,\x\}^m$, the state $\rho_{\Test AE}$ is negligibly close (in
    $m$) to a state
    $$ 
    \tilde{\rho}_{\Test AE} = \sum_{\test} P_{\Test}(\test)
    \proj{\test} \otimes \proj{\tilde{\varphi}^{\test}_{AE}} \, ,
    $$ 
    where for any $\test = (T,T')$, we have
    $$ 
    \ket{\tilde{\varphi}^{\test}_{AE}} = \sum_{x \in B_{\test}}
    \alpha^{\test}_x \ket{x}_{\hat{\theta}} \ket{\psi_E^x}\, ,
    $$ 
    for $B_{\test} = \{x \in \{0,1\}^m \,|\,
    r_H(x|_{\bar{T}},\hat{x}|_{\bar{T}}) \leq
    r_H(x|_{T'},\hat{x}|_{T'}) + \eps \}$ and arbitrary coefficients
    $\alpha^x_{test} \in \C$.
  \end{lemma}
We want to point out that the ``ideal state''
$\ket{\tilde{\varphi}^{\test}_{AE}}$ in the remaining subsystem after
the test is a superposition of states with relative Hamming distance
upper bounded by the test estimate (plus a small error $\eps$). This
is the case, since we sum over all $x$ restricted to the set
specifying exactly that, and also note that $\dB$'s subsystem
$\ket{\psi_E^x}$ depends on $x$, which means, informally speaking,
only such states survive. Yet in other words, we are indeed left with
a superposition over all strings that have relative Hamming
distance $\eps$-close to the estimate of the test.\\

  \begin{proof}
    For any $\test$, we let $\ket{\tilde{\varphi}^{\test}_{AE}}$ be
    the renormalized projection of $\ket{\varphi_{AE}}$ into the
    subspace $\mathrm{span}\{\ket{x}_{\hat{\theta}} \,|\, x \in
    B_{\test} \} \otimes {\cal H}_E$ and we let
    $\ket{\tilde{\varphi}^{\test \perp}_{AE}}$ be the renormalized
    projection of $\ket{\varphi_{AE}}$ into the orthogonal complement,
    such that $\ket{\varphi_{AE}} = \eps_{\test}
    \ket{\tilde{\varphi}^{\test}_{AE}} + \eps_{\test}^\perp
    \ket{\tilde{\varphi}^{\test \perp}_{AE}}$ with $\eps_{\test} =
    \braket{\tilde{\varphi}^{\test}_{AE}}{\varphi_{AE}}$ and
    $\eps_{\test}^\perp = \braket{\tilde{\varphi}^{\test
    \perp}_{AE}}{\varphi_{AE}}$. By construction,
    $\ket{\tilde{\varphi}^{\test}_{AE}}$ is of the form as required in
    the statement of the lemma. A basic property of the trace norm of
    pure states leads to
    \begin{equation*}
      \delta\bigl( \proj{\varphi_{AE}},
      \proj{\tilde{\varphi}^{\test}_{AE}} \bigr) 
      = \sqrt{1
	- |\braket{\tilde{\varphi}^{\test}_{AE}}{\varphi_{AE}}|^2} 
      = \sqrt{1 - |\eps_{\test}|^2} 
      = |\eps_{\test}^\perp| \, .
    \end{equation*}
    This last term corresponds to the square root of the probability,
    when given $\test$, to observe a string $x \not\in B_{\test}$ when
    measuring subsystem $A$ of $\ket{\varphi_{AE}}$ in basis
    $\hat{\theta}$. Furthermore, using elementary properties of the
    trace norm with Jensen's inequality\footnote{In this context, we
    use Jensen's inequality with $f\big( \sum_i p_i x_i \big) \leq
    \sum_i p_i f(x_i)$, for positive $p_i$ and real convex function
    $f$.} gives
    \begin{align*}
      \delta\bigl( \rho_{\Test AE},\tilde{\rho}_{\Test AE} \bigr)^2 &=
      \bigg( \sum_{\test} P_{\Test}(test) \,\delta\bigl(
      \proj{\varphi_{AE}}, \proj{\tilde{\varphi}^{\test}_{AE}} \bigr)
      \bigg)^2 \\ &= \bigg( \sum_{\test} P_{\Test}(test) \,
      |\eps_{\test}^\perp| \bigg)^2 \leq \sum_{\test} P_{\Test}(test)
      \, |\eps_{\test}^\perp|^2 \, ,
    \end{align*}
    where the last term is the probability to observe a string $x
    \not\in B_{\test}$ when choosing $\test$ according to $P_{\Test}$
    and measuring subsystem $A$ of $\ket{\varphi_{AE}}$ in basis
    $\hat{\theta}$. This situation, though, is a classical sampling
    problem, for which it is well known that for any measurement
    outcome $x$, the probability (over the choice of $\test$) that $x
    \not\in B_{\test}$ is negligible in $m$
    (see~e.g.~\cite{Hoeffding63}). Thus, it follows that state
    $\rho_{\Test AE}$ is negligibly close (in $m$) to state
    $\tilde{\rho}_{\Test AE}$.
  \end{proof}

  Next, we need a preliminary lemma, stating that a pure state can be
  written as a ``small superposition'' of basis vectors.
  \begin{lemma}
    \label{lemma:small.superposition}
    Let $\ket{\varphi_{AE}} \in {\cal H}_A \otimes {\cal H}_E$ be a
    state of the form $\ket{\varphi_{AE}} = \sum_{i \in J} \alpha_i
    \ket{i}\ket{\varphi_E^i}$, where $\set{\ket{i}}_{i \in I}$ is a
    basis of ${\cal H}_A$ and $J \subseteq I$. Then, the following
    holds.
    \begin{enumerate}
    \item\label{it:boundHoo} 
      Let $\tilde{\rho}_{AE} = \sum_{i \in J} |\alpha_i|^2
      \proj{i}\otimes\proj{\varphi_E^i}$, and let $W$, and $\tilde{W}$,
      be the outcome of measuring $A$ of $\ket{\varphi_{AE}}$,
      respectively of $\tilde{\rho}_{AE}$, in some basis
      $\set{\ket{w}}_{w \in \cal W}$. Then,
      $$ 
      \Hmin(W) \geq \Hmin(\tilde{W}) - \log|J| \, .
      $$
    \item\label{it:boundHo} 
      The reduced density matrix $\rho_E = \tr_A(\proj{\varphi_{AE}})$
      has max-entropy
      $$ 
      \Hmax(\rho_E) \leq \log |J| \, .
      $$
    \end{enumerate}
  \end{lemma}
  Note that when using Renner's definition for conditional min-entropy
  of~\cite{Renner05} under Point~(\ref{it:boundHoo}.), one can
  actually show that $H_{\infty}(W|E) \geq H_{\infty}(\tilde{W}|E) -
  \log|J|$.\\

  \begin{proof}
    For Point~(\ref{it:boundHoo}.), we may understand
    $\tilde{\rho}_{AE}$ as being in state $\ket{i}\ket{\varphi_E^i}$
    with probability $|\alpha_i|^2$, so that we easily see that
    \begin{align*}
      P_{\tilde{W}}(w) 
      &= \sum_{i \in J} |\alpha_i|^2 |\braket{w}{i}|^2 = \sum_{i \in
      J} |\alpha_i|^2 |\braket{w}{i}|^2 \cdot \sum_{i \in J} 1^2 \cdot
      \frac{1}{|J|} \\ 
      &\geq \bigg|\sum_{i \in J} \alpha_i \braket{w}{i}\bigg|^2 \cdot
      \frac{1}{|J|} = \bigg|\bra{w}\sum_{i \in J} \alpha_i
      \ket{i}\bigg|^2 \cdot \frac{1}{|J|}
      = P_{W}(w) \cdot \frac{1}{|J|} \, ,
    \end{align*}
    where the inequality is Cauchy-Schwartz\footnote{In this context,
    we use the inequality phrased as $\sum_i | x_i|^2 |y_i|^2 \geq
    |\sum_i x_i y_i|^2$}. The claim follows (with
    Definition~\ref{def:min.entropy}).

    For Point~(\ref{it:boundHo}.), note that $\rho_E = \tr_A(
    \proj{\varphi_{AE}} ) = \sum_{i \in J} |\alpha_i|^2
    \proj{\varphi_E^i}$. The claim follows immediately from the
    sub-additivity of the rank, i.e.,
    $$ 
    \rank{\rho_E} \leq \sum_{i \in J} \rank{|\alpha_i|^2
    \proj{\varphi_E^i}} \leq \sum_{i \in J} 1 = |J| \, ,
    $$ 
    where we use that all $\proj{\varphi_E^i}$ have rank (at most) 1.
  \end{proof}

  Now, combining the fact that it holds for the binary entropy $h$
  that $\big|\{ y \in \{0,1\}^n \,|\, d_H(y,\hat{y}) \leq \mu n
  \}\big| \leq 2^{h(\mu) n}$ for any $\hat{y} \in \set{0,1}^n$ and $0
  \leq \mu \leq \frac12$ with Lemma~\ref{lemma:small.superposition} on
  ``small superpositions of product states'', we can conclude the
  following corollary.

  \begin{corollary}
    \label{corSerge}
    Let $\tilde{\rho}_{\Test AE}$ be of the form as in
    Lemma~\ref{lemma:ideal.state} (for given $\eps$, $\hat{x}$ and
    $\hat{\theta}$). For any fixed $\test = (T,T')$ and for any fixed
    $x|_T \in \{0,1\}^{\alpha m}$ with \smash{$\err :=
    r_H(x|_{T'},\hat{x}|_{T'}) \leq \frac12$}, let $\ket{\psi_{AE}}$
    be the state to which \smash{$\ket{\tilde{\varphi}^{\test}_{AE}}$}
    collapses when, for every $i \in T$, subsystem $A_i$ is measured
    in basis $\hat{\theta}_i$ and $x_i$ is observed, where we
    understand $A$ in $\ket{\psi_{AE}}$ to be restricted to the
    registers $A_i$ with $i \in \bar T$. Finally, let $\sigma_E =
    \tr_A(\proj{\psi_{AE}})$ and let the random variable $X$ describe
    the outcome when measuring the remaining $n = (1-\alpha) m$
    subsystems of $A$ in basis $\theta|_{\bar{T}} \in \set{\+,\x}^n$.
    Then, for any subset $I \subseteq \set{1,\ldots,n}$ and any
    $x|_I$,\footnote{Below, $\theta|_I$ (and similarly
    $\hat{\theta}|_I$) should be understood as first restricting the
    $m$-bit vector $\theta$ to $\bar{T}$, and then restricting the
    resulting $n$-bit vector $\theta|_{\bar{T}}$ to $I$: $\theta|_I:=
    {(\theta|_{\bar{T}})|}_I$.}
    $$ 
    \Hmin\bigl(X|_I \,\big|\,X|_{\bar{I}} = x|_{\bar{I}}\bigr) \geq
    d_H\bigl(\theta|_I,\hat{\theta}|_I\bigr) - h(\err + \eps) n
    $$ 
    and 
    $$
    \Hmax\bigl(\sigma_E\bigr) \leq h(\err + \eps) n \, .
    $$
  \end{corollary}
  Thus, the number of errors between the measured $x|_{T'}$ and the
  given $\hat{x}|_{T'}$ gives us a bound on the min-entropy of the
  outcome, when measuring the remaining subsystems of $A$, and on the
  max-entropy of the state of subsystem $E$.\\

  \begin{proof}
    To simplify notation, we write $\vartheta = \theta|_{\bar{T}}$ and
    $\hat{\vartheta} = \hat{\theta}|_{\bar{T}}$. By definition of
    $\tilde{\rho}_{\Test AE}$, for any fixed values of $\eps,\hat{x}$
    and $\hat{\theta}$, the state $\ket{\psi_{AE}}$ is of the form
    $\ket{\psi_{AE}} = \sum_{y \in {\cal Y}} \alpha_y
    \ket{y}_{\hat{\vartheta}} \otimes \ket{\psi^y_E}$, where ${\cal Y}
    = \Set{y \in \set{0,1}^n}{d_H(y,\hat{x}|_{\bar{T}}) \leq (\err +
    \eps)n}$. Recall here that $r_H(y,\hat{x}|_{\bar{T}}) =
    d_H(y,\hat{x}|_{\bar{T}})/n$. Consider the corresponding mixture
    $\tilde{\sigma}_{AE} = \sum_{y \in {\cal Y}} |\alpha_y|^2
    \ket{y}_{\hat{\vartheta}} \bra{y}_{\hat{\vartheta}}
    \otimes\proj{\psi_E^y}$ and define $\tilde{X}$ as the random
    variable for the outcome when measuring register $A$ of
    $\tilde{\sigma}_{AE}$ in basis $\vartheta$. Note that
    $H_{\infty}(\tilde{X})\geq d_H(\vartheta,\hat{\vartheta})$, since
    any state $\ket{y}_{\hat{\vartheta}}$, when measured in basis
    $\vartheta$, produces a random bit for every position $i$ with
    $\vartheta \neq \hat{\vartheta}$ (see also the definition of the
    min-entropy (Definition~\ref{def:min.entropy}) and note that there
    exist $2^{d_H(\vartheta,\hat{\vartheta})}$ possible
    outcomes). Lemma~\ref{lemma:small.superposition} allows us to
    conclude that
    $$
    H_0(\sigma_E) \leq \log |{\cal Y}| \leq \log 2^{h(err+\eps)n} =
    h(err+\eps)n \, ,
    $$ and similarly, 
    $$ 
    H_{\infty}(X) \geq H_{\infty}(\tilde{X})-\log{|{\cal Y}|} \geq
    d_H(\vartheta,\hat{\vartheta})-h(err+\eps)n \, .
    $$  
    This proves the claim for $I = \set{1,\ldots,n}$. For arbitrary $I
    \subset \set{1,\ldots,n}$ and $x|_I$, we can consider the pure
    state, obtained by measuring the registers $A_i$ with $i \not\in
    I$ in basis $\vartheta_i$, when $x|_{\bar{I}}$ is observed. This
    state is still a superposition of at most $|{\cal Y}|$ vectors and
    thus we can apply the exact same reasoning to obtain
    Eq.~\eqref{eq:staterequirements}.
  \end{proof}

  The initial claim to be shown now follows by combining
  Lemma~\ref{lemma:ideal.state} and Corollary~\ref{corSerge}. Indeed,
  the ideal state $\rho_{Ideal}$ we promised, for which
  \eqref{eq:staterequirements} holds, is produced by putting $\A$ and
  $\dB$ in state $\tilde{\rho}_{TestAE}$, defined in
  Lemma~\ref{lemma:ideal.state}, and then running
  Steps~(\ref{step:EPRcheck}.)~and~(\ref{step:EPRrest}.) of the
  verification phase. This state is negligibly close to the real
  state, since by Lemma~\ref{lemma:ideal.state}, we were negligibly
  close to the real state before these
  operations. Corollary~\ref{corSerge} ensures that the bounds for
  benign Bob as stated in the definition of benignity in
  Eq.~\eqref{eq:staterequirements} are satisfied.


\section{In the Presence of Noise} 
\label{sec:compiler.noise}

In the description of the compiler and in its analysis in the previous
section, we assume the quantum communication to be noise-free. Indeed,
in the case of transmission errors, honest Alice is likely to reject
an execution with honest Bob. However, it is straightforward to
generalize the result to noisy quantum communication as follows.

In Step~(\ref{step:check}.)~in the verification phase of
$\compile(\Pi)$, Alice rejects and aborts if the relative number of
errors between $x_i$ and $\hat{x}_i$ for $i \in T$ with $\theta_i =
\hat{\theta}_i$ exceeds the error probability $\phi$, induced by the
noise in the quantum communication, by some small $\eps' > 0$. By
Hoeffding's inequality~\cite{Hoeffding63}, this guarantees that honest
Alice does not reject honest Bob, except with exponentially small
probability. Furthermore, proving the security of this
``noise-resistant'' compiler goes along the exact same lines as for
the original compiler. The only difference is that when applying
Corollary~\ref{corSerge}, the parameter $\err$ has to be chosen as
$\err = \phi + \eps'$, such that the bounds in
Eq.~\eqref{eq:staterequirements} hold for
$$
\beta = h(\err+\eps) = h(\phi+\eps'+\eps) \, .
$$ 
Thus, the claim of our compiler-theorem (Theorem~\ref{thm:compiler})
holds for any $\beta$-benign Bob with $\beta > h(\phi)$ (by choosing
$\eps,\eps' > 0$ small enough).


\section{Bounded-Quantum-Storage Security}
\label{sec:compiler.hybrid.security}
\index{bounded-quantum-storage model}

In this section we show that our compiler preserves security in the
bounded-quantum-storage model. Recall that in the BQSM, one of the
players, in our case it is Bob, is assumed be able to store only a
limited number of qubits beyond a certain point in the
protocol. BQSM-secure OT- and ID-protocols are known~\cite{DFSS07},
but can be efficiently broken, if the memory bound does not
hold. Therefore, we show here that applying the compiler produces
protocols with better security, namely the adversary needs large
quantum storage {\em and} large computing power to succeed. In
Chapter~\ref{chap:hybrid.security.applications}, we will then discuss
the compiled protocols with hybrid security in more detail.

Consider a BB84-type protocol $\Pi$. For a constant $0 < \gamma < 1$,
let $\dBobBQSM^\gamma(\Pi)$ be the set of dishonest players $\dB$ that
store only $\gamma n$ qubits after a certain point in $\Pi$, where $n$
is the number of qubits sent initially. Protocol $\Pi$ is said to be
unconditionally secure against such a $\gamma$-BQSM Bob, if it
satisfies Definition~\ref{def:qmemoryBob} with the restriction that
the quantification is over all dishonest $\dB \in
\dBobBQSM^\gamma(\Pi)$. 

\begin{theorem}
\label{thm:BQSM}
  If protocol $\Pi$ is unconditionally secure against $\gamma$-BQSM
  Bob, then the compiled protocol $\compile(\Pi)$ is unconditionally
  secure against $\gamma(1\!-\!\alpha)$-BQSM Bob, where $0 < \alpha <
  1$.
\end{theorem}

\begin{proof}
  The proof proceeds as the proof for our compiler-theorem
  (Theorem~\ref{thm:compiler}). We have a dishonest $\dB$ that attacks
  $\compile(\Pi)$, and we construct a $\dB_\circ$ that attacks the
  original protocol $\Pi$. The only difference here is that we let
  $\dB_\circ$ generate the common reference string ``correctly'' as
  $\pkH$ sampled according to $\GH$.

  It follows by construction of $\dB_\circ$ that
  $out^{\compile(\Pi)}_{\A,\dB} = out^{\Pi}_{\A_\circ, \dB_\circ} \,
  .$ Furthermore, since $\dB_\circ$ requires the same amount of
  quantum storage as $\dB$ but communicates an $\alpha$-fraction fewer
  qubits, it follows that $\dB_\circ \in \dBobBQSM^\gamma(\Pi)$, if
  $\dB \in \dBobBQSM^{\gamma(1-\alpha)}(\compile(\Pi))$. Thus, it
  follows that there exists $\dhB$ such that $out^{\Pi}_{\A_\circ,
  \dB_\circ} \stackrel{s}{\approx} out^{\F}_{\hA,\dhB} \, .$ This
  proves the claim.
\end{proof}


\section{Composability}
\label{sec:composability.compiler}

Several composition frameworks for the quantum setting have been
proposed, for instance, sequential composability in a classical
environment~\cite{FS09}, sequential composability in a quantum
environment but restricted to the BQSM~\cite{WW08}, or attempts of
generalizing the universal classical composability
framework (UC in~\cite{Canetti01}) to universal quantum
composability~\cite{BM04,Unruh04,Unruh10}. Here, we will briefly
investigate our protocols in the particular composition frameworks, we
consider most appropriate for our setting.


\subsection{Sequential Composition}
\label{sec:sequential.composition.compiler}
\index{composition!sequential}

  All our definition for correctness and security of our two-party
  quantum protocols comply with the composition framework
  of~\cite{FS09} as described in detail in
  Section~\ref{sec:security.definition}. In particular, we will show
  in the next chapter that all of our quantum protocols $\pi$ securely
  implements their corresponding ideal functionality $\F$. Thus,
  according to the Composition
  Theorems~\ref{thm:composition.unconditional}
  and~\ref{thm:composition.computational}, we arrive at a situation
  where an outer protocol $\Sigma^{\pi_1\cdots\pi_\ell}$, composed of
  possibly different inner sub-protocols $\pi_i$, is essentially as
  secure as any hybrid protocol $\Sigma^{\F_1\cdots\F_\ell}$ with
  sequential calls to the corresponding ideal functionalities
  $\F_i$. Sequential composition in a classical environment follows
  immediately.


\subsection{General Composition}
\label{sec:general.composition.compiler}
\index{composition!general}

  Our strong simulation-based security approach is clearly closely
  related to the concept of universal composability, but our
  definitions do not imply UC-security. The definitions of
  unconditional security leading to sequential composability do not
  require the running time of ideal-world adversaries to be polynomial
  whenever the real-life adversaries run in polynomial
  time. Fortunately, by extending our basic commitment construction,
  we can achieve efficient ideal-life adversaries. Therewith, we get
  efficient simulation on both sides without rewinding any dishonest
  player.

  Note that it might still be the case that our compilation preserves
  efficiency of the simulator, namely if protocol $\Pi$ is secure
  against dishonest $\dA$ with efficient simulator $\dhA$, then so is
  $\compile(\Pi)$. Although this would be desirable, it does not seem
  to be the case for our basic construction for the following
  reason. In order to show such a result, we would need to simulate
  the pre-processing phase against dishonest $\dA$ efficiently and
  without measuring the qubits that are not opened during
  pre-processing. Then after preparation and verification, we could
  give the remaining qubits to $\dhA$ to simulate the rest of the
  protocol as specified previously. However, the whole point of the
  pre-processing is to ensure that a real Bob measures all qubits,
  unless he can break the binding property of the commitments. Thus,
  the only way to resolve this situation is to give the simulator some
  trapdoor with which it can make commitments and open them any way it
  wants, or in other words, to equip the simulator with the
  possibility of equivocate\index{commitment!equivocability} its
  commitments.

  With such a equivocability trapdoor, the simulation of the
  verification phase is straightforward. $\dhA$ just waits until $\dA$
  reveals the test subset, measures the qubits in the test subset, and
  opens the commitments according to the measurement results. Then,
  $\dhA$ simulates the protocol with the remaining unopened
  qubits. Our basic commitment construction, introduced in
  Section~\ref{sec:mixed.commit}, does not provide such an
  equivocability trapdoor. However, we can extend the scheme as
  discussed in Section~\ref{sec:extended.commit.compiler} by first
  extending our mixed commitment to the multi-bit crypto-system
  of~\cite{PVW08} and then combining it with an HVZK-$\Sigma$-protocol
  construction for some quantumly hard $\NP$-relation $\Rel$. As
  previously shown, equivocability emerges in this construction with
  the simulator's knowledge of a valid witness $w$ such that $(x,w)\in
  \Rel$. In that case, the simulator can compute two accepting
  conversations for the $\Sigma$-protocol, and therewith, answer both
  challenges. The extension preserves the different but
  indistinguishable dual-modes of the underlying commitment scheme
  such that the committed bit can still be extracted by a simulator
  $\dhB$, decrypting both commitments $C_0,C_1$ to determine, which
  contains a valid reply in the $\Sigma$-protocol.

  In~\cite{Unruh10} a special case of our generic construction, namely
  the quantum oblivious transfer protocol of
  Section~\ref{sec:hybrid.security.ot} is related to the quantum-UC
  context\index{composition!quantum-UC}. It is shown that the protocol
  \emph{statistically quantum-UC-emulates} its ideal functionality in
  the case of corrupted Alice and corrupted Bob, if it is instantiated
  with an \emph{ideal commitment functionality}. Furthermore, it is
  established that security as specified in~\cite{FS09} implies
  quantum-UC-security in the case of our OT-protocol.\footnote{The
  security we achieve here is called \emph{quantum stand-alone
  security} in~\cite{Unruh10}, but we prefer to describe the
  statements in the terms used throughout this work.} Last, the
  OT-protocol in its randomized version and when instantiated by an
  unconditionally binding commitment scheme implements its
  corresponding ideal functionality with \emph{statistical security}
  in the case of corrupted Bob. Even though for the last result the
  protocol is based on an actual commitment, the case considers only a
  dishonest Bob, and by using an unconditionally binding scheme in the
  real world, we would loose unconditional security against dishonest
  Alice.

  However, by combining our extended construction as described above
  with the results of Section~\ref{sec:extended.commit.compiler} and
  with~\cite[Theorem 20]{Unruh10}, we get the following stronger
  result that applies to our generic compiler construction: Let $\Pi$
  be a BB84-type protocol as specified in Theorem~\ref{thm:compiler}
  and let $\compile(\Pi)$ be its compilation, instantiated with an
  extended mixed commitment construction in the CRS-model as described
  above. Then, $\compile(\Pi)$ \emph{computationally
  quantum-UC-emulates} its corresponding ideal functionality $\F$ for
  \emph{both dishonest players}.


\chapter{Applications}
\label{chap:hybrid.security.applications}

Our compiler, discussed in the previous chapter, can be easily applied
to known protocols. Here, we show two example applications, namely
oblivious transfer and password-based identification. Since the
original protocols are BQSM-secure, we also obtain hybrid security by
compilation. These results appeared in~\cite{DFLSS09}. We then show
that the compiled identification protocol is secure against
man-in-the-middle attacks, which was sketch in~\cite{DFLSS09} but
formal proofs were omitted.


\section{Oblivious Transfer}
\label{sec:hybrid.security.ot}
\index{oblivious transfer}

Oblivious transfer, as introduced in Section~\ref{sec:primitives.ot},
constitutes a highly relevant cryptographic primitive, which is
complete for general two-party computation and reduces to classical
commitment in its quantum variant. As a building block it can be
securely used in outer quantum or classical protocols and extends, for
instance, to quantum identification.


\subsection{Motivation and Related Work}
\label{sec:hybrid.security.ot.history}

  As mentioned already, the idea of introducing a $\CO$-step to
  improve the security of quantum protocols was suggested in the first
  quantum OT protocol of Cr\'epeau and Kilian~\cite{CK88}, which---in
  its original form---proposes a protocol for $\tt{Rabin\text{--}OT}$,
  and in the practical follow-up work of Bennett, Brassard, Cr\'epeau
  and Skubiszewska~\cite{BBCS91}, implementing $\ot{1}{2}^1$. The
  $\CO$ approach is sketched as a ``conceptually easy
  fix''~\cite[p.~14]{BBCS91} in a situation where a dishonest Bob has
  large quantum storage.

  Various partial results for OT in that context followed. For
  instance, in~\cite{Yao95} such a construction is proven secure
  against any receiver in the case of noiseless communication. To make
  the proof work, however, an ideal black-box commitment scheme is
  assumed. This approach was then generalized for noisy channels and
  perfect string commitments in~\cite{Mayers96}. Another approach in
  the computational setting was taken in~\cite{CDMS04}. There it was
  shown that a computationally binding quantum string commitment would
  enforce an apparent collapse of Bob's quantum information, which in
  turn would imply secure OT. The paper concludes with the open
  question of how to construct an actual commitment scheme as required
  to get an applicable protocol.

  Based on our analysis of Section~\ref{sec:compiler.hybrid.security},
  we can now rather simply apply our compiler to (a variant of) the
  protocol in~\cite{BBCS91}, and therewith, give a complete proof for
  a concrete unconditionally hiding commitment scheme.


\subsection{The Protocol} 
\label{sec:hybrid.security.ot.protocol}

  The variant we consider here achieves $\ot{1}{2}^\ell$. Recall that
  in such a protocol, the sender $\A$ sends two $\ell$-bit strings
  $s_0$ and $s_1$ to the receiver $\B$. $\B$ can select a string to
  receive, $s_k$, but he does not learn anything about
  $s_{1-k}$. Finally, $\A$ does not learn $\B$'s choice bit $k$. The
  ideal oblivious transfer functionality $\otF$ is shown in
  Figure~\ref{fig:ideal.functionality.ot}.

\begin{figure}
  \begin{framed}
    \noindent\hspace{-1.5ex}{\sc Functionality $\otF$ :}\\ Upon
    receiving $\tt s_0$ and $\tt s_1$ from Alice and $\tt k$ from Bob,
    $\otF$ outputs $\tt s_k$ to Bob.
    \vspace{-1ex}
  \end{framed}
  \vspace{-1.5ex}
  \caption{The Ideal Functionality for String OT.}\label{fig:otF}
  \label{fig:ideal.functionality.ot}
\end{figure}

  Our protocol is almost identical to $\ot{1}{2}^1$ introduced
  in~\cite{BBCS91}, but instead of using parity values to mask the
  bits in the last protocol message, we follow the approach
  of~\cite{DFRSS07}. Their BQSM-secure protocol $\tt RAND \ \QOT$ for
  the randomized version uses hash-functions that allow for
  transferring an $\ell$-bit string instead of a bit as final message.

  Let $\F$ denote a suitable family of two-universal hash-functions
  with range $\set{0,1}^n \rightarrow \set{0,1}^\ell$ as specified in
  Definition~\ref{def:hashing}. Note that if the input to the function
  is smaller than $n$, we can pad it with zeros without decreasing its
  entropy. We further assume that $\ell = \lfloor \lambda n \rfloor$
  for some constant $\lambda > 0$. Then, after the modifications
  described above, we obtain the basic $\QOT$ protocol, depicted in
  Figure~\ref{fig:basic.ot}.

\begin{figure}
  \begin{framed}
    \noindent\hspace{-1.5ex} {\sc Protocol $\QOT:$ } \\[-4ex]
    \begin{description}\setlength{\parskip}{0.5ex}
    \item[{\it Preparation:}]$ $ \\ 
      $\A$ chooses $x \in_R \set{0,1}^n$
      and $\theta \in_R \set{+,\x}^n$ and sends $\ket{x}_{\theta}$
      to~$\B$. $\B$ chooses $\hat{\theta} \in_R \set{0,1}^n$ and
      obtains $\hat{x} \in \set{0,1}^n$ by measuring
      $\ket{x}_{\theta}$ in bases $\hat{\theta}$.
    \item[{\it Post-processing:}]
      \begin{enumerate}
      \item[]
      \item
	$\A$ sends $\theta$ to $\B$.
      \item
	$\B$ partitions all positions $1\leq i \leq n$ in two subsets
	according to his choice bit $k \in \{ 0,1 \}$: the ``good''
	subset $I_k := \{ i: \theta_i = \hat{\theta}_i \}$ and the
	``bad'' subset $I_{1-k} := \{ i: \theta_i \neq \hat{\theta}_i
	\}$. $\B$ sends $(I_0,I_1)$ in this order.
      \item
	$\A$ sends descriptions of $f_0,f_1\in_R {\cal F}$ together
	with $m_0 := s_0 \oplus f_0(x|_{I_0})$ and $m_1 := s_1 \oplus
	f_1(x|_{I_1})$.
      \item 
	$\B$ computes $s_k = m_k \oplus f_k(\hat{x}|_{I_k})$.
      \end{enumerate}
    \end{description}
    \vspace{-1.5ex}
  \end{framed}
  \vspace{-1.5ex}
  \caption{Basic Protocol for String OT.}
  \label{fig:basic.ot} 
\end{figure}

  \begin{proposition}
    Protocol $\QOT$ satisfies correctness and achieves unconditional
    security against dishonest Alice, according to
    Definitions~\ref{def:correctness}
    and~\ref{def:unboundedAliceNiceOrder}, respectively.
  \end{proposition}

  \begin{proof}
    Correctness for honest players is obvious: $\B$ selects one string
    to receive, which is masked by the hashed bit-string of outcomes,
    measured in the matching basis. In the positions with non-matching
    bases, he does not know the outcomes, and therewith he does not
    learn anything about $s_{1-k}$. Finally, $\A$ does not learn
    which is the ``good'' subset, and hence, which is $\B$'s choice
    $k$.

    Security against dishonest Alice is derived in a straightforward
    way from $\tt RAND \ \QOT$ of~\cite{DFRSS07} as follows. Note that
    in $\tt RAND \ \QOT$, the receiver measures all his qubits in one
    basis, depending on his choice bit $k$, i.e.\ $\theta \in
    [0,1]_{k}$. As described previously in
    Chapter~\ref{chap:hybrid.security}, our compiler requires
    measurement in random bases $\theta \in_R \{ 0,1 \}^n$. Otherwise,
    the opened and tested positions during $\CO$ would obviously leak
    $k$.

    Due to the non-interactivity in $\tt RAND \ \QOT$, $\dA$ cannot
    learn $\B$'s choice bit $k$, so the protocol is perfectly
    receiver-secure. More formally, the proof compares the real output
    to an ideal output, which is obtained by letting $\dA$ run the
    protocol with an unbounded receiver who measures his qubits in
    $\dA$'s bases $\theta$, samples \emph{independent} $K$ from the
    correct distribution, and sets $S_K$ correspondingly. The only
    difference between the two executions is the point in time and the
    choice of bases, in which positions $i \in I_{1-k}$ is
    measured. However, these parameters do not influence the output
    states, once $K$ is fixed.

    Now, the preparation phase combined with Step~(2.)~of the
    post-processing in $\QOT$ is equivalent to $\B$ measuring all
    qubits in the basis, dictated by $K$. Thus, the same analysis can
    be applied to $\QOT$, achieving unconditional security against
    $\dA$.
  \end{proof}

  \begin{theorem}\label{thm:ot}
    Protocol $\QOT$ is unconditionally secure against $\beta$-benign Bob
    for any $\beta < \frac18 - \frac{\lambda}{2}$.
  \end{theorem}

  \begin{proof}
    For any given benign Bob $\dB$, we construct $\dhB$ the following
    way: $\dhB$ runs locally a copy of $\dB$ and simulates an
    execution by running $\hA$ up to but not including
    Step~(3.). Since $\dB$ is benign, $\dhB$ obtains $\hat{\theta}$
    after the preparation phase. When the simulation of $\hA$ reaches
    the point just after the announcement of $f_0$ and $f_1$ in
    Step~(3.), $\dhB$ finds $k'$ such that
    $d_H(\hat{\theta}|_{I_{k'}},\theta|_{I_{k'}})$ is minimum for
    $k'\in\{0,1\}$. $\dhB$ then calls $\otF $ with input $k'$ and
    obtains output $s_{k'}$. $\dhB$ sets $m'_{k'} = s_{k'} \oplus
    f_{k'}(x|_{I_{k'}})$ and $m'_{1-k'} \in_R\{0,1\}^\ell$ before
    sending $(m_0,m_1)$ to $\dB$. Finally, $\dhB$ outputs whatever
    $\dB$ outputs.
 
    We now argue that the state output by $\dhB$ is statistically
    close to the state output by $\dB$ when executing $\QOT$ with the
    real $\A$. The only difference is that, while $\dhB$ outputs
    $m'_{1-k'}\in_R\{0,1\}^\ell$, $\dB$ outputs $m_{1-k'} = s_{1-k'}
    \oplus f_{1-k'}(x|_{I_{1-k'}})$. Thus, we simply have to show that
    $m_{1-k'}$ is statistically indistinguishable from uniform in the
    view of $\dB$.

    Note that, since $\theta$ and $\hat{\theta}$ are independent and
    $\theta$ is a uniform $n$-bit string, we have that for any
    $\epsilon>0$,
    $$ d_H(\theta,\hat{\theta})> \frac{(1-\epsilon)n}{2} \, ,$$ 
    except with negligible probability. We can now claim that with
    overwhelming probability
    $$ 
    d_H(\theta|_{I_{1-k'}},\hat{\theta}|_{I_{1-k'}})\geq
    \frac{(1-\epsilon)n}{4} \, .
    $$ 
    Now, since $\dB$ is $\beta$-benign, we get with
    Definition~\ref{def:BenignBob} that
    $$ 
    \Hmin \bigl( X|_{I_{1-k'}} \,\big|\, X|_{I_{k'}} = x|_{I_{k'}}
    \bigr) \geq \frac{(1-\epsilon)n}{4}-\beta n \ \text{ and } \
    H_0(\rho_E)\leq \beta n \, .
    $$ 
    It follows from privacy amplification
    (Theorem~\ref{theo:privacy-amplification}) that
    $f_{1-k'}(x|_{I_{1-k'}})$ is statistically indistinguishable from
    uniform for $\dB$, provided that
    $$ \frac{\ell}{n} < \frac{1}{4}-2\beta-\epsilon'$$ 
    for any $\epsilon'>0$. Finally, by the properties of exclusive-OR,
    we can now also conclude that $m_{1-k'}$ is statistically close to
    uniform. Solving the last inequality for $\beta$, we obtain
    $$ \beta < \frac18 - \frac{\lambda}{2} - \frac{\epsilon'}{2} \, ,$$
    and Theorem~\ref{thm:ot} follows.
  \end{proof}

  Informally, the next Corollary~\ref{cor:ot} states that, when
  compiling the basic protocol $\QOT$, we obtain an improved protocol
  $\compile(\QOT)$ with \emph{hybrid security}\index{hybrid security},
  such that a dishonest Bob is required to have large quantum
  computing power \emph{and} large quantum storage to succeed. For
  completeness, $\compile(\QOT)$ is given explicitly in
  Figure~\ref{fig:compiled.ot}.

\begin{figure}
  \begin{framed}
    \noindent\hspace{-1.5ex} {\sc Protocol $\compile(\QOT):$ }\\[-4ex]
    \begin{description}\setlength{\parskip}{0.5ex}
    \item[{\it Preparation:}]$ $ \\ 
      $\A$ chooses $x \in_R \set{0,1}^m$
      and $\theta \in_R \set{+,\x}^m$ and sends $\ket{x}_{\theta}$
      to~$\B$. $\B$ chooses $\hat{\theta} \in_R \set{0,1}^m$ and
      obtains $\hat{x} \in \set{0,1}^m$ by measuring
      $\ket{x}_{\theta}$ in bases $\hat{\theta}$.
    \item[{\it Verification:}]
      \begin{enumerate}
      \item[]
      \item 
      $\B$ sends $c_i :=
      \mathtt{commit_{pkH}}((\hat{\theta}_i,\hat{x}_i),r_i)$ with
      randomness $r_i$ for all $i = 1,\ldots,m$.
      \item 
      $\A$ sends random $T \subset \{1,\ldots,m \}$ with $|T| = \alpha
      m$. $\B$ opens $c_i \ \forall \ i \in T$, and $\A$ checks that
      the openings were correct and that $x_i = \hat{x}_i$ whenever
      $\theta_i = \hat{\theta}_i$. If all tests are passed, $\A$
      accepts. Otherwise, she rejects and aborts.
      \item 
      $\A$ and $\B$ restrict $x, \theta$ and $\hat{x}, \hat{\theta}$,
      respectively, to the remaining $n$ positions $i \in \bar{T}$.
      \end{enumerate}
    \item[{\it Post-processing:}]
      \begin{enumerate}
      \item[]
      \item
	$\A$ sends $\theta$ to $\B$.
      \item
	$\B$ partitions all positions $1\leq i \leq n$ in two subsets
	according his choice bit $k \in \{ 0,1 \}$: the ``good''
	subset $I_k := \{ i: \theta_i = \hat{\theta}_i \}$ and the
	``bad'' subset $I_{1-k} := \{ i: \theta_i \neq \hat{\theta}_i
	\}$. $\B$ sends $(I_0,I_1)$ in this order.
      \item
	$\A$ sends descriptions of $f_0,f_1\in_R {\F}$ together
	with $m_0 := s_0 \oplus f_0(x|_{I_0})$ and $m_1 := s_1 \oplus
	f_1(x|_{I_1})$.
      \item 
	$\B$ computes $s_k = m_k \oplus f_k(\hat{x}|_{I_k})$.
      \end{enumerate}
    \end{description}
    \vspace{-1.5ex}
  \end{framed}
  \vspace{-1.5ex}
  \caption{Improved Protocol for String OT.}
  \label{fig:compiled.ot} 
\end{figure}
  
  \begin{corollary}\label{cor:ot}
    If $\lambda < \frac14$, then protocol $\compile(\QOT)$ is
    computationally secure against dishonest Bob and unconditionally
    secure against $\gamma (1\!-\!\alpha)$-BQSM Bob with $\gamma <
    \frac14 - 2 \lambda$. Correctness and unconditional security
    against dishonest Alice is maintained during compilation.
  \end{corollary}

  \begin{proof}
    The corollary is obtained by the following steps: First, we sketch
    that the analysis of protocol $\tt RAND \ \QOT$ in~\cite{DFRSS07}
    can be almost analogously applied to $\QOT$. Then, we combine this
    result with our BQSM-theorem (Theorem~\ref{thm:BQSM}). And finally,
    we apply Theorem~\ref{thm:ot} with our compiler-theorem
    (Theorem~\ref{thm:compiler}). Note that, by definition, all these
    transformations do not touch correctness nor unconditional
    security against $\dA$.

    In more detail, the main difference from $\tt RAND \ \QOT$ to
    $\QOT$ is that $\B$ measures all his qubits in the basis
    corresponding to his choice bit $k$, i.e.\ $\theta \in
    [0,1]_{k}$. Since we require these measurements to be in
    random bases $\theta \in_R \{ 0,1 \}^n$, we loose the
    non-interactivity and must include the additional message
    $(I_0,I_1)$ from $\B$ to $\A$ in Step~(2.), so that $\A$ obtains
    the same partitions. However, the partitions are send in fixed
    order and do not allow to conclude on the ``good'' subset
    $I_k$. No other message is sent by $\B$. Furthermore, recall that
    in randomized OT, $\A$ does not input the two messages $s_0,
    s_1$ herself by masking them with the hashed output of the
    measurement outcomes. Instead, only these hash-values, generated
    uniformly at random during the protocol, are output. However, due
    to the characteristic of exclusive-OR, the security properties in
    this aspect do not change. 

    Thus, $\QOT$ inherits the BQSM-security of $\tt RAND \ \QOT$, and
    we can claim that $\QOT$ is unconditionally secure against
    $\gamma$-BQSM Bob for all $\gamma$ strictly smaller than
    $\frac{1}{4} - 2\lambda$. Then, by Theorem~\ref{thm:BQSM}, we
    obtain unconditional security for $\compile(\QOT)$ against
    $\gamma(1\!-\!\alpha)$-BQSM Bob. 

    Last, we know from Theorem~\ref{thm:ot} that $\QOT$ is
    unconditionally secure against a $\beta$-benign Bob for $\beta <
    \frac18 - \frac{\lambda}{2}$. It follows with
    Theorem~\ref{thm:compiler} that $\CO$, instantiated by our
    dual-mode commitment scheme, leads to quantum-computational
    security for $\compile(\QOT)$ against any $\dB$.
  \end{proof}


\section{Password-Based Identification}
\label{sec:hybrid.security.id}
\index{identification}

Password-based identification is introduced in
Section~\ref{sec:primitives.id}, where we also describe a construction
from randomized $\ot{1}{n}^\ell$, and the therewith inherited
OT-security aspects. Secure identification is highly significant in
any authenticated set-up of outer protocols, and may provide
re-usability of the initial user-memorizable passwords, if cleverly
implemented.


\subsection{Motivation and Related Work}
\label{sec:hybrid.security.id.motivation}

  There exist various approaches for classical and quantum
  identification, based on different techniques, e.g.\ on
  zero-knowledge~\cite{FS86,FFS87}, on password-based
  key-agreement~\cite{KOY01}, and on quantum memory
  restrictions~\cite{DFSS07}. Here, we will subject the quantum
  identification scheme of~\cite{DFSS07}, denoted in the following by
  $\tt BQSM\text{--}QID$, to our compiler technique, yielding more
  diverse security assumptions. $\tt BQSM\text{--}QID$ was proven to
  be unconditionally secure against arbitrary dishonest Alice and
  against quantum-memory-bounded dishonest Bob by using
  quantum-information-theoretic security definitions. In~\cite{FS09}
  it was then shown that these security definitions imply
  simulation-based security as considered here, with respect to the
  functionality $\idF$ given in
  Figure~\ref{fig:ideal.functionality.id}. Actually, the definition
  and proof from~\cite{DFSS07} guarantee security only for a slightly
  weaker functionality, which gives some unfair advantage to dishonest
  Alice in case she guesses the password correctly. However, as
  discussed in~\cite{FS09}, the protocol from~\cite{DFSS07} does
  implement functionality $\idF$.


\subsection{The Protocol}
\label{sec:hybrid.security.id.protocol}

  Recall that we require from an identification scheme that a user
  $\A$ succeeds in identifying herself to a server $\B$, if she knows
  an initial, secret password $w$. Additionally, a dishonest user
  $\dA$ should not succeed with higher probability than at a guess,
  and similarly, a dishonest server $\dB$ should be only able to guess
  $\A$'s password without learning anything beyond the (in)correctness
  of his guess. These last requirements provide re-usability of the
  password. To achieve security under realistic assumptions, we
  further want to allow memorizable passwords with low entropy.

  Let $\cal W$ be the set of possible keys, not necessarily large in
  size, with $w \in \cal{W}$ denoting the initially shared
  password. For clarity, we will often use $w_A$ and $w_B$ to
  indicate $\A$'s and $\B$'s input to the protocol, and only accept if
  $w_A = w_B$, which implies equality to $w$. Let
  $\mathfrak{c}:{\cal W} \rightarrow \set{+,\x}^n$ be the encoding
  function of a binary code of length $n$ with $|{\cal W}|$ codewords
  and minimal distance $d$. Families of codes as required for our
  subsequent results, correcting a constant fraction of errors
  efficiently and with constant information rate are indeed known
 ~\cite{SS96}. And finally, let $\F$ and $\G$ denote suitable families
  of (strongly) two-universal hash-functions, as specified in
  Definition~\ref{def:hashing}, with range $\F: \set{0,1}^n
  \rightarrow \set{0,1}^\ell$ and $\G: \mathcal{W} \rightarrow
  \set{0,1}^\ell$, respectively. Again we stress that we can pad the
  input to the functions with zero, if it is smaller than expected.

\begin{figure}
  \begin{framed}
    \noindent\hspace{-1.5ex}{\sc Functionality $\idF$ :}\\ Upon
    receiving $\mathtt{w_A, w_B} \in \cal W$ from Alice and Bob,
    respectively, $\idF$ outputs the bit \smash{$\mathtt{y}:=
    (\mathtt{w_A} \eqq \mathtt{w_B})$} to Bob. In case Alice is
    dishonest, she may choose $\mathtt{w_A} = \,\perp$ (where $\perp
    \, \not\in \cal W$), and (for any choice of $\tt w_A$) the bit
    $\tt y$
    is also output to Alice.
    \vspace{-1ex}
  \end{framed}
  \vspace{-1.5ex}
  \caption{The Ideal Functionality for Password-Based Identification.}
  \label{fig:ideal.functionality.id}
\end{figure}

  We cannot directly apply our compiler to the original $\tt
  BQSM\text{--}QID$, since it is {\em not} a BB84-type
  protocol. Similar to $\tt RAND \ \QOT$ described in the previous
  Section~\ref{sec:hybrid.security.ot}, $\B$ does not measure the
  qubits in a random basis but in a basis-string $c$ determined by
  his password $w_B \in {\cal W}$ by $c =
  \mathfrak{c}(w_B)$. After $\A$'s basis announcement, both players
  compute set $I_w = \set{i:\theta_i = \mathfrak{c}(w)_i}$
  with the positions on which they base the last steps of the
  post-processing.

  We briefly sketch now the transformation from $\tt BQSM\text{--}QID$
  into a BB84-type protocol, without affecting security and without
  loosing efficiency. The first step is naturally to let $\B$ measure
  in random basis $\hat{\theta} \in_R \set{\+,\x}^n$. The most
  straightforward next step would be to include a new message from
  $\B$ to $\A$ during post-processing, in which $\B$ announces $I_B =
  \set{i:\hat{\theta}_i = \mathfrak{c}(w)_i}$. Then, $\A$ sends
  $\theta$ and the remaining post-processing could be conducted on
  $I_w = \set{i \in I_B: \theta_i = \hat{\theta}_i}$. Note, however,
  that this solution here is less efficient than in the original
  protocol, since only approx.\ $1/4$ of all measurement outcomes
  could be used. So instead, we let Bob apply a {\em random shift}
  $\kappa$ to the code, which $\B$ announces to $\A$ in the
  post-processing phase, namely $\hat{\theta}= \mathfrak{c}(w) \oplus
  \kappa$ with $\kappa\in \{ +, \times \}^n$ and $\+ \equiv 0$ and $\x
  \equiv 1$ for computing the $\oplus$-operation. Then, we define
  $\mathfrak{c'}(w):= \mathfrak{c}(w) \oplus \kappa$. Finally, after
  $\A$'s announcements of $\theta$ the protocol is completed with the
  shifted code, i.e., based on positions in $I_w :=
  \set{i:\theta=\mathfrak{c'}(w)_i}$. This has the effect that the
  post-processing is actually based on positions $i$ with $\theta_i =
  \hat{\theta}_i$, and thus, on approx.\ $1/2$ of all qubits as in
  protocol $\tt BQSM\text{--}QID$. Our resulting protocol $\QID$ is
  described in Figure~\ref{fig:basic.id}. We show in the following
  proofs that the modification does not affect security as given
  in~\cite{DFSS07} (and~\cite{FS09}).

\begin{figure}
  \begin{framed}
    \noindent\hspace{-1.5ex} {\sc Protocol $\QID:$ } \\[-4ex]
    \begin{description}
    \item[{\it Preparation:}] $ $ \\ 
      $\A$ chooses $x \in_R \set{0,1}^n$ and $\theta \in_R
      \set{+,\x}^n$ and sends $\ket{x}_{\theta}$ to~$\B$. $\B$ chooses
      $\hat{\theta} \in_R \set{0,1}^n$ and obtains $\hat{x} \in
      \set{0,1}^n$ by measuring $\ket{x}_{\theta}$ in bases
      $\hat{\theta}$.
    \item[{\it Post-processing:}] 
      \begin{enumerate}
	\item[]
	\item
	  $\B$ computes a string $\kappa\in \{ +, \times \}^n$ such
	  that $\hat{\theta}= \mathfrak{c}(w) \oplus \kappa$, where we
	  think of $\+$ as 0 and $\x$ as 1 so that $\oplus$ makes
	  sense. He sends $\kappa$ to $\A$ and we define
	  $\mathfrak{c'}(w):= \mathfrak{c}(w) \oplus \kappa$.
	\item 
	  $\A$ sends $\theta$ and $f \in_R \mathcal{F}$ to $\B$. Both
	  compute $I_w := \set{i:
	  \theta_i=\mathfrak{c'}(w)_i}$.
	\item 
	  $\B$ sends $g \in_R \mathcal{G}$.
	\item 
	  $\A$ sends $z:= f(x|_{I_{w}}) \oplus g(w)$
	  to $\B$.
	\item $\B$ accepts if and only if
	$z=f(\hat{x}|_{I_w}) \oplus g(w)$.
      \end{enumerate}
    \end{description}
    \vspace{-1.5ex}
  \end{framed}
  \vspace{-1.5ex}
  \caption{Basic Protocol for Password-Based Identification}
  \label{fig:basic.id} 
\end{figure}

  \begin{proposition}
    Protocol $\QID$ satisfies correctness and achieves unconditional
    security against dishonest Alice, according to
    Definitions~\ref{def:correctness}
    and~\ref{def:unboundedAliceNiceOrder}, respectively.
  \end{proposition}

  \begin{proof}
    Correctness for honest players is obvious: If both $\A$ and $\B$
    know $w$, i.e.~$w_A = w_B$, they can compute $\mathfrak{c}(w)$ and
    $\mathfrak{c'}(w)$. Following the last steps as supposed to, they
    conclude with $f(x|_{I_w}) \oplus g(w_A) = f(\hat{x}|_{I_w})
    \oplus g(w_B)$.

    Security against dishonest $\dA$ is derived in a straightforward
    way from $\tt BQSM\text{--}QID$ as follows. Recall that in $\tt
    BQSM\text{--}QID$, $\B$ measures all his qubits in one basis,
    depending on $c = \mathfrak{c}(w)$. In $\QID$, the preparation
    phase combined with Step~(1.)~of the post-processing, where $\B$
    sends $\kappa$, can be seen as an equivalent situation from the
    view of $\dA$. The important positions are now defined by
    $\mathfrak{c'}(w)$, which is however only deducible if
    $\mathfrak{c}(w)$ is known in addition, since otherwise, $\kappa$
    looks completely random. All subsequent steps are exactly as in
    $\tt BQSM\text{--}QID$, and thus, the same analysis can be applied
    to $\QID$. In the following, we will sketch the intuitive idea
    thereof. $\dA$ runs the protocol with a memory-unbounded server
    who measures his qubits in $\dA$'s bases $\theta$ and therefore
    obtains $x$. He then computes $s_j = f(x|_{I_j}) \oplus g(j)$ for
    all codewords $j = 1,\ldots,|\cal W|$, where $s_{w}$ would be
    expected from $\dA$ for an accepting run of the protocol. By the
    strongly universal-two property of $g$, all $s_j$ are pairwise
    independent, and thus, it follows that all $s_j$'s are distinct,
    except with some negligible probability. Assume that the accepting
    message is one of the $s_j$'s for a \emph{random} variable $w'$,
    i.e.~$z = s_{w'}$. $\dA$ will only succeed, if $w' = w$, and $\dA$
    does not learn anything beyond that. A further analysis of $\dA$'s
    state before the final accept/reject-message shows its
    independence from $w$, given $w'$ and conditioned on $w' \neq w$
    and on the pairwise distinction of all $s_j$'s. And finally, for
    $\dA$'s state after the final message it is shown that the event
    of all distinct $s_j$'s is independent of $w$ and
    $w'$. Statistical security against $\dA$ follows.
  \end{proof}

  \begin{theorem}\label{thm:id}
    If $\mathfrak{c}:{\cal W} \rightarrow \set{+,\x}^n$ has minimal
    distance $d \geq \delta n$ and is polynomial-time decodeable, then
    protocol $\QID$ is unconditionally secure against $\beta$-benign
    Bob for any $\beta < \frac{\delta}{4}$.
  \end{theorem}

  \begin{proof}
    For any given benign Bob $\dB$, we construct $\dhB$ as
    follows. $\dhB$ runs locally a copy of $\dB$ and simulates Alice's
    actions by running $\A$ faithfully except for the following
    modifications.

    After the preparation phase, $\dhB$ gets $\hat{\theta}$ and
    $\kappa$ from $\dB$. It then computes $w' \in \cal W$ such that
    $\mathfrak{c}(w')$ has minimal Hamming distance to $\hat{\theta}
    \oplus \kappa$. Note that this can be done in polynomial-time by
    assumption on the code. Then, $\dhB$ submits $w'$ as input $w_B$
    to $\idF$ and receives output $y \in \set{0,1}$. If $y = 1$, then
    $\dhB$ faithfully completes $\A$'s simulation using $w'$ as
    $w$. Otherwise, $\dhB$ completes the simulation by using a random
    $z'$ instead of $z$. In the end, $\dhB$ outputs whatever $\dB$
    outputs.

    We need to show that the state output by $\dhB$ (respectively
    $\dB$) above is statistically close to the state output by $\dB$
    when executing $\QID$ with real $\A$. For simpler notation, we use
    $w$ for honest Alice's input $w_A$. Note that if $w' = w$, then
    the simulation of $\A$ is perfect and thus the two states are
    equal. If $w' \neq w$ then the simulation is not perfect, as
    the real $\A$ would use $z= f(x|_{I_{w}}) \oplus g({w})$ instead
    of random $z'$. It thus suffices to argue that $f(x|_{I_{w}})$ is
    statistically close to random and independent of the view of $\dB$
    for any fixed ${w} \neq w'$. Note that this is also what had to be
    proven in~\cite{DFSS07}, but under a different assumption, namely
    that $\dB$ has bounded quantum memory, rather than that he is
    benign. Nevertheless, we can recycle part of the proof.

    Recall from the definition of a benign Bob that the common state
    after the preparation phase is statistically close to a state for
    which it is guaranteed that $\Hmin(X|_I) \geq
    d_H(\theta|_I,\hat{\theta}|_I) - \beta n$ for any $I \subseteq
    \set{1,\ldots,n}$, and $\Hmax(\rho_E) \leq \beta n$. By the
    closeness of these two states, switching from the real state of
    the protocol to the ideal state satisfying these bounds, has only
    a negligible effect on the state output by $\dhB$. Thus, we may
    assume these bounds to hold.

    Recall that $\hat{\theta} \oplus \kappa$ is at Hamming distance at
    most $d/2 $ from $\mathfrak{c}(w')$. Since the distance from here
    to the (distinct) codeword $\mathfrak{c}({w})$ is greater than
    $d$, we see that $\hat{\theta}\oplus\kappa$ is at least $d/2$ away
    from $\mathfrak{c}({w})$. It follows that $ \mathfrak{c}'(w) =
    \mathfrak{c}(w) \oplus \kappa$ has Hamming distance at least $d/2$
    from $\hat{\theta}$. Furthermore, for arbitrary $\eps > 0$ and
    except with negligible probability, the Hamming distance between
    $\theta|_{I_{w}} = \mathfrak{c}'({w})|_{I_w}$ and
    $\hat{\theta}|_{I_{w}}$ is at least $d/4 - \eps n$. Therefore, we
    can conclude that
    $$
    \Hmin(X|_{I_{w}}) \geq d/4 - \eps n - \beta n \
    \text { and } \
    \Hmax(\rho_E) \leq \beta n \, .
    $$ 
    We require $\Hmin(X|_{I_{w}}) - \Hmax(\rho_E) - \ell$ to be
    positive and linear in $n$, which is the case here for parameters
    $$
    \beta n \leq d/8 - (\eps/2) \ n - \ell/2 \, .
    $$ 
    We conclude by privacy amplification that $f(x|_{I_w})$ and
    therewith $z$ are close to random and independent of $E$,
    conditioned on $w \neq w$. This concludes the proof.
  \end{proof}

  The next corollary informally states that, when applying our
  compiler to the basic protocol $\QID$, we obtain a hybrid-secure
  protocol\index{hybrid security} $\compile(\QOT)$. Thus, any
  dishonest Bob needs large quantum computing power \emph{and} large
  quantum storage to launch a successful attack. For completeness, we
  again give $\compile(\QID)$ explicitly in
  Figure~\ref{fig:compiled.id}.

  \begin{corollary}\label{cor:id}
    If $|{\cal W}| \leq 2^{\nu n}$, and if $\mathfrak{c}:{\cal W}
    \rightarrow \set{+,\x}^n$ has minimal distance $d \geq \delta n$
    for $\delta > 0$ and is polynomial-time decodeable, then protocol
    $\compile(\QID)$ is computationally secure against dishonest Bob
    and unconditionally secure against $\gamma (1\!-\!\alpha)$-BQSM
    Bob with $\gamma < \frac{\delta}{4} - \nu$. Correctness and
    unconditional security against dishonest Alice is maintained
    during compilation.
  \end{corollary}

  \begin{proof}
    We can show hybrid security by first adapting and connecting the
    results of~\cite{DFSS07} with our BQSM-Theorem~\ref{thm:BQSM}, and
    then combining Theorem~\ref{thm:id} with our compiler theorem
    (Theorem~\ref{thm:compiler}). All definitions preserve correctness
    and unconditional security against $\dA$.

    In more detail, the main difference from $\tt BQSM\text{--}QID$
    of~\cite{DFSS07} to $\QID$ is that $\B$ measures all his qubits in
    the bases corresponding to $c = \mathfrak{c}(w_B)$. Then after
    $\A$'s basis announcement, both players base the remaining
    post-processing on $I_w = \set{i:\theta_i =
    \mathfrak{c}(w)_i}$. In $\QID$ instead, $\B$ measures in
    random bases, computes $\hat{\theta}= \mathfrak{c}(w)
    \oplus \kappa$, and announces $\kappa$ to $\A$. Then after $\A$'s
    basis announcements, the protocol is completed based on positions
    in $I_w := \set{i:\theta=\mathfrak{c'}(w)_i}$ with
    $\mathfrak{c'}(w):= \mathfrak{c}(w) \oplus
    \kappa$. Note that both situations however are equivalent. First,
    the important positions are those $i$ where $\theta_i =
    \hat{\theta}_i$ in both cases. And second, $\kappa$ looks
    completely random and is of no value without the knowledge of
    $\mathfrak{c}({w})$.

    Thus, $\QID$ inherits the BQSM-security of $\tt BQSM\text{--}QID$,
    and we can claim that $\QID$ is unconditionally secure against
    $\gamma$-BQSM Bob for all $\gamma < \frac{\delta}{4} - \nu$. From
    Theorem~\ref{thm:BQSM} unconditional security of $\compile(\QOT)$
    against $\gamma(1\!-\!\alpha)$-BQSM Bob follows. $\QID$ is
    guaranteed by Theorem~\ref{thm:id} to achieve unconditional
    security against a $\beta$-benign Bob for $\beta <
    \frac{\delta}{4}$ and it follows with Theorem~\ref{thm:compiler}
    that $\CO$, instantiated by our dual-mode commitment scheme,
    yields quantum-computational security for $\compile(\QID)$ against
    any $\dB$.
\end{proof}

\begin{figure}
  \begin{framed}
    \noindent\hspace{-1.5ex} {\sc Protocol $\compile(\QID):$ } \\[-4ex]
    \begin{description}
    \item[{\it Preparation:}] $ $ \\ 
      $\A$ chooses $x \in_R \set{0,1}^n$ and $\theta \in_R
      \set{+,\x}^n$ and sends $\ket{x}_{\theta}$ to~$\B$. $\B$ chooses
      $\hat{\theta} \in_R \set{0,1}^n$ and obtains $\hat{x} \in
      \set{0,1}^n$ by measuring $\ket{x}_{\theta}$ in bases
      $\hat{\theta}$.
    \item[{\it Verification:}]
      \begin{enumerate}
      \item[]
      \item 
      $\B$ sends $c_i :=
      \mathtt{commit_{\pkH}}((\hat{\theta}_i,\hat{x}_i),r_i)$ with
      randomness $r_i$ for all $i = 1,\ldots,m$.
      \item 
      $\A$ sends random $T \subset \{1,\ldots,m \}$ with $|T| = \alpha
      m$. $\B$ opens $c_i \ \forall \ i \in T$, and $\A$ checks that
      the openings were correct and that $x_i = \hat{x}_i$ whenever
      $\theta_i = \hat{\theta}_i$. If all tests are passed, $\A$
      accepts. Otherwise, she rejects and aborts.
      \item 
      $\A$ and $\B$ restrict $x, \theta$ and $\hat{x}, \hat{\theta}$,
      respectively, to the remaining $n$ positions $i \in \bar{T}$.
      \end{enumerate}
    \item[{\it Post-processing:}] 
      \begin{enumerate}
	\item[]
	\item
	  $\B$ computes a string $\kappa\in \{ +, \times \}^n$ such
	  that $\hat{\theta}= \mathfrak{c}(w) \oplus \kappa$,
	  where we think of $\+$ as 0 and $\x$ as 1 so that $\oplus$
	  makes sense. He sends $\kappa$ to $\A$ and we define
	  $\mathfrak{c'}(w):= \mathfrak{c}(w) \oplus
	  \kappa$.
	\item 
	  $\A$ sends $\theta$ and $f \in_R \mathcal{F}$ to $\B$. Both
	  compute $I_w := \set{i:
	  \theta_i=\mathfrak{c'}(w)_i}$.
	\item 
	  $\B$ sends $g \in_R \mathcal{G}$.
	\item 
	  $\A$ sends $z:= f(x|_{I_w}) \oplus g(w)$
	  to $\B$.
	\item $\B$ accepts if and only if
	$z=f(\hat{x}|_{I_w}) \oplus g(w)$.
      \end{enumerate}
    \end{description}
    \vspace{-1.5ex}
  \end{framed}
  \vspace{-1.5ex}
  \caption{Improved Protocol for Password-Based Identification}
  \label{fig:compiled.id} 
\end{figure}


\section{Man-in-the-Middle Security for Identification}
\label{sec:hybrid.security.mitm}
\index{player!man-in-the-middle}

In a man-in-the-middle attack, we assume an external adversary who
attacks an execution of the protocol with honest communicating
parties, while having full control over the classical and the quantum
communication.


\subsection{Motivation}
\label{sec:hybrid.security.mitm.motivation}

  The compiled quantum protocols from
  Sections~\ref{sec:hybrid.security.ot} and
  \ref{sec:hybrid.security.id} protect against (arbitrary) dishonest
  Alice and against (computationally or quantum-storage bounded)
  dishonest Bob. However, in particular in the context of
  identification, it is also important to protect against a
  man-in-the-middle attacker Eve ($\E$). Both, $\QID$ and
  $\compile(\QID)$, are insecure in this model. Eve might measure one
  of the transmitted qubits, say, in the $\+$-basis, and this way
  learn information on the basis $\hat{\theta}_i$ used by $\B$, and
  thus on the password $w$, simply by observing if $\B$ accepts or
  rejects in the end.\footnote{Note that this attack does not
  immediately apply to the scheme sketched in the previous section,
  but similar, however more sophisticated, attacks may still apply.}

  In~\cite{DFSS07} it was shown how to enhance $\tt BQSM\text{--}QID$
  in order to obtain security (in the bounded-quantum-storage model)
  against man-in-the-middle attacks. The very same techniques can also
  be used to obtain {\em hybrid security} against man-in-the-middle
  attacks for $\compile(\QID)$. The techniques from~\cite{DFSS07}
  consist of the following two add-on's to the original protocol.
  \begin{enumerate}
    \item 
      A test on a random subset of qubits in order to detect
      disturbance of the quantum communication. 
    \item 
      Authentication of the classical communication. 
  \end{enumerate}
  First note that $\compile(\QID)$ already does such a check as
  required in Point~(1.), namely in the verification phase, so this is
  already taken care of here. Point~(2.)~requires that Alice and Bob,
  in addition to the password, share a high-entropy key $k$ that could
  be stored, e.g.\ on a smart-card. This key will be used for a
  so-called {\em extractor MAC}. Besides being a MAC, i.e.\ a message
  authentication code, such a construction has the additional property
  that it also acts as an extractor. This means that if the message to
  be authenticated has high enough min-entropy, then the key-tag pair
  is close to randomly and independently distributed. As a
  consequence, the tag gives away (nearly) no information on $k$, and
  thus, $k$ can be re-used in the next execution of the
  protocol.\footnote{This is in sharp contrast to the standard way of
  authenticating the classical communication, where the authentication
  key can only be used a bounded number of times.} For further
  details, we refer to~\cite{DFSS07,DKRS06}.

  More specifically, in order to obtain hybrid security against
  man-in-the-middle attacks for $\compile(\QID)$, $\A$ will, in her
  last move of the protocol, use the extractor MAC to compute an
  authentication tag on all the classical messages exchanged plus the
  string $x|_{I_w}$. This, together with the test of a random subset,
  prevents Eve from interfering with the (classical and quantum)
  communication without being detected, and security against Eve
  essentially follows from the security against impersonation
  attacks. Note that including the $x|_{I_w}$ into the authenticated
  message guarantees the necessary min-entropy, and as such the
  re-usability of the key $k$.

  We emphasize that the protocol is still secure against impersonation
  attacks (i.e.\ dishonest Alice or Bob), even if the adversary knows
  $k$, but with slightly weaker parameters due to the ``entropy-loss''
  within $x|_{I_w}$, caused by the additional information for
  authentication and private error correction that is now available.


\subsection{The Set-Up}
\label{sec:hybrid.security.mitm.setup}

  In addition to the previous setting in
  Section~\ref{sec:hybrid.security.id}, we now have the following
  assumptions. Let $MAC^*: \mathcal{K} \times \mathcal{M} \rightarrow
  \{0, 1\}^\ell$ be the extractor MAC with arbitrary key space
  $\mathcal{K}$, message space $\mathcal{M}$ and error probability
  $2^{-\ell}$. Its extractor property guarantees that for any message
  $M$ and quantum state $E$ (which may depend on $M$), the tag $T =
  MAC^*(K,M)$ of $M$ is such that $\rho_{T K E}$ is
  $2^{-(\Hmin(M)-\Hmax(\rho_E)-\ell)/2}$-close to $\frac{1}{2^\ell}
  \mathbb{I} \otimes \rho_K \otimes \rho_E$. Recall that
  $\mathfrak{c}:{\cal W} \rightarrow \set{+,\x}^n$ is the encoding
  function of a binary code with minimal distance $d$, and we have
  strongly universal-2 classes of hash-functions $\mathcal{F}:
  \{0,1\}^n \rightarrow \{ 0,1 \}^{\ell}$ and $\mathcal{G}:
  \mathcal{W} \rightarrow \{ 0,1 \}^{\ell}$.

  In order to do error correction, let $\set{syn_j}_{j \in \cal J}$ be
  the family of syndrome functions\footnote{Note that we have the
  following convention: For a bit string $y$ of arbitrary length,
  $syn_j(y)$ is to be understood as $syn_j(y0\cdots0)$ with enough
  padded zeros if its bit length is smaller than $n'$, and as
  $\big(syn_j(y'),y''\big)$, where $y'$ consist of the first $n'$ and
  $y''$ of the remaining bits of $y$, if its bit length is bigger than
  $n'$.}, corresponding to a family ${\cal C} = \set{C_j}_{j \in \cal
  J}$ of linear error correcting codes of size $n' = n/2$, where $n =
  (1-\alpha)m$. We require the property that any $C_j$ allows to
  efficiently correct a $\phi''$-fraction of errors for some constant
  $\phi'' > 0$. For a random $j \in \cal J$, the syndrome of a string
  with $t$ bits of min-entropy is $2^{-\frac14(t-2q)}$-close to
  uniform given $j$ and any quantum state with max-entropy at most
  $q$. We refer to~\cite{DS05,DFSS07,FS08} for the existence of such
  families and example constructions. Protocol $\QIDplus$ can
  tolerate a noisy quantum communication up to any error rate $\phi <
  \phi''$. We stress that for security against man-in-the-middle
  attacks, error correction with $\phi'' > 0$ needs to be done even if
  we assume perfect quantum communication (with $\phi = 0$), as should
  become clear from the analysis of the protocol given below. Finally,
  we let $\phi'$ be a constant such that $\phi < \phi' < \phi''$.\\

  The ideal functionality $\F_\IDplus$ is given in
  Figure~\ref{fig:ideal.functionality.mitm}. The following definition
  captures unconditional security against a man-in-the-middle
  attacker, where $\E$ gets classical $W'$ and quantum state $E$ as
  input and both honest players $\A$ and $\B$ get classical input $W$
  and $K$. The joint state is then of the form
  $$ 
  \rho_{K W W' E|W' \neq W} = \rho_K \otimes \rho_{W\leftrightarrow
  W'\leftrightarrow E| W' \neq W} \, .
  $$ 
  Note that we require that the adversary's quantum register $E$ is
  correlated with the honest players' parts only via her classical
  input $W'$, conditioned on $W \neq W'$.
  \begin{definition}[Unconditional security against a Man-in-the-middle]
    \label{def:unconditional.Eve} 
    A protocol \\ $\Pi$ implements an ideal classical functionality $\F$
    unconditionally securely against a man-in-the-middle attacker, if
    for any real-world adversary $\E$, there exists an ideal-world
    adversary $\dhE$, such that, for any input state as specified
    above, it holds that the outputs in the real and the ideal world
    are statistically indistinguishable, i.e.,
    $$
    out_{\A,\B,\E}^\Pi \approxs out_{\hA,\hB,\dhE}^\F \, .
    $$
  \end{definition}

\begin{figure}
  \begin{framed}
    \noindent\hspace{-1.5ex}{\sc Functionality $\F_{\IDplus}$ :}\\ The
    ideal functionality $\F_{\IDplus}$ receives pairs of strings
    $(w_A,k_A)$ and $(w_B,k_B)$ from honest Alice and Bob, and a
    string $w_E$ from Eve, where $w_A,w_B \in \cal W$ and $k \in
    \mathcal{K}$. If $w_E = w_A$, then $\F_{\IDplus}$ sends $({\sf
    correct},k_A)$ to Eve. Otherwise, $\F_{\IDplus}$ sends ${\sf
    incorrect}$. Last, Eve is asked to input an ``override bit'' $d$,
    and $\F_{\IDplus}$ outputs the bit $(w_A \eqq w_B) \wedge d$ to
    Bob and to Eve.
    \vspace{-1ex}
  \end{framed}
  \vspace{-1.5ex}
  \caption{The Ideal Functionality with Man-in-the-Middle Security.}
  \label{fig:ideal.functionality.mitm}
\end{figure}

  Computational security against a man-in-the-middle is defined as
  follows. For a given value of the security parameter $m$, the
  common reference string $\crs$ is chosen at first. The
  polynomial-size input sampler takes as input $m$ and $\crs$ and
  samples an input state of the form
  $$
  \rho_{W_A K_A W_B K_B Z E} = \rho_{\MC{W_A K_A W_B K_B}{Z}{E}}
  \, ,
  $$ 
  where honest Alice gets as input password $W_A$, honest Bob gets
  $W_B$, and Eve's quantum register $E$ is correlated with the honest
  player's part only via her classical input $Z$. In addition to their
  passwords $W_A,W_B$, the honest players are given high-entropy keys
  $K_A,K_B$. We restrict the input sampler to choose $K_A$ uniformly
  at random from $\mathcal{K}$ and guarantee that $K_A=K_B$ whenever
  $W_A = W_B$.
  \begin{definition}[Computational Security against a Man-in-the-middle] 
    \label{def:computational.Eve}
    A protocol $\\ \Pi$ implements an ideal classical functionality $\F$
    computationally securely against a man-in-the-middle attacker, if
    for any poly-time real-world adversary $\E$ who has access to the
    common reference string $\crs$, there exists a poly-time
    ideal-world adversary $\dhE$, not using $\crs$, such that for any
    input sampler as described above, it holds that the outputs in the
    real and the ideal world are quantum-computationally
    indistinguishable, i.e., 
    $$
    out_{\A,\B,\E}^\Pi \approxq out_{\hA,\hB,\dhE}^\F \, .
    $$
  \end{definition}


\subsection{The Protocol}
\label{sec:hybrid.security.mitm.protocol}

  The extended and compiled protocol $\QIDplus$ is depicted in
  Figure~\ref{fig:qid.plus}. Corollary~\ref{cor:Eve} states hybrid
  security against man-in-the-middle attacks, such that a
  computationally or quantum-storage bounded Eve can do no better than
  trying to guess the password. If the guess is incorrect, she learns
  (essentially) nothing.

  \begin{figure}
    \begin{framed}
      \noindent\hspace{-1.5ex} {\sc Protocol $\QIDplus$} \\[-4ex]
      \begin{description}\setlength{\parskip}{0.5ex}
      \item[{\it Preparation:}] $ $\\
	$\A$ chooses $x \in_R \zo^m$ and $\theta \in_R \set{+,\x}^m$
	and sends $\ket{x}_{\theta}$ to~$\B$. $\B$ chooses
	$\hat{\theta} \in_R \set{0,1}^m$ and obtains $\hat{x} \in
	\set{0,1}^m$ by measuring $\ket{x}_{\theta}$ in bases
	$\hat{\theta}$.
      \item[{\it Verification:}] 
	\begin{enumerate}
	\item[]
 	\item 
      	  $\B$ sends $c_i :=
	  \commitk{(\hat{\theta}_i,\hat{x}_i)}{r_i}{\pkH}$ with
	  randomness $r_i$ for $i = 1,\ldots,m$.
	\item\label{step:check} 
	  $\A$ sends random $T \subset \set{1,\ldots,m}$ with $|T| =
	  \alpha m$. $\B$ opens $c_i \ \forall \ i \in T$, and $\A$ checks
	  that the openings were correct and that $x_i = \hat{x}_i$
	  whenever $\theta_i = \hat{\theta}_i$. $\A$ accepts, if this
	  is the case for all but a $\phi^\prime$-fraction of the
	  tested bits. Otherwise, she rejects and aborts.
	\item 
	  $\A$ and $\B$ restrict $x,\theta$ and
	  $\hat{\theta},\hat{x}$,respectively, to the remaining $n$
	  positions $i \in \bar{T}$.
	\end{enumerate}
      \item[{\it Post-processing:}]
	\begin{enumerate}
	\item[]
	\item
	  $\B$ computes a string $\kappa\in \{ +, \times \}^n$ such
	  that $\hat{\theta}= \mathfrak{c}(w) \oplus \kappa$. He sends
	  $\kappa$ to $\A$ and we define $\mathfrak{c'}(w):=
	  \mathfrak{c}(w) \oplus \kappa$.
	\item 
      	  $\A$ sends $\theta$, $f \in_R \mathcal{F}$, $j \in_R \cal
	  J$, and $syn=syn_j(x|_{I_w})$, where $I_w := \{ i: \theta_i
	  = \mathfrak{c'}(w)_i \}$.
	\item 
	  $\B$ sends $g \in_R \mathcal{G}$.
	\item 
      	  $\A$ sends $z:= f(x|_{I_w}) \oplus g(w)$ to
	  $\B$. Additionally, she sends the authentication tag of all
	  previously transmitted classical information, i.e.~$tag^* :=
	  MAC^*_k (\theta, j, syn, f, g, z, \kappa, T, test,
	  x|_{I_w})$ with $test = \{(c_i, \hat{x}_i, \hat{\theta}_i,
	  r_i)\}_{i \in T}$.
	\item 
      	  $\B$ uses $syn$ to correct the errors within
	  $\hat{x}|_{I_w}$, and he accepts if and only if $tag^*$
	  verifies correctly and $z = f(\hat{x}|_{I_w}) \oplus g(w)$.
	\end{enumerate}
      \end{description}
      \vspace{-1.5ex}
    \end{framed}
    \vspace{-1.5ex}
    \caption{Extended and Compiled Protocol for Password-Based
    Identification.}
    \label{fig:qid.plus} 
  \end{figure}

  \begin{corollary}
    \label{cor:Eve}  
    Assume that $|{\cal W}| \leq 2^{\nu n}$ and that
    $\mathfrak{c}:{\cal W} \rightarrow \set{+,\x}^n$ has minimal
    distance $d \geq \delta n$ for $\delta > 0$ and is polynomial-time
    decodeable. Then, protocol $\QIDplus$ is computationally secure
    against Eve with $\beta < \frac{\delta}{6}$, and unconditionally
    secure against $\gamma (1\!-\!\alpha)$-BQSM Eve with $\gamma <
    \frac{\delta}{2} - \nu - 2\ell$.
  \end{corollary}

  We split the proof of Corollary~\ref{cor:Eve} into two parts. First,
  we show computational security in
  Proposition~\ref{pro:computational.Eve}, and second, we show
  unconditional security in the bounded-quantum-storage model in
  Proposition~\ref{pro:unconditional.Eve}.

  \begin{proposition}
    \label{pro:computational.Eve}
    Let $\mathfrak{c}:{\cal W} \rightarrow \set{+,\x}^n$ have minimal
    distance $d \geq \delta n$ and be polynomial-time decodeable.
    Then, $\QIDplus$ is computationally secure against Eve with $\beta
    < \frac{\delta}{6}$, according to
    Definition~\ref{def:computational.Eve}.
  \end{proposition}

  \begin{proof}
    We start with the real-life execution of $\QIDplus$ with honest
    $\A$ and $\B$ with respective inputs $(w_A,k_A)$ and $(w_B,k_B)$,
    and a man-in-the-middle attacker $\E$. We then modify it step by
    step without (essentially) changing the common output state, such
    that in the end we have a simulation of the protocol as required.

    First, we change the action of $\B$ in that we assume that $\B$
    learns in the final step of $\QIDplus$ ``by magic'' whether one of
    the classical messages communicated was modified by $\E$ and
    whether $w_A = w_B$ or not. He accepts the execution if none of
    the messages was modified, if $w_A = w_B$, and if $z$ verifies
    correctly. This changes the outcome of the protocol only by a
    negligible amount. Indeed, if $w_A = w_B$, the restriction on the
    input sampler guarantees that $k_A = k_B$ and the claim follows
    from the security of the MAC. If $w_A \neq w_B$, then $\B$ rejects
    anyway in both versions, except with negligible probability.

    Next, we further change the action of $\B$ in that $\B$ accepts
    the execution in the final step of $\QIDplus$ if none of the
    messages was modified and if $w_A = w_B$ (without verifying
    $z$). We argue that this modification does not change the common
    output state, up to a negligible amount. Note that by
    Lemma~\ref{lemma:ideal.state}, we may replace the real state
    consisting of the qubits obtained by $\A$ and the choice of $T$
    and $T' = \Set{i \in T}{\theta_i = \hat{\theta}_i}$ by a
    negligibly close ideal state (with the same $T$ and $T'$) such
    that the error rate within $T'$, i.e.\ the fraction of $i \in T'$
    with $x_i \neq \hat{x}_i$, gives an exact upper bound on the error
    rate outside of $T$. Thus, it follows that if $\A$ does not reject
    during verification, then $\B$ will recover the correct string
    $x|_{I_w}$ in the final step (except with negligible
    probability) and correctly verify $z$ if and only if $w_A = w_B$.

    The next modification is that $\B$ runs the modified protocol with
    some ``dummy input" instead of his real input $w_B$, but he still
    accepts only if $w_A$ equals his real input $w_B$ and no
    transmitted message was modified by $\E$. Since $\B$ does not
    reveal any information on his input before the last step, this
    modification does not change the common output state at all. We
    write $\B^*$ for this modified $\B$.

    As last modification, we choose an unconditionally binding key
    $\pkB$ as reference string, together with the decryption key
    $\sk$. The new common output state is computationally
    indistinguishable from the previous one by assumption on the
    commitment keys.

    Now, this modified protocol can be simulated by an ideal-life
    adversary $\dhE$ via the following two arguments.\\

    \noindent(1) $\dhE$ can simulate $\A$ as $\dhB$ does in
    the proof of security against dishonest Bob (see
    Theorem~\ref{thm:id}) by sampling unconditionally binding key
    $\pkB$, such that $\dhE$ also knows the decryption key $\sk$,
    extracting $w_B$ from $\B$'s commitments, and inquiring the ideal
    functionality $\F_{\IDplus}$. In more detail, upon receiving
    $\kappa$ from $\B$, $\dhE$ attempts to decode the string
    $\hat{\theta} \oplus \kappa$. If this is successful (a codeword at
    distance at most $d/2$ is returned), it computes the password $w'$
    such that $\mathfrak{c}(w')$ is the decoded codeword. If decoding
    fails, $\dhE$ chooses an arbitrary $w'$. It then sends $w'$ to
    $\F_{\IDplus}$.

    If the functionality replies by $({\sf correct},k_A)$, then $\dhE$
    completes the simulation by following the protocol with inputs $w'
    = w_A$ and $k_A$. In that case, the simulation is perfect and the
    final outputs are equal.
    
    In case the extracted password $w'$ does not match $w_A$, $\dhE$
    follows the protocol but uses random values $syn'$, $tag^{*'}$ and
    $z'$. Note that the real $\A$ would use $z= f(x|_{I_{w_{A}}})
    \oplus g(w_A)$ instead of random $z'$. Thus, we have to argue that
    $f(x|_{I_{w_{A}}})$ is statistically close to random and
    independent of the view of $\E$ (for any fixed $w'$). Recall that
    the common state after the verification phase is statistically
    close to a state for which it is guaranteed that $\Hmin(X|_I) \geq
    d_H(\theta|_I, \hat{\theta}|_I) - \beta n$ for any $I \subseteq
    \set{1,\ldots,n}$, and $\Hmax(\rho_E) \leq \beta n$. Hence,
    switching between these two states has only a negligible effect on
    the final output, and thus we may assume these bounds also to hold
    here. By the way $w'$ was chosen, it is guaranteed that
    $\hat{\theta} \oplus \kappa$ has Hamming distance at most $d/2$
    from $\mathfrak{c}(w')$, which is at distance greater than $d$
    from $\mathfrak{c}(w)$. Thus, the Hamming distance between
    $\hat{\theta} \oplus \kappa$ and $\mathfrak{c}(w)$ is at least
    $d/2$, except with negligible probability. The same holds if
    decoding fails, since $\hat{\theta} \oplus \kappa$ is at least
    $d/2$ away from any codeword and $\mathfrak{c}(w) \oplus \kappa$
    has distance at least $d/2$ from $\hat{\theta}$. It follows that
    the Hamming distance between $\theta|_{I_{w_A}}$ and
    $\hat{\theta}|_{I_{w_B}}$ is at least $(d/4 - \eps)n$. Therefore,
    we can conclude that $\Hmin(X|_{I_{w_A}}) \geq d/4 - \eps n -
    \beta n$.

    Finally, note that by the property of the code family as described
    previously, it follows that if $\Hmin(X|_{I_{w_A}}) > 2
    \Hmax(\rho_E)$ with a linear gap, then $syn$ is close to uniformly
    distributed from $\E$'s point of view. Then, from the extractor
    property of $MAC^*$, it follows that $tag^*$ is essentially random
    and independent of $k, f, test, T, w, w', \theta$ and $E$,
    conditioned on $w \neq w'$. And further, privacy amplification
    guarantees that $f(x|_{I_{w_A}})$ is uniformly distributed and
    thus $z$ is close to random and independent of $E$ (conditioned on
    $w_A \neq w'$). Now, given the two $\ell$-bit strings $tag^*$ and
    $z$, the bound on the min-entropy is slightly reduced by $2\ell$.\\

    \noindent(2) $\dhE$ can also simulate modified $\B^*$ up to before
    the final step, as $\B^*$ uses a ``dummy input". If simulated
    $\A$ rejects in the verification, or $\E$ has modified one of the
    communicated messages, then $\dhE$ sends ``override bit" $d = 0$
    to the ideal functionality. Otherwise, it sends $d = 1$ and
    therewith learns, whether $w_A = w_B$ or not. In both cases,
    $\dhE$ can easily complete the simulation for $\B^*$. The claim
    follows.
  \end{proof}

  \begin{proposition}
  \label{pro:unconditional.Eve}
    If $|{\cal W}| \leq 2^{\nu n}$, then protocol $\QIDplus$ is
    unconditionally secure against $\gamma (1\!-\!\alpha)$-BQSM Eve
    with $\gamma < \frac{\delta}{2} - \nu - 2\ell$, according to
    Definition~\ref{def:unconditional.Eve}.
  \end{proposition}
 
  \begin{proof}
    Here, we can reason similarly to the proof in~\cite{DFSS07}
    against a man-in-the-middle. By the security of the MAC, $\E$
    cannot modify any classical message without being detected (and
    the extractor property guarantees re-usability). Therefore, one
    can show security against $\E$ up to the point {\em before} $\B$
    announces whether to accept the protocol execution or not. 

    In order to show security even after $\B$ has announced his
    decision, one can make the following case distinction. If $\E$
    modifies the quantum communication in such a way that she only
    introduces a few errors in the test set, then she also only
    introduced a few errors in the remaining positions, except with
    small probability. Those positions will be corrected by the error
    correction, and thus, $\B$ accepts---independent of what $w$
    is. In the other case, namely if $\E$ modifies the quantum
    communication in such a way that she introduces many errors in the
    test set, then $\A$ rejects already early in the
    protocol---independent of what $w$ is. Hence, this case
    distinction does not depend on $w$. It follows that $\B$'s
    announcement of whether he accepts or rejects gives away no
    information on $w$.

    Let $w'$ denote $\E$'s guess on the password. Then, if $w' \neq
    w$, $x_{I|_w}$ has $d/4 - \nu$ bits of entropy, given $w$, $w'$
    and $\theta$. Furthermore, given $tag^*$ and $f(x|_{I_w})$, the
    min-entropy is reduced by $2\ell$. By the properties of the code
    family and the privacy amplification property of $MAC^*$ and the
    hash-function, we get that $syn$ as well as $tag^*$ and $f$ are
    essentially random and independent, conditioned on $w \neq w'$,
    for $\gamma < d/4 - \nu - 2\ell$.
  \end{proof}


\clearemptydoublepage
\part{Cryptography in the Quantum World}
\label{part:cryptography.in.quantum.world}


\clearemptydoublepage
\chapter{Introduction}
\label{chap:intro.cryptography.in.quantum.world}

In this part of the thesis, we want to investigate classical
cryptography in the quantum world, which means that we consider the
security of classical protocols subject to quantum attacks. This
scenario is of practical importance and independent of any progress
towards large-scale quantum computing. In the following sections, we
introduce various commitment schemes and extended variants thereof,
which we will use as underlying constructions of the protocols in the
subsequent chapters.

In Chapter~\ref{chap:coin.flip}, we show that a quantum-secure bit
commitment, as discussed in Section~\ref{sec:bit.commit}, implies a
quantum-secure single coin-flip. Then, we will use the mixed
commitments, described in Part~\ref{part:quantum.cryptography},
Section~\ref{sec:mixed.commit}, together with a variation of its
extended construction (described in
Section~\ref{sec:extended.commit.coin}) to equip the underlying
commitment construction with extraction and equivocability such that
we achieve an efficiently simulatable and more general composable
single coin-flip.

In Chapter~\ref{chap:framework}, we propose a framework for the
quantum-secure amplification of the security degree of coins, where we
rely on the mixed commitments of Section~\ref{sec:mixed.commit}. One
step towards a fully simulatable coin-flipping protocol, however,
requires an extended construction allowing for an untypical way of
opening a commitment in that, instead of sending the plaintext, we do
a trapdoor opening (Section~\ref{sec:mixed.commit.trapdoor.opening}).

In Chapter~\ref{chap:coin.flip.applications}, we show different
example applications, where the interactive generation of coins at
the beginning or during outer protocols results in implementations
without set-up assumptions and allows for quantum-secure realizations
of classical schemes.


\section{Regular Bit Commitment}
\label{sec:bit.commit}
\index{commitment}

In Chapter~\ref{chap:coin.flip}, we will show a natural and direct
translation of standard coin-flipping to the quantum world. Recall
from Section~\ref{sec:primitives.commitment} that commitments imply
coin-flipping. More specifically, we require an {\it unconditionally
binding} and {\it quantum-computationally hiding} bit commitment
scheme from $\A$ to $\B$ that takes a bit and some randomness $r$ of
length $\ell$ as input, i.e.\ ${\tt commit}: \set{0,1} \times
\set{0,1}^\ell \rightarrow \set{0,1}^*$. As discussed, the
unconditionally binding property is fulfilled, if it is impossible for
any forger $\Forg$ to open one commitment to both 0 and 1, i.e.\ to
compute $r,r'$ such that $\commitx{0}{r} =
\commitx{1}{r'}$. Quantum-computationally hiding is ensured, if no
quantum distinguisher $\Dist$ can distinguish between $\commitx{0}{r}$
and $\commitx{1}{r'}$ for random $r,r'$ with non-negligible
advantage. Note that we will use this simple notation for the
commitments in the following sections. For a specific scheme, the
precise notation has to be naturally adapted.

For an actual instantiation we can use, for instance, Naor's
commitment based on a pseudorandom generator~\cite{Naor91}. A
pseudorandom generator is a function that maps a short, randomly
chosen seed to a long pseudorandom sequence, which is computationally
indistinguishable from a truly random string for any polynomial-time
bounded adversary. Informally speaking, pseudorandomness ensures
unpredictability of the next bit in the sequence after learning the
previous one. There are two main arguments for commitments based on
pseudorandomness. First, this construction does not require any
initially shared information between the players. This aspect is of
particular importance, when we later propose sequential coin-flipping
for actually implementing the CRS-model assumption, and therewith,
implementing other functionalities from scratch without \emph{any}
set-up assumptions. The second reason relates to our claim of
quantum security. Given any one-way function, pseudorandom generators
can be constructed, where its security parameter is defined by the
length of the seeding key. A brute-force search through the key space
would find all seeds, and thus, all pseudorandom sequences could be
computed. Now, under the assumption of a quantum-secure one-way
function, Grover's optimal quantum search algorithm provides only
quadratic speed-up for brute-searching. More efficient attacks are not
known, and therewith, we claim that for any poly-time bounded quantum
adversary, we achieve quantum-secure schemes.

More formally~\cite{Naor91}, let $f(n)$ denote a function with $f(n) >
n$. Then, $G: \{0,1\}^n \rightarrow \{0,1\}^{f(n)}$ defines a
pseudorandom generator, if for all polynomial-time (quantum)
distinguisher $\Dist$, it holds that
$$|Pr[\Dist(y) = 1] - Pr[\Dist(G(s)) = 1]| \leq \eps \, ,$$ where $y \in_R
\{0,1 \}^{f(n)}$, $s \in_R \{ 0,1 \}^n$, and $\eps$ is negligible in
the security parameter $n$. A bit commitment scheme using
pseudorandomness is now constructed as follows. Let $a$ be the bit to
which Alice wants to commit, and let $G_i(s)$ denote the $i$th bit of
the pseudorandom sequence on seed $s$. To ensure the binding property,
the receiver Bob sends a random vector $R_B = (r_1, \ldots, r_{3n})$
where $r_i \in_R \{ 0,1 \}$ for $1 \leq i \leq 3n$. Alice selects $s
\in_R \{ 0,1 \}^n$ and sends the vector $R_A = (r^\prime_1, \ldots,
r^\prime_{3n})$, where
    $$ r^\prime_i =
      \begin{cases}
	G_i(s) &  \text{if }  r_i = 0 \\
	G_i(s) \oplus a & \text{if } r_i = 1 \, .
      \end{cases}
    $$ 
To open the commitment, Alice sends $s$ and Bob then verifies
that for all $i$, $r^\prime_i = G_i(s)$ for $r_i = 0$ and $r^\prime_i
= G_i(s) \oplus a$ for $r_i = 1$.

Assuming that a dishonest receiver is polynomial-time bounded, he
cannot learn anything about $a$. Otherwise, he could be used to
construct a distinguisher $\Dist$ between pseudorandom and truly
random outputs. This also holds in the quantum world, since the
reduction does not require rewinding. It follows that any
quantum-computationally bounded receiver can only guess $a$ with
probability essentially $1/2$, so the commitment scheme is {\it
quantum-computationally hiding}.

For any (unbounded) dishonest committer, opening a commitment to both
values 0 and 1, requires a seed pair $(s_1, s_2)$, such that sequences
$G_{3n}(s_1)$ and $G_{3n}(s_2)$ agree for all $i$ where $r_i = 0$ and
disagree for all $i$ where $r_i = 1$, i.e.~$r_i = G_i(s_1) \oplus
G_i(s_2)$ for exactly one $R_B$ chosen by the other player. The
probability for the existence of such a pair is at most $2^{2n}/2^{3n}
= 2^{-n}$. It follows that the committer can reveal only one possible
$a$, except with probability less than $2^{-n}$, which satisfies {\it
statistical binding}.


\section{Extended Construction for Mixed Commitments}
\label{sec:extended.commit.coin}
\index{commitment!extended} \index{commitment!extractability}
\index{commitment!equivocability}

We will, also in the context of a single coin-flip, need an extended
construction, which is similar to the extension of
Section~\ref{sec:extended.commit.compiler} but adapted to the case of
an underlying commitment from $\A$ to $\B$ with flavors {\it
unconditionally binding} and {\it quantum-computationally hiding}. We
again aim at providing the respective simulator with a trapdoor for
either extraction to efficiently simulate in case of $\dA$ or
equivocability to avoid rewinding $\dB$. As in
Section~\ref{sec:extended.commit.compiler}, we require a
$\Sigma$-protocol for a (quantumly) hard relation $\Rel = \{(x,w)\}$
with conversations $\tt \big(a^{\Sigma}, c^{\Sigma}, z^{\Sigma}
\big)$. Furthermore, we will also use the keyed dual-mode commitment
scheme described in Section~\ref{sec:mixed.commit.idea}, based on the
multi-bit version of~\cite{PVW08} with keys $\pkH$ and $\pkB$, where
it holds that $\pkH \approxq \pkB$.

In the real protocol, the common reference string consists of
commitment key $\pkB$ and an instance $x'$ for which it holds that
$\nexists \ w'$ such that $(x',w') \in \Rel$, where we assume that $x
\approxq x'$. To commit to bit $a$, $\A$ runs the honest-verifier
simulator to get a conversation $\big( {\tt a^{\Sigma}}, a, {\tt
z^{\Sigma}} \big)$. She then sends $\tt a_{\Sigma}$ and two
commitments $C_0, C_1$ to $\B$, where $C_a = \commitk{{\tt
z^{\Sigma}}}{r_a}{\pkB}$ and $C_{1-a} =
\commitk{0^{z'}}{r_{1-a}}{\pkB}$ with randomness $r_a,r_{1-a}$ and $z'
= |{\tt z^\Sigma}|$. Then, $\big( a, ({\tt z^{\Sigma}}, r_a) \big)$ is
send to open the relevant commitment $C_a$, and $\B$ checks that
$\big( {\tt a^{\Sigma}}, a, {\tt z^{\Sigma}} \big)$ is an accepting
conversation. Assuming that the $\Sigma$-protocol is honest-verifier
zero-knowledge and $\pkB$ leads to unconditionally binding
commitments, the new commitment construction is again unconditionally
binding.

During simulation, $\dhA$ chooses a $\pkB$ such that it knows the
matching decryption key $sk$. Then, it can extract $\dA$'s choice bit
$a$ by decrypting both $C_0$ and $C_1$ and checking which contains a
valid $\tt z^\Sigma$. Again, not both $C_0$ and $C_1$ can contain a
valid reply, since otherwise, $\dA$ would know a $w'$ such that
$(x',w') \in \Rel$. In order to simulate in case of $\dB$, $\dhB$
chooses $\pkH$ and $x$. Hence, the commitment is unconditionally
hiding in this simulation. Furthermore, it can be equivocated, since
now $\exists \ w$ with $(x,w) \in \Rel$ and therefore, $C_0, C_1$ can
both be computed with valid replies, i.e.~$C_0 = \commitk{{\tt
z^\Sigma}_0}{r_0}{\pkH}$ and $C_1 = \commitk{{\tt
z^\Sigma}_1}{r_1}{\pkH}$. Quantum-computational security against $\dB$
follows from the indistinguishability of the keys $\pkB$ and $\pkH$
and the indistinguishability of the instances $x$ and $x'$, and
efficiency of both simulations is ensured, due to extraction and
equivocability.


\section{Trapdoor Opening for Mixed Commitments}
\label{sec:mixed.commit.trapdoor.opening}
\index{commitment!trapdoor opening}
\index{commitment!equivocability}

The typical notion of mixed commitment schemes is stronger than we
require for our basic construction of mixed commitments, namely, it
postulates trapdoors for both extraction and equivocability. As
previously discussed, it suffices in our basic construction to only
rely on an extraction trapdoor. This aspect is very convenient, since
it allows us to weaken the assumption on its underlying construction,
i.e., we can build it from a public-key crypto-system with regular
keys $pk$ and $sk$ as binding commitment key and extraction key, and
require only an indistinguishable hiding key, generated as a random
string in the key space. This, in turn, offers the possibility of
generating the hiding key solely by a precedent interactive
coin-flipping procedure without any set-up assumptions. For a more
advanced usage of commitments as in our strong coin-flipping notion in
Chapter~\ref{chap:framework}, however, we have (in some sense) the
requirement of equivocability. We want to maintain the interactive
generation of the key at any rate, which means that we do not have
enough control of its generation and even less control to equip it
with a trapdoor (as done in
Sections~\ref{sec:extended.commit.compiler}
and~\ref{sec:extended.commit.coin}).

We therefore develop a special notion of trapdoor opening, where the
ability to do a trapdoor opening is not associated to a special
knowledge of the hiding key, but is rather done by cheating in the
opening phase. Specifically, we do the opening not by sending the
plaintext and the randomness, committed to in the first phase but
instead by sending only the plaintext and then doing an interactive
proof that this plaintext is indeed what was committed to. The ability
to do trapdoor openings will then be associated with being able to
control the challenge in the interactive proof. We will get this
control by using a weak coin-flipping protocol as sub-protocol. This will
be one of the essential steps in bootstrapping fully simulatable
strong coin-flipping from weak coin-flipping.

As before, we denote the mixed string commitment scheme of
Section~\ref{sec:mixed.commit} by $\commit_{pk}$. Let $\kappa$ be the
security parameter defining the key space $\zo^\kappa$ and let
$\sigma$ be the secondary security parameter controlling the soundness
error in the interactive proof, which we want to be negligible in
$\sigma$ when $\commit_{pk}$ is unconditionally binding. We equate the
plaintext space $\zo^\ell$ of $\commit_{pk}$ with the Galois field
$\FF = \FF_{2^\kappa}$. The new extended commitment scheme, equipped
with the possibility to do trapdoor openings, is denoted by
$\Commit_{pk}$. We assume its plaintext space to be $\FF^\sigma$ and
denote by $\sss$ a secret sharing scheme over $\FF$.

Given message $m = (m_1, \ldots, m_\sigma) \in \FF^\sigma$ and
randomizer $s = (s_1, \ldots, s_\sigma) \in \FF^\sigma$, let
$f_{m,s}({\tt X})$ denote the unique polynomial of degree $2\sigma-1$,
for which $f_{m,s}(-i+1) = m_i$ for $i = 1, \ldots, \sigma$ and
$f_{m,s}(i) = s_i$ for $i = 1, \ldots, \sigma$. Furthermore, we ``fill
up'' positions $i = \sigma+1, \ldots, \Sigma$, where $\Sigma = 4
\sigma$, by letting $s_i = f_{m,s}(i)$. The shares are now $s =
(s_1,\ldots,s_\Sigma)$. The new commitment scheme $\Commit_{pk}$ is
described in Figure~\ref{fig:sss.commit}.

We stress two simple facts about this scheme. First, for any message
$m \in \FF^\sigma$ and any subset $S \subset \{ 1, \ldots, \Sigma \}$
of size $\vert S \vert = \sigma$, the shares $s|_S$ are uniformly
random in $\FF^\sigma$, when $S$ is chosen uniformly at random in
$\FF^\sigma$ and independent of $m$. This aspect is trivial for $S =
\{ 1, \ldots, \sigma \}$, as we defined it that way, and it extends to
the other subsets using Lagrange interpolation. And second, if $m^1,
m^2 \in \FF^\sigma$ are two distinct messages, then $\sss(m^1;s^1)$
and $\sss(m^2;s^2)$ have Hamming distance at least $\Sigma - 2
\sigma$. Again, this follows by Lagrange interpolation, since the
polynomial $f_{m^1,s^1}({\tt X})$ has degree at most $2\sigma-1$, and
hence, can be computed from any $2 \sigma$ shares $s_i$ using Lagrange
interpolation. The same holds for $f_{m^2,s^2}({\tt X})$. Thus, if
$2\sigma$ shares are the same, then $f_{m^1,s^1}({\tt X})$ and
$f_{m^2,s^2}({\tt X})$ are the same, which implies that the messages
$m^1 = f_{m^1,s^1}(-\sigma+1),\ldots, f_{m^1,s^1}(0)$ and $m^2 =
f_{m^2,s^2}(-\sigma+1),\ldots, f_{m^2,s^2}(0)$ are the same.

\begin{figure}
  \begin{framed}
  \noindent\hspace{-1.5ex} {\sc Commitment Scheme $\Commit_{pk}$:}
  \begin{itemize}
    \item[]{\sc Commitment Phase:}
      \begin{enumerate}
      \item
	Let message $m \in \FF^\sigma$ be the message to get committed
	to. The committer samples uniformly random $s \in \FF^\sigma$
	and computes the shares $\sss(m;s) = (s_1, \ldots, s_\Sigma)$,
	where $s_i \in \FF$.
      \item
	He computes $\Commitk{m}{(s,r)}{pk} = \big(
	M_1,\ldots,M_\Sigma \big)$. In more detail, for $i = 1,
	\ldots, \Sigma$, the committer computes $M_i =
	\commitk{s_i}{r_i}{pk}$ with shares $s =
	(s_1,\ldots,s_\Sigma)$ and randomness $r =
	(r_1,\ldots,r_\Sigma)$.
      \item 
	The committer sends $(M_1,\ldots,M_\Sigma)$.\\[-2ex]
      \end{enumerate}
    \item[]{\sc Opening Phase:}
      \begin{enumerate}
      \item
	The committer sends the shares $s = (s_1, \ldots, s_\Sigma)$
	to the receiver.
      \item
	If the shares are not consistent with a polynomial of degree
	at most $2\sigma-1$, the receiver aborts. Otherwise, he picks
	a uniformly random subset $S \subset \set{1,\ldots,\Sigma}$ of
	size $\vert S \vert = \sigma$ and sends $S$ to the committer.
      \item
	The committer sends $r|_S$.
      \item
	The receiver verifies that $M_i = \commitk{s_i}{r_i}{pk}$ for
	all $i \in S$. If the test fails, he aborts. Otherwise, he
	computes the message $m \in \FF^\sigma$ consistent with $s$.
      \end{enumerate}
  \end{itemize}
    \vspace{-1ex}
  \end{framed}
  \vspace{-2ex}
  \caption{The Commitment Scheme $\Commit_{pk}$.}
  \label{fig:sss.commit}
  \vspace{-1ex}
\end{figure}

First note that if the underlying commitment $\commit_{pk}$ is
unconditionally hiding, then so is $\Commit_{pk}$. In the following,
we investigate the extraction property of $\Commit_{pk}$, under the
assumption that we work in the unconditionally binding mode of
$\commit_{pk}$. Given any commitment $M = \big( M_1,\ldots,M_\Sigma
\big)$, we extract
$$ 
\big( \xtr{M_1}{sk},\ldots,\xtr{M_\Sigma}{sk} \big) =
(s_1,\ldots,s_\Sigma) = s \, .
$$ Assume $s' = (s'_1,\ldots,s'_\Sigma)$ is the consistent sharing
closest to $s$. That means that $s'$ is the vector which is consistent
with a polynomial $f_{m',s'}({\tt X})$ of degree at most $2\sigma-1$
and which at the same time differs from $s$ in the fewest
positions. Note that we can find $s'$ in poly-time when using a Reed
Solomon code, which has efficient minimal distance decoding. We then
interpolate this polynomial $f_{m',s'}({\tt X})$, let $m' =
f_{m',s'}(-\sigma+1),\ldots,f_{m',s'}(0)$, and define $m'$ to be the
message committed to by $\Commit_{pk}$. Any other sharing $s''=
(s_1'',\ldots,s_\Sigma'')$ must have Hamming distance at least $2
\sigma$ to $s'$. Now, since $s$ is closer to $s'$ than to any other
consistent sharing, it must, in particular, be closer to $s'$ then to
$s''$. This implies that $s$ is at distance at least $\sigma$ to
$s''$.

We will use this observation for proving soundness of the opening
phase. To determine the soundness error, assume that $\Commit_{pk}$
does not open to the shares $s'$ consistent with $s$. As observed,
this implies that $\big( \xtr{M_1}{sk},\ldots,\xtr{M_\Sigma}{sk}
\big)$ has Hamming distance at least $\sigma$ to $s'$. However, when
$\commit_{pk}$ is unconditionally binding, all $M_i$ can only be
opened to $\xtr{M_i}{sk}$. From the above two facts, we have that
there are at least $\sigma$ values $i \in \set{1,\ldots,\Sigma}$ such
that the receiver cannot open $M_i$ to $s_i$ for $i \in S$. Since
$\Sigma = 4 \sigma$, these $\sigma$ bad indices (bad for a dishonest
sender) account for a fraction of $\frac14$ of all points in
$\set{1,\ldots,\Sigma}$. Thus, the probability that none of the
$\sigma$ points in $S$ is a bad index is at most $(\frac34)^\sigma$,
which is negligible. Lemma~\ref{lemma:soundness.sss} follows.
  \begin{lemma}
    \label{lemma:soundness.sss}
    If $pk$ is unconditionally binding, then the probability that an
    unbounded cheating committer can open $M = \Commitk{m}{(s,r)}{pk}$
    to a plaintext different from $\xtr{M}{sk}$ is at most
    $(\frac34)^\sigma$, assuming that the challenge $S$ is picked
    uniformly at random and independent of $M$.
  \end{lemma}

In the context of simulation, we will use the challenge $S$ as the
simulators trapdoor, allowing him to equivocally open his
commitments. In such a simulation, the ideal-world adversary $\hat{S}$
can---by means discussed later---enforce a specific challenge, i.e.,
it is guaranteed that this will be the challenge in the opening
phase. Thus, for simplicity, we assume here that it simply gets a
fixed challenge $S$ as input. The simulation is described in
Figure~\ref{fig:simulation.sss}. Lemma \ref{lemma:simulation.sss}
follows via a hybrid argument, which relies on the
quantum-computational indistinguishability in switching
unconditionally binding and unconditionally hiding commitment keys. We
omit a proof here but refer to Chapter~\ref{chap:framework}, where the
construction will be explicitly proven within its outer construction.

\begin{figure}
  \begin{framed}
  \noindent\hspace{-1.5ex} {\sc Simulating $\Commit_{pk}$ with
  Trapdoor $S$:} \\[-4ex]
  \begin{enumerate}
  \item
    $\hat{S}$ gets as input a uniformly random subset $S \subset
    \set{1,\ldots,\Sigma}$ of size $\sigma$ and an initial message $m
    \in \FF^\sigma$.
  \item
    $\hat{S}$ commits honestly to $m \in \FF^\sigma$ by $M =
    \Commitk{m}{(s,r)}{sk}$, as specified in the commitment phase.
  \item
    $\hat{S}$ is given an alternative message $\tilde{m} \in
    \FF^\sigma$, i.e., the aim is opening $M$ to $\tilde{m}$.
  \item
    $\hat{S}$ lets $s|_S$ be the $\sigma$ messages committed to by
    $M|_S$. Then it interpolates the unique polynomial
    $f_{\tilde{m},s}$ of degree at most $2\sigma-1$ for which
    $f_{\tilde{m},s}(i) = s_i$ for $i \in S$ and for which
    $f_{\tilde{m},s}(-i+1) = \tilde{m}_i$ for $i =
    1,\ldots,\sigma$. Note that this is possible, as we have exactly
    $2\sigma$ points which restrict our choice of
    $f_{\tilde{m},s}$. $\hat{S}$ sends $s = \big(
    f_{\tilde{m},s}(1),\ldots,f_{\tilde{m},s}(\Sigma) \big)$ to the
    receiver.
  \item
    The receiver sends the challenge $S$.
  \item
    For all $i \in S$, the sender opens $M_i$ to
    $f_{\tilde{m},s}(i)$. This is possible, since $f_{\tilde{m},s}(i)
    = s_i$ is exactly the message committed to by $M_i$ when $i \in
    S$.
  \end{enumerate}
  \vspace{-1.5ex}
  \end{framed}
  \vspace{-1.5ex}
  \caption{The Ideal-World Simulation of $\Commit_{pk}$.}
  \label{fig:simulation.sss}
\end{figure}
   
  \begin{lemma}
    \label{lemma:simulation.sss}
    If $\tilde{m} = m$, then the transcript of the protocol is
    identical to that of an honest commitment to $m$, followed by an
    honest opening phase to $m$, and run with a uniformly random
    challenge $S$.

    If $\tilde{m} \ne m$, then the transcript of the protocol is
    quantum-computationally indistinguishable to that of an honest
    commitment to $\tilde{m}$, followed by an honest opening phase to
    $\tilde{m}$, and run with a uniformly random challenge $S$.
  \end{lemma}


\clearemptydoublepage
\chapter{Quantum-Secure Coin-Flipping}
\label{chap:coin.flip}
\index{coin-flipping}

Coin-flipping is introduced in Section~\ref{sec:primitives.coin.flip}
and allows two parties to agree on a uniformly random bit in a fair
way. Security for both parties follows, if neither party can influence
the value of the coin to his advantage. Thus, it enables the parties
to interactively generate true randomness from scratch. The chapter is
based on parts of~\cite{DL09}.


\section{Motivation and Related Work}
\label{sec:coin.flip.motivation}

We are interested in the standard coin-flipping protocol~\cite{Blum81}
with classical message exchange but we here assume that the adversary
is capable of quantum computing. As already mentioned, bit commitment
implies a secure coin-flipping, but even when basing the embedded
commitment on a computational assumption that withstands quantum
attacks, the security proof of the entire coin-flipping (and its
integration into other applications) could previously not be
translated from the classical to the quantum world.

Typically, security against a classical adversary is argued in such a
context by rewinding the adversary in a simulation. Recall that, in
general, rewinding as a proof technique cannot be directly applied in
the quantum world. Based on a recent result of
Watrous~\cite{Watrous09}, which originally allowed to prove
unconditionally that quantum zero-knowledge of certain interactive
proofs is possible and that the classical definitions can be
translated into the quantum world, we show the most natural and direct
quantum analogue of the classical security proof for standard
coin-flipping.

We want to mention an alternative approach, which was independently
investigated but never published~\cite{Smith09}. They propose a
classical protocol for zero-knowledge proofs of knowledge secure
against quantum adversaries. The protocol consists of a commitment
phase and two zero-knowledge proofs. Instead of opening the
commitment, the committer claims the value of the committed coins and
gives the first zero-knowledge proof that the claim is correct. To
simulate this zero-knowledge proof, Watrous' technique is used. Note
that this approach allows for flipping a string of coins in the
commitments, and thus, arrives at a coin-flipping protocol with round
complexity independent of the length of the flipped string at
first. However, the required zero-knowledge proof has round complexity
depending on the security parameter, i.e.\ how many proofs must be
completed to achieve a negligible soundness error. Finally, the
coin-string is used as key to encode the witness and the second
zero-knowledge proof is given that this statement is actually true. As
encryption scheme, they suggest a scheme with similar properties as in
our mixed commitment constructions---but at least to our best
knowledge, the question of its actual secure implementation was left
open.

We stress that we aim at establishing coin-flipping as a stand-alone
tool that can be used in several contexts and different generic
constructions. Some example applications thereof are discussed in
Chapter~\ref{chap:coin.flip.applications}, including an independently
proposed zero-knowledge proof of knowledge. In order to include
coin-flipping securely in other applications, we conclude this chapter
by proving the basic construction secure under sequential composition and
propose an extended construction for general composability.


\section{The Protocol}
\label{sec:coin.flip.protocol}

The standard coin-flipping protocol $\COIN$ is shown in
Figure~\ref{fig:coin.flip}, allowing players $\A$ and $\B$ to
interactively generate a random and fair $\coin$ in one execution
without any set-up requirements. As underlying commitment scheme, we
use the \emph{unconditionally binding} and
\emph{quantum-computationally hiding} scheme described in
Section~\ref{sec:bit.commit} with security parameter $n$. We will use
its simpler notation here, namely $\commitx{a}{r}$ with input $a \in
\zo$, randomness $r \in \ell$ and output in $\zo^*$. To indicate
the opening phase, where $\A$ sends $a$ and $r$, we will write
$\open{a}{r}$. The corresponding ideal coin-flipping functionality $\cF$
is depicted in Figure~\ref{fig:cF}. Note that dishonest $\dA$ may
refuse to open $\commitx{a}{r}$ in the real world after learning
$\B$'s input. For this case, $\cF$ allows her a second input $\bot$,
modeling the abort of the protocol.

\begin{figure}
  \begin{framed}
  \noindent\hspace{-1.5ex} {\sc Protocol $\COIN$:} \\[-4ex]
      \begin{enumerate}\setlength{\parskip}{0.5ex}
      \item $\A$ chooses $a \in_R \set{0,1}$ and computes
      $\commitx{a}{r}$. She sends $\commitx{a}{r}$ to $\B$.
      \item $\B$ chooses $b \in_R \set{0,1}$ and sends $b$ to $\A$.
      \item $\A$ sends $\open{a}{r}$ and $\B$ checks if the opening is
      valid.
    \item Both compute $\coin = a \oplus b$.
    \end{enumerate}
    \vspace{-1.5ex}
    \end{framed}
    \vspace{-1.5ex}
  \caption{The Coin-Flipping Protocol.}
  \label{fig:coin.flip}
\end{figure}

\begin{figure}
  \normalsize
  \begin{framed}
    \noindent\hspace{-1.5ex} {\sc Functionality $\cF:$}\\ Upon
    receiving requests $\start$ from Alice and Bob, $\cF$ outputs
    uniformly random $\coin$ to Alice. It then waits to receive
    Alice's second input $\top$ or $\bot$ and outputs $\coin$ or
    $\bot$ to Bob, respectively.
    \vspace{-1ex}
  \end{framed}
  \vspace{-2ex}
  \caption{The Ideal Functionality for a Coin-Flip.}
  \label{fig:cF}
  \vspace{-1ex}
\end{figure}

\begin{proposition}
  Protocol $\COIN$ satisfies correctness, according to
  Definition~\ref{def:correctness}.
\end{proposition}

Correctness is obvious by inspection of the protocol: If both players
are honest, they independently choose random bits $a$ and $b$. These
bits are then combined via exclusive disjunction, resulting in a
uniformly random $\coin$.

\begin{theorem}
  Protocol $\COIN$ is unconditionally secure against any unbounded
  dishonest Alice according to
  Definition~\ref{def:unboundedAliceNiceOrder}, provided that the
  underlying commitment scheme is unconditionally binding.
\end{theorem}

\begin{proof}
  We construct an ideal-world adversary $\dhA$, such that the real
  output of the protocol is statistically indistinguishable from the
  ideal output produced by $\dhA$, $\cF$ and $\dA$. The ideal-world
  simulation is depicted in Figure~\ref{fig:simulationA}.

\begin{figure}
  \begin{framed}
  \noindent\hspace{-1.5ex} {\sc Simulation $\dhA$ :} \\[-4ex]
    \begin{enumerate}
    \item Upon receiving $\commitx{a}{r}$ from $\dA$, $\dhA$ sends
    $\start$ and then $\top$ to $\cF$ as first and second input,
    respectively, and receives a uniformly random $\coin$.
    \item\label{step:compute-a} $\dhA$ computes $a$ and $r$ from
    $\commitx{a}{r}$.
    \item\label{step:compute-b} $\dhA$ computes $b = \coin \oplus a$
    and sends $b$ to $\dA$.
    \item $\dhA$ waits to receive $\dA$'s last message and outputs
    whatever $\dA$ outputs.
    \end{enumerate}
    \vspace{-1ex}
  \end{framed}
  \vspace{-2ex}
  \caption{The Ideal-World Simulation against dishonest Alice.}
  \label{fig:simulationA}
  \vspace{-1ex}
\end{figure}

  First note that $a, r$ and $\commitx{a}{r}$ are chosen and computed
  as in the real protocol. From the statistically binding property of
  the commitment scheme, it follows that $\dA$'s choice bit $a$ is
  uniquely determined from $\commitx{a}{r} = c$, since for any $c$,
  there exists at most one pair $(a,r)$ such that $c =
  \commitx{a}{r}$, except with probability negligible in the security
  parameter $n$. Hence in the real world, $\dA$ is unconditionally
  bound to her bit before she learns $\B$'s choice bit, which means
  $a$ is independent of $b$. Therefore in
  Step~(\ref{step:compute-a}.), the simulator can correctly (but not
  necessarily efficiently) compute $a$ (and $r$). Note that, in the
  case of unconditional security, we do not have to require the
  simulation to be efficient. However, we show in
  Section~\ref{sec:general.composition.coin} how to extend the
  underlying commitment in order to extract $\dA$'s inputs. This
  extraction requires a extraction trapdoor and yields an efficient
  simulation in the CRS-model. Finally, due to the properties of XOR,
  $\dA$ cannot tell the difference between the random $b$ computed
  from the ideal, random $\coin$ in the simulation in
  Step~(\ref{step:compute-b}.)~and the randomly chosen $b$ of the real
  world. It follows that the simulated output is statistically
  indistinguishable from the output in the real protocol.
\end{proof}

To prove security against any dishonest quantum-computationally
bounded $\dB$, we will follow the lines of argument as in
Section~\ref{sec:security.definition.computational}, in particular
Definition~\ref{def:polyboundedBobCRS}, with slight
modifications. More specifically, we do not require a common reference
string, so we can omit this part of the definition. Thus, we show that
there exists an ideal-world simulation $\dhB$ with output
quantum-computationally indistinguishable from the output of the
protocol in the real world. For the ideal world, we consider the
poly-size input sampler, which takes as input only the security
parameter and produces a valid input state $\rho_{U ZV'} =
\rho_{\MC{U}{Z}{V'}}$ as specified in
Section~\ref{sec:security.definition.computational}.

In a simulation against a \emph{classical} adversary, a classical
poly-time simulator would work as follows. It inquires $\coin$ from
$\cF$, chooses random $a$ and $r$, and computes $b' = \coin \oplus a$
as well as $\commitx{a}{r}$. It then sends $\commitx{a}{r}$ to $\dB$
and receives $\dB$'s choice bit $b$. If $b = b'$, the simulation was
successful. Otherwise, the simulator rewinds $\dB$ and repeats the
simulation. For a security proof against any \emph{quantum}
adversary\index{rewinding!quantum}, we construct a poly-time quantum
simulator proceeding similarly to its classical analogue. However, it
requires quantum registers as work space and relies on Watrous' {\it
quantum rewinding lemma} (see Lemma~\ref{lemma:qrewind}). Recall from
Section~\ref{sec:quantum.rewinding} that Watrous constructs the quantum
simulator for a $\Sigma$-protocol, i.e.\ a protocol in three-move
form, where the verifier flips a single coin in the second step and
sends this challenge to the prover. Since these are the essential
aspects also in our protocol $\COIN$, we can apply Watrous' quantum
rewinding technique (with slight modifications) as a black-box to our
protocol. We also follow his notation and line of argument here. For a
more detailed description and proofs, we refer to~\cite{Watrous09} and
Section~\ref{sec:quantum.rewinding}.

\begin{theorem}
  \label{thm:computational.security.single.coin}
  For $p_0 \geq \frac14$, protocol $\COIN$ is quantum-computationally
  secure against any poly-time bounded dishonest Bob (according to
  Definition~\ref{def:polyboundedBobCRS} but with the modification
  described above), provided that the underlying commitment scheme is
  quantum-computationally hiding.
\end{theorem}

\begin{proof}
  Let $\ket{\psi}$ denote $\dB$'s $n$-qubit auxiliary input. Let $W$
  denote $\dB$'s auxiliary input register, containing
  $\ket{\psi}$. Let $V$ and $B$ denote $\dB$'s work space, where $V$
  is an arbitrary polynomial-size register and $B$ is a single qubit
  register. $\A$'s classical messages are considered in the following
  as being stored in quantum registers $A_1$ and $A_2$. In addition,
  the quantum simulator uses registers $R$, containing all possible
  choices of a classical simulator, and $G$, representing its guess
  $b'$ on $\dB$'s message $b$ in the second step. Finally, let $X$
  denote a working register of size $k$, which is initialized to the
  state $\ket{0^{\tilde{k}}}$ and corresponds to the collection of all
  registers as described above except $W$.

  The quantum rewinding procedure is implemented by a general quantum
  circuit $R_{\coin}$ with input $(W,X, \dB, \coin)$. As a first step,
  it applies a unitary $(n,k)$-quantum circuit $Q$ to $(W,X)$ to
  simulate the conversation, obtaining registers $(G,Y)$. Then, a test
  takes place to observe whether the simulation was successful. In
  that case, $R_{\coin}$ outputs the resulting quantum
  register. Otherwise, it {\it quantumly rewinds} by applying the
  reverse circuit $Q^\dag$ on $(G,Y)$ to retrieve $(W,X)$ and then a
  phase-flip transformation on $X$ before another iteration of $Q$ is
  applied. Note that $R_{\coin}$ is essentially the same circuit as
  $R$ described in~\cite{Watrous09} (and Section~\ref{sec:quantum.rewinding}),
  but in our application it depends on the value of a given $\coin$,
  i.e., we apply $R_0$ or $R_1$ for $\coin = 0$ or $\coin = 1$,
  respectively. \\

  \noindent In more detail, $Q$ transforms $(W,X)$ to $(G,Y)$ by the
  following unitary operations:
  \begin{enumerate}
  \item[(1.)] 
    It constructs a superposition over all possible random
    choices of values in the real protocol, i.e.,    
    $$ 
    \frac{1}{\sqrt{2^{\ell+1}}} \sum_{a,r} \ket{a,r}_{R}
    \ket{\commitx{a}{r}}_{A_1} \ket{b' = \coin \oplus a}_{G}
    \ket{\open{a}{r}}_{A_2} \ket{0}_{B} \ket{0^{k^*}}_{V}
    \ket{\psi}_{W} \, ,
    $$ 
    where $k^* < k$. Note that the state of registers $\big( A_1, G,
    A_2 \big)$ corresponds to a uniform distribution of possible
    transcripts of the interaction between the players.
  \item[(2.)] 
    For each possible $\commitx{a}{r}$, it simulates $\dB$'s
    possible actions by applying a unitary operator to $\big(
    W,V,B,A_1 \big)$ with register $A_1$ as control, i.e.,
    $$
    \frac{1}{\sqrt{2^{\ell+1}}} \sum_{a,r} 
    \ket{a,r}_{R}
    \ket{\commitx{a}{r}}_{A_1} 
    \ket{b'}_{G} 
    \ket{\open{a}{r}}_{A_2}
    \ket{b}_{B} 
    \ket{\tilde{\phi}}_{V} 
    \ket{\tilde{\psi}}_{W} \, ,
    $$
    where ${\tilde{\phi}}$ and ${\tilde{\psi}}$ describe modified
    quantum states. Note that register $B$ now includes $\dB$'s reply
    $b$. 
  \item[(3.)] 
    Finally, a $\op{CNOT}$-operation is applied to pair $\big( B,G
    \big)$ with $B$ as control to check whether the simulator's guess
    of $\dB$'s choice was correct. The result of the
    $\op{CNOT}$-operation is stored in register $G$.
    $$\frac{1}{\sqrt{2^{\ell+1}}} \sum_{a,r} 
    \ket{a,r}_{R}
    \ket{\commitx{a}{r}}_{A_1} 
    \ket{b' \oplus b}_{G} 
    \ket{\open{a}{r}}_{A_2}
    \ket{b}_{B} 
    \ket{\tilde{\phi}}_{V} 
    \ket{\tilde{\psi}}_{W} \, .$$
  \end{enumerate}
  Note that the qubit in register $G$ gives the information about
  success or failure of the simulated run, and the other registers are
  combined in the residual $n + k - 1$-qubit register $Y$. 

  Since the commitment scheme in the protocol is only
  quantum-computationally hiding, we must allow for small
  perturbations in the quantum rewinding procedure, according to
  Lemma~\ref{lemma:qrewind} : Bound $\eps$ indicates $\dB$'s advantage
  over a random guess on the committed value with $q = 1/2$ (and
  therefore, his advantage to bias the outcome), due to his computing
  power, i.e.~$\eps = |p - 1/2|$. From the hiding property of the
  commitment scheme, it follows that $\eps$ is negligible in the
  security parameter $n$. Thus, we can argue that probability $p$ is
  \emph{close} to independent of the auxiliary input. As a lower bound
  on the success probability, we chose $p_0 \geq 1/4$, which matches
  our setting. 
  
  Thus, we have circuit $Q$ as described above and our setting
  achieves the given bounds. Lemma~\ref{lemma:qrewind} applies. We can
  now construct an ideal-world quantum simulator $\dhB$ (see
  Figure~\ref{fig:simulationB}), interacting with $\dB$ and the ideal
  functionality $\cF$ and executing Watrous' quantum rewinding
  algorithm. We then compare the output states of the real process and
  the ideal process. In case of indistinguishable outputs,
  quantum-computational security against $\dB$ follows.

\begin{figure}
  \begin{framed}
    \noindent\hspace{-1.5ex} {\sc Simulation $\dhB$ :} \\[-4ex]   
    \begin{enumerate}
    \item $\dhB$ gets $\dB$'s auxiliary quantum input $W$ and working
    registers $X$.
    \item\label{step:get-coin} $\dhB$ sends $\start$ and then $\top$
    to $\cF$. It receives a uniformly random $\coin$.
    \item Depending on the value of $\coin$, $\dhB$ applies the
    corresponding circuit $R_{\coin}$ with input $W, X ,\dB$ and
    $\coin$.
    \item $\dhB$ receives output register $Y$ with
    $\ket{\phi_{good}(\psi)}$ and ``measures the conversation'' to
    retrieve the corresponding $\big( \commitx{a}{r},b,\open{a}{r}
    \big)$. It outputs whatever $\dB$ outputs.
  \end{enumerate}
  \vspace{-1.5ex}
  \end{framed}
  \vspace{-1.5ex}
  \caption{The Ideal-World Simulation against dishonest Bob.}
  \label{fig:simulationB}
\end{figure}

  First note that the superposition constructed as described above in
  circuit $Q$ in Step~(1.) corresponds to all possible random choices
  of values in the real protocol. Furthermore, the circuit models any
  possible strategy of quantum $\dB$ in Step~(2.), depending on
  control register $\ket{\commitx{a}{r}}_{A_1}$. The
  $\op{CNOT}$-operation on $(B,G)$ in Step~(3.), followed by a
  standard measurement of $G$, indicate whether the guess $b'$ on
  $\dB$'s choice $b$ was correct. If that was not the case (i.e.~$b
  \neq b'$ and measurement result 1), the system gets quantumly
  rewound by applying reverse transformations (3)-(1), followed by a
  phase-flip operation. The procedure is repeated until the
  measurement outcome is 0 and hence $b=b'$. Watrous' technique then
  guarantees that, for negligible advantage $\eps$ and a lower bound
  $p_0 \geq \frac14$, $\eps'$ is negligible. Thus, the final output of
  the simulation is close to the ``good'' state of a successful
  simulation. More specifically, the output $\rho(\psi)$ of $R_{coin}$
  has square-fidelity close to 1 with state $\ket{\phi_{good}(\psi)}$
  of a successful simulation,
  i.e.~
  $$\bra{\phi_{good}(\psi)}\rho(\psi)\ket{\phi_{good}(\psi)} \geq
  1 - \eps' \, ,
  $$ where $\eps' = 16 \ \eps \log^2(1 / \eps) / (p_0^2 \
  (1-p_0)^2)$. Last, note that all operations in $Q$ (and therewith in
  $R_{coin}$) can be performed by polynomial-size circuits, and thus,
  the simulator has polynomial size (in the worst case). It follows
  that the output of the ideal simulation is indistinguishable from
  the output in the real world for any quantum-computationally bounded
  $\dB$.
\end{proof}


\section{Composability}
\label{sec:composability.coin}

As already discussed in the previous part, there are several
composition frameworks proposed for the quantum setting, but for
sequential composition we will argue along the lines of our security
framework (Section~\ref{sec:sequential.composition.coin}). In
Section~\ref{sec:general.composition.coin}, we will use an extend
commitment construction to achieve a more general composability in the
CRS-model. Note that only sequential composition allows us to do
coin-flipping from scratch.


\subsection{Sequential Composition}
\label{sec:sequential.composition.coin}
\index{composition!sequential}

  Recall that we prove correctness and security for our single
  coin-flip according to the security framework as described in
  Section~\ref{sec:security.definition}, with the one modification
  that we do not assume a common reference string in the simulation
  against a dishonest Bob (see
  Theorem~\ref{thm:computational.security.single.coin}). However, we
  can still apply the Composition Theorems I and II
  (Theorems~\ref{thm:composition.unconditional}
  and~\ref{thm:composition.computational}), where we also omit the
  reference string in the latter. We will state the composition result
  explicitly here.
  
  \begin{corollary}
    \label{cor:sequential.coin.flipping}
    Let $\pi_i = \cPi{\A}{\B}$ and $\F_i = \cF$, and let
    $\Sigma^{\F_1\cdots\F_\ell}$ be a classical two-party hybrid
    protocol which makes at most $\ell=\poly(n)$ calls to the
    functionalities. Then, for every $i \in \set{1,\ldots,\ell}$,
    each protocol $\pi_i$ is a statistically secure implementation of
    $\F_i$ against $\dAlice$ and a computationally secure
    implementation of $\F_i$ against $\dBobPoly$.

    \noindent It holds that there exists an ideal-world adversary
    $\dhA \in \dAlice$ such that
    $$ 
    out_{\dA,\B}^{\Sigma^{\pi_1\cdots\pi_\ell}} \approxs
    out_{\dhA,\dB}^{\Sigma^{\F_1\cdots\F_\ell} } \, ,
    $$ and an ideal-world adversary $\dhB \in \dBobPoly$ such that for
    every efficient input sampler, we have
    $$ 
    out_{\A,\dB}^{\Sigma^{\pi_1\cdots\pi_\ell}} \approxq
    out_{\hA,\dhB}^{\Sigma^{\F_1\cdots\F_\ell} } \, .
    $$
  \end{corollary}
 
  The ideal functionality for sequential coin-flipping,
  i.e.~$\cxF{\ell} = \Sigma^{\F_1\cdots\F_\ell}$, is depicted in
  Figure~\ref{fig:clambdaF}. Note that $\cxF{\ell}$ is in fact derived
  from composing the functionality $\cF$ of a single coin-flip
  sequentially but interpreted more directly, e.g.\ it does not output
  the bits one after another but as a string, and thus, does not
  output the precedent coins in case of an intermediate abort.

\begin{figure}
  \begin{framed}
    \noindent\hspace{-1.5ex} {\sc Functionality} $\cxF{\ell}$
    :\\[-4ex]
    \begin{enumerate}
    \item 
      Upon receiving requests $\start$ from both Alice and Bob,
      $\cxF{\ell}\,$ outputs uniformly random $h \in_R \zo^\ell$ to
      Alice.
    \item
      It then waits to receive her second input $\top$ or $\bot$ and
      outputs $h$ or $\bot$ to Bob, respectively.
    \end{enumerate}
    \vspace{-1ex}
  \end{framed}
  \vspace{-2ex}
  \small
  \caption{The Ideal Functionality for Sequential $\ell$-bit
  Coin-Flipping.}
  \label{fig:clambdaF}
  \vspace{-1ex}
\end{figure}
 

\subsection{General Composition}
\label{sec:general.composition.coin}
\index{composition!general}

  For our coin-flipping protocol without set-up, we cannot claim
  universal composability. We do not require (nor obtain) an efficient
  simulator in case of unconditional security against dishonest Alice
  and furthermore, we allow rewinding in case of dishonest Bob. These
  two aspects contradict the universal composability framework.

  Efficient simulation requires some trapdoor information in the
  commitment construction, which is available only to a simulator, so
  that it is able to extract dishonest Alice's choice bit
  efficiently. Therefore, we have to extend the commitment scheme by
  including an extraction trapdoor. To circumvent the necessity of
  rewinding dishonest Bob, we further extend the scheme with respect
  to equivocability, i.e., the simulator can now construct a valid
  commitment, which can later be opened to both bit values as
  desired. Note that with such requirements, the CRS-model seems
  unavoidable.

  An appropriate extended construction is proposed in
  Section~\ref{sec:extended.commit.coin}. The real-world key consists
  of commitment key $\pkB$ and (invalid) instance $x'$. During
  simulation against $\dA$, $\dhA$ chooses $\pkB$ with matching
  decryption key $\sk$ and therefore, it can extract $\dA$'s choice
  bit $a$ by decrypting both commitments $C_0$ and $C_1$. In both
  worlds, the commitment is unconditionally binding. During simulation
  against $\dB$, $\dhB$ chooses commitment key $\pkH$ and (valid)
  instance $x$. Hence, the commitment is unconditionally hiding and
  can be equivocated by using $w$ to compute two valid replies in the
  underlying $\Sigma$-protocol. Quantum-computational security in real
  life follows from the indistinguishability of the keys $\pkB$ and
  $\pkH$ and the indistinguishability of the instances $x$ and $x'$,
  and efficiency of both simulations is ensured due to extraction and
  equivocability.

  Again, by combining our extended construction in the CRS-model
  providing efficient simulations on both sides with the results of
  Section~\ref{sec:extended.commit.coin} and~\cite[Theorem
  20]{Unruh10}, we get the following
  result\index{composition!quantum-UC} that $\cPi{\A}{\B}$
  \emph{computationally quantum-UC-emulates} its corresponding ideal
  functionality $\cF$ for \emph{both dishonest players}. In the next
  Chapter~\ref{chap:framework}, we will show another method of
  achieving fully simulatability in the plain model without any set-up
  assumption, when both players are poly-time bounded.


\clearemptydoublepage
\chapter{Amplification Framework for Strong Coins}
\label{chap:framework}
\index{coin-flipping!amplification}

Here, we present a framework that amplifies weak security requirements
on coins into very strong properties, with the final result of a
quantum-secure and fully simulatable coin-flipping protocol, which can be
implemented in the plain model from scratch. The results in this
chapter are joint work with Nielsen~\cite{LN10}.


\section{Motivation}
\label{sec:framework.motivation}

Coin-Flipping of a single coin is in itself an intriguing and prolific
primitive in cryptographic protocol theory. Its full potential is
tapped in the possibility of flipping a string of coins, which opens
up for various applications and implementations without any set-up
assumptions. We will later in
Chapter~\ref{chap:coin.flip.applications} discuss some examples
thereof.

In this chapter, we first investigate the different degrees of
security that a string of coins can acquire. Then, we propose and
prove constructions that allow us to amplify the respective degrees of
security such that weaker coins are converted into very strong ones in
a straightforward way.\footnote{For the sake of clarity, we note that
we use the (intuitive) literal interpretation of ``weak'' and
``strong'' coins related to their degrees of security, which differs
from their definitions in the quantum literature (see also
Section~\ref{sec:primitives.coin.flip}).} Our method only assumes
mixed commitment schemes, which we know how to construct with quantum
security, no other assumptions are put forward. Our final result is a
coin-flipping protocol, which is fully simulatable in polynomial time,
even against poly-sized \emph{quantum} adversaries on both sides, and
which can be implemented with quantum-computational security in the
plain model from scratch.

Our method of amplifying the security of coin-flipping also applies to
potential \emph{constant round} coin-flipping. Such a strong and
efficient construction would require a basic quantum-secure coin-flip
protocol with long outcomes (in constant round), and poly-time
simulatability on one side. Its construction, however, is still a
fascinating open problem in the quantum world.


\section{Security Notions}
\label{sec:notions.coin.flip}

We denote a generic protocol with a $\lambda$-bit coin-string as
output by $\cxPi{\A}{\B}{\lambda}$, corresponding to an ideal
functionality $\cxF{\lambda}$. Recall that the outcome of such a
protocol is $c \in \zo^\lambda \cup \set{\bot}$, i.e., either an
$\lambda$-bit string or an error message.\footnote{We want to stress
that throughout the chapter, a reference to any \emph{coin-flip} is
understood as one run of coin-flipping with a coin-string outcome.} We
will use several security parameters, indicating the length of
coin-strings for different purposes. The length of a coin-flip
yielding a key and a challenge are denoted by $\kappa$ and $\sigma$,
respectively, and the length of a final coin-flip is indicated by
$\ell$, i.e., we allow that $\lambda$ is a function of the respective
parameter, e.g.\ $\lambda(\kappa)$, but we write $\kappa$ instead.

Throughout this chapter, we restrict both players Alice and Bob to the
families $\dAlicePoly$ and $\dBobPoly$ of classical polynomial-time
strategies, i.e.\ for the honest case, we require $\A , \hA \in
\dAlicePoly$ and $\B, \hB \in \dBobPoly$, as well as for possibly
quantum dishonest entities, we demand $\dA, \dhA \in \dAlicePoly$ and
$\dB, \dhB \in \dBobPoly$. We want to stress here that, in contrast to
previous chapters, both players are poly-time bounded. This means, in
particular, that the ideal functionality is defined symmetric such
that always the respective dishonest party has an option to abort. For
clarity, we will explicitly show the ideal functionalities in the
case of both players being honest (Figure~\ref{fig:clambdaF.honest})
and in the case of dishonest Alice and honest Bob
(Figure~\ref{fig:clambdaF.dishonest}). The latter then also applies to
honest Alice and dishonest Bob by simply switching sides and names.

\begin{figure}[here]
  \begin{framed}
    \noindent\hspace{-1.5ex} {\sc Functionality} $\cxF{\lambda}$ {\sc
    with honest players:}\\ Upon receiving requests $\start$ from both
    Alice and Bob, $\cxF{\lambda}\,$ outputs uniformly random $h \in_R
    \zo^\lambda$ to Alice and Bob.
    \vspace{-1ex}
  \end{framed}
  \vspace{-1.5ex}
  \caption{The Ideal Functionality for $\lambda$-bit Coin-Flipping
  (without Corruption).}
  \label{fig:clambdaF.honest}
  \vspace{-1ex}
\end{figure}

\begin{figure}[here]
  \begin{framed}
    \noindent\hspace{-1.5ex} {\sc Functionality} $\cxF{\lambda}$ {\sc with
    dishonest Alice:}\\[-4ex]
    \begin{enumerate}
    \item 
      Upon receiving requests $\start$ from both Alice and Bob,
      $\cxF{\lambda}\,$ outputs uniformly random $h \in_R \zo^\lambda$
      to Alice.
    \item
      It then waits to receive her second input $\top$ or $\bot$ and
      outputs $h$ or $\bot$ to Bob, respectively.
    \end{enumerate}
    \vspace{-1ex}
  \end{framed}
  \vspace{-1.5ex}
  \caption{The Ideal Functionality for $\lambda$-bit Coin-Flipping
  (with Corruption).}
  \label{fig:clambdaF.dishonest}
  \end{figure}

Recall that the \emph{joint output representation} of a protocol
execution is denoted by $out_{\A,\B}^\Pi \,$ with $\Pi =
\cxPi{\A}{\B}{\lambda}$ and given here for the case of honest
players. The same notation with $\F = \cxF{\lambda}$ and $\hA, \hB$
applies in the ideal world as $out_{\hA,\hB}^\F$, where the players
invoke the ideal functionality $\cxF{\lambda}$ and output whatever
they obtain from it. We need an additional notation here, describing
the \emph{outcome} of a protocol run between e.g.\ honest $\A$ and
$\B$, namely $c \la \cxPi{\A}{\B}{\lambda}$.

\index{coin-flipping!uncontrollable} 
\index{coin-flipping!random}
\index{coin-flipping!enforceable} 
We will define three flavors of security for coin-flipping protocols,
namely \defterm{uncontrollable (uncont)}, \defterm{random} and
\defterm{enforceable (force)}. The two sides can have different
flavors. Then, if a protocol $\cxPi{\A}{\B}{\lambda}$ is, for
instance, enforceable against Alice and random against Bob, we write
$\pi^{(\force,\random)}$, and similarly for the eight other
combinations of security. Note that for simplicity of notation, we
will then omit the indexed name as well as the length of the coin, as
they are clear from the context. Similar to the ideal functionality
for the case of dishonest Alice, we define all three flavors for
Alice's side only, as the definitions for Bob are analogue. The
flavors are defined along the lines of the security framework
introduced in Section~\ref{sec:security.definition} but with adaptions
to reflect the particular context here. Recall that $U'$, $Z$, and $V$
denote dishonest Alice's quantum and classical input, and honest Bob's
classical input, respectively. Note that an honest player's input is
empty but models the invocation $\start$. Any input state $\rho_{U' Z
V}$ is restricted to $ \rho_{U' Z V} = \rho_{\MC{U'}{Z}{V}}$, such
that Alice's quantum and Bob's classical part are only correlated via
Alice's classical~$Z$. We assume again a poly-size input sampler,
which takes as input the security parameter, and then produces a valid
input state $\rho_{U' Z V} = \rho_{\MC{U'}{Z}{V}}$ (and analogous
$\rho_{U Z V'}$ in case of dishonest Bob).

We stress that we require for all three security flavors and for all
$c \in \zo^{\lambda}$ that
$$\prob{c \la \cxPi{\A}{\B}{\lambda}} = 2^{-\lambda} \, ,$$ 
which implies that when both parties are honest, then the coin is
unbiased. Below we only define the extra properties required for each
of the three flavors.\\

We call a coin-flip \defterm{uncontrollable} against Alice, if she
cannot force the coin to hit some negligible subset, except with
negligible probability.
  \begin{definition}[Uncontrollability against dishonest Alice]
    \label{def:uncont}
    We say that the protocol $\ \cxPi{\A}{\B}{\lambda} \ $
    implements an \defterm{uncontrollable} coin-flip against dishonest
    Alice, if it holds for any poly-sized adversary $\dA \in
    \dAlicePoly$ with inputs as specified above and all negligible
    subsets $Q \subset \zo^\lambda$ that the probability
    $$ \prob{c \la \cxPi{\dA}{\B}{\lambda}\, : \, c \in Q} \in
    \negl{\kappa} \, .
    $$
  \end{definition}
Note that we denote by $Q \subset \zo^\lambda$ a family of subsets
$\set{Q(\kappa) \subset \zo^{\lambda(\kappa)}}_{\kappa \in \NN}$ for
security parameter $\kappa$. Then we call $Q$ negligible, if $\vert
Q(\kappa) \vert 2^{-\lambda(\kappa)}$ is negligible in $\kappa$. In
other words, we call a subset negligible if it contains a negligible
fraction of the elements in the set in which it lives.\\

We call a coin-flip \defterm{random} against Alice, if she cannot
enforce a non-uniformly random output string in $\zo^\lambda$, except
by making the protocol fail on some chosen runs. That means she can at
most lower the probability of certain output strings compared to the
uniform case.
  \begin{definition}[Randomness against dishonest Alice]
    \label{def:random}
    We say that $ $  protocol $\ \cxPi{\A}{\B}{\lambda} \ $ implements a
    \defterm{random} coin-flip against dishonest Alice, if it holds
    for any poly-sized adversary $\dA \in \dAlicePoly$ with inputs as
    specified above that there exists an event E such that $\prob{E}
    \in \negl{\kappa}$ and for all $x \in \zo^\lambda$ it holds that
    $$
    \prob{c \la
    \cxPi{\dA}{\B}{\lambda} \, : \, c = x \, \vert \, \bar{E}} \leq
    2^{-\lambda} \, .
    $$
  \end{definition}
It is obvious that if a coin-flip is random against Alice, then it is also
an uncontrollable coin-flip against her. We will later discuss a generic
transformation going in the other direction from uncontrollable to
random coin-flipping.\\

We call a coin-flip \defterm{enforceable} against Alice, if it is
possible, given a uniformly random $c$, to simulate a run of the
protocol hitting exactly the outcome $c$, though we still allow that
the corrupted party forces abort on some outcomes.
  \begin{definition}[Enforceability against dishonest Alice]
    \label{def:force}
    We call a protocol $ \ \cxPi{\A}{\B}{\lambda} \ $
    \defterm{enforceable} against dishonest Alice, if it implements
    the ideal functionality $\ \cxF{\lambda} \ $ against her.
  \end{definition}
In more detail, that means that for any poly-sized adversary $\dA \in
\dAlicePoly$, there exists an ideal-world adversary $\dhA \in
\dAlicePoly$ that simulates the protocol with $\dA$ as follows.
$\dhA$ requests output $h \in \zo^\lambda$ from $\cxF{\lambda}$. Then
it simulates a run of the coin-flipping protocol with $\dA$ and tries to
enforced output $h$. If $\dhA$ succeeds, it inputs $\top$ as $\dA$'s
second input to $\cxF{\lambda}$. In that case, $\cxF{\lambda}$ outputs
$h$. Otherwise, $\dhA$ inputs $\bot$ to $\cxF{\lambda}$ as second
input and $\cxF{\lambda}$ outputs $\bot$. The simulation is such that
the ideal output is quantum-computationally indistinguishable from the
output of an actual run of the protocol, i.e.,
$$ 
out_{\dA,\B}^\Pi \approxq out_{\dhA,\hB}^\F \, , 
$$ where $\Pi = \cxPi{\dA}{\B}{\lambda}$ and $\F = \cxF{\lambda}$.

Note that an enforceable coin-flip is not necessarily a random
coin-flip, as it is allowed that the outcome of an enforceable
coin-flip is only quantum-computationally indistinguishable from
uniformly random, whereas a random coin-flip is required to produce
truly random outcomes on the non-aborting runs.\\

We defined an enforceable coin-flip against dishonest Alice to be a
coin-flip, simulatable on her side and implementing the corresponding
ideal functionality against her. The same result with switched sides
also holds for any poly-time bounded Bob. Thus, we obtain a coin-flip
protocol, for which we can simulate both
sides\index{coin-flipping!fully simulatable} in polynomial
time. Corollary~\ref{cor:double.simulatable} follows.
  \begin{corollary}
    \label{cor:double.simulatable}
    Let $\cxPi{\A}{\B}{\lambda}$ be an enforceable coin-flip against
    both parties Alice and Bob with $\A \in \dAlicePoly$ and $\B \in
    \dBobPoly$, i.e.~$\cxPi{\A}{\B}{\lambda} =
    \pi^{(\force,\force)}$. Then $\pi^{(\force,\force)}$ is a
    \defterm{fully poly-time simulatable} coin-flipping protocol for the
    ideal functionality $\ \cxF{\lambda} \, $ with
    quantum-computational indistinguishability between the real and
    the ideal output.
  \end{corollary}

Combining the part regarding simulatability in
Corollary~\ref{cor:sequential.coin.flipping}, where we again omit the
common reference string, in contrast to the original Composition
Theorem II (Theorem~\ref{thm:composition.computational}), with the
results of Corollary~\ref{cor:double.simulatable}, we can show that
each protocol $\pi^{(\force,\force)}$ is a computationally secure
implementation of $\cxF{\lambda}$ against both $\dAlicePoly$ and
$\dBobPoly$.
  \begin{corollary}
    \label{cor:double.simulatable.sequential.composition}
    Protocol $\pi^{(\force,\force)}$ composes sequentially. 
  \end{corollary}


\section{Amplification Theorems}
\label{sec:amplification.coin.flip}

We now propose and prove theorems, which allow us to amplify the
security strength of coins. Ultimately, we aim at constructing a
strong coin-flipping protocol $\pi^{(\force,\force)}$ with outcomes of any
polynomial length $\ell$ in $\lambda$ from any weaker coin-flip
protocol, i.e., either from a protocol $\pi^{(\force,\random)}$
producing one-bit outcomes
(Section~\ref{sec:amplification.short.long}), or from a protocol
$\pi^{(\force,\uncont)}$ giving outcomes of length $\kappa$, as
described in Section~\ref{sec:amplification.uncont.random}. In both
cases, the first step towards $\pi^{(\force,\force)}$ is to build a
protocol $\pi^{(\force,\random)}$ with outcomes of length $\ell$. 

We want to stress that if the underlying protocol already produces
$\ell$-bit outcomes and is constant round, then the resulting protocol
$\pi^{(\force,\force)}$ will also be constant round. If we start from
a protocol only producing constant-sized outcomes, then
$\pi^{(\force,\force)}$ will use $O(\ell)$ times the number of rounds
used by the underlying scheme.

We note here that we do not know of any candidate protocol with flavor
$(\force,\uncont)$ but not $(\force,\random)$. However, we consider it
as a contribution in itself to find the weakest security notion for
coin-flipping that allows to amplify to the final strong
$(\force,\force)$ notion using a constant round reduction.


\subsection{From Short Outcomes to Long Outcomes}
\label{sec:amplification.short.long}

  To obtain long coin-flip outcomes, we can repeat a given protocol
  $\pi^{(\force,\random)}$ with one-bit outcomes $\ell$ times in
  sequence to get a protocol $\pi^{(\force,\random)}$ with $\ell$-bit
  outcomes. A candidate for $\pi^{(\force,\random)}$ with one-bit
  outcomes is the protocol of Chapter~\ref{chap:coin.flip}, which
  is---in terms of this context---enforceable against one side in
  poly-time and random on the other side, with empty event $E$
  according to Definition~\ref{def:random}, and the randomness
  guarantee even withstanding an unbounded adversary. The protocol was
  argued to be sequentially composable according to
  Corollary~\ref{cor:sequential.coin.flipping}.

  Note that this protocol is previously described and proven as
  $\pi^{(\random,\force)}$. However, due to the symmetric coin-flip
  definitions here and the restriction of entities to families of
  classical polynomial-time strategies, we can easily switch sides
  between $\A$ and $\B$.


\subsection{From $(\force,\uncont)$ to $(\force,\random)$}
\label{sec:amplification.uncont.random}

  Assume that we are given a protocol $\pi^{(\force,\uncont)}$, that
  only guarantees that Bob cannot force the coin to hit a negligible
  subset (except with negligible probability). We now amplify the
  security on Bob's side from $\defterm{uncontrollable}$ to
  $\defterm{random}$ and therewith obtain a protocol
  $\pi^{(\force,\random)}$, in which Bob cannot enforce a
  non-uniformly random output string, except by letting the protocol
  fail on some occasions. The stronger protocol
  $\pi^{(\force,\random)}$ is given in
  Figure~\ref{fig:force.random}. The underlying commitment $\commit$
  denotes the commitment algorithm of the keyed mixed string
  commitment scheme as described in
  Section~\ref{sec:mixed.commit}. Recall that $\commit$ does not
  require actual unconditionally hiding keys, but rather it suffices
  to use uniformly random strings from $\zo^\kappa$, which
  unconditionally hide the plaintext, except with negligible
  probability. The possibility of using random strings ensures that
  most keys of the given domain are in that sense unconditionally
  hiding keys.

\begin{figure}
  \begin{framed}
    \noindent\hspace{-1.5ex} {\sc Protocol} $\pi^{(\force,\random)}$:
    \\[-4ex]
    \begin{enumerate}
    \item
      $\A$ and $\B$ run $\pi^{(\force,\uncont)}$ to produce a public
      key $pk \in \zo^\kappa$.
    \item
      $\A$ samples $a \in_R \zo^\ell$, commits to it with $A =
      \commitk{a}{r}{pk}$ and randomizer $r \in_R \zo^\ell$, and sends
      $A$ to $\B$.
    \item
      $\B$ samples $b \in_R \zo^\ell$ and sends $b$ to $\A$.
    \item
      $\A$ opens $A$ towards $\B$.
    \item
      The outcome is $c = a \oplus b$.
    \end{enumerate}
    \vspace{-1.5ex}
  \end{framed}
  \vspace{-1.5ex}
  \small
  \caption{Amplification from $(\force,\uncont)$ to
  $(\force,\random)$.}
  \label{fig:force.random}
  \end{figure}

  \begin{proposition}
    Protocol $\pi^{(\force,\random)}$ satisfies correctness, according
    to Definition~\ref{def:correctness}.
  \end{proposition}
  
  Correctness is obvious by inspection of the protocol. If both
  players are honest, they independently choose random strings $a$ and
  $b$. The result of these strings combined by the XOR-operation gives
  a uniformly random coin $c$ of length $\ell$.

  \begin{theorem}
    \label{thm:force.random}
    If $\pi^{(\force,\uncont)}$ is enforceable against Alice and
    uncontrollable against Bob, then protocol $\pi^{(\force,\random)}$
    is enforceable against Alice and random for Bob.
  \end{theorem}

  \begin{proof}[ (\emphbf{Enforceability against Alice})] 
    In case of corrupted $\dA$, $\dhA$ samples $(pk,sk) \la \GB$ as
    input. It then requests a uniformly random value $h$ from
    $\clF$. It runs $\pi^{(\force,\uncont)}$ with $\dA$, in which
    $\dhA$ enforces the outcome $pk$ in the first step. When $\dA$
    sends commitment $A$, $\dhA$ uses $sk$ to decrypt $A$ to learn the
    unique string $a$ that $A$ can be opened to. $\dhA$ computes $b =
    h \oplus a$ and sends $b$ to $\dA$. If $\dA$ opens commitment $A$
    correctly, then the result is $c = a \oplus b = a \oplus (h \oplus
    a) = h$ as desired. In case she does not open correctly, $\dhA$
    aborts with result $\bot$. Otherwise, $\dhA$ outputs whatever
    $\dA$ outputs.
  
    Since $h$ is uniformly random and independent of $A$ and $a$, it
    follows that $b = h \oplus a$ is uniformly random and independent
    of $A$, exactly as in the protocol. Therefore, the transcript of
    the simulation has the same distribution as the real protocol,
    except that $pk$ is uniform in $\X$ and not in $\zo^\kappa$. This
    is, however, quantum-computationally indistinguishable, as
    otherwise, $\dA$ could distinguish random access to samples from
    $\X$ from random access to samples from $\zo^\kappa$. The formal
    proof proceeds through a series of hybrids as described in full
    detail in the proof for Theorem~\ref{thm:force.force} in the next
    Section~\ref{sec:amplification.force.force}.

    The above two facts, that first we hit $h$ when we do not abort,
    and second that the transcript of the simulation is
    quantum-computationally indistinguishable from the real protocol,
    show that the resulting protocol is enforceable against Alice and
    simulatable on Alice's side for functionality $\clF$, according to
    Definition~\ref{def:force} combined with Theorem~\ref{def:force}.
  \end{proof}

  \begin{proof}[ (\emphbf{Randomness against Bob})] 
    For any $\dB$, $pk$ is uncontrollable, i.e.~$pk \in \zo^\kappa
    \setminus \X$, except with negligible probability, as $\X$ is
    negligible in $\zo^\kappa$. This, in particular, means that the
    commitment $A$ is perfectly hiding the value $a$. Therefore, $a$
    is uniformly random and independent of $b$, and thus, $h = a
    \oplus b$ is uniformly random. This proves that the resulting
    coin-flip is random against Bob, according to
    Definition~\ref{def:random}.
  \end{proof}


\subsection{From $(\force,\random)$ to $(\force,\force)$}
\label{sec:amplification.force.force}

  We now show how to obtain a coin-flipping protocol, which is enforceable
  against both parties. Then, we can also claim by
  Corollary~\ref{cor:double.simulatable} that this protocol is a
  strong coin-flipping protocol, poly-time simulatable on both sides for
  the natural ideal functionality $\clF$. The protocol
  $\pi^{(\force,\force)}$ is described in Figure~\ref{fig:force.force}.

  Note that the final protocol makes two calls to a subprotocol with
  random flavor on one side and enforceability on the other side, but
  where the sides are interchanged for each instance,
  i.e.~$\pi^{(\force,\random)}$ and $\pi^{(\random,\force)}$. That
  means that we switch the players' roles as well as the direction of
  the messages. Furthermore, note that we use here the possibility of
  trapdoor openings in our extended commitment construction $\Commit$,
  based on secret sharing and mixed commitments, as described in
  detail in Section~\ref{sec:mixed.commit.trapdoor.opening}.

\begin{figure}
  \begin{framed}
    \noindent\hspace{-1.5ex} {\sc Protocol} $\pi^{(\force,\force)}$:
    \\[-4ex]
    \begin{enumerate}
    \item
      $\A$ and $\B$ run $\pi^{(\force,\random)}$ to produce a random
      public key $pk \in \zo^\kappa$.
    \item
      $\A$ computes and sends commitments $\Commitk{a}{(s,r)}{pk} =
      (A_1,\ldots,A_\Sigma) $ to $\B$. In more detail, $\A$ samples
      uniformly random $a, s \in \FF^\sigma$. She then computes
      $\sss(a;s) = (a_1,\ldots,a_\Sigma)$ and $A_i =
      \commitk{a_i}{r_i}{pk}$ for all $i = 1, \ldots, \Sigma$.
    \item
      $\B$ samples uniformly random $b \in \zo^\ell$ and sends $b$ to
      $\A$.
    \item
      $\A$ sends secret shares $(a_1,\ldots,a_\Sigma)$ to $\B$. If
      $(a_1, \ldots, a_\Sigma)$ is not consistent with a polynomial of
      degree at most $(2\sigma-1)$, $\B$ aborts.
    \item
      $\A$ and $\B$ run $\pi^{(\random,\force)}$ to produce a
      challenge $S \subset \set{1,\ldots,\Sigma}$ of length $\vert S
      \vert = \sigma$.
    \item
      $\A$ sends $r|_S$ to $\B$.
    \item
      $\B$ checks if $A_i = \commitk{a_i}{r_i}{pk}$ for all $i \in
      S$. If that is the case, $\B$ computes message $a \in
      \FF^\sigma$ consistent with $(a_1, \ldots, a_\Sigma)$ and the
      outcome of the protocol is $c = a \oplus b$. Otherwise, $\B$
      aborts and the outcome is $c = \bot\,$.
    \end{enumerate}
    \vspace{-1.5ex}
  \end{framed}
  \vspace{-1.5ex}
  \small
  \caption{Amplification from $(\force,\random)$ to
  $(\force,\force)$.}
  \label{fig:force.force}
\end{figure}

  \begin{proposition}
    Protocol $\pi^{(\force,\force)}$ satisfies correctness, according
    to Definition~\ref{def:correctness}.
  \end{proposition}
 
  Again, correctness can be trivially checked, first by observing that
  honest players independently input uniformly random strings $a$ and
  $b$, and second by verifying that these strings combined by XOR
  result in a uniformly random coin $c$ of length $\ell$.

  \begin{theorem}
    \label{thm:force.force}
    If $\pi^{(\force,\random)}$ is enforceable against Alice and
    random against Bob, then protocol $\pi^{(\force,\force)}$ is
    enforceable against both Alice and Bob.
  \end{theorem}

  \begin{proof}[ (\emphbf{Enforceability against Alice})] 
    If $\dA$ is corrupted, $\dhA$ samples $(pk,sk) \leftarrow \GB$ as
    input and enforces $\pi^{(\force,\random)}$ in the first step to
    hit the outcome $pk$. It then requests value $h$ from $\clF$. When
    $\dA$ sends commitments $(A_1,\ldots,A_\Sigma)$, $\dhA$ uses $sk$
    to extract $a'$ with $\big( a'_1,\ldots,a'_\Sigma \big) = \big(
    \xtr{A_1}{sk},\ldots,\xtr{A_\Sigma}{sk} \big)$. $\dhA$ then sets
    $b = h \oplus a'$, and sends $b$ to $\dA$. Then $\dhA$ finishes
    the protocol honestly. In the following, we will prove that the
    transcript is quantum-computationally indistinguishable from the
    real protocol and that if $c \neq \bot$, then $c = h$, except with
    negligible probability.

    First, we show indistinguishability. The proof proceeds via a
    hybrid\index{hybrid argument} argument.\footnote{Briefly, a hybrid
    argument is a proof technique to show that two (extreme)
    distributions are computationally indistinguishable via proceeding
    through several (adjacent) hybrid distributions. If all adjacent
    distributions are pairwise computationally indistinguishability,
    it follows by transitivity that the two end points are so as
    well. We want to point out that we are not subject to any
    restrictions in how to obtain the hybrid distributions as long as
    we maintain indistinguishability.} Let $\D{0}$ denote the
    distribution of the output of the simulation as described. We now
    change the simulation such that, instead of sending $b = h \oplus
    a'$, we simply choose a uniformly random $b \in \zo^\ell$ and then
    output the corresponding $h = a' \oplus b$. Let $\D{1}$ denote the
    distribution of the output of the simulation after this
    change. Since $h$ is uniformly random and independent of $a'$ in
    the first case, it follows that then $b = h \oplus a'$ is
    uniformly random. Therefore, the change to choose a uniformly
    random $b$ in the second case actually does not change the
    distribution at all, and it follows that $\D{0} = \D{1}$.

    By sending a uniformly random $b$, we are in a situation where we
    do not need the decryption key $sk$ to produce $\D{1}$, as we no
    longer need to know $a'$. So we can now make the further change
    that, instead of forcing $\pi^{(\force,\random)}$ to produce a
    random public key $pk \in \X$, we force it to hit a random public
    key $pk \in \zo^\kappa$. This produces a distribution $\D{2}$ of
    the output of the simulation. Since $\D{1}$ and $\D{2}$ only
    differ in the key we enforce $\pi^{(\force,\random)}$ to hit and
    the simulation is quantum poly-time, there exists a poly-sized
    circuit $Q$, such that $Q(\U(\X)) = \D{1}$ and $Q(\U(\zo^\kappa))
    = \D{2}$, where $\U(\X)$ and $\U(\zo^\kappa)$ denote the uniform
    distribution on $\X$ and the uniform distribution on $\zo^\kappa$,
    respectively. As $\U(\X)$ and $\U(\zo^\kappa)$ are
    quantum-computationally indistinguishable, and $Q$ is poly-sized,
    it follows that $Q(\U(\X))$ and $Q(\U(\zo^\kappa))$ are
    quantum-computationally indistinguishable, and therewith, $\D{1}
    \approxq \D{2}$.

    A last change to the simulation is applied by running
    $\pi^{(\force,\random)}$ honestly instead of enforcing a uniformly
    random $pk \in \zo^\kappa$. Let $\D{3}$ denote the distribution
    obtained after this change. As given in
    Definition~\ref{def:force}, real runs of $\pi^{(\force,\random)}$
    and runs enforcing a uniformly random value are
    quantum-computationally indistinguishable. Using a similar
    argument as above, where $Q$ is the part of the protocol following
    the run of $\pi^{(\force,\random)}$, we get that $\D{2} \approxq
    \D{3}$. Finally by transitivity, it follows that $\D{0} \approxq
    \D{3}$. The observation that $\D{0}$ is the distribution of the
    simulation and $\D{3}$ is the actual distribution of the real
    protocol concludes the first part of the proof.

    We now argue the second part, i.e., if $c \neq \bot$, then $c =
    h$, except with negligible probability. This follows by arguing
    soundness of the commitment scheme $\Commit$, according to
    Lemma~\ref{lemma:soundness.sss}. Recall that, if $pk \in \X$, then
    the probability that $\dA$ can open any $A$ to a plaintext
    different from $\xtr{A}{sk}$ is at most $(\frac34)^\sigma$ when
    $S$ is picked uniformly at random and independent of $A$. The
    requirement on $S$ is however guaranteed (except with negligible
    probability) by the $\random$ flavor of the underlying protocol
    $\pi^{(\random,\force)}$ producing $S$. This concludes the proof
    of enforceability against Alice, as given in
    Definition~\ref{def:force}.
  \end{proof}

  \begin{proof}[ (\emphbf{Enforceability against Bob})] 
    To prove enforceability against corrupted $\dB$, we construct a
    simulator $\dhB$ as shown in Figure~\ref{fig:force.force.dhB}. It
    is straightforward to verify that the simulation always ensures
    that $c = h$, if $\dB$ does not abort. However, we must explicitly
    argue that the simulation is quantum-computationally
    indistinguishable from the real protocol.

\begin{figure}
  \begin{framed}
    \noindent\hspace{-1.5ex} {\sc Simulation} $\dhB$ for
    $\pi^{(\force,\force)}$: \\[-4ex]
    \begin{enumerate}
    \item
      $\dhB$ requests $h$ from $\clF$ and runs
      $\pi^{(\force,\random)}$ honestly with $\dB$ to produce a
      uniformly random public key $pk \in \zo^\kappa$.
    \item
      $\dhB$ computes $\Commitk{a'}{(s,r)}{pk} =
      (A_1,\ldots,\A_\Sigma)$ for uniformly random $a',s \in
      \FF^\sigma$ and sends $(A_1,\ldots,A_\Sigma)$ to $\dB$.
    \item
      $\dhB$ receives $b$ from $\dB$.
    \item\label{step:trapdoor}
      $\dhB$ computes $a = b \oplus h$. It then picks a uniformly
      random subset $S \subset \set{1,\ldots,\Sigma}$ with $|S| =
      \sigma$, and lets $a'|_S$ be the $\sigma$ messages committed to
      by $A|_S$. Then, it interpolates the unique polynomial $f$ of
      degree at most $(2\sigma-1)$ for which $f(i) = a'_i$ for $i \in
      S$ and for which $f(-i+1) = a_i$ for $i \in
      \set{1,\ldots,\Sigma} \setminus S$. Finally, it
      sends $(f(1), \ldots, f(\Sigma))$ to $\dB$.
    \item
      During the run of $\pi^{(\random,\force)}$, $\dhB$ enforces the
      challenge $S$.
    \item\label{step:test.on.S}
      $\dhB$ sends $r|_S$ to $\dB$.
    \item
      $\dhB$ outputs whatever $\dB$ outputs.
    \end{enumerate}
    \vspace{-1.5ex}
  \end{framed}
  \vspace{-1.5ex}
  \caption{Simulation for Bob's $\force$ in $\pi^{(\force,\force)}$.} 
  \label{fig:force.force.dhB}
\end{figure}
   
      Indistinguishability follows by first arguing that the
      probability for $pk \notin \zo^\kappa \setminus \X$ is
      negligible. This follows from $\X$ being negligible in
      $\zo^\kappa$ and $pk$ produced with flavor $\random$ against
      $\dB$ by $\pi^{(\force,\random)}$ being uniformly random in
      $\zo^\kappa$, except with negligible probability. 

      Second, we have to show that if $pk \in \zo^\kappa \setminus
      \X$, then the simulation is quantum-computationally close to the
      real protocol. This can be shown via the following hybrid
      argument. Let $\D{0}$ be the distribution of the output of the
      simulation and let $\D{1}$ be the distribution of the output of
      the simulation where we send all $a'_i$ for all $i =
      \set{1,\ldots,\Sigma}$ at the end of
      Step~(\ref{step:trapdoor}.). Since commitments by
      $\commitk{\cdot}{\cdot}{pk}$ are unconditionally hiding in case
      of $pk\in \zo^\kappa \setminus \X$, commitments by
      $\Commitk{\cdot}{\cdot}{pk}$ are unconditionally hiding as well.
      Furthermore, both $a'$ and $a$ are uniformly random, so we
      obtain statistical closeness between
      $(a',\Commitk{a'}{(s,r)}{pk})$ and
      $(a,\Commitk{a'}{(s,r)}{pk})$. Note further that distributions
      $\D{0}$ and $\D{1}$ can be produced by a poly-sized circuit
      applied to either $(a',\Commitk{a'}{(s,r)}{pk})$ or
      $(a,\Commitk{a'}{(s,r)}{pk}$, it holds that $\D{0} \approxq
      \D{1}$.

      Now, let $\D{2}$ be the distribution obtained by not simulating
      the opening via the trapdoor, but instead doing it honestly to
      the value committed to, i.e.~$(a',r)$. We still use the
      challenge $S$ from the forced run of $\pi^{(\random,\force)}$
      though. However, for uniformly random challenges, real runs are
      quantum-computationally indistinguishable from simulated runs,
      and we get $\D{1} \approxq \D{2}$.

      Next, let $\D{3}$ be the distribution of the output of the
      simulation where we run $\pi^{(\random,\force)}$ honestly
      instead of enforcing outcome $S$. We then use the honestly
      produced $S'$ in the proof in
      Step~(\ref{step:test.on.S}.)~instead of the enforced $S$. We can
      do this, as we modified the process leading to $\D{2}$ towards
      an honest opening without any trapdoor, so we no longer need to
      enforce a particular challenge. Under the assumption that
      $\pi^{(\random,\force)}$ is enforceable against $\dB$, and
      observing that real runs are quantum-computationally
      indistinguishable from runs enforcing uniformly random outcomes,
      we obtain $\D{2} \approxq \D{3}$.
      
      Finally, we get by transitivity that $\D{0} \approxq \D{3}$ and
      conclude the proof by observing that after our changes, the
      process producing $\D{3}$ is the real protocol. This concludes
      the proof of enforceability against Bob, according to
      Definition~\ref{def:force} with switched sides.
  \end{proof}


\chapter{Applications}
\label{chap:coin.flip.applications}

Coin-flipping as a stand-alone tool allows us to use it rather freely
in several contexts. Shared randomness is a crucial ingredient in many
cryptographic implementations. Applications in the
common-reference-string-model, that assumes a random public string
before communication, achieve great efficiency and composability, and
many protocols have been proposed in the model. In this chapter, we
will discuss example applications that rely on shared randomness. Two
applications relate to the context of zero-knowledge. First, we show a
simple transformation from non-interactive zero-knowledge to
interactive quantum zero-knowledge. This result appeared
in~\cite{DL09}. Then, we propose a quantum-secure zero-knowledge proof
of knowledge, which is interesting also in that the construction
relies not only on initial randomness but also on enforceable
randomness as discussed in Chapter~\ref{chap:framework}. This
construction is part of the results in~\cite{LN10}. Last, we discuss
the interactive generation of a common reference string for the
proposed lattice-based instantiation of the compiler construction,
proposed in Chapter~\ref{chap:hybrid.security} and applied in
Chapter~\ref{chap:hybrid.security.applications}. This result appeared
in~\cite{DFLSS09} and~\cite{DL09}.


\section{Interactive Quantum Zero-Knowledge}
\label{sec:coin.iqzk}
\index{zero-knowledge!proofs}

Zero-knowledge proofs, as described in
Section~\ref{sec:primitives.zk}, are an important building block for
larger cryptographic protocols, capturing the definition of convincing
the verifier of the validity of a statement with no information beyond
that.


\subsection{Motivation and Related Work}
\label{sec:iqzk.motivation}

  As in the classical case, where ZK protocols exist if one-way
  functions exist, quantum zero-knowledge (QZK) is possible under the
  assumption that quantum one-way functions exist. In~\cite{K03},
  Kobayashi showed that a common reference string or shared
  entanglement is necessary for non-interactive quantum
  zero-knowledge. Interactive quantum zero-knowledge protocols in
  restricted settings were proposed by Watrous in the honest-verifier
  setting~\cite{Watrous02} and by Damg{\aa}rd \emph{et al.} in the
  CRS-model~\cite{DFS04}, where the latter introduced the first
  $\Sigma$-protocols for QZK withstanding even active quantum
  attacks. In~\cite{Watrous09}, Watrous then proved that several
  interactive protocols are zero-knowledge against general quantum
  attacks.

  It has also been shown that any honest-verifier zero-knowledge
  protocol can be made zero-knowledge against any classical and
  quantum verifier~\cite{HKSZ08}. In more detail, they showed how to
  transform a $\Sigma$-protocol with stage-by-stage honest-verifier
  zero-knowledge into a new $\Sigma$-protocol that is zero-knowledge
  against all verifiers. Special bit commitment schemes are proposed
  to limit the number of rounds, and each round is viewed as a stage
  in which an honest-verifier simulator is assumed. Then, by using a
  technique of~\cite{DGW94}, each stage can be converted to obtain
  zero-knowledge against any classical verifier. Finally, Watrous'
  quantum rewinding lemma is applied in each stage to prove
  zero-knowledge also against any quantum verifier. We now show a
  simple transformation from non-interactive (quantum) zero-knowledge
  to interactive quantum zero-knowledge by combining the coin-flip
  protocol with any non-interactive ZK protocol. Note that a
  non-interactive zero-knowledge proof system can be trivially turned
  into an interactive honest-verifier zero-knowledge proof system by
  just letting the verifier choose the reference string, and
  therefore, this consequence of our result also follows
  from~\cite{HKSZ08}. However, our proof is much simpler and the
  coin-flipping is not restricted to a specific zero-knowledge
  construction. In addition, we obtain the corollary that if there
  exist mixed commitments, then we can achieve interactive quantum
  zero-knowledge against any poly-sized quantum adversary without any
  set-up assumptions.


\subsection{Formal Definition of Zero-Knowledge Proofs}
\label{sec:iqzk.definitions}

  In Section~\ref{sec:primitives.zk}, we gave an intuitive
  introduction to zero-knowledge proof systems. Here, we make this
  description formal. Recall that a zero-knowledge proof for set $\cal
  L$ on common input $x$ yields no other knowledge than the validity
  of membership $x \in \cal{L}$. An interactive proof system must
  fulfill completeness and soundness, as given in
  Definitions~\ref{def:iqzk.complete} and~\ref{def:iqzk.sound}, and is
  quantum zero-knowledge (IQZK), if in addition
  Definition~\ref{def:iqzk.zero-knowledge} holds. Note that in the
  following, we let $\A$ be the prover and let $\B$ denote the
  verifier.
  \index{zero-knowledge!proofs!completeness}
  \index{zero-knowledge!proofs!soundness}

  \begin{definition}[Completeness]
    \label{def:iqzk.complete}
    If $x \in \cal L$, the probability that $(\A,\B)$ rejects $x$ is
    negligible in the length of $x$.
  \end{definition}

  \begin{definition}[Soundness]
    \label{def:iqzk.sound}
    If $x \notin \cal L$, then for any unbounded prover $\dA$, the
    probability that $(\dA,\B)$ accepts $x$ is negligible in the
    length of $x$.
  \end{definition}

  \begin{definition}[Zero-Knowledge]
    \label{def:iqzk.zero-knowledge}
    An interactive proof system $(\A,\dB)$ for language $\cal L$ is
    quantum zero-knowledge, if for any quantum verifier $\dB$, there
    exists a simulator $\hat{S}$ with output quantum-computationally
    indistinguishable from the real output, i.e.,
    $$
    out^{\hat{S}} \approxq out^{\pi(x,\crs)}_{\A,\dB} \, ,
    $$ 
    on common input $x \in \cal L$ and arbitrary additional
    (quantum) input to $\dB$.
  \end{definition}

  According to~\cite{BFM88}, the interaction between prover and
  verifier can be replaced by a common reference string. Then, there
  is only a single message sent from prover to verifier, who makes the
  final decision weather to accept or not. More precisely, both
  parties $\A$ and $\B$ get common input $x$. A common reference
  string $\crs$ of size $\kappa$ allows the prover $\A$, who knows a
  witness $w$, to give a non-interactive zero-knowledge proof
  $\pi(\crs,x)$ to a (possibly quantum) verifier, poly-time bounded in
  $\kappa$. For simplicity, we consider the proof of a single theorem
  of size smaller than $n$ (and $n \leq \kappa$, i.e.~${\cal L}_\kappa
  = \set{ x \in {\cal L} \; \vert \; |x| \leq \kappa}$. The extension
  to a more general notion is rather straightforward (see~\cite{BFM88}
  for details).\\
  
  Completeness and soundness hold as defined above, but we explicitly
  state the definitions as given in~\cite{BFM88} and adapted to our
  context.
  \index{zero-knowledge!non-interactive}

  \begin{definition}[Completeness in NIZK]
    \label{def:nizk.complete}
    There exists a constant $c > 0$ such that for all $x \in {\cal
    L}_\kappa$, the acceptance probability is overwhelming, i.e.,
    $$
    \pro{complete} = \prob{\crs \la \zo^{n^c}, \pi(x,\crs) \la
    \A(\crs,x,w): \B(\crs,x,\pi(x,\crs)) = 1} > 1 - \eps
    $$ where $\eps$ is negligible in $n$ (and $\kappa$).
  \end{definition}

  \begin{definition}[Soundness in NIZK]
    \label{def:nizk.sound}
    There exists a constant $c > 0$ such that for all $x \notin
    {\cal L}_\kappa$ and for all provers $\dA$, the acceptance
    probability is negligible, i.e.,
    $$ 
    \pro{sound} = \prob{\crs \la \zo^{n^c}, \pi(x,\crs) \la
    \dA(\crs,x): \B(\crs,x,\pi(x,\crs)) = 1} \leq \eps'
    $$ where $\eps'$ is negligible in $n$ (and $\kappa$). 
  \end{definition}

  The non-interactive zero-knowledge requirement is simpler than for
  general zero-knowledge for the following reason. Since all
  information is communicated mono-directional from prover to verifier
  in the protocol, the verifier does not influence the distribution in
  the real world. Thus, in the ideal world, we require a simulator
  that only outputs pairs that are (quantum) computationally
  indistinguishable from the distribution of pairs
  $(\crs,\pi(x,\crs))$ in the real world, where $\pi$ is generated
  with uniformly chosen $\crs$ and random
  $x$.\footnote{Indistinguishability, in turn, implies that the proof
  construction withstands quantum-computationally bounded verifiers.}
  In other words, we can eliminate the quantification over all $\dB$
  in the zero-knowledge definition.
  
  \begin{definition}[Non-Interactive Zero-Knowledge]
    \label{def:nizk.zero-knowledge}
    There exist a constant $c > 0$ and a simulator $\hat{S}$ with
    output quantum-computationally indistinguishable from the real
    output, i.e.,
    $$ 
    out^{\hat{S}(x)} \approxq out^{\pi(x,\crs)}_{\A,\dB} \, ,
    $$ 
    where $out^{\hat{S}(x)} = \set{ \crs \la \zo^{|x|^c}, \pi(x,\crs)
    \la \A(x,\crs): (\crs,\pi(x,\crs))}$.
  \end{definition}


\subsection{The Transformation}
\label{sec:iqzk.transformation}

  We obtain a generic transformation of non-interactive zero-knowledge
  into interactive quantum zero-knowledge as follows. In each
  invocation, protocol $\COIN$ generates a truly random $\coin$ even
  in the case of a malicious quantum $\dB$. A string of such coins,
  obtained by sequential composition as described in
  Section~\ref{sec:sequential.composition.coin} by the ideal
  functionality in Figure~\ref{fig:clambdaF}, is then used as reference
  string in any ($\NIZK$)-subprotocol with properties as defined
  previously.

  The final protocol $\IQZK$ is shown in Figure~\ref{fig:iqzk}. To
  prove that it is an interactive quantum zero-knowledge protocol, we
  first construct an intermediate protocol $\IQZKF$ (see
  Figure~\ref{fig:iqzkf}) that runs with the ideal functionality
  $\cxF{\kappa}$. Then we prove that $\IQZKF$ satisfies completeness,
  soundness and zero-knowledge according to
  Definitions~\ref{def:iqzk.complete}
  -~\ref{def:iqzk.zero-knowledge}. To complete the proof, the calls to
  $\cxF{\kappa}$ are replaced with actual invocations of
  $\cxPi{\A}{\B}{\kappa}$, and we arrive at $\IQZK$.

  \begin{figure}
    \begin{myprotocol}{\sc protocol $\IQZKF$}
      \item[$(\COIN)$]
      \item\label{step:coin} $\A$ and $\B$ invoke $\cxF{\kappa}$. If
	$\A$ aborts by sending $\bot$ as second input, $\B$ aborts the
	protocol. Otherwise, $\A$ and $\B$ set $\crs = h$. 
      \item[$(\NIZK)$]
      \item $\A$ sends $\pi(x,\crs)$ to $\B$. $\B$ checks the proof and
        accepts or rejects accordingly.
    \end{myprotocol}
    \vspace{1ex} 
    \caption{Intermediate Protocol for IQZK.}
    \label{fig:iqzkf}
  \end{figure}

  \begin{claim}
    \label{claim:zk.complete} 
    Protocol $\IQZKF$ satisfies completeness, according to
    Definition~\ref{def:iqzk.complete}.
  \end{claim}

  \begin{proof}
    From the ideal functionality $\cxF{\kappa}$ it follows that $\crs$
    is uniformly random. Then by Definition~\ref{def:nizk.complete} of
    any ($\NIZK$)-subprotocol, we know that, for $x \in {\cal
    L}_\kappa$, $\B$ accepts, except with negligible probability (in
    the length of $x$. Thus, completeness for the $\IQZKF$ follows.
  \end{proof}

  \begin{claim}
    \label{claim:sound}
    Protocol $\IQZKF$ satisfies soundness, according to
    Definition~\ref{def:iqzk.sound}.
  \end{claim}
  
  \begin{proof}
    Assume that $x \notin {\cal L}_\kappa$. Any dishonest $\dA$ might
    stop $\IQZKF$ at any point during execution. For example, she can
    block the output in Step~(\ref{step:coin}.)~or she can refuse to
    send a proof $\pi$ in $(\NIZK)$. Furthermore, $\dA$ can use an
    invalid $\crs$ (or $x$) for $\pi$. In all of these cases, $\B$
    will abort without even checking the proof.

    Therefore, $\dA$'s best strategy is to ``play the entire game'',
    i.e.\ to execute $\IQZKF$ without making obvious cheats. $\dA$
    can only convince $\B$ in the $(\NIZK)$-subprotocol of a $\pi$ for
    any given (i.e.\ normally generated) $\crs$ with a probability
    that is negligible in the length of $x$ (see
    Definition~\ref{def:nizk.sound}). Therefore, the probability that
    $\dA$ can convince $\B$ in the full $\IQZKF$ in case of $x \notin
    {\cal L}_\kappa$ is also negligible and its soundness follows.
  \end{proof}

  \begin{claim}
    \label{claim:qzk}
    Protocol $\IQZKF$ is an interactive zero-knowledge proof,
    according to Definition~\ref{def:iqzk.zero-knowledge}.
  \end{claim}

 \begin{proof}
    We construct a simulator $\SIQZKF$, interacting with dishonest
    $\dB$ and a simulator $\SNIZK$. As given in
    Definition~\ref{def:nizk.zero-knowledge}, such a simulator
    generates, on input $x \in \cal L$, a randomly looking $\crs$
    together with a valid proof $\pi$ for $x$ (without knowing witness
    $w$). $\SIQZKF$, described in Figure~\ref{fig:simulationZKF},
    receives a random string $\tilde{\crs}$ from $\SNIZK$, which now
    replaces the coin-string $h$ produced by $\cxF{\kappa}$ in
    protocol $\IQZKF$. By assumption on $\SNIZK$, this is
    quantum-computationally indistinguishable for $\dB$. Thus, the
    simulated proof $\pi(\crs,x)$ is indistinguishable from a real
    proof, which proves that the $\IQZKF$ is zero-knowledge.
  \end{proof}

  \begin{figure}
    \begin{framed}
      \noindent\hspace{-1.5ex} {\sc Simulation $\SIQZKF$:}\\[-4ex]
      \begin{enumerate}
        \item 
	  $\SIQZKF$ gets input $x$ and invokes $\SNIZK$ with $x$ to
	  receives $\pi(\crs, x)$.
	\item Let $\crs = h$. $\SIQZKF$ sends $h$ to $\dB$.
	\item $\SIQZKF$ sends $\pi(\crs, x)$ to $\dB$ and outputs
	  whatever $\dB$ outputs.
      \end{enumerate}
      \vspace{-1ex}
    \end{framed}
    \vspace{-1.5ex}
    \caption{The Simulation of the Intermediate Protocol for IQZK.}
    \label{fig:simulationZKF}
   \end{figure}

  It would be natural to think that $\IQZK$ could be proven secure
  simply by showing that $\IQZKF$ implements some appropriate
  functionality and then use a composition theorem from
  Section~\ref{sec:security.definition}. Recall, however, that a
  zero-knowledge protocol---which is not necessarily a proof of
  knowledge---cannot be modeled by a functionality in a natural
  way. Instead, we prove the standard properties of a zero-knowledge
  proof system explicitly and therewith the following
  Theorem~\ref{thm:iqzk}.

  \begin{theorem}[Interactive Quantum Zero-Knowledge]
    \label{thm:iqzk}
    Protocol $\IQZK$ is an interactive proof system, satisfying
    completeness and soundness. Since, for any quantum verifier $\dB$,
    there exists a simulator $\SIQZK$ with output
    quantum-computationally indistinguishable from the real output, we
    additionally achieve quantum zero-knowledge.
  \end{theorem}
 
  \begin{figure}
    \begin{myprotocol}{\sc protocol $\IQZK$}
      \item[($\COIN$)] 
	$\A$ and $\B$ run $\cxPi{\A}{\B}{\kappa}$ and
	set $\crs = h$.\medskip
      \item[$(\NIZK)$]
	$\A$ sends $\pi(\crs,x)$ to $\B$. $\B$ checks the proof and
	accepts or rejects accordingly.
    \end{myprotocol}
    \vspace{1ex} 
    \caption{Interactive Quantum Zero-Knowledge.}
    \label{fig:iqzk}
   \end{figure}

  \begin{proof}
    From the analysis of protocol $\COIN$, its sequential
    composability, and the indistinguishability from the ideal
    functionality $\cxF{\kappa}$, it follows that if both players are
    honest $\crs$ is a random common reference string of size $\kappa$
    and the acceptance probability of the $(\NIZK)$-subprotocol as
    given previously holds. Completeness of $\IQZK$ follows.

    To show soundness, we again only consider the case where $\dA$
    executes the entire protocol without making obvious cheats, since
    otherwise, $\B$ immediately aborts. Assume that $\dA$ could cheat
    in $\IQZK$, i.e., $\B$ would accept an invalid proof with
    non-negligible probability. Then we could combine $\dA$ with
    simulator $\dhA$ of protocol $\COIN$ (Figure~\ref{fig:simulationA})
    to show that $\IQZKF$ was not sound. This, however, is
    inconsistent with the previously given soundness argument in the
    proof of Claim~\ref{claim:sound}, and thus proves by contradiction
    that $\IQZK$ is sound.

    To further prove that the interactive proof system is also quantum
    zero-knowledge, we compose a simulator $\SIQZK$ of simulator
    $\SIQZKF$ (Figure~\ref{fig:simulationZKF}) and simulator $\dhB$ of
    protocol $\COIN$ (Figure~\ref{fig:simulationB}). In more detail,
    $\SIQZK$ gets classical input $x$ as well as quantum input $W$ and
    $X$. It then receives a valid proof $\pi$ and a random string
    $\crs$ from $\SNIZK$. $\crs$ is split into $coin_1 \ldots
    coin_k$. For each $coin_i$, it will then invoke $\dhB$ to simulate
    one coin-flip execution with $\coin = coin_i$ as result. In other
    words, whenever $\dhB$ asks $\cF$ to output a bit
    (Step~(\ref{step:get-coin}.), Figure~\ref{fig:simulationB}), it
    instead receives this $coin_i$. We see that the transcript of the
    simulation is indistinguishable from the transcript of the
    protocol $\IQZK$ for any quantum-computationally bounded
    $\dB$. This concludes the proof.
  \end{proof}\\

  We conclude this section by the corollary, immediately following
  from the previous proof and stating that quantum-secure commitments,
  as defined in Section~\ref{sec:bit.commit}, imply interactive
  quantum zero-knowledge.
  \begin{corollary}
    If there exist quantum-secure commitment schemes, then we can
    obtain interactive quantum zero-knowledge against any quantum
    adversary $\dP \in \dPlayerPoly$ without any set-up assumptions.
  \end{corollary}


\section{Zero-Knowledge Proof of Knowledge}
\label{sec:coin.zkpk}
\index{zero-knowledge!proofs of knowledge}

A zero-knowledge proof of knowledge is a special case of
zero-knowledge proof systems, introduced in
Section~\ref{sec:primitives.zk}. Here, we propose a quantum-secure
construction based on witness encoding, which we define in the context
of simulation.


\subsection{Motivation and Related Work}
\label{sec:zkpk.motivation}

  Recall that the purpose of a zero-knowledge proof of knowledge is to
  verify in classical poly-time in the length of the instance, whether
  $w$ is a valid witness for instance $x$ in relation $\Rel$,
  i.e.~$(x,w) \in \Rel$. We call $\Rel$ an \NP-relation, as the
  language ${\cal L(R)} = \set{ x \in \zo^* \vert \ \exists \, w \
  \text{s.t.}  \ (x,w) \in \Rel }$ is seen to be an
  \NP-language. Interestingly, such a zero-knowledge proof of
  knowledge, in contrast to zero-knowledge proofs, can be modeled by
  an ideal functionality.

  Our protocol is based on a witness encoding scheme, providing a
  certain degree of extractability and simulatability, defined in
  Section~\ref{sec:zkpk.witness.encoding}. We want to stress that the
  extractability requirement resembles special soundness in proof
  systems, which are secure in the classical world and typically come
  along with a knowledge error negligible in the length of the
  challenge. We have to reformulate this aspect in stronger terms in
  the quantum world, since special soundness seems to be impossible to
  use in the quantum realm, due to the restrictions within
  rewinding. However, we obtain a similar result also with knowledge
  error negligible in the length of the challenge.

  Furthermore, our construction requires a mixed bit commitment (see
  Section~\ref{sec:mixed.commit}) and two calls to the coin-flip
  protocol $\pi^{(\force,\force)}$, described in
  Figure~\ref{fig:force.force}, Chapter~\ref{chap:framework}, which is
  poly-time simulatable for both sides even against quantum
  adversaries. Since this protocol only assumes mixed commitments as
  well, we get the corollary that if there exists a mixed commitment
  scheme, then we can construct a classical zero-knowledge proof of
  knowledge against any poly-sized quantum adversary. This is of
  particular interest, as the problems of rewinding in the quantum
  realm complicate implementing proofs of knowledge from scratch.

  As already mentioned in Chapter~\ref{chap:coin.flip}, the
  unpublished approach of~\cite{Smith09} suggest another solution for
  this concept. Instead of composing the coin-string from single
  coins, they use a string commitment with special opening and compose
  the subsequent zero-knowledge proof. The coin-string is used as key
  to encode the witness and the second zero-knowledge proof is given
  to prove it.


\subsection{Simulatable Witness Encodings of \NP}
\label{sec:zkpk.witness.encoding}
\index{simulatable witness encoding}

  We first specify a simulatable encoding scheme for binary relation
  $\Rel \subset \zo^* \times \zo^*$, which consists of five classical
  poly-time algorithms $(E,D,S,J,\hat{E})$. Then, we define
  completeness, extractability and simulatability for such a scheme in
  terms of the requirements of our zero-knowledge proof of knowledge.

  Let $E: \Rel \times \zo^m \ra \zo^n$ denote an \defterm{encoder},
  such that for each $(x,w) \in \Rel$, the $n$-bit output $e \la
  E(x,w,r')$ is a random encoding of $w$, with randomness $r' \in
  \zo^m$ and polynomials $m(\vert x \vert)$ and $n(\vert x
  \vert)$. The corresponding \defterm{decoder} $D: \zo^* \times \zo^n
  \ra \zo^*$ takes as input an instance $x \in \zo^*$ and an encoding
  $e \in \zo^n$ and outputs $w \la D(x,e)$ with $w \in \zo^*$. 

  Next, let $S$ denote a \defterm{selector} with input $s \in
  \zo^\sigma$ (with polynomial $\sigma(\vert x \vert)$) specifying a
  challenge, and output $S(s)$ defining a poly-sized subset of
  $\set{1,\ldots,n}$ corresponding to challenge $s$. We will use
  $S(s)$ to select which bits of an encoding $e$ to reveal to the
  verifier. For simplicity, we use $e_s$ to denote the collection of
  bits $e|_{S(s)}$. 

  We denote with $J$ the \defterm{judgment} that checks a potential
  encoding $e$ by inspecting only bits $e_s$. In more detail, $J$
  takes as input instance $x \in \zo^*$, challenge $s \in \zo^\sigma$
  and the $\vert S(s) \vert$ bits $e_s$, and outputs a judgment $j \la
  J(x,s,e_s)$ with $j \in \set{\abort, \success}$. 

  Finally, the \defterm{simulator} is called $\hat{E}$. It takes as
  input instance $x \in \zo^*$ and challenge $s \in \zo^\sigma$ and
  outputs a random collection of bits $t|_{S(s)} \la
  \hat{E}(x,s)$. Again for simplicity, we let $t|_{S(s)} = t_s$. Then,
  if this set has the same distribution as bits of an encoding $e$ in
  positions $S(s)$, the bits needed for the judgment to check an
  encoding $e$ can be simulated given just instance $x$ (see
  Definition~\ref{def:simulate}).
  \index{simulatable witness encoding!completeness}

  \begin{definition}[Completeness]
    \label{def:zkpk.complete}
    If an encoding $e\la~E(x,w,r)$ is generated correctly, then
    $\success \la J(x,s,e_s)$ for all $s \in_R \zo^\sigma$.
  \end{definition}

  We will call an encoding $e$ \defterm{admissible} for $x$, if there
  \emph{exist} two distinct challenges $s,s' \in \zo^\sigma$ for which
  $\success \la J(x,s,e_s)$ and $\success \la J(x,s',e_{s'})$.
  \index{simulatable witness encoding!extractability}
  \index{simulatable witness encoding!admissible}

  \begin{definition}[Extractability]
    \label{def:extract}
    If an encoding $e$ is admissible for $x$, then
    $\big( x,D(x,e) \big)\in \Rel$.
  \end{definition}
  We want to stress that extractability is similarly defined to the
  special soundness property of a classical $\Sigma$-protocol, which
  allows to extract $w$ from two accepting conversations with
  different challenges. Such a requirement would generally be
  inapplicable in the quantum setting, as the usual rewinding
  technique is problematic and in particular in the context here, we
  cannot measure two accepting conversations during rewinding in the
  quantum world. Therefore, we define the stronger requirement that if
  there \emph{exist} two distinct answerable challenges for one
  encoding $e$, then $w$ can be extracted given only $e$. This
  condition works nicely in the quantum world, since we can obtain $e$
  without rewinding, as we will show in our quantum-secure proof
  construction.
  \index{simulatable witness encoding!simulatability}

  \begin{definition}[Simulatability]
    \label{def:simulate}
    For all $(x,w) \in \Rel$ and all $s \in_R \zo^\sigma$, the
    distribution of $e \leftarrow E(x,w,r')$ restricted to positions
    $S(s)$ is identical to the distribution of $t_s \leftarrow
    \hat{E}(x,s)$, i.e.,
    $$
    {\cal D}(e_s) = {\cal D}(t_s) \, .
    $$
  \end{definition}
   
  There are several commit\&open proofs for $\NP$. One can, for
  instance, start from the commit\&open protocol for circuit
  satisfiability, where the bits of the randomized circuit committed
  to by the sender is easy to see as a simulatable encoding of a
  witness being a consistent evaluation of the circuit to output
  $1$. The challenge in the protocol is one bit $e$ and the prover
  replies by showing either the bits corresponding to some positions
  $S'(0)$ or positions $S'(1)$. The details can be found
  in~\cite{BCC88}. This gives us a simulatable witness encoding for
  any $\NP$-relation $\Rel$ with $\sigma = 1$, using a reduction from
  $\NP$ to circuit simulatability. By repeating it $\sigma$ times in
  parallel we get a simulatable witness encoding for any $\sigma$. For
  $i = 1, \ldots, \sigma$, compute an encoding $e^i$ of $w$ and let $e
  = (e^1, \ldots, e^\sigma)$. Then for $s \in \zo^\sigma$, let $S(s)$
  specify that the bits $S'(s_i)$ should be shown in $e^i$ and check
  these bits. Note, in particular, that if two distinct $s$ and $s'$
  passes this judgment, then there exists $i$ such that $s_i \ne
  s_i'$, so $e^i$ passes the judgment for both $s_i = 0$ and $s_i =
  1$, which by the properties of the protocol for circuit
  satisfiability allows to compute a witness $w$ for $x$ from
  $e^i$. One can find $w$ from $e$ simply by trying to decode each
  $e^j$ for $j = 1, \ldots, \sigma$ and check if $(x,w_j) \in \Rel$.


\subsection{The Protocol}
\label{sec:zkpk.protocol}

  We now construct a quantum-secure zero-knowledge proof of knowledge
  from prover $\A$ to verifier $\B$. Recall that we are interested in
  the \NP-language ${\cal L(R)} = \set{ x \in \zo^* \, \vert \,
  \exists \, w \ \text{s.t.} \ (x,w) \in \Rel }$, where $\A$ has input
  $x$ and $w$, and both $\A$ and $\B$ receive positive or negative
  judgment of the validity of the proof as output. We assume in the
  following that on input $(x,w) \notin \Rel$, honest $\A$ aborts. The
  final protocol $\ZKPK$ is describe in Figure~\ref{fig:zkpk}.

  As already mentioned, unlike zero-knowledge proofs, proofs of
  knowledge can be modeled by an ideal functionality, given as
  $\zkpkF$ in Figure~\ref{fig:zkpkF}. $\zkpkF$ can be thought of as a
  channel which only allows to send messages in the language $\cal
  L(R)$. It models \emph{zero-knowledge}, as it only leaks instance
  $x$ and judgment $j$ but not witness $w$. Furthermore, it models a
  \emph{proof of knowledge}, since Alice has to know and input a valid
  witness $w$ to obtain output $j = \success$.

  \begin{figure}
    \begin{framed}
      \noindent\hspace{-1.5ex} {\sc Functionality $\zkpkF$:}\\[-4ex]
      \begin{enumerate}
      \item
	On input $(x,w)$ from Alice, $\zkpkF$ sets $j = \success$ if
	$(x,w) \in \Rel$. Otherwise, it sets $j = \abort$.
      \item
	$\zkpkF$ outputs $(x,j)$ to Alice and Bob.
      \end{enumerate}
      \vspace{-1.5ex}
    \end{framed}
    \vspace{-1.5ex}
    \caption{The Ideal Functionality for a Zero-Knowledge Proof of
      Knowledge.}\label{fig:zkpkF}
   \end{figure}

  Protocol $\ZKPK$ is based on our fully simulatable coin-flip
  protocol $\pi^{(\force,\force)}$, which we analyze here in the
  hybrid model by invoking the ideal functionality of sequential
  coin-flipping twice (but with different output lengths). Note that
  in the hybrid model, a simulator can enforce a particular outcome to
  hit also when invoking the ideal coin-flipping functionality. We can
  then use Definition~\ref{def:force} to replace the ideal
  functionality by the actual protocol $\pi^{(\force,\force)}$.

  One call to the ideal functionality $\cxF{\kappa}$ with output
  length $\kappa$ is required to instantiate a mixed bit commitment
  scheme $\Commit$ as discussed in
  Section~\ref{sec:mixed.commit.trapdoor.opening}. Recall that it is
  therewith possible to sample an unconditionally binding key $pk \in
  \zo^\kappa$ along with an extraction key $sk$. Since such keys are
  quantum-computationally indistinguishable from random values in
  $\zo^\kappa$, the latter serves us as unconditionally hiding
  instantiations of $\Commit$. The second call to the functionality
  $\cxF{\sigma}$ produces $\sigma$-bit challenges for a simulatable
  witness encoding scheme with $(E,D,S,J,\hat{E})$ as specified in the
  previous Section~\ref{sec:zkpk.witness.encoding}.

  \begin{figure}
    \begin{framed}
      \noindent\hspace{-1.5ex} {\sc Protocol $\ZKPK$ :}\\[-4ex]
      \begin{enumerate}
      \item
	$\A$ and $\B$ invoke $\cxF{\kappa}$ to get a commitment key
	$pk \in \zo^\kappa$.
      \item
	$\A$ samples $e \leftarrow E(x,w,r')$ with randomness $r' \in
	\zo^m$ and commits position-wise to all $e_i$ for $i =
	1,\ldots,n$, by computing random commitments $E_i =
	\Commitk{e_i}{r_i}{pk}$ with randomness $r \in \zo^n$. She
	sends $x$ and all $E_i$ to $\B$.
      \item
	$\A$ and $\B$ invoke $\cxF{\sigma}$ to flip a challenge $s
	\in_R \zo^\sigma$.
      \item  
	$\A$ opens her commitments to all $e_s$.
      \item
	If any opening is incorrect, $\B$ outputs $\abort$. Otherwise, he
	outputs $j \la J(x,s,e_s)$ with $j \in \set{\success,\abort}$.
      \end{enumerate}
      \vspace{-1ex}
    \end{framed}
    \vspace{-1.5ex}
    \caption{Zero-Knowledge Proof of Knowledge.}
    \label{fig:zkpk}
  \end{figure}

  \begin{theorem}[Zero-Knowledge Proof of Knowledge]\label{thm:zkpk}
    For any simulatable witness encoding scheme $(E,D,S,J,\hat{E})$,
    satisfying completeness, extractability, and simulatability
    according to Definitions~\ref{def:zkpk.complete}
    -~\ref{def:simulate}, and for negligible knowledge error
    $2^{-\sigma}$, protocol $\ZKPK$ is a zero-knowledge proof of
    knowledge and securely implements $\zkpkF$.
  \end{theorem}

  Completeness is obvious. A honest party $\A$, following the protocol
  with $(x,w) \in \Rel$ and any valid encoding $e$, will be able to
  open all commitments in the positions specified by any challenge
  $s$. Honest Bob then outputs $J(x,s,e_s) =
  \success$.\\

  \begin{proof}[ (\emphbf{Security against dishonest Alice})]
    To prove security in case of corrupted $\dA$, we construct a
    simulator $\dhA$ that simulates a run of the actual protocol with
    $\dA$ and $\zkpkF$. The proof is then twofold. First, we show
    indistinguishability between the distributions of simulation and
    protocol. And second, we verify that the extractability property
    of the underlying witness encoding scheme (see
    Definition~\ref{def:extract}) implies a negligible knowledge
    error. Note that if $\dA$ sends $\abort$ at any point during the
    protocol, $\dhA$ sends some input $(x',w') \notin \Rel$ to $\zkpkF$
    to obtain output $(x,j)$ with $j = \abort$, and the simulation
    halts. Otherwise, the simulation proceeds as shown in
    Figure~\ref{fig:simulation.zkpk.dhA}.

  \begin{figure}
    \begin{framed}
      \noindent\hspace{-1.5ex} {\sc Simulation $\dhA$ for
      $\ZKPK$ :}\\[-4ex]
      \begin{enumerate}
      \item 
	$\dhA$ samples a random key $pk$ along with the extraction key
	$sk$. Then it enforces $pk$ as output from $\cxF{\kappa}$
      \item
	When $\dhA$ receives $x$ and $(E_1,\ldots,E_n)$ from $\dA$, it
	extracts $e = (\xtr{E_1}{sk},\ldots,\xtr{E_n}{sk})$.
      \item
	$\dhA$ completes the simulation by following the protocol
	honestly. If any opening of $\dA$ is incorrect, $\dhA$
	aborts. Otherwise, $\dhA$ inputs $\big( x,D(x,e) \big)$ to
	$\zkpkF$ and receives $(x,j)$ back. $\dhA$ outputs the final
	state of $\dA$ as output in the simulation.
      \end{enumerate}
      \vspace{-1.5ex} 
    \end{framed}
    \vspace{-1.5ex}
    \caption{Simulation against dishonest Alice.}
    \label{fig:simulation.zkpk.dhA}
   \end{figure}

    Note that the only difference between the real protocol and the
    simulation is that $\dhA$ uses a random public key $pk$ sampled
    along with an extraction key $sk$, instead of a uniformly random
    $pk \in \zo^\kappa$. It then enforces $\cxF{\kappa}$ to hit
    $pk$. However, by assumption on the commitment keys and by the
    properties of the ideal coin-flipping functionality, the transcripts
    of simulation and protocol remain quantum-computationally
    indistinguishable under these changes.

    Next, we analyze the output in more detail. It is clear that
    whenever honest $\B$ would output $\abort$ in the actual protocol,
    also $\dhA$ aborts, namely, if $\dA$ does deviate in the last
    steps of protocol and simulation, respectively. Furthermore,
    $\dhA$ accepts if and only if $(x,D(x,e)) \in \Rel$ or in other
    words, the judgment of the functionality is positive, denoted by
    $j_\F = \success$.

    It is therefore only left to prove that the case of $j_\F =
    \abort$ but $j_J = \success$ is negligible, where the later
    denotes the judgment of algorithm $J(x,s,e_s)$ as in the
    protocol. In that case, we have $(x,D(x,e)) \notin \Rel$. This means
    that $w$ is not extractable from $D(x,e)$, which in turn implies
    that $(\xtr{E_1}{sk},\ldots,\xtr{E_n}{sk}) = e$ is not
    admissible. Thus, there are no two distinct challenges $s$ and
    $s'$, in which $\dA$ could correctly open her commitment to $e$.
    It follows by contradiction that there exists at most one
    challenge $s$ which $\dA$ can answer. We produce $s \in
    \zo^\sigma$ uniformly at random, from which we obtain an
    acceptance probability of at most $2^{-\sigma}$. Thus, we conclude
    the proof with negligible knowledge error, as desired.\\
  \end{proof}

  \begin{proof}[ (\emphbf{Security against dishonest Bob})]
    To prove security in case of corrupted $\dB$, we construct
    simulator $\dhB$ as shown in
    Figure~\ref{fig:simulation.zkpk.dhB}. Our aim is to verify that
    this simulation is quantum-computationally indistinguishable from
    the real protocol. The key aspect will be the simulatability
    guarantee of the underlying witness encoding scheme, according to
    Definition~\ref{def:simulate}.
    
  \begin{figure}
    \begin{framed}
      \noindent\hspace{-1.5ex} {\sc Simulation $\dhB$ for $\ZKPK$
      :}\\[-4ex]
      \begin{enumerate}
      \item
	$\dhB$ invokes $\cxF{\kappa}$ to receive a uniformly random
	$pk$.
      \item
	$\dhB$ samples a uniformly random challenge $s \in \zo^\sigma$
	and computes $t_s \la \hat{E}(x,s)$. $\dhB$ then computes
	commitments $E_i$ as follows: For all $i \in S(s)$, it commits
	to the previously sampled $t_s$ via $E_i =
	\Commitk{t_i}{r_i}{pk}$. For all other positions $i \in
	\bar{S}$ (where $\bar{S} = \set{1,\ldots,n} \setminus S(s)$),
	it commits to randomly chosen values $t'_i \in_R \zo$,
	i.e.~$E_i = \Commitk{t'_i}{r_i}{pk}$. It sends $x$ and all
	$E_i$ to $\dB$.
      \item
	$\dhB$ forces $\cxF{\sigma}$ to hit $s$.
      \item
	$\dhB$ opens $E_i$ to $t_i$ for all $i \in S(s)$, i.e.\ to all
	$t_s$.
      \item $\dhB$ outputs whatever $\dB$ outputs.
      \end{enumerate}
      \vspace{-1.5ex} 
    \end{framed}
    \vspace{-1.5ex}
    \caption{Simulation against dishonest Bob.}
    \label{fig:simulation.zkpk.dhB}
  \end{figure}

    The proof proceeds via a hybrid argument. Let $\D{0}$ be the
    distribution of the simulation as described in
    Figure~\ref{fig:simulation.zkpk.dhB}. Let $\D{1}$ be the
    distribution obtained from the simulation but with the following
    change: We inspect $\zkpkF$ to get a valid witness $w$ for
    instance $x$, and let $e \la E(x,w,r')$ be the corresponding
    encoding. Note that this is possible as a thought experiment for
    any adjacent distribution in a hybrid argument. From $e$ we then
    use bits $e_s$ for the same $S(s)$ as previously, instead of bits
    $t_s$ sampled by $\hat{E}(x,s)$. All other steps are simulated as
    before. By the simulatability of the encoding scheme
    (Definition~\ref{def:simulate}), it holds that the bits $t_s$ in
    $\D{0}$ and the bits $e_s$ in $\D{1}$ have the same
    distribution. Thus, we obtain $\D{0} = \D{1}$.

    We further change the simulation in that we compute the bits in
    all positions $i \in \bar{S}$ by $e_i$ of the encoding $e$ defined
    in the previous step. Again, all other steps of the simulation
    remain unchanged. Let $\D{2}$ denote the new distribution. The
    only difference now is that for $i \in \bar{S}$, the commitments
    $E_i$ are to the bits $e_i$ of a valid $e$ and not to uniformly
    random bits $t'_i$. This, however, is quantum-computationally
    indistinguishable to $\dB$ for $pk \in_R \zo^\kappa$, as $\Commit$
    is quantum-computationally hiding towards $\dB$. Note that $pk$ is
    guaranteed to be random by an honest call to $\cxF{\kappa}$ and
    recall that we do not have to open the commitments in these
    positions. Hence, we get that $\D{1} \approxq \D{2}$.\\

    Note that after the two changes, leading to distributions $\D{1}$
    and $\D{2}$, the commitment step and its opening now proceed as in
    the actual protocol, namely, we commit to the bits of $e \la
    E(x,e,r')$ and open the subset corresponding to $S(s)$. The
    remaining difference to the real protocol is the enforcement of
    challenge $s$, whereas $s$ is chosen randomly in the
    protocol. Now, let $\D{3}$ be the distribution of the modified
    simulation, in which we implement this additional change of
    invoking $\cxF{\sigma}$ honestly and then open honestly to the
    resulting $s$. Note that both processes, i.e., first choosing a
    random $s$ and then enforcing it from $\cxF{\sigma}$, or invoking
    $\cxF{\sigma}$ honestly and receiving a random $s$, result in a
    uniformly random distribution on the output of
    $\cxF{\sigma}$. Thus, we obtain $\D{2} = \D{3}$.

    By transitivity, we conclude that $\D{0} \approxq \D{3}$, and
    therewith, that the simulation is quantum-computationally
    indistinguishable from the actual protocol.
  \end{proof}\\

  We conclude this section by the corollary that follows
  straightforward from the above construction and proof and states
  that mixed commitments, as defined in
  Section~\ref{sec:mixed.commit.trapdoor.opening}, imply classical
  zero-knowledge proofs of knowledge against any poly-sized quantum
  adversary.
    \begin{corollary}
      If there exist mixed commitment schemes, then we can construct a
      classical zero-knowledge proof of knowledge against any quantum
      adversary $\dP \in \dPlayerPoly$ without any set-up assumptions.
    \end{corollary}


\section{Generation of Commitment Keys}
\label{sec:key.generation.coin}

Here, we briefly describe the initial generation of a common reference
string for the proposed lattice-based instantiation of the generic
compiler, introduced in Chapter~\ref{chap:hybrid.security}, according
to the specific requirements of its underlying mixed commitment
scheme, discussed in Section~\ref{sec:mixed.commit}.


\subsection{Motivation}
\label{sec:generation.motivation}
  
  The compiler is constructed in the CRS-model to achieve high
  efficiency. We now aim at circumventing the CRS-assumption to
  achieve the potential of allowing the implementation of
  \emph{complete} protocols in the quantum world without any set-up
  assumptions. More specifically, we integrate the generation of a
  common reference string from scratch based on our quantum-secure
  coin-flipping, which will then be used during compilation as
  commitment key. We want to stress, however, that implementing the
  entire process comes at the cost of a non-constant round
  construction, added to otherwise very efficient protocols under the
  CRS-assumption.


\subsection{The Generation}
\label{sec:generation.generation}
\index{commitment!key generation}

  Recall that the argument for computational security in
  Section~\ref{sec:compiler} proceeds along the following lines. After
  the preparation phase $\B$ commits to all his measurement bases and
  outcomes. The keyed dual-mode commitment scheme has the special
  properties that the key can be generated by one of two possible
  key-generation algorithms $\GH$ or $\GB$. Depending on the key in
  use, the scheme provides both flavors of security. Namely, with key
  $\pkH$ generated by $\GH$, respectively $\pkB$ produced by $\GB$,
  the commitment scheme is unconditionally hiding respectively
  unconditionally binding. Furthermore, the commitment is secure
  against a quantum adversary and it holds that $\pkH \approxq \pkB$.
  In the real-world protocol, $\B$ uses the unconditionally hiding key
  $\pkH$ to maintain unconditional security against any unbounded
  $\dA$. To argue security against a computationally bounded $\dB$, an
  information-theoretic argument involving the simulator $\dhB$ is
  given (in the proof of Theorem~\ref{thm:compiler}) to prove that
  $\dB$ cannot cheat with the unconditionally binding key
  $\pkB$. Security in real life then follows from the
  quantum-computational indistinguishability of $\pkH$ and $\pkB$.

  We want to repeat that we can even weaken the assumption on the
  hiding key in that we do in fact not require an actual
  unconditionally hiding key, if the public-key encryption scheme
  guarantees that a random public key looks pseudo-random to poly-time
  quantum circuits. As discussed in Section~\ref{sec:mixed.commit}, the
  lattice-based crypto-system of Regev~\cite{Regev05}, which is
  considered to withstand quantum attacks, is a good candidate to
  construct such a dual-mode commitment scheme. The public key of a
  regular key pair can be used as the unconditionally binding key
  $\pkB'$ in our commitment scheme for the ideal-world simulation, and
  for the real protocol, an unconditionally hiding commitment key
  $\pkH'$ can simply be constructed by uniformly choosing numbers in
  the same domain.

  The idea is now the following. Let $k$ denote the length of a
  regular key $\pkB'$. We add (at least) $k$ executions of our
  protocol $\COIN$ as a first step to the compiler-construction to
  generate a uniformly random sequence $coin_1 \ldots coin_k$. These
  $k$ random bits produce a $\pkH'$ as sampled by
  $\mathcal{G}_\mathtt{H}$, except with negligible probability. Hence,
  in the real world, Bob can use key $coin_1 \ldots coin_k = \pkH'$
  for committing with $c_i =
  \commitk{\hat{\theta}_i,\hat{x}_i)}{r_i}{\pkH'}$ on all positions
  $i$. Since an ideal-world adversary $\dhB$ is free to choose any
  key, it can generate $(\pkB', \sk')$, i.e., a regular public key
  together with a secret key according to Regev's crypto-system. For
  the security proof, write $\pkB' = coin_1 \ldots coin_k$. In the
  simulation, $\dhB_{\tt compile}$ (as described in the proof of
  Theorem~\ref{thm:compiler}) first invokes $\dhB_{\tt coin}$
  (Figure~\ref{fig:simulationB}) for each $coin_j$ to simulate one
  coin-flip with $coin_j$ as result. Whenever $\dhB_{\tt coin}$ asks
  $\cF$ to output a bit, it instead receives this $coin_i$. Then
  $\dhB_{\tt compile}$ has the possibility to decrypt dishonest
  $\dB$'s commitments $c_i =
  \commitk{(\hat{\theta}_i,\hat{x}_i)}{r_i}{\pkB'}$ during simulation,
  which binds $\dB$ unconditionally to his committed measurement bases
  and outcomes. Finally, since we proved in the analysis of protocol
  $\COIN$ that $\pkH'$ is a uniformly random string, Regev's proof of
  semantic security applies, namely that a random public key, chosen
  independently from a secret key, is indistinguishable to a regular
  key and that such encodings carry essentially no information about
  the message. Thus, we obtain $\pkH' \approxq \pkB'$ and
  quantum-computational security in real life follows.


\clearemptydoublepage
\addcontentsline{toc}{chapter}{Bibliography}
\bibliographystyle{alpha} 
\bibliography{crypto,qip,procs,personal}

\index{BQSM|see{bounded-quantum-\\storage model}}
\index{CRS-model|see{common-reference-\\string-model}}
\index{OT|see{oblivious transfer}}
\index{ID|see{identification}}
\index{ZK|see{zero-knowledge}}
\index{HVZK|see{honest-verifier}}
\index{commitment!dual-mode|see{mixed}}

\clearemptydoublepage
\addcontentsline{toc}{chapter}{Index}
\printindex

\end{document}